\documentclass{report}
\usepackage{subfigure}
\usepackage{longtable}
\usepackage{amsmath}
\usepackage{amsthm}
\usepackage{amsfonts}
\usepackage{amssymb}
\usepackage{graphics}
\usepackage{graphicx}
\usepackage{appendix}
\usepackage{color}
\usepackage{makeidx}
\usepackage{hyperref,refcount}

\usepackage{mythesis}  

\renewcommand{\baselinestretch}{1} 
\large
\normalsize

\newcommand{\ZZ}{\mathbb{Z}}

\newcommand{\AC}{\mathcal{A}}
\newcommand{\BC}{\mathcal{B}}
\newcommand{\CC}{\mathcal{C}}
\newcommand{\DC}{\mathcal{D}}
\newcommand{\EC}{\mathcal{E}}

\newcommand{\GC}{\mathcal{G}}
\newcommand{\HC}{\mathcal{H}}

\newcommand{\JC}{\mathcal{J}}
\newcommand{\KC}{\mathcal{K}}
\newcommand{\LC}{\mathcal{L}}

\newcommand{\OC}{\mathcal{O}}
\newcommand{\PC}{\mathcal{P}}
\newcommand{\QC}{\mathcal{Q}}
\newcommand{\RC}{\mathcal{R}}
\newcommand{\SC}{\mathcal{S}}
\newcommand{\TC}{\mathcal{T}}
\newcommand{\UC}{\mathcal{U}}

\newcommand{\WC}{\mathcal{W}}
\newcommand{\XC}{\mathcal{X}}

\newcommand{\GCbar}{\overline{\mathcal{G}}}

\newcommand{\ad}{^\dagger }			 			  
\newcommand{\ket}[1]{|#1\rangle}                  
\newcommand{\bra}[1]{\left\langle #1 \right|}     
\newcommand{\dyad}[2]{\ket{#1}\bra{#2}}           
\newcommand{\ip}[2]{\langle #1|#2\rangle}         
\newcommand{\matl}[3]{\langle #1|#2|#3\rangle}    
\newcommand{\Tr}{{\rm Tr}}                        
\newcommand{\vect}[1]{\mbox{\textbf{#1}}}         

\newcommand{\ii}{\mathrm{i}}					  
\newcommand{\ee}{\mathrm{e}}
\newcommand{\expo}[1]{\mathrm{e}^{#1}} 

\def\avg#1{\langle #1\rangle }

\def\dya#1{|#1\rangle \langle#1|}

\def\mat#1{\left(\begin{matrix}#1\end{matrix}\right)}
\def\om{\omega }
\def\ot{\otimes}

\newcommand{\pt}{B}   			
\newcommand{\ptc}{\bar\pt}	 	

\def\QZ{\mbox{$Q$}}
\def\XZ{\mbox{$\XC_0$}}

\long\def\ca#1\cb{} 

\def\CP{C}
\def\dist{\delta}
\def\vc{\vect{c}}
\def\olabel#1{}
\def\hsp{\hphantom{x}}

\def\ftna{\footnotemark[1]}
\def\ftnb{\footnotemark[2]}
\def\ftnc{\footnotemark[3]}
\def\ftnd{\footnotemark[4]}

\newtheorem{theorem}{Theorem}[chapter]
\newtheorem{lemma}[theorem]{Lemma}
\newtheorem{corollary}[theorem]{Corollary}

\newtheorem*{QGS}{Quadratic Gauss Sums}
\newtheorem*{LS}{Landsberg-Schaar Identity}
\newtheorem*{GQGS}{Generalized Quadratic Gauss Sums}
\newtheorem*{RFGQGS}{Reciprocity Formula for Generalized Quadratic Gauss Sums}

\title{Separable Operations, Graph Codes and the Location of Quantum Information}
\author{Vlad Gheorghiu}
\date{June 1, 2010}

\makeindex

\begin{document}

\pagenumbering{roman}
\maketitle
\makecopyright

\begin{acknowledgements}
I would like to thank to my advisor Robert B. Griffiths for his guidance during my PhD, as well as to the members of the quantum information group at Carnegie Mellon University for useful and stimulating discussions. 
\end{acknowledgements}

\begin{abstract}
In the first part of this Dissertation, I study the differences between LOCC (local operations and classical communication) and the more general class of separable operations. I show that the two classes coincide for the case of pure bipartite state input, and derive a set of important consequences. Using similar techniques I also generalize the no-cloning theorem when restricted to separable operations and show that cloning becomes much more restrictive, by providing necessary (and sometimes sufficient) conditions.

In the second part I investigate graph states and graph codes with carrier qudits of arbitrary dimensionality, and extend the notion of stabilizer to any dimension, not necessarily prime. I further study how and where information is located in the various subsets of the qudit carriers of arbitrary additive graph codes, and provide efficient techniques that can be used in deciding what types of information a subset contains.
\end{abstract}

\tableofcontents
\listoftables
\listoffigures

\chapter{Introduction and preliminary concepts\label{chp1}}
\pagenumbering{arabic}

\section{Historical remarks}
Classical computers are indispensable in today's world and appear in almost all imaginable scenarios, ranging from simple MP3 players to sophisticated supercomputers used in weather prediction. Although their architecture is strongly device-dependent, they all have one thing in common: every classical computer is a physical realization of a Turing machine\index{Turing machine}, a mathematical model introduced by Alan Turing in 1937 \cite{A.M.Turing01011937} which formalized the concept of computation. Intuitively, any computer with a certain minimum capability is, in principle, capable of performing the same tasks that any other computer can perform, if sufficient time and memory are provided. 

The fundamental processing unit of a classical computer is the \emph{bit}\index{bit}: a two-state system commonly denoted by 0 and 1. A computation is performed whenever some input bits are processed by the computer, resulting in another sequence of bits that represent the output of the computation.  Any two-state classical physical system can in principle represent a bit, and a classical computer able to operate on these bits and to produce a valid output can (in principle) be built using only classical devices, as billiard balls, pulleys and so on, although in real world applications everything tend to be miniaturized for efficiency purposes. 

A fruitful idea is that \emph{computation is physical}: any physical system performs some kind of computation during its evolution. Consider for example a rock that is falling down from some height. The total falling time is directly proportional to the square root of the height, so in a sense by simply measuring the time one effectively computes the square root of the height. With a bit of imagination more sophisticated examples can be constructed. What if we make use of the intimate structure of quantum mechanics and quantum systems to perform computations? Are there any fundamentally new possibilities relative to the classical case, or, is a \emph{quantum computer}\index{quantum computer} potentially more powerful than a classical one? This idea was first introduced  in 1982 by Richard Feynman \cite{feynman-82}, when the notion of a quantum computer was born. One may argue that a rock is also a (quite large) quantum mechanical system, so why should quantum mechanics be more ``powerful" in performing computations than its classical counterpart? The main and fundamental difference between classical and quantum mechanics is that the latter has a much richer structure that allows for novel effects not present in classical physics due to its use of a Hilbert space. On larger scales the quantum effects tend to be smoothed out and the system becomes classical, or \emph{decoheres}\index{decoherence}, but on smaller scales the effects can be quite significant. 

Although it was widely believed that a quantum computer can indeed be more powerful than a classical one, in cases such as simulation of complex quantum systems, the first quantum algorithm able to outperform any known classical version was discovered only in 1994 \cite{shor:1484}  by the computer scientist Peter Shor. He invented a quantum algorithm able to factor large numbers in polynomial time and that did not (and still does not) have any efficient classical counterpart.

Even though a genuinely good quantum algorithm existed, a great deal of skepticism was displayed with respect to an actual physical realization. It was already known that quantum systems are very sensitive to external noise that induce \emph{errors} in the computation, and the main problem seemed to be the inability to correct these errors. Classical error correction was well understood \cite{MacWilliamsSloane:TheoryECC}, the basic idea behind the whole field being the usage of \emph{redundancy}, or duplication of information for better protection against errors. On the other hand, the ``no-cloning" theorem of Wootters and Zurek \cite{Nature.299.802}, discovered in 1982, forbids the duplication of quantum information, e.g. non-orthogonal states cannot be duplicated. Therefore the perspectives for good \emph{quantum error correction} schemes looked extremely unpromising. The solution was provided by the same Peter Shor in 1995 \cite{PhysRevA.52.R2493}, who showed that good quantum error correction schemes, which are not simply based on duplication of quantum information, exist. Quantum error correcting codes are of extreme importance in quantum information processing since they allow for high-fidelity transmission of quantum information and reduction of decoherence.

A novel field of science, Quantum Computation and Quantum Information, suddenly became a hot topic at the intersection of physics, mathematics and computer science. It has developed  rapidly since 1994 and remarkable theoretical as well as experimental progress has been achieved. Today the subject is still far from being well understood, and I hope this Dissertation contains some nontrivial contributions to it.

\section{Preliminary concepts}
\subsection{Qubits and qudits}

The fundamental processing unit of a quantum computer is the \emph{qubit}\index{qubit}: a quantum system
with two energy levels. The state of the qubit is often denoted by
\begin{equation}\label{intro_eqn1}
\ket{\psi}=\alpha\ket{0}+\beta\ket{1},
\end{equation}
where $\alpha$ and $\beta$ are complex coefficients satisfying $|\alpha|^2+|\beta|^2=1$ and $\ket{0}$ and $\ket{1}$ are orthonormal basis vectors of a 2-dimensional complex Hilbert space $\HC$.
The \emph{qudit}\index{qudit} is the natural generalization of the qubit to $D$ level quantum systems and is represented by a complex Hilbert space of dimension $D$. 

\subsection{Evolution and quantum channels}

The evolution of a quantum system interacting with an external environment is in general non-unitary. A standard way of describing such evolution uses the Kraus formalism: if the initial state of the system is given in terms of a density operator $\rho$, then after a time $t$ the state evolves as 
\begin{equation}\label{intro_eqn2}
\rho\rightarrow\sum_{k}F_k^{}\rho F_k^\dagger,
\end{equation}
where the $F_k$ are called \emph{Kraus operators}\index{Kraus operators} that satisfy the \emph{closure condition}\index{closure condition}
\begin{equation}\label{intro_eqn3}
\sum_kF_k^\dagger F_k^{}=I,
\end{equation}
where $I$ denotes the identity operator. The Kraus representation \eqref{intro_eqn2} can be formally derived by considering a unitary evolution of the combined system-environment, and then tracing away (or measuring) the environment degrees of freedom. A non-trivial result is that any open quantum evolution can be represented in the form \eqref{intro_eqn2}. If the result of the measurement on the environment is known, e.g. equals $k$, then conditioned on this $k$ the initial density operator of the system is transformed to 
\begin{equation}\label{intro_eqn4}
\rho\rightarrow \frac{F_k^{}\rho F_k^\dagger}{\Tr[F_k^{}\rho F_k^\dagger]}.
\end{equation}
However, measurement is not a deterministic operation and the result $k$ is obtained with some probability $p_k=\Tr[F_k\rho F_k^\dagger]$. Therefore whenever the measurement results on the environment are not discarded, the initial density operator of the system is transformed to 
an \emph{ensemble}\index{ensemble}
\begin{equation}\label{intro_eqn5}
\rho\rightarrow\{p_k,\rho_k\},\quad \rho_k=\frac{F_k^{}\rho F_k^\dagger}{\Tr[F_k^{}\rho F_k^\dagger]}.
\end{equation}
If the system starts out in a pure state $\ket{\psi}$, then \eqref{intro_eqn5} reduces to
\begin{equation}\label{intro_eqn6}
\ket{\psi}\rightarrow\{p_k,\ket{\psi_k}\},\text{ with } F_k\ket{\psi}=\sqrt{p_k}\ket{\psi_k}\text{ and } p_k=\matl{\psi}{F_k^\dagger F_k^{}}{\psi}.
\end{equation}

Any evolution of the form \eqref{intro_eqn2} is also called a \emph{quantum channel}\index{quantum channel}, and \eqref{intro_eqn2} represents the Kraus representation of a quantum channel. Technically the map \eqref{intro_eqn2} is a completely positive trace preserving (CPTP) map\index{completely positive trace preserving (CPTP)}, which intuitively means that it maps positive operators to positive operators, remains positive whenever the system is trivially enlarged to a larger one, and preserves the trace of an operator. The latter condition is imposed by the closure condition \eqref{intro_eqn3} and ensures that probability is conserved. The study of quantum channels is the subject of intense theoretical investigations as they are much less understood than their classical counterpart. The interested reader can consult Chapter 8 of \cite{NielsenChuang:QuantumComputation} for a good introduction to the subject.

\subsection{Bipartite (multipartite) quantum systems}

In quantum theory a multipartite quantum system is described by a Hilbert space constructed as a tensor product of the individual Hilbert spaces, and any vector in this tensor product space represents a valid quantum state. This is one instance in which quantum mechanics has a much richer structure than classical mechanics, because the tensor product description allows the existence of \emph{entangled states}\index{entangled states}, i.e. quantum states that cannot be written as a tensor product of individual states. For example, in a bipartite qubit quantum system represented by a Hilbert space $\HC_A\otimes \HC_B$, a state of the form
\begin{equation}\label{intro_eqn7}
\ket{\psi}=\alpha\ket{00}_{AB}+\beta\ket{11}_{AB}
\end{equation}
with $|\alpha|^2+|\beta|^2=1$
is entangled as long as $0<|\alpha|^2<1$. Otherwise it is called a \emph{product state}\index{product states}. 
Whenever $|\alpha|=|\beta|=1/\sqrt{2}$ the state is called \emph{maximally entangled}\index{maximally entangled states}. Throughout this Dissertation I will use the shorthand notation  $\ket{00}_{AB}$ to denote the tensor product $\ket{0}_A\otimes\ket{0}_B$. 

It turns out that any normalized bipartite entangled state $\ket{\psi}\in\HC_A\otimes\HC_B$ can be written in a canonical form known as the \emph{Schmidt form}\index{Schmidt form}: there always exist orthonormal bases $\{\ket{a_j}\}$ and $\{\ket{b_k}\}$ of $\HC_A$ and $\HC_B$, respectively, in which the $\ket{\psi}$ has the form
\begin{equation}\label{intro_eqn8}
\ket{\psi}=\sum_{r=0}^{D-1}\sqrt{\lambda_r}\ket{a_r}_A\ket{b_r}_B.
\end{equation}
Here $D$ is the minimum of the dimensions of $\HC_A$ and $\HC_B$.
The $\lambda_r$'s can always be chosen to be positive real numbers that satisfy $\sum_r\lambda_r=1$ and are called \emph{Schmidt coefficients}\index{Schmidt coefficients}. An alternative definition is that a bipartite pure state is maximally entangled if and only if all Schmidt coefficients are equal. When all Schmidt coefficients except one are zero then the state is a product state, otherwise it is partially entangled. The number of positive Schmidt coefficients is called the \emph{Schmidt rank}\index{Schmidt rank} of $\ket{\psi}$. If all Schmidt coefficients are strictly positive than we say that $\ket{\psi}$ has \emph{full Schmidt rank}\index{full Schmidt rank}. There is no analog of the Schmidt form for multipartite pure states, and this constitutes a major obstacle in understanding them.

Entanglement is often considered a precious resource, since it can be ``consumed" by various non-classical protocols such as quantum teleportation \cite{PhysRevLett.70.1895}, quantum dense coding \cite{PhysRevLett.69.2881} etc. It also constitutes a key ingredient (although by no means the only one) in the construction of good quantum error codes and exponentially faster quantum algorithms, and therefore its study constitutes an important part of quantum information theory. 

\subsection{LOCC and Separable Operations}

An important paradigm in quantum information is that of local operations and classical communication (LOCC)\index{LOCC}. Consider for example two spatially separated parties, traditionally named Alice and Bob, each having access to local quantum systems described by Hilbert spaces $\HC_A$ and $\HC_B$, respectively. Both Alice and Bob are allowed to perform arbitrary quantum operations on their individual quantum systems, and can also make use of a classical channel to communicate. They are not allowed, however, to exchange quantum systems between themselves nor to use a quantum channel. What kind of tasks can be performed in this paradigm? What are the restrictions compared to the global case (i.e. when global quantum operations are allowed on the combined system $\HC_A\otimes\HC_B$)? Understanding LOCC constitutes an extremely important research program, since in an actual realization of a quantum computer many qubits may be well separated in space, and performing a global operation on them, in contrast to LOCC, will often be challenging, at least from an experimental point of view.

Every LOCC operation can be regarded as a composition of local operations conditioned on particular measurement results that may be communicated through the classical channel. It is not hard to see that the initial state $\rho$ of a quantum system transforms under LOCC as
\begin{equation}\label{intro_eqn9}
\rho\rightarrow\sum_{k}(A_k\otimes B_k)\rho(A_k\otimes B_k)^{\dagger},
\end{equation}
with 
\begin{equation}\label{intro_eqn10}
\sum_{k}(A_k^\dagger A_k^{}\otimes B_k^\dagger B_k^{})=I_A\otimes I_B.
\end{equation}
The concepts above generalize to more than two parties in a straightforward manner.

One may be tempted to say that any quantum operation of the form \eqref{intro_eqn9} represents a valid LOCC, but this is not true! There are operations of this form that are not LOCC \cite{PhysRevA.59.1070}, that is, cannot be implemented by an LOCC paradigm, and which are called \emph{separable operations}\index{separable operations}. There exists a simple example \cite{PhysRevA.59.1070} of a set of basis states in a bipartite qutrit ($D=3$) system that cannot be distinguished by LOCC but can be distinguished by separable operations. Hence LOCC is a proper subset of this more general class of separable operations. Although separable operations are not always implementable by LOCC, studying them is worthwhile since in general they have a cleaner mathematical formulation than LOCC and any result proven to be true for all separable operations will automatically be valid for LOCC, since the latter is a subset of the former. 

In the bipartite setting, a maximally entangled state plays the role of a universal resource,  i.e. from an operational point of view can be transformed deterministically by LOCC to any bipartite partially entangled state \cite{PhysRevLett.83.436}. This is not true anymore in the multipartite regime; there is no universal multipartite quantum state that can be transformed to any arbitrary multipartite state by LOCC \cite{RevModPhys.81.865}.

Entanglement can be quantified by various measures, and the measure is called an \emph{entanglement monotone}\index{entanglement monotone} if it is non-increasing under LOCC. In general entanglement measures have a simple form only for pure bipartite states, and in this case depend only on the Schmidt coefficients of the state. For multipartite pure states or even for bipartite mixed states such a Schmidt decomposition does not exist, and entanglement in these cases is far from being understood.
For a good introduction to the theory of entanglement see \cite{RevModPhys.81.865}.

\subsection{Graph states and graph codes}

Quantum states and quantum entanglement are much better understood in the bipartite setting than in the multipartite setting. Two main difficulties that arise in the latter case are: i) the exponential growth, with the number of constituent parts, of the number of complex amplitudes used to describe a multipartite quantum state, and ii) the non-existence of a Schmidt representation.

However, \emph{graph states}\index{graph state} form a class of multipartite states with a fairly simple structure. 
Given a graph $G=(V,E)$ with $n$ vertices $V$, each corresponding to a qubit, and a collection $E$ of undirected edges connecting pairs of distinct vertices (no self loops are allowed),  a graph state is obtained by preparing a set of initial qubits in the
 $\ket{+}=(\ket{0}+\ket{1})/\sqrt{2}$ state, then applying controlled-phase gates between any two neighbors that are connected by an edge in the corresponding graph. Graph states can be generalized to higher dimensional qudits in a direct manner, using graphs with multiple edges, as described in detail in Chapters~\ref{chp5} and~\ref{chp6}.

Graph states were first introduced by Raussendorf and Briegel \cite{PhysRevLett.86.5188} in their \emph{measurement-based computation model}\index{measurement-based computation}(often called the \emph{one-way model}\index{one-way model}). This new model of quantum computation is fundamentally different from the well-known circuit model, and is also universal. Any desired quantum computation can be achieved by only performing local measurements (in a basis that is conditioned on previous measurement results) on the qubits of a sufficiently large cluster state, a graph state in which the graph is a finite part of a lattice such as the square lattice. Graph states are therefore universal resources for quantum computation and intense theoretical as well as experimental work has been dedicated to their study. For a comprehensive introduction to the subject see \cite{quantph.0602096}.

Graph states are also an instance of the so-called \emph{stabilizer states}\index{stabilizer states} \cite{quantph.9705052}. A stabilizer state is a multipartite quantum state that can be described by an Abelian group (the stabilizer group) of Pauli-like operators on the Hilbert space of the carriers, each of which leaves the state invariant. Instead of describing a stabilizer state of $n$ qubits by $2^n$ complex amplitudes, it is enough to specify the generators of its stabilizer group of which there are no more than $n$. Hence stabilizer states allow for a very compact description. Stabilizer states play an extremely important role in the theory of quantum error-correction codes and they extend the notion of linear classical error correcting codes \cite{MacWilliamsSloane:TheoryECC} to the quantum domain.

\subsection{Location of information}
First let us define what we mean by (classical) information. Information is embedded in correlations between two systems, e.g. we say that information about a system $A$ is ``located" in another system $B$ if the statistical correlations between $A$ and $B$ are such that information about some properties of $A$ can be recovered from $B$.
 For example, the photons that bounce off the Sunday newspaper hit the reader's retina and correlations between the letters on the newspaper page and reader's brain are established: the information about the latest news is now located in the reader's brain.
If the systems are perfectly correlated then we say that all information about $A$ is perfectly present in $B$, and if they are totally uncorrelated than we say that no information about $A$ is present in $B$ (or, equivalently, that all information about $A$ is absent from $B$). Traditionally $A$ is considered to be the input of a communication channel and $B$ its output at a later time, but information theory is not restricted to channels.

 Whereas the concept seems to be quite natural and simple, the rigorous mathematical theory of information was founded by Claude Shannon only in 1948 \cite{Shannon48} and  has been continuously developed since then (see \cite{CoverThomas:ElementsOfInformationTheory} for a comprehensive introduction). Information theory is a very important subject in modern communication, cryptography, classical error-correcting codes etc, with applications ranging from audio CD error-correction to military satellite communication.

Unlike classical information, quantum information can be present in more than one ``type", a terminology introduced by Griffiths \cite{PhysRevA.76.062320}, and various types can be incompatible, e.g. associated to operators that do not commute. Formally, a type of information is associated with a projective decomposition of the identity 
\begin{equation}\label{intro_eqn11}
I=\sum_j P_j,\quad P_j=P_j^\dagger=P_j^2.
\end{equation}
We also associate a type of information with a normal operator through its spectral decomposition, and, for example, call the information corresponding to the $x$ component of an angular momentum of a spin one-half particle, represented by the Pauli operator $\sigma_x$, usually denoted by $X$, the $X$-type of information, and the information corresponding to the $z$ component as the $Z$-type of information. Of course the $X$ and $Z$ types are incompatible, since their corresponding operators do not commute. 

\section{Overview of the Dissertation}

\subsection{Separable operations}

An LOCC acting on a pure bipartite state $\ket{\psi}_{AB}$ will in general produce an ensemble of states $\{p_k,\ket{\phi_k}_{AB}\}$, where $p_k$ is the probability of obtaining the result $k$ through a measurement of the environment. Finding necessary and sufficient conditions for when such a transformation is possible represents an important problem, and a complete solution for bipartite systems was first provided by Nielsen for an ensemble with only one output state \cite{PhysRevLett.83.436} and generalized by Jonathan and Plenio to ensembles with a finite number of states \cite{PhysRevLett.83.1455}. Both of these necessary and sufficient conditions are given in terms of \emph{majorization relations}\index{majorization relations} \cite{Bhatia:MatrixAnalysis} between Schmidt coefficients, and they completely characterize LOCC operations acting on pure bipartite states.

In Chapter~\ref{chp2} and Chapter~\ref{chp3} we study separable operations acting on a pure bipartite state, trying to generalize previously known results. 
The most important result is that any output ensemble produced by a separable operation acting on a pure bipartite state can in fact be produced by some LOCC acting on the same state. Our result effectively says that LOCC and separable operations are the same class of quantum operations when acting on pure bipartite states. In particular, we prove that the majorization conditions of Jonathan and Plenio \cite{PhysRevLett.83.1455} are necessary and sufficient in the more general case of separable operations.

Our result also has important consequences in the theory of entanglement, implying that a large number of mixed-state entanglement monotones remain monotone under separable operations.   Since such monotonicity under LOCC has long been considered a
necessary, or at least a very desirable condition for any ``reasonable''
entanglement monotone, one wonders whether monotonicity under separable operations, in principle a
stronger condition, might be an equally good or even superior desideratum.

An interesting question that follows from our result is: are separable operations and LOCC the same class of quantum operations when acting on multipartite pure states? It might be, but proving it would require very different methods than the ones we used, since there are no simple analogs of the Schmidt decomposition and majorization conditions. Necessary and sufficient conditions are not known even for LOCC.

\subsection{Local cloning of bipartite entangled states}
In Chapter~\ref{chp4} we consider the slightly different problem of cloning orthogonal entangled states by separable operations, a problem that belongs to the more general framework of deterministic mapping of an ensemble of pure states (and not just one state, as before) into another ensemble of pure states by a separable operation. As summarized by the ``no-cloning" theorem\index{no-cloning
 theorem} of \cite{Nature.299.802}, any set of quantum states can be deterministically cloned if and only if the states in the set are mutually orthogonal. When the states are not orthogonal, there is no deterministic apparatus capable of performing such a cloning. However, probabilistic cloning may still be possible and a significant amount of work has been dedicated to studying this case \cite{RevModPhys.77.1225}.

Formally, a set of quantum states $\{\ket{\psi_j}\}$ is cloned whenever there exist a quantum operation that performs 
\begin{equation}\label{intro_eqn12}
\ket{\psi_j}\otimes\ket{\phi}\rightarrow\ket{\psi_j}\otimes\ket{\psi_j},\quad\forall j.
\end{equation}
If the transformation is deterministic for all $j$, then the cloning is deterministic. Otherwise is probabilistic. The state $\ket{\phi}$ plays the role of a resource or a ``blank state" in which the copy of $\ket{\psi_j}$ is to be imprinted.

When the set consists of bipartite entangled states, and the cloning is restricted to LOCC (or to separable operations), the problem becomes much more difficult, and further restrictions have to be imposed. The mere orthogonality of the states no longer implies that they can be (locally) cloned. An LOCC analog of the ``no-cloning" theorem was not yet found, and finding it may prove useful.

It turns out that any two (and no more than two) orthogonal maximally entangled two-qubit states can be locally cloned by LOCC, using a maximally entangled ``blank state" \cite{NewJPhys.6.164} (on which the copy is to be imprinted). 
A generalization to $D$ maximally entangled states of two qudits of prime dimension $D$ was given in \cite{PhysRevA.74.032108},  which showed  that a set of $D$ such states can be locally cloned using a maximally entangled resource if and only if the states in the set are locally (cyclically) shifted
\begin{equation}\label{intro_eqn13}
\ket{\psi_i}=\frac{1}{\sqrt{D}}\sum_{r=0}^{D-1}\ket{r}^A\ket{r\oplus i}^B,
\end{equation}
where the $\oplus$ symbol denotes addition modulo $D$. 
Kay and Ericsson \cite{PhysRevA.73.012343} extended the above results to the LOCC cloning of full Schmidt rank partially entangled states using a maximally entangled blank state. They presented an explicit protocol for the local cloning of a set of $D\times D$ cyclically shifted partially entangled states
of the form
\begin{equation}\label{intro_eqn14}
\ket{\psi_i}=\sum_{r=0}^{D-1}\sqrt{\lambda_r}\ket{r}^A\ket{r\oplus i}^B
\end{equation}
using a maximally entangled blank state, but failed to prove that any clonable set of states must be of this form.

We investigate the conditions under which a set of pure bipartite quantum states on a $D\times D$ system can be locally cloned deterministically by separable operations when at least one of the states is full Schmidt rank. We do not assume that $D$ is necessarily a prime number and we also allow for the possibility of cloning using a resource state that is less than maximally entangled. We derive a set of necessary conditions, that are also sufficient in the case of qubits. In this latter case we proved a long-standing conjecture that a maximally entangled state is a necessary resource for such local cloning, even if the states to be cloned are partially entangled. We also generalize the protocol of Kay and Ericsson and show that any set of partially entangled ``group-shifted" states 
\begin{equation}\label{intro_eqn15}
\ket{\psi_f}=\sum_{g\in G}\sqrt{\lambda_g}\ket{g}^A\ket{fg}^B,
\end{equation}
where $G$ is a group of order $D$ and the elements of the group label an orthonormal basis of the (local) Hilbert space,
can be locally cloned using a maximally entangled blank state, by providing an explicit LOCC circuit. Our protocol reduces to the one of Kay and Ericsson in the case of cyclic groups, since the latter is isomorphic to the additive group of integers mod $D$ and the set of states defined by \eqref{intro_eqn14} and \eqref{intro_eqn15} are the same.

Our results significantly extend previous work in the literature (limited only to LOCC) \cite{PhysRevA.69.052312,NewJPhys.6.164,PhysRevA.74.032108,PhysRevA.73.012343,PhysRevA.76.052305}.

\subsection{Graph codes}
Using graph states provides a fruitful approach for constructing good quantum error correcting codes, called \emph{graph codes}\index{graph code}. A graph code is a subspace of the Hilbert space of $n$ carrier qubits spanned by a collection of \emph{graph basis states}\index{graph basis}: a quantum state obtained from a fixed graph state by applying local $Z$ operators on some of the qubits. See Chapter~\ref{chp5} for a detailed introduction to graph codes.

There have been extensions to higher dimensional qudits, but all of them considered only qudits of prime dimensionality $D$. The main difficulty in the non-prime case is that $\ZZ_D$ is a ring, not a finite field, and the lack of the multiplicative inverse operation poses some technical problems, which we have successfully solved. 

Most known quantum error correcting codes are graph codes or are equivalent to graph codes under local unitary transformations, hence understanding them is extremely important.
In Chapter~\ref{chp5} we present an elegant method for constructing such graph codes, allowing for carrier qudits of arbitrary dimensionality, not necessarily prime. Our method allows for \emph{additive}\index{additive} (or stabilizer \cite{quantph.9705052}) as well as non-additive graph codes and was simultaneously developed by Cross et al \cite{IEEE.55.433} for qubits and Chen et al \cite{PhysRevA.78.062315} for higher dimensional qudits. We use
simple graphical methods and computer searches to construct both additive and non-additive quantum error correcting codes, but computer searches are much faster for additive codes. In a number of cases we have been able to construct
what we call quantum Singleton codes that saturate
the quantum Singleton bound \cite{PhysRevA.55.900}. Our numerical techniques are based on finding a \emph{maximum clique}\index{maximum clique} \cite{GareyJohnson:CompIntract} in some related graph. The maximum clique problem on a general graph is
known to be NP-complete, but our method may still be useful for constructing quantum codes with relatively small number of carriers.
 We also generalize the concept of a stabilizer group to the non-prime case and derive an elegant duality (known before only in the prime case)  between the coding space and its corresponding stabilizer group.

\subsection{Location of quantum information}
In Chapter~\ref{chp6} we develop a mathematical formalism that can be successfully applied in studying how quantum information is encoded in additive graph codes and where is it located, a subject closely related to the one of Chapter~\ref{chp5}. Studying additive graph codes is worthwhile since the vast majority of quantum error correcting codes are stabilizer codes, and stabilizer codes are locally equivalent to additive graph codes \cite{quantph.0111080}.

We show how to encode some input quantum information in the \emph{carrier}\index{carrier qudits} qudits of an additive graph code and demonstrate how to use the concept of types of information to study the location of quantum information in arbitrary subsets of the carrier qudits. What types and how much information about the input can be then recovered? To various types of information we associate a collection of operators on the coding space which form what we call the information group. It represents the input information through an encoding operation constructed as an explicit quantum circuit (hence generalizing the encoding methods developed before only for prime dimensional qudits). Our formalism is very general and works for arbitrary additive graph codes of arbitrary dimension (not necessarily prime). We also present an efficient numerical algorithm that can be successfully used in deciding where information is located and which types of information are present.  As a side remark, note that we have not studied the ``recovery problem", i.e. finding the decoding operation that effectively ``extracts" the quantum information from some carrier qudits, but in principle such a decoding always exist, provided all quantum information is located in these qudits. This recovery operation is interesting, but is not included in this Dissertation.

The methods presented here allow for a better understanding of the intimate nature of quantum codes and may be of use in constructing better quantum error-correcting codes or quantum secret sharing schemes \cite{PhysRevA.78.042309}.

\subsection{Bipartite equientangled bases}
Chapter~\ref{chp7} is not closely related to the other chapters but presents a solution to a problem posed in \cite{PhysRevA.73.012329} of constructing a family of ``equientangled bases" for a bipartite system of two qudits of arbitrary (but equal) dimension: (i) The basis continuously changes from a product basis to a maximally entangled basis, by varying a parameter $t$, and (ii) for a fixed $t$, all basis states are equally entangled. 

We actually construct two solutions to the problem, one based on quadratic Gauss sums and the other using qudit graph states. These bases may find applications in various quantum information protocols including quantum cryptography, optimal Bell tests, investigation of the enhancement of channel capacity due to entanglement and the study of multipartite entanglement.

\section{The structure of the Dissertation}

All chapters of this Dissertation are self contained and consist of published (or accepted for publication) articles in refereed journals. The contents of each chapter is almost the same as that of the published paper, with minor modifications made for the sake of consistency of notation throughout the Dissertation. Most of the chapters represent collaborative work with different persons in our research group, as described below. 

\begin{itemize}
\item Chapter~\ref{chp2}, \textbf{Entanglement transformations using separable operations}: published in Physical Review A \cite{PhysRevA.76.032310}. Collaboration with Robert B. Griffiths. Both authors made major contributions. My most important contributions were: i) the application of map-state duality formalism ii) the idea of using an inequality by Minkowski in deriving necessary conditions for pure state transformations under separable operations, and iii) the investigation of random separable unitary channels.

\item Chapter~\ref{chp3}, \textbf{Separable operations on pure states}: published in Physical Review A as a Rapid Communication \cite{PhysRevA.78.020304}. Collaboration with Robert B. Griffiths. Both authors made major contributions. My most important contributions consisted in: i) extensive numerical studies that led us to the conjecture that separable operations are implementable by LOCC in the case of pure bipartite states; ii) the application of map-state duality formalism; iii) parts of the proof of the main majorization theorem that was conjectured by Griffiths; iv) derivation of the consequences of our result, the most important being that all convex-roof mixed state entanglement measures remain monotone under the more general class of separable operations.

\item Chapter~\ref{chp4}, \textbf{Local cloning of entangled states by separable operations}: accepted for publication in Physical Review A and available on arXiv \cite{quantph.1004.5126}. Collaboration with Scott M. Cohen and Li Yu. I made major contributions to this work, including: i) the map-state duality formalism used in our investigation of the problem; ii) various necessary conditions, such as the necessary form of qubit entangled states, equality of $G$-concurrence, information-theoretical observation etc.; and iii) various proofs of theorems. Scott Cohen introduced the idea of using finite groups to study sets of clonable states, and Li Yu proved that a maximally entangled state is necessary for the local cloning of group shifted states in $D=2$ and $D=3$.

\item Chapter~\ref{chp5}, \textbf{Quantum error correcting codes using qudit graph states}: published in Physical Review A \cite{PhysRevA.78.042303}. Collaboration with Shiang Yong Looi, Li Yu and Robert B. Griffiths. This work was not one of my main projects. My main contributions consisted in developing the stabilizer formalism for non-prime qudits in Sec.~\ref{quditgrcodes_sct15}, and proof of the $X-Z$ rule for qudit graph states in Sec.~\ref{quditgrcodes_sct17}. I was not involved in the numerical work for searching good quantum error-correcting codes.

\item Chapter~\ref{chp6}, \textbf{Location of quantum information in additive graph codes}: published in Physical Review A \cite{PhysRevA.81.032326}. Collaboration with Shiang Yong Looi and Robert B. Griffiths. I had a major part in this paper, together with my co-authors. My most important contributions were: i) introducing a set of useful Clifford gates for arbitrary dimensions; ii) the use of Smith diagonal forms over rings of integers, which allows one to deal with problems that appear in the non-prime case; iii) the discovery of a general encoding operation in terms of an explicit quantum circuit that extended previous work restricted to qudits with prime $D$; iv) implementation of an efficient linear-algebra algorithm used to decide where and which types of informations are present in a given subset of the carriers; v) proofs of various theorems.

\item Chapter~\ref{chp7}, \textbf{Bipartite equientangled bases}: accepted for publication in Physical Review A and available on arXiv \cite{quantph.1004.1633}. Collaboration with Shiang Yong Looi. This work was split into two parts, between myself and my co-author. I found the solution based on Gauss sums (first part), whereas my co-author found the one based on graph states (second part). I also studied the entanglement properties of the second solution in terms of $G$-concurrence, and proved its monotonicity as a function of $t$.
\end{itemize}

\chapter{Entanglement transformations using separable operations\label{chp2}}

\section{Introduction}
\label{sep1_sct1}
A separable operation $\Lambda$ on a bipartite quantum system is a
transformation of the form
\begin{equation}
\label{sep1_eqn1} \rho' = \Lambda(\rho)=\sum_m
\bigl(A_m^{}\otimes B_m^{}\bigr)\rho (A_m^\dagger \otimes B_m^\dagger\bigr),
\end{equation}
where $\rho$ is an initial density operator on the Hilbert space
$\HC_A\otimes\HC_B$.  The Kraus operators $A_m \otimes B_m$ are
arbitrary product operators satisfying the closure condition
\begin{equation}
\label{sep1_eqn2}
\sum_m A_m^\dagger A_m^{}\otimes B_m^\dagger B_m^{}=I\otimes I.
\end{equation}
The extension of \eqref{sep1_eqn1} and \eqref{sep1_eqn2} to multipartite systems is
obvious, but here we will only consider the bipartite case.  To avoid
technical issues the sums in \eqref{sep1_eqn1} and \eqref{sep1_eqn2} and the dimensions
of $\HC_A$ and $\HC_B$ are assumed to be finite.

Various kinds of separable operations play important roles in quantum
information theory. When $m$ takes on only one value the operators $A_1$ and
$B_1$ are (or can be chosen to be) unitary operators, and the operation is a
\emph{local unitary} transformation. When every $A_m$ and every $B_m$ is
proportional to a unitary operator, we call the operation a \emph{separable
  random unitary channel}\index{separable random unitary channel}. Both of these are members of the well-studied class
of \emph{local operations with classical communication} (LOCC)\index{LOCC}, which can be
thought of as an operation carried out by Alice on $\HC_A$ with the outcome
communicated to Bob.  He then uses this information to choose an operation
that is carried out on $\HC_B$, with outcome communicated to Alice, who uses it
to determine the next operation on $\HC_A$, and so forth.  For a precise
definition and a discussion, see [\cite{RevModPhys.81.865}, Sec.~XI]. 
While any LOCC is a separable operation,
i.e., can be written in the form \eqref{sep1_eqn1}, the reverse is not true: there
are separable operations which fall outside the LOCC class
\cite{PhysRevA.59.1070}.

Studying properties of general separable operations seems worthwhile because
any results obtained this way then apply to the LOCC subcategory, which is
harder to characterize from a mathematical point of view.  However, relatively
little is known about separable operations, whereas LOCC has been the subject
of intensive studies, with many important results.  For example, an LOCC
applied to a pure entangled state $\ket{\psi}$ (i.e., $\rho=\dyad{\psi}{\psi}$
in \eqref{sep1_eqn1}) results in an ensemble of pure states (labeled by $m$) whose
average entanglement cannot exceed that of $\ket{\psi}$, [\cite{RevModPhys.81.865}, Sec.~XV D].  One suspects that the same is true of a
general separable operation $\Lambda$, but this has not been proved.  All that
seems to be known is that $\Lambda$ cannot ``generate'' entanglement when
applied to a product pure state or a separable mixed state: the outcome (as is
easily checked) will be a separable state.

If an LOCC is applied to a pure (entangled) state $\ket{\psi}$, Lo and Popescu
\cite{PhysRevA.63.022301} have shown that the same result, typically an
ensemble, can be achieved using a different LOCC (depending both on the
original operation and on $\ket{\psi}$) in which Alice carries out an
appropriate operation on $\HC_A$ and Bob a unitary, depending on that outcome,
on $\HC_B$.  This in turn is the basis of a condition due to Nielsen
\cite{PhysRevLett.83.436} which states that there is an LOCC operation
deterministically (probability 1) mapping a given bipartite state
$\ket{\psi}$ to another pure state $\ket{\phi}$ if and only if $\ket{\phi}$
majorizes $\ket{\psi}$ 
\footnote{
By ``${|\phi\rangle}$ majorizes ${|\psi\rangle}$'' we mean that the
  vector of eigenvalues of the reduced density operator $\rho(\phi)$ of
  ${|\phi\rangle}$ on $\mathcal{H}_A$ majorizes that of the reduced density operator
  $\rho(\psi)$ of ${|\psi\rangle}$ in the sense discussed in
  \cite{PhysRevLett.83.436}, or in [\cite{NielsenChuang:QuantumComputation}, Sec.~12.5.1]: the sum of the $k$ largest
  eigenvalues of $\rho(\phi)$ is never smaller than the corresponding sum for
  $\rho(\psi)$. A helpful discussion of majorization is also found in
  \cite{HornJohnson:MatrixAnalysis} (see the index), with, however, the
  opposite convention from Nielsen for ``$A$ majorizes $B$''
\label{sep1_ftn1}
}.

In this chapter we derive a necessary condition for a separable operation to
deterministically map $\ket{\psi}$ to $\ket{\phi}$ in terms of their Schmidt
coefficients, the inequality \eqref{sep1_eqn5}.  While it is weaker than Nielsen's
condition (unless either $\HC_A$ or $\HC_B$ is two dimensional, in which case
it is equivalent), it is not trivial. In the particular case that the Schmidt
coefficients are the same, i.e., $\ket{\psi}$ and $\ket{\phi}$ are equivalent
under local unitaries, we show that all the $A_m$ and $B_m$ operators in
\eqref{sep1_eqn1} are proportional to unitaries, so that in this case the separable
operation is also a random unitary channel.  For this situation we
also study the conditions under which a whole \emph{collection}
$\{\ket{\psi_j}\}$ of pure states are deterministically mapped to pure states,
a problem which seems not to have been previously studied either for LOCC or
for more general separable operations.

The remainder of this chapter is organized as follows.  Section~\ref{sep1_sct2} has
the proof, based on a inequality by Minkowski, p.~482 of
\cite{HornJohnson:MatrixAnalysis}, of the relationship between the Schmidt
coefficients of $\ket{\psi}$ and $\ket{\phi}$ when a separable operation
deterministically maps $\ket{\psi}$ to $\ket{\phi}$, and some consequences of
this result.  In Section~\ref{sep1_sct3} we derive and discuss the conditions under
which a separable random unitary channel will map a collection of pure
states to pure states.  A summary and some discussion of open questions will
be found in Section~\ref{sep1_sct6}.

\section{Local transformations of bipartite entangled states}
\label{sep1_sct2}

We use the term \emph{Schmidt coefficients}\index{Schmidt coefficients} for the \emph{nonnegative}
coefficients $\{\lambda_j\}$ in the Schmidt expansion
\begin{equation}
\label{sep1_eqn3}
\ket{\psi}=\sum_{j=0}^{D-1}\sqrt{\lambda_j}\ket{a_j}\otimes\ket{b_j},
\end{equation}
of a state $\ket{\psi}\in \HC_A\otimes\HC_B$, using appropriately chosen
orthonormal bases $\{\ket{a_j}\}$ and $\{\ket{b_j}\}$, with the
order chosen so that
\begin{equation}
\label{sep1_eqn4}
 \lambda_0 \geq \lambda_1 \geq \cdots \geq \lambda_{D-1} \geq 0.
\end{equation}
The number $r$ of positive (nonzero) Schmidt coefficients is called the
\emph{Schmidt rank}\index{Schmidt rank}.  We call the subspace of $\HC_A$ spanned by $\ket{a_0},\ket{a_1}\ldots \ket{a_{r-1}}$, i.e., the basis kets for which
the Schmidt coefficients are positive, the $\HC_A$ \emph{support}\index{support} of
$\ket{\psi}$, and that spanned by $\ket{b_0},\ket{b_1}\ldots \ket{b_{r-1}}$ its
$\HC_B$ \emph{support}.

 Our main result is the following:
\begin{theorem}
\label{sep1_thm1}
Let $\ket{\psi}$ and $\ket{\phi}$ be two bipartite entangled states on
$\mathcal{H}_A\otimes\mathcal{H}_B$ with positive Schmidt coefficients
$\{\lambda_j\}$ and $\{\mu_j\}$, respectively, in decreasing order, and let
$r$ be the Schmidt rank of $\ket{\psi}$.  If
$\ket{\psi}$ can be transformed to $\ket{\phi}$ by a deterministic separable
operation, then

i) The Schmidt rank of $\ket{\phi}$ is less than or equal to $r$.

ii)
\begin{equation}
\label{sep1_eqn5}
\prod_{j=0}^{r-1}\lambda_j\geq\prod_{j=0}^{r-1}\mu_j.
\end{equation}

iii) If \eqref{sep1_eqn5} is an equality with both sides positive,
the Schmidt coefficients of $\ket{\psi}$ and $\ket{\phi}$ are identical,
$\lambda_j = \mu_j$, and the operators $A_m$ and $B_m$ restricted to the
$\HC_A$ and $\HC_B$ supports of $\ket{\psi}$, respectively, are proportional
to unitary operators.

iv) The reverse deterministic transformation of $\ket{\phi}$ to
$\ket{\psi}$ by a separable operation is only possible when the Schmidt
coefficients are identical, $\lambda_j = \mu_j$.
\end{theorem}

\begin{proof}

For the proof it is convenient to use map-state duality (see
\cite{OSID.11.3,PhysRevA.73.052309} and [\cite{BengtssonZyczkowski:GeometryQuantumStates}, Chap.~11])
defined in the following way.  Let $\{\ket{b_j}\}$ b	e an orthonormal basis
of $\HC_B$ that will remain fixed throughout the following discussion.
Any ket $\ket{\chi}\in\HC_A\otimes\HC_B$ can be expanded in this basis in the
form
\begin{equation}
\label{sep1_eqn6}
  \ket{\chi} = \sum_j \ket{\alpha_j}\otimes\ket{b_j},
\end{equation}
where the $\{\ket{\alpha_j}\}$ are the (unnormalized) expansion coefficients.
We define the corresponding dual map $\chi:\HC_B\rightarrow \HC_A$ to be
\begin{equation}
\label{sep1_eqn7}
 \chi = \sum_j\ket{\alpha_j}\bra{b_j}.
\end{equation}
Obviously any map from $\HC_B$ to $\HC_A$ can be written in the form
\eqref{sep1_eqn7}, and can thus be transformed into a ket on $\HC_A\otimes\HC_B$ by
the inverse process: replacing $\bra{b_j}$ with $\ket{b_j}$.  The
transformation depends on the choice of basis $\{\ket{b_j}\}$, but this will
not matter, because our results will in the end be independent of this choice.
Note in particular that the \emph{rank}\index{rank of map-state duality operator} of the operator $\chi$ is exactly the
same as the \emph{Schmidt rank} of $\ket{\chi}$.

For a separable operation that deterministically maps $\ket{\psi}$ to
$\ket{\phi}$ (or, to be more specific, $\dyad{\psi}{\psi}$ to
$\dyad{\phi}{\phi}$) it must be the case that
\begin{equation}
\label{sep1_eqn8}
\bigl(A_m\otimes B_m\bigr)\ket{\psi}=\sqrt{p_m}\ket{\phi},
\end{equation}
for every $m$, as otherwise the result of the separable operation acting on
$\ket{\psi}$ would be a mixed state.  (One could also include a complex phase
factor depending on $m$, but this can be removed by incorporating it in
$A_m$---an operation is not changed if the Kraus operators are multiplied by
phases.)  By using map-state duality we may rewrite \eqref{sep1_eqn8} in the form
\begin{equation}
\label{sep1_eqn9} A_m\psi \bar B_m = \sqrt{p_m}\phi,
\end{equation}
where by $\bar B_m$ we mean the \emph{transpose} of this operator in the basis
$\{\ket{b_j}\}$---or, to be more precise, the operator whose matrix in this
basis is the transpose of the matrix of $B_m$. From \eqref{sep1_eqn9} one sees at
once that since the rank of a product of operators cannot be larger than the
rank of any of the factors, the rank of $\phi$ cannot be greater than that of
$\psi$.  When translated back into Schmidt ranks this proves (i).

For the next part of the proof let us first assume that $\HC_A$ and
$\HC_B$ have the same dimension $D$, and that the Schmidt ranks of both
$\ket{\psi}$ and $\ket{\phi}$ are equal to $D$; we leave until later the
modifications necessary when these conditions are not satisfied.  In light of
the previous discussion of \eqref{sep1_eqn9}, we see that $\bar B_m$ has rank
$D$, so is invertible. Therefore one can solve \eqref{sep1_eqn9} for $A_m$, and if
the solution is inserted in \eqref{sep1_eqn2} the result is

\begin{equation}
\label{sep1_eqn10}
I\otimes I= \sum_m p_m\big[\psi^{-1\dagger}
\bar B_m^{-1\dagger}(\phi^\dagger\phi)\bar B_m^{-1}\psi^{-1}\big]
\otimes \big[B_m^\dagger B_m\big]
\end{equation}

The Minkowski inequality\index{Minkowski inequality} (\cite{HornJohnson:MatrixAnalysis}, p. 482) for a sum of positive semidefinite operators on a $S$-dimensional space
is
\begin{equation}
\label{sep1_eqn11}
{\Bigg[\det\Big(\sum_m
Q_m\Big)\Bigg]}^{1/S}\geq\sum_m{\Big(\det{Q_m}\Big)}^{1/S},
\end{equation}
with equality if and only if all $Q_m$'s are proportional, i.e.
$Q_i=f_{ij}Q_j$, where the $f_{ij}$ are positive constants.  Since
$A_m^\dagger A_m\otimes B_m^\dagger B_m$ is a positive operator on a
$S=D^2$ dimensional space, \eqref{sep1_eqn10} and \eqref{sep1_eqn11} yield
\begin{eqnarray}
\label{sep1_eqn12}
  1&\geq&{\Bigg[\det\Big(\sum_m p_m\big[\psi^{-1\dagger}
\bar B_m^{-1\dagger}(\phi^\dagger\phi){\bar B_m}^{-1}\psi^{-1}\big]
\otimes \big[B_m^\dagger B_m\big]\Big)\Bigg]}^{1/D^2}\nonumber \\
  &\geq&\sum_m {\Bigg[\det\Big(p_m\big[{\psi^{-1}}^\dagger
\bar B_m^{-1\dagger}(\phi^\dagger\phi){\bar B_m}^{-1}\psi^{-1}\big]\otimes
\big[B_m^\dagger B_m\big]\Big)\Bigg]}^{1/D^2}\nonumber\\
 &=& \sum_m p_m
\frac{\det(\phi^\dagger\phi)^{1/D}}{\det(\psi^\dagger\psi)^{1/D}} =
\frac{\det(\phi^\dagger\phi)^{1/D}}{\det(\psi^\dagger\psi)^{1/D}},
\end{eqnarray}
which is equivalent to
\begin{equation}
\label{sep1_eqn13}
\det(\psi^\dagger\psi)\geq\det(\phi^\dagger\phi).
\end{equation}
The relation $\det(A\otimes B)$=$(\det A)^b (\det B)^a$ , where
$a,b$ are the dimensions of $A$ and $B$, was used in deriving
\eqref{sep1_eqn12}.
Since \eqref{sep1_eqn13} is the square of \eqref{sep1_eqn5}, this proves part (ii).

If \eqref{sep1_eqn5} is an equality with both sides positive,
$\det(\phi^\dagger\phi)/\det(\psi^\dagger\psi)=1$ and the inequality
\eqref{sep1_eqn12} becomes an equality, which implies that all positive
operators in \eqref{sep1_eqn11} are proportional, i.e.
\begin{equation}
\label{sep1_eqn14}
A_m^\dagger A_m^{}\otimes B_m^\dagger B_m^{} =
f_{mn}^{}A_n^\dagger A_n^{}\otimes B_n^\dagger B_n^{},
\end{equation}
where the $f_{mn}$ are positive constants. Setting $n=1$ in
\eqref{sep1_eqn14} and inserting it in \eqref{sep1_eqn2} one gets
\begin{equation}
\label{sep1_eqn15}
(\sum_m f_{m1}) A_1^\dagger A_1\otimes B_1^\dagger B_1 = I\otimes I.
\end{equation}
This implies that both $A_1^\dagger A_1^{}$ and $B_1^\dagger B_1^{}$ are
proportional to the identity, so $A_1$ and $B_1$ are proportional to unitary
operators, and of course the same argument works for every $m$.  Since local
unitaries cannot change the Schmidt coefficients, it is obvious that
$\ket{\psi}$ and $\ket{\phi}$ must share the same set of Schmidt coefficients,
that is $\lambda_j=\mu_j$, for every $j$, and this proves (iii).

To prove (iv), note that if there is a separable operation carrying
$\ket{\psi}$ to $\ket{\phi}$ and another carrying $\ket{\phi}$ to
$\ket{\psi}$, the Schmidt ranks of $\ket{\psi}$ and $\ket{\phi}$ must be equal
by (i), and \eqref{sep1_eqn5} is an equality, so (iii) implies equal Schmidt
coefficients.

Next let us consider the modifications needed when the Schmidt ranks of
$\ket{\psi}$ and $\ket{\phi}$ might be unequal, and are possibly less than the
dimensions of $\HC_A$ or $\HC_B$, which need not be the same.  As noted
previously, \eqref{sep1_eqn9} shows that the Schmidt rank of $\ket{\phi}$ cannot be
greater than that of $\ket{\psi}$.  If it is less, then the right side of
\eqref{sep1_eqn5} is zero, because at least one of the $\mu_j$ in the product will
be zero, so part (ii) of the theorem is automatically satisfied, part (iii)
does not apply, and (iv) is trivial. Thus we only need to discuss the case in
which the Schmidt ranks of $\ket{\psi}$ and $\ket{\phi}$ have the same value
$r$.  Let $P_A$ and $P_B$ be the projectors on the $\HC_A$ and $\HC_B$
supports $\SC_A$ and $\SC_B$ of $\ket{\psi}$ (as defined at the beginning of
this section), and let $\TC_A$ and $\TC_B$ be the corresponding supports of
$\ket{\phi}$. Note that each of these subspaces is of dimension $r$.  Since
$(P_A\otimes P_B)\ket{\psi}=\ket{\psi}$, \eqref{sep1_eqn8} can be rewritten as
\begin{equation}
\label{sep1_eqn16}
\bigl(A'_m\otimes B'_m\bigr)\ket{\psi}=\sqrt{p_m}\ket{\phi},
\end{equation}
where
\begin{equation}
\label{sep1_eqn17}
 A'_m =  A_m P_A,\quad B'_m =  B_m P_B
\end{equation}
are the operators $A_m$ and $B_m$ restricted to the supports of $\ket{\psi}$.
In fact, $A'_m$ maps $\SC_A$ onto $\TC_A$, and $B'_m$ maps $\SC_B$ onto
$\TC_B$, as this is the only way in which \eqref{sep1_eqn16} can be satisfied
when $\ket{\phi}$ and $\ket{\psi}$ have the same Schmidt rank. Finally,
by multiplying \eqref{sep1_eqn2} by $P_A\otimes P_B$ on both left and right
one arrives at the closure condition
\begin{equation}
\label{sep1_eqn18}
\sum_m {A'_m}^\dagger {A'_m}^{}\otimes {B'_m}^\dagger {B'_m}^{}=P_A\otimes P_B.
\end{equation}
Thus if we use the restricted operators $A'_m$ and $B'_m$ we are back to the
situation considered previously, with $\SC_A$ and $\TC_A$ (which are
isomorphic) playing the role of $\HC_A$, and $\SC_B$ and $\TC_B$ the role of
$\HC_B$, and hence the previous proof applies.
\end{proof}

Some connections between LOCC and the more general category of separable
operations are indicated in the following corollaries:

\begin{corollary}
\label{sep1_crl1}
When $\ket{\psi}$ is majorized by $\ket{\phi}$, so there is a deterministic
LOCC mapping $\ket{\psi}$ to $\ket{\phi}$, there does not exist a separable
operation that deterministically maps $\ket{\phi}$ to $\ket{\psi}$, unless
these have equal Schmidt coefficients (are equivalent under local unitaries).
\end{corollary}

This is nothing but (iv) of Theorem 1 applied when the $\ket{\psi}$ to
$\ket{\phi}$ map is LOCC, and thus separable.  It is nonetheless worth
pointing out because majorization provides a very precise characterization of
what deterministic LOCC operations can accomplish, and the corollary
provides a connection with more general separable operations.

\begin{corollary}
\label{sep1_crl2}
If either $\HC_A$ or $\HC_B$ is 2-dimensional, then $\ket{\psi}$ can be
deterministically transformed to $\ket{\phi}$ if and only if this is possible
using LOCC, i.e., $\ket{\psi}$ is majorized by $\ket{\phi}$.
\end{corollary}

The proof comes from noting that when there are only two nonzero Schmidt
coefficients, the majorization condition is $\mu_0\geq \lambda_0$, and this is
equivalent to \eqref{sep1_eqn5}.

\section{Separable random unitary channel}
\label{sep1_sct3}

\subsection{Condition for deterministic mapping}
\label{sep1_sct4}

Any quantum operation (trace-preserving completely positive map) can be
thought of as a quantum channel, and if the Kraus operators are proportional
to unitaries, the channel is bistochastic (maps $I$ to $I$) and is called a
random unitary channel or a random external field in Sec.~10.6
of \cite{BengtssonZyczkowski:GeometryQuantumStates}.  Thus a separable operation in which the $A_m$
and $B_m$ are proportional to unitaries $U_m$ and $V_m$, so \eqref{sep1_eqn1} takes
the form
\begin{equation}
\label{sep1_eqn19}
\rho' = \Lambda(\rho)=\sum_m p_m\big(U_m\otimes
V_m\big)\rho\big(U_m\otimes V_m\big)^\dagger,
\end{equation}
with the $p_m>0$ summing to 1, can be called a separable random unitary
channel.  We shall be interested in the case in which $\HC_A$ and $\HC_B$ have
the same dimension $D$, and in which the separable unitary channel
deterministically maps not just one but a collection $\{\ket{\psi_j}\}$,
$1\leq j\leq N$ of pure states of full Schmidt rank $D$ to pure states.  This
means that \eqref{sep1_eqn8} written in the form
\begin{equation}
\label{sep1_eqn20}
\bigl(U_m \otimes V_m\bigr)
\ket{\psi_j} \doteq \ket{\phi_j},
\end{equation}
must hold for all $j$ as well as for all $m$.  The dot equality $\doteq$ means
the two sides can differ by at most a complex phase.  Here such phases cannot
simply be incorporated in $U_m$ or $V_m$, because \eqref{sep1_eqn20} must hold for
all values of $j$, even though they are not relevant for the map carrying
$\dyad{\psi_j}{\psi_j}$ to $\dyad{\phi_j}{\phi_j}$.

\begin{theorem}
\label{sep1_thm2}
Let $\{\ket{\psi_j}\}$, $1\leq j\leq N$ be a collection of states of full
Schmidt rank on a tensor product $\HC_A\otimes\HC_B$ of two spaces of equal
dimension, and let $\Lambda$ be the separable random unitary channel defined by
\eqref{sep1_eqn19}. Let $\psi_j$ and $\phi_j$  be the operators dual to
$\ket{\psi_j}$ and $\ket{\phi_j}$---see \eqref{sep1_eqn6} and \eqref{sep1_eqn7}.

i) If every $\ket{\psi_j}$ from the collection is deterministically mapped to
a pure state, then
\begin{equation}
\label{sep1_eqn21}
 U_m^\dagger U_n^{} \psi_j^{}\psi_k^\dagger \doteq
 \psi_j^{}\psi_k^\dagger U_m^\dagger U_n^{}
\end{equation}
for every $m, n, j,$ and $k$.

ii) If \eqref{sep1_eqn21} holds for a \emph{fixed} $m$ and every $n, j,$ and $k$,
it holds for every $m, n, j,$ and $k$.  If in addition \emph{at least one} of
the states from the collection $\{\ket{\psi_j}\}$ is deterministically mapped
to a pure state by $\Lambda$, then every state in the collection is mapped to
a pure state.

iii) Statements (i) and (ii) also hold when \eqref{sep1_eqn21} is replaced
with
\begin{equation}
\label{sep1_eqn22}
 V_m^\dagger V_n^{} \psi_j^\dagger\psi_k^{} \doteq
 \psi_j^\dagger\psi_k^{} V_m^\dagger V_n^{}.
\end{equation}

\end{theorem}

\begin{proof}

Part (i).
By map-state duality \eqref{sep1_eqn20} can be rewritten as
\begin{equation}
\label{sep1_eqn23}
 U_m\psi_j\bar V_m \doteq \phi_j,
\end{equation}
where $\bar V_m$ is the transpose of $V_m$---see the remarks following
\eqref{sep1_eqn9}. By combining \eqref{sep1_eqn23} with its adjoint with $j$ replaced by
$k$, and using the fact that $\bar V_m$ is unitary, we arrive at
\begin{equation}
\label{sep1_eqn24}
 U_m^{}\psi_j^{}\psi_k^\dagger U_m^\dagger \doteq \phi_j^{}\phi_k^\dagger.
\end{equation}
Since the right side is independent of $m$, so is the left, which means that
\begin{equation}
\label{sep1_eqn25}
 U_n^{}\psi_j^{}\psi_k^\dagger U_n^\dagger  \doteq
 U_m^{}\psi_j^{}\psi_k^\dagger U_m^\dagger.
\end{equation}
Multiply on the left by $U_m^\dagger$ and on the right by $U_n$ to obtain
\eqref{sep1_eqn21}.

Part (ii).
If \eqref{sep1_eqn25}, which is equivalent to \eqref{sep1_eqn21}, holds for $m=1$ it
obviously holds for all values of $m$. Now assume that $\ket{\psi_1}$ is
mapped by $\Lambda$ to a pure state $\ket{\phi_1}$, so \eqref{sep1_eqn23} holds
for all $m$ when $j=1$. Take the adjoint of this equation and multiply by
$\bar V_m$ to obtain
\begin{equation}
\label{sep1_eqn26}
 \psi_1^\dagger U_m^\dagger \doteq \bar V_m^{}\phi_1^\dagger.
\end{equation}
Set $k=1$ in \eqref{sep1_eqn25}, and use  \eqref{sep1_eqn26} to rewrite it as
\begin{equation}
\label{sep1_eqn27}
 U_n^{}\psi_j^{}\bar V_n^{}\phi_1^\dagger \doteq
 U_m^{}\psi_j^{}\bar V_m^{}\phi_1^\dagger.
\end{equation}
Since by hypothesis $\ket{\psi_1}$ has Schmidt rank $D$, the same is true of
$\psi_1$, and since $U_m$ and $\bar V_m$ in \eqref{sep1_eqn23} are unitaries,
$\phi_1$ and thus also $\phi_1^\dagger$ has rank $D$ and is invertible.
Consequently, \eqref{sep1_eqn27} implies that
\begin{equation}
\label{sep1_eqn28}
 U_n^{}\psi_j^{}\bar V_n^{} \doteq U_m^{}\psi_j^{}\bar V_m^{},
\end{equation}
and we can define $\phi_j$ to be one of these common values, for example
$U_1\psi_j\bar V_1$.  Map-state duality transforms this $\phi_j$
into $\ket{\phi_j}$ which, because of \eqref{sep1_eqn28}, satisfies \eqref{sep1_eqn20}.

Part (iii). The roles of $U_m$ and $V_m$ are obviously symmetrical, but our
convention for map-state duality makes $\psi_j$ a map from $\HC_B$ to $\HC_A$,
which is the reason why its adjoint appears in \eqref{sep1_eqn22}.
\end{proof}

\subsection{Example}
\label{sep1_sct5}

Let us apply Theorem \ref{sep1_thm2} to see what pure states of full Schmidt rank
are deterministically mapped onto pure states by the following separable
random unitary channel on two qubits:
\begin{equation}
\label{sep1_eqn29}
\Lambda(\rho)=p\rho+(1-p)(X\otimes Z)\rho (X\otimes Z).
\end{equation}
The Kraus operators are $I\otimes I$ and $X\otimes Z$, so $U_1=I$ and $U_2=X$.
Thus the condition \eqref{sep1_eqn21} for a collection of states $\{\ket{\psi_j}\}$
to be deterministically mapped to pure states is
\begin{equation}
\label{sep1_eqn30}
X\psi_j\psi_k^\dagger\doteq \psi_j\psi_k^\dagger X.
\end{equation}
It is easily checked that
\begin{equation}
\label{sep1_eqn31}
 \ket{\psi_1}= (\ket{+}\ket{0}+\ket{-}\ket{1})/\sqrt{2}
\end{equation}
 is mapped to itself by \eqref{sep1_eqn29}.  If the corresponding
\begin{equation}
\label{sep1_eqn32} \psi_1= \frac{1}{2}\left(
\begin{array}{cc}
1 & 1\\
1 & -1\\
\end{array}
\right)
\end{equation}
is inserted in \eqref{sep1_eqn30} with $k=1$, one can show that \eqref{sep1_eqn30} is
satisfied for any $2\times2$ matrix
\begin{equation}
\label{sep1_eqn33}
\psi_j=\left(
\begin{array}{cc}
a_j & b_j\\
c_j & d_j\\
\end{array}
\right)
\end{equation}
having $c_j=\pm a_j$ and $d_j=\mp b_j$, and that in turn these
satisfy \eqref{sep1_eqn30} for every $j$ and $k$. Thus all states of the form
\begin{equation}
\label{sep1_eqn34}
\ket{\psi_\pm}=a\ket{00}+b\ket{01}\pm a\ket{10}\mp b\ket{11}
\end{equation}
with $a$ and $b$ complex numbers, are mapped by this channel into pure states.

\section{Conclusions}
\label{sep1_sct6}

Our main results are in Theorem~\ref{sep1_thm1}: if a pure state on a bipartite
system $\HC_A\otimes\HC_B$ is deterministically mapped to a pure state by a
separable operation $\{A_m\otimes B_m\}$, then the product of the Schmidt
coefficients can only decrease, and if it remains the same, the two sets of
Schmidt coefficients are identical to each other, and the $A_m$ and $B_m$
operators are proportional to unitaries.  (See the detailed statement of the
theorem for situations in which some of the Schmidt coefficients vanish.)
This \emph{product condition} is necessary but not sufficient: i.e., even if
it is satisfied there is no guarantee that a separable operation exists which
can carry out the specified map.  Indeed, we think it is likely that when both
$\HC_A$ and $\HC_B$ have dimension 3 or more there are situations in which the
product condition is satisfied but a deterministic map is not possible.  The
reason is that \eqref{sep1_eqn5} is consistent with $\ket{\phi}$ having a larger
entanglement than $\ket{\psi}$, and we doubt whether a separable operation can
increase entanglement.  While it is known that LOCC cannot increase the
average entanglement [\cite{RevModPhys.81.865}, Sec.~XV D], there
seems to be no similar result for general separable operations. This is an
important open question.

It is helpful to compare the product condition \eqref{sep1_eqn5} with Nielsen's
majorization condition, which says that a deterministic separable operation of
the LOCC type can map $\ket{\psi}$ to $\ket{\phi}$ if and only if $\ket{\phi}$
majorizes $\ket{\psi}$, see~\ref{sep1_ftn1}. Corollary~\ref{sep1_crl2} of Theorem~\ref{sep1_thm1}
shows that the two are identical if system $A$ or system $B$ is 2-dimensional.
Under this condition a general separable operation can deterministically map
$\ket{\psi}$ to $\ket{\phi}$ only if it is possible with LOCC.  This
observation gives rise to the conjecture that when either $A$ or $B$ is
2-dimensional \emph{any} separable operation is actually of the LOCC form.
This conjecture is consistent with the fact that the well-known example
\cite{PhysRevA.59.1070} of a separable operation that is \emph{not} LOCC uses
the tensor product of two 3-dimensional spaces.  But whether separable and LOCC
coincide even in the simple case of a $2\times 2$ system is at present an
open question (see note added in proof).

When the dimensions of $A$ and $B$ are both 3 or more the product condition of
Theorem~\ref{sep1_thm1} is weaker than the majorization condition: if $\ket{\phi}$
majorizes $\ket{\psi}$ then \eqref{sep1_eqn5} will hold 
\footnote{
The general argument that \eqref{sep1_eqn5} is implied by (though it
  does not imply) majorization will be found in [\cite{Nielsen:MajorizationNotes}, Sec.~4], or as an exercise on [\cite{HornJohnson:MatrixAnalysis}, p.~199]
\label{sep1_ftn2}
},
but the converse is in general not true. Thus there might be situations in which a
separable operation deterministically maps $\ket{\psi}$ to $\ket{\phi}$ even
though $\ket{\phi}$ does not majorize $\ket{\psi}$.  If such cases exist,
Corollary~\ref{sep1_crl1} of Theorem~\ref{sep1_thm1} tells us that $\ket{\psi}$ and
$\ket{\phi}$ must be incomparable under majorization: neither one majorizes the
other.  Finding an instance, or demonstrating its impossibility, would help
clarify how general separable operations differ from the LOCC subclass.

When a separable operation deterministically maps $\ket{\psi}$ to $\ket{\phi}$
and the product of the two sets of Schmidt coefficients are the same, part
(iii) of Theorem~\ref{sep1_thm1} tells us that the collections of Schmidt
coefficients are in fact identical, and that the $A_m$ and $B_m$ operators
(restricted if necessary to the supports of $\ket{\psi}$) are proportional to
unitaries.  Given this proportionality (and that the map is deterministic),
the identity of the collection of Schmidt coefficients is immediately evident,
but the converse is not at all obvious.  The result just mentioned can be used
to simplify part of the proof in some interesting work on local copying,
specifically the unitarity of local Kraus operators in [\cite{NewJPhys.6.164}, Sec.~3.1]. It might have applications in other cases
where one is interested in deterministic nonlocal operations.

Finally, Theorem~\ref{sep1_thm2} gives conditions under which a separable random
unitary operation can deterministically map a whole collection of pure states
to pure states.  These conditions [see \eqref{sep1_eqn21} or \eqref{sep1_eqn22}] involve
both the unitary operators and the states themselves, expressed as operators
using map-state duality, in an interesting combination.  While these results
apply only to a very special category, they raise the question whether
simultaneous deterministic maps of several pure states might be of interest
for more general separable operations. The nonlocal copying problem, as
discussed in
\cite{NewJPhys.6.164,PhysRevA.73.012343,PhysRevA.69.052312,PhysRevA.74.032108},
is one situation where results of this type are relevant, and there may be
others.

\textit{Note added in proof.} Our conjecture on the equivalence of separable operations and LOCC for low dimensions has been shown to be false \cite{quantph.0705.0795}.

\chapter{Separable operations on pure states\label{chp3}}

\section{Introduction\label{sep2_sct1}} 
A separable operation $\Lambda$ on a
bipartite quantum system is a transformation of the form
\begin{equation}
\label{sep2_eqn1} \rho'=\Lambda(\rho)=\sum_{k=0}^{N-1}(A_k\otimes B_k)\rho(A_k\otimes
B_k)^\dagger,
\end{equation} where $\rho$ is an initial density operator on the Hilbert
space $\HC_A\otimes \HC_B$. The Kraus operators $A_k\otimes B_k$ are arbitrary
product operators satisfying the closure condition
\begin{equation}
\label{sep2_eqn2} \sum_{k=0}^{N-1}A_k^\dagger A_k\otimes B_k^\dagger B_k=I_A\otimes
I_B,
\end{equation} with $I_A$ and $I_B$ the identity operators. The extension to
multipartite systems is obvious, but here we will only consider the bipartite
case. To avoid technical issues the sums in \eqref{sep2_eqn1} and \eqref{sep2_eqn2} as
well as the dimensions $D_A$ and $D_B$ of $\HC_A$ and $\HC_B$ are assumed to
be finite.

Local operations with classical communication (LOCC)\index{LOCC} form a subset of
separable operations in which the Kraus operators $A_k\otimes B_k$ are
restricted by the requirement that they be generated in the following fashion.
Alice carries out an operation $\{A^{(1)}_i\}$, $\sum_i A^{(1)\dagger}_i
A^{(1)}_i=I_A$, in the usual way with the help of an ancilla, the measurement
of which yields the value of $i$, which is then transmitted to Bob.  He uses
$i$ to choose an operation $\{B^{(2,i)}_j\}$, the result $j$ of which is
transmitted back to Alice, whose next operation can depend on $j$ as well as
$i$, and so forth.  While it is (fairly) easy to see that the end result after
an arbitrary number of rounds is of the form \eqref{sep2_eqn1}, it is difficult to
characterize in simple mathematical or physical terms precisely what it is
that distinguishes LOCC from more general separable operations.  Examples show
that separable operations can be more effective than LOCC in distinguishing
certain sets of orthogonal states \cite{PhysRevA.59.1070}, even in a system as
simple as two qubits \cite{quantph.0705.0795}, but apart from this little is
known about the difference.

What we demonstrate in Sec.~\ref{sep2_sct2} of this chapter is that the ensemble
$\{p_k,\ket{\phi_k}\}$ produced by a separable operation acting on a pure
state $\ket{\psi}$, see \eqref{sep2_eqn5}, satisfies a majorization condition
\eqref{sep2_eqn7}, which is already known to be a necessary and sufficient
condition for producing the same ensemble from the same $\ket{\psi}$ by LOCC.
Among the consequences discussed in Sec.~\ref{sep2_sct5} are: a separable
operation acting on a pure state can be ``simulated'' by LOCC; a necessary
condition for a deterministic transformation $\ket{\psi}\rightarrow\ket{\phi}$
given in \cite{PhysRevA.76.032310} can be replaced by a necessary and
sufficient majorization condition; and certain entanglement measures are
nonincreasing under separable operations.  Section~\ref{sep2_sct6} summarizes our
main result and indicates some open questions.

\section{Ensembles produced by separable operations on pure bipartite
states}\label{sep2_sct2}
\subsection{Majorization conditions\label{sep2_sct3}}

Let $\{A_k\otimes B_k\}_{k=1}^{N}$ be a separable operation on
$\HC_A\otimes\HC_B$, specified by $N$ Kraus operators satisfying the closure
condition \eqref{sep2_eqn2}.  Let $\ket{\psi}$ be a normalized entangled state on
$\HC_A\otimes\HC_B$ with Schmidt form
\begin{equation}
\label{sep2_eqn3} \ket{\psi}=\sum_{j=0}^{D-1}\sqrt{\lambda_j}\ket{a_j}\ket{b_j},
\end{equation} where $D=D_B$, and we assume without loss of generality that
$D_A\geq D_B$.  Here $\{\ket{a_j}\}$ and $\{\ket{b_j}\}$ are orthonormal bases
chosen so that the Schmidt weights (coefficients) $\lambda_j$ are in
increasing order, i.e.
\begin{equation}
\label{sep2_eqn4}
0\leqslant\lambda_0\leqslant\lambda_1\leqslant\cdots\leqslant\lambda_{D-1}.
\end{equation} The separable operation acting on $\ket{\psi}$ will produce an
ensemble $\{p_k,\ket{\phi_k}\}_{k=1}^N$, where
\begin{equation}
\label{sep2_eqn5} (A_k\otimes B_k)\ket{\psi}=\sqrt{p_k}\ket{\phi_k}
\end{equation} and
\begin{equation}
\label{sep2_eqn6} p_k=\matl{\psi}{A_k^\dagger A_k\otimes B_k^\dagger B_k}{\psi}.
\end{equation}

In \cite{PhysRevLett.83.1455} it was shown that such an ensemble
$\{p_k,\ket{\phi_k}\}_{k=0}^{N-1}$ can be produced from $\ket{\psi}$ by a suitable LOCC if
and only if the majorization inequalities
\begin{equation}
\label{sep2_eqn7} \sum_{k=0}^{N-1}p_kE_n(\ket{\phi_k})\leqslant E_n(\ket{\psi})
\end{equation} 
hold for $0\leqslant n\leqslant D-1$, where
\begin{equation}
\label{sep2_eqn8}
E_n(\ket{\psi})=\chi_n\big(\Tr_A\dyad{\psi}{\psi}\big)=\sum_{j=0}^{n}\lambda_j,
\end{equation} 
and similarly for the $\ket{\phi_k}$. Here $\Tr_A(\dyad{\psi}{\psi})$ is the
reduced density operator of $\dyad{\psi}{\psi}$ on Bob's side, and
$\chi_n(\cdot)$ is defined to be the sum of the first $n$ smallest eigenvalues
of its argument.  Note that we are assuming that $D=D_B\leq D_A$, because if
$D_B$ were greater than $D_A$ the extra zero eigenvalues in
$\Tr_A\dyad{\psi}{\psi}$ would cause confusion when using $\chi_n$.

Our main result is the following.
\begin{theorem}
  \label{sep2_thm1} The ensemble $\{p_k,\ket{\phi_k}\}_{k=0}^{N-1}$ can be produced by
  a bipartite separable operation acting on the normalized state $\ket{\psi}$
  if and only if the majorization condition defined by the collection of
  inequalities in \eqref{sep2_eqn7} is satisfied.
\end{theorem}

\begin{proof}

  To simplify the proof we assume that $D_A=D_B=D$.  If $D_A$ is larger, one
always modify each $A_k$ by following it with a suitable local unitary which
has the result that as long as the Kraus operators are acting on a fixed
$\ket{\psi}$ the action on the $A$ side takes place in a subspace of $\HC_A$
of dimension $D$.  These local unitaries do not change the Schmidt weights of
the $\ket{\phi_k}$ or alter the closure condition \eqref{sep2_eqn2}. For more
details about this ``decoupling'' see \cite{PhysRevA.76.032310}.

  When the majorization condition \eqref{sep2_eqn7} holds the result in
\cite{PhysRevLett.83.1455} guarantees the existence of an LOCC (hence
separable operation) which will produce the ensemble out of $\ket{\psi}$.  The
reverse inference, that the ensemble $\{p_k,\ket{\phi_k}\}_{k=0}^{N-1}$ defined in
\eqref{sep2_eqn5} and \eqref{sep2_eqn6} satisfies \eqref{sep2_eqn7}, follows from noting that
\begin{equation}
\label{sep2_eqn9} p_kE_n(\ket{\phi_k})=\chi_n \big(\Tr_A[A_k\otimes
B_k\dyad{\psi}{\psi}A_k^\dagger\otimes B_k^\dagger]\big),
\end{equation} 
and applying Theorem~\ref{sep2_thm2} below with $R=I_A\otimes I_B$,
corresponding to \eqref{sep2_eqn2}, 
so $\|R\|=1$.
\end{proof}

\subsection{A majorization theorem\label{sep2_sct4}}

\begin{theorem}
\label{sep2_thm2} Let $\HC_A$ and $\HC_B$ have the same dimension $D$, let
$\ket{\psi}$ be some pure state on $\HC_A \otimes\HC_B$, and let $\{A_k\otimes
B_k\}_{k=0}^{N-1}$ be any collection of product operators on $\HC_A\otimes \HC_B$.
Then for every $0\leqslant n\leqslant {D-1}$
\begin{equation}
\label{sep2_eqn10} \sum_{k=0}^{N-1}\chi_n\big(\Tr_A[A_k\otimes
B_k\dyad{\psi}{\psi}A_k^\dagger\otimes B_k^\dagger]\big)
\leqslant\|R\|\chi_n\big(\Tr_A\dyad{\psi}{\psi}\big),
\end{equation} 
where $\|R\|=\sup_{\|\omega\|=1}\|R\ket{\omega}\|$ is the
largest eigenvalue of the positive operator
\begin{equation}
\label{sep2_eqn11} R=\sum_{k=0}^{N-1} A_k^\dagger A_k\otimes B_k^\dagger B_k.
\end{equation}
\end{theorem}

\begin{proof} By map-state duality
\cite{PhysRevA.76.032310,OSID.11.3,PhysRevA.73.052309}, using the Schmidt bases
of $\ket{\psi}$, we transform the state $A_k\otimes B_k\ket{\psi}$ to a map
$A_k\psi \bar B_k$, where
\begin{equation}
\label{sep2_eqn12} \psi=\sum_{j=0}^{D-1}\sqrt{\lambda_j}\dyad{a_j}{b_j}.
\end{equation} 
denotes an operator mapping $\HC_B$ to $\HC_A$, and $\bar B_k=
B_k^T$ is the transpose of $B_k$. The matrix of $\psi$ using the Schmidt bases
of $\ket{\psi}$ is diagonal, with the entries on the diagonal in
increasing order. (See Sec. II of \cite{PhysRevA.76.032310} for more details
on map-state duality.)
Upon writing the partial traces as
\begin{equation}
\label{sep2_eqn13} \Tr_A\dyad{\psi}{\psi}=\psi\psi^\dagger,\quad
 \Tr_A[A_k\otimes B_k\dyad{\psi}{\psi}A_k^\dagger\otimes
B_k^\dagger]= A_k\psi {\bar B_k}{\bar B_k}^\dagger\psi^\dagger A_k^\dagger,
\end{equation} 
the inequalities \eqref{sep2_eqn10} become:
\begin{equation}
\label{sep2_eqn14} \sum_{k=0}^{N-1}\chi_n(A_k \psi \bar B_k{\bar
B_k}^\dagger\psi^\dagger A_k^\dagger) \leqslant\|R\|\chi_n(\psi\psi^\dagger).
\end{equation}

For some $n$ between $0$ and $D-1$ write the diagonal matrix $\psi$ as
\begin{equation}
\label{sep2_eqn15} \psi=\psi_n+\tilde\psi_{n},
\end{equation} where $\psi_n$ is the same matrix but with $\lambda_{n},
\lambda_{n+1},\ldots$ set equal to zero, while $\tilde\psi_n$ is obtained by
setting $\lambda_0, \lambda_1,\ldots \lambda_{n-1}$ equal to zero.
Lemma~\ref{sep2_lma1}, below, tells us that for each $k$,
\begin{equation} \chi_n(A_k \psi \bar B_k {\bar B_k}^\dagger\psi^\dagger
A_k^\dagger ) \leqslant \Tr(A_k \psi_n \bar B_k {\bar
B_k}^\dagger\psi_n^\dagger A_k^\dagger ).
\label{sep2_eqn16}
\end{equation}
By map-state duality,
\begin{equation}
\label{sep2_eqn17} \Tr(A_k \psi_n \bar B_k {\bar B_k}^\dagger\psi_n^\dagger
A_k^\dagger )= \matl{\psi_n}{A_k^\dagger A_k\otimes B_k^\dagger B_k}{\psi_n}
\end{equation} where $\ket{\psi_n}$, the counterpart of $\psi_n$, is given by
\eqref{sep2_eqn3} with $D$ replaced by $n$. Inserting \eqref{sep2_eqn17} in
\eqref{sep2_eqn16} and summing over $k$, see \eqref{sep2_eqn11}, yields
\begin{equation}
\label{sep2_eqn18} \sum_{k=0}^{N-1}\chi_n(A_k \psi \bar B_k{\bar
B_k}^\dagger\psi^\dagger A_k^\dagger )
\leqslant\matl{\psi_n}{R}{\psi_n}
\leqslant\|R\|\ip{\psi_n}{\psi_n}=\|R\|\chi_n(\psi^\dagger\psi).
\end{equation}
This establishes \eqref{sep2_eqn14}, which is equivalent to
\eqref{sep2_eqn10}.
\end{proof}

\begin{lemma}
\label{sep2_lma1} Let $A$, $B$, and $\psi$ be $D\times D$ matrices, where $\psi$ is
diagonal with nonnegative diagonal elements in increasing order, and for some
$0\leqslant n\leqslant D-1$ let $\psi_n$ be obtained from $\psi$ by setting all
but the $n$ smallest diagonal elements equal to 0, as in \eqref{sep2_eqn15}. Then
\begin{equation}
\label{sep2_eqn19} \chi_n(A \psi B{B}^\dagger\psi^\dagger A^\dagger )\leqslant
\Tr(A \psi_n B{B}^\dagger\psi_n^\dagger A^\dagger ).
\end{equation}
\end{lemma}

\begin{proof}

The inequality
\begin{equation}
\label{sep2_eqn20} \chi_n(A \psi B{B}^\dagger\psi^\dagger A^\dagger )\leqslant
\Tr(P_nA \psi B{B}^\dagger\psi^\dagger A^\dagger P_n),
\end{equation} 
where $P_n$ is a projector (orthogonal projection operator) of rank at least
$n$, follows from the fact that for any Hermitian operator $T$ the sum of its
$n$ smallest eigenvalues is the minimum of $\Tr(P_nTP_n)$ over such $P_n$, see
page~24 of \cite{Bhatia:MatrixAnalysis}.  Choose $P_n$ to be the
projector onto the orthogonal complement of the range of $A\tilde\psi_{n}$,
where $\tilde\psi_n=\psi-\psi_n$, as in \eqref{sep2_eqn15}.  The rank of
$A\tilde\psi_{n}$ is no larger than the rank of $\tilde\psi_n$, which is
smaller than or equal to $D-n$. Thus the dimension of the range of
$A\tilde\psi_{n}$ cannot exceed $D-n$, so the rank of $P_n$ is at least $n$.
By construction, $P_nA\tilde\psi_{n}=0$, so
\begin{equation}
\label{sep2_eqn21} P_nA\psi=P_nA(\psi_n+\tilde\psi_{n})=P_nA\psi_n.
\end{equation} 
Thus with this choice of $P_n$ the right side of \eqref{sep2_eqn20} is
\begin{equation}
\label{sep2_eqn22} \Tr(P_n A \psi_n B{B}^\dagger\psi_n^\dagger A^\dagger P_n),
\end{equation} 
and this implies \eqref{sep2_eqn19}, since $P_n\leqslant I$ and $A \psi_n B
{B}^\dagger\psi_n^\dagger A^\dagger $ is positive.
\end{proof}

\section{Consequences\label{sep2_sct5}} The following are some consequences of
Theorem~\ref{sep2_thm1}.
\begin{itemize}
 \item[i)] An ensemble
$\{p_k,\ket{\phi_k}\}$ can be produced by a separable operation acting on
a bipartite entangled state
$\ket{\psi}$ if and only if it can be produced by some LOCC acting on the same
state $\ket{\psi}$.

\item[ii)] For a given bipartite $\ket{\psi}$ and separable operation
  $\{A_k\otimes B_k\}_{k=0}^{N-1}$ there is another operation of the form 
$\{\hat A_l\otimes U_l\}_{l=0}^{M-1}$, where the $U_l$ are unitary operators (and the closure condition is $\sum_{l=0}^{M-1}\hat A_l^\dagger \hat A_l=I_A$), which 
produces the same ensemble when applied to $\ket{\psi}$. Here $M$ could be
different from $N$, as two Kraus operators might yield the same $\ket{\phi_k}$. For more details about the relation between the $\{A_k,B_k\}_{k=0}^{N-1}$ set and the $\{\hat A_l\otimes U_l\}_{k=0}^{M-1}$ set see \cite{PhysRevA.63.022301}.

\item[iii)] A deterministic transformation $\ket{\psi}\rightarrow\ket{\phi}$
  by a separable operation is possible if and only if
  $E_n(\ket{\phi})\leqslant E_n(\ket{\psi})$ for every $n$ between $0$ and
  $D-1$, with $E_n(.)$ defined in \eqref{sep2_eqn8} This is often written as
  $\lambda_\psi\prec\lambda_\phi$, where $\lambda_\psi$ and $\lambda_\phi$ are
  vectors of the corresponding Schmidt weights.  (This extends 
  Theorem~1 in \cite{PhysRevA.76.032310}.) 

\item[iv)] The maximum probability of success for the
transformation $\ket{\psi}\rightarrow\ket{\phi}$ by a separable operation is
given by
\begin{equation}
\label{sep2_eqn23}
p_{max}^{SEP}(\ket{\psi}\rightarrow\ket{\phi})=\min_{n\in[0,D-1]}\frac{E_n(\ket{\psi})}{E_n(\ket{\phi})},
\end{equation} where $E_n(\cdot)$ was defined in \eqref{sep2_eqn8}.

\item[v)] An entanglement measures $E$ defined on pure bipartite
states is nonincreasing on average under separable
operations, which is to say
\begin{equation}
\label{sep2_eqn24} E(\ket{\psi})\geqslant\sum_{k=0}^{N-1} p_kE(\ket{\phi_k})
\end{equation}
if and only if it is similarly nonincreasing under LOCC.

\item[vi)] Let
\begin{equation}
\label{sep2_eqn25} \hat E(\rho)=\inf\sum_{i}p_iE(\ket{\psi_i}),
\end{equation} 
with the infimum over all ensembles $\{p_i,\ket{\psi_i}\}$ yielding the
density operator $\rho=\sum_ip_i\dyad{\psi_i}{\psi_i}$, be the convex roof
extension of a pure state entanglement measure $E$ that is monotone on pure
states in the sense of \eqref{sep2_eqn24}.  Then $\hat E$ is monotone on mixed
states in the sense that
\begin{equation}
\label{sep2_eqn26} \hat E(\rho)\geqslant\sum_{k=0}^{N-1} p_k\hat E(\sigma_k)
\end{equation} 
for any ensemble $\{p_k,\sigma_k\}$ produced from $\rho$ by separable
operations.

\end{itemize}

The result (i) is an immediate consequence of Theorem~\ref{sep2_thm1}, as the
same majorization condition applies for both separable and LOCC.  Then
(ii), (iii), and (iv) are immediate consequences of known results, in 
 \cite{PhysRevA.63.022301}, \cite{PhysRevLett.83.436}, and 
\cite{PhysRevLett.83.1046}, respectively, for LOCC.  The result (v) is an
obvious consequence of (i), whereas (vi) follows from general arguments
about convex roof extensions; see Sec.~XV.C.2
of \cite{RevModPhys.81.865}.

\section{Conclusion\label{sep2_sct6}} 

We have shown that possible ensembles of states produced by applying a separable
operation to a bipartite entangled pure state can be exactly characterized
through a majorization condition, the collection of inequalities \eqref{sep2_eqn7} 
for different $n$.  These have long been known to be necessary and sufficient
conditions for producing such an ensemble using LOCC, so their extension to
the full class of separable operations is not altogether surprising, even if
our proof is not altogether straightforward.  

Connecting the full set of separable operations with the more specialized LOCC
class immediately yields several significant consequences for the former, as
indicated in the list in Sec.~\ref{sep2_sct5}, because much is already known about
the latter.  Of particular significance is that various entanglement measures
are monotone, meaning they cannot increase, under separable
operations---something expected on intuitive grounds, but now rigorously
proved.  Since such monotonicity under LOCC has long been considered a
necessary, or at least a very desirable condition for any ``reasonable''
entanglement measure on mixed states (see Sec.~XV.B of \cite{RevModPhys.81.865}), one
wonders whether monotonicity under separable operations, in principle a
stronger condition, might be an equally good or even superior desideratum.

Our results apply only to bipartite states, but separable operations and the
LOCC subclass can both be defined for multipartite systems.  Might it be that
in the multipartite case the ensemble produced by applying a separable
operation to a pure entangled state could also be produced by some LOCC
applied to the same state?  It might be, but proving it would require very
different methods than used here.  There are no simple multipartite analogs of
the Schmidt representation \eqref{sep2_eqn3}, the majorization condition
\eqref{sep2_eqn7}, or map-state duality.

Even in the bipartite case we still know very little about separable
operations which are \emph{not} LOCC, aside from the fact that they exist and
can be used to distinguish certain collections of orthogonal states more
effectively than LOCC.  The results in this chapter contribute only indirectly
to a better understanding of this matter: looking at what a separable
operation does when applied to a \emph{single} entangled state will not help;
one must ask what it does to several different states.

\chapter{Local cloning of entangled states by separable operations\label{chp4}}

\section{Introduction\label{lcl_sct1}}
As summarized by the ``no-cloning" theorem of \cite{Nature.299.802}\index{no-cloning
 theorem}, any set of quantum states can be deterministically cloned if and only if the states in the set are mutually orthogonal\index{cloning}. When the set consists of bipartite entangled states, and the cloning is restricted to local operations and classical communication (LOCC), the problem becomes much more difficult, and further restrictions have to be imposed. The mere orthogonality of the states no longer implies that they can be (locally) cloned. 

The local cloning protocol\index{local cloning} of a set of bipartite entangled states $\SC=\{\ket{\psi_i}^{AB}\}$ is schematically represented as
\begin{equation}\label{lcl_eqn1}
\ket{\psi_i}^{AB}\otimes\ket{\phi}^{ab}\longrightarrow\ket{\psi_i}^{AB}\otimes\ket{\psi_i}^{ab},\text{ }\forall i,
\end{equation}
where the letters $A,a$ label Alice's systems and $B,b$ label Bob's systems. Both parties are assumed to have access to ancillary qudits and may share a classical communication channel, so that in principle any LOCC operation can be performed. The state $\ket{\phi}$ is shared in advance between the parties, and it plays the role of a ``blank state" on which the copy of $\ket{\psi_i}$ is to be imprinted. 

The local cloning problem has recently  received a great deal of attention \cite{PhysRevA.69.052312,NewJPhys.6.164,PhysRevA.74.032108,PhysRevA.73.012343,PhysRevA.76.052305}, and was partially extended to tripartite systems in \cite{PhysRevA.76.062312}. 
The question addressed in all previous work was  which sets of states $\SC$ can be locally cloned (by LOCC) using a given blank state $\ket{\phi}$. 

Note that if one can use LOCC to transform $\ket{\phi}$ into three maximally entangled states of sufficient Schmidt rank, then the local cloning of any set of bipartite orthogonal entangled states becomes trivially possible, using teleportation: Alice uses one maximally entangled state to teleport her part of $\ket{\psi_i}$ to Bob, who then distinguishes it (i.e. learns $i$), and next communicates the result back to Alice. Now both Alice and Bob know which state was fed into the local cloning machine. Finally they transform deterministically the two remaining maximally entangled states into $\ket{\psi_i}\otimes\ket{\psi_i}$ by LOCC, which is always possible, 
according to \cite{PhysRevLett.83.436}.

Another possible scenario that uses only two entangled blank states involves using LOCC to deterministically distinguish which state $\ket{\psi_i}$ was fed into the local cloning machine, which can always be done if there are only two states in the set $\SC$ \cite{PhysRevLett.85.4972}. Then, knowing the state, one can deterministically transform the two blank states into $\ket{\psi_i}\otimes\ket{\psi_i}$ (by LOCC). In this case, one needs at least two maximally entangled resource states, one for each of the two copies that must now be created, since in general the entanglement of the original state will have been destroyed in the process of distinguishing the states \cite{PhysRevA.75.052313}. 

One might hope, however, that local cloning can be performed using even less entanglement. As first shown in \cite{PhysRevA.69.052312}, this hope is sometimes correct. Any two (and not more) two-qubit Bell states can be locally cloned using a two-qubit maximally entangled state.

This result was further extended in \cite{NewJPhys.6.164} and \cite{PhysRevA.74.032108}, which considered local cloning of maximally entangled states on higher-dimensional $D\times D$ systems using a maximally entangled resource of Schmidt rank $D$. First, necessary and sufficient conditions for the local cloning of two maximally entangled states were provided in \cite{NewJPhys.6.164}, which also proved that for $D=2$ (qubits) or $D=3$ (qutrits), any pair of maximally entangled states can be locally cloned with a maximally entangled blank state. Whenever $D$ is not prime the authors showed that there always exist pairs of maximally entangled states that cannot be locally cloned with a maximally entangled blank state. A generalization to more than 2 states but prime $D$ was given in \cite{PhysRevA.74.032108},  which showed  that a set of $D$ maximally entangled states can be locally cloned using a maximally entangled resource if and only if the states in the set are locally (cyclically) shifted
\begin{equation}\label{lcl_eqn2}
\ket{\psi_i}=\frac{1}{\sqrt{D}}\sum_{r=0}^{D-1}\ket{r}^A\ket{r\oplus i}^B,
\end{equation}
where the $\oplus$ symbol denotes addition modulo $D$.

Kay and Ericsson \cite{PhysRevA.73.012343} extended the above results to the LOCC cloning of full Schmidt rank partially entangled states using a maximally entangled blank state. They presented an explicit protocol for the local cloning of a set of $D\times D$ cyclically shifted partially entangled states
\begin{equation}\label{lcl_eqn3}
\ket{\psi_i}=\sum_{r=0}^{D-1}\sqrt{\lambda_r}\ket{r}^A\ket{r\oplus i}^B,
\end{equation}
and asserted that \eqref{lcl_eqn3} is also a necessary condition for such cloning; that the states to be cloned must be of this form. Unfortunately, the proof is not correct  
\footnote{The matter was discussed with Kay \cite{GheorghiuAlastair}. The fact that the argument is not correct can be observed after a careful reading of the paragraph following Eq. (3) in \cite{PhysRevA.73.012343}.
 The authors claim that the local cloning of partially entangled states is equivalent to the cloning of maximally entangled states, but this statement is incorrect, because the authors implicitly modified the Kraus operators that defined the local cloning, i.e. changed $A_k$ to $A_k'=A_kM_0$, where $M_0$ (defined in Eq. (3) of \cite{PhysRevA.73.012343}) is the operator that transforms the maximally entangled state ${(1/\sqrt{D})\sum_{r=0}^{D-1}\ket{r}^A\ket{r}^B}$ to the partially entangled state ${|\psi_0\rangle=\sum_{r=0}^{D-1}\sqrt{\lambda_r}\ket{r}^A\ket{r}^B}$. 
The new Kraus operators do not satisfy the closure condition anymore (necessary for a deterministic transformation), since $\sum_k {A_k'}^\dagger A_k' \otimes {B_k}^\dagger B_k=\sum_k M_0^\dagger({A_k}^\dagger A_k)M_0 \otimes {B_k}^\dagger B_k=M_0^\dagger M_0\otimes I\neq I\otimes I$, because $M_0$ is not a unitary operator (unless ${\ket{\psi_0}}$ is maximally entangled, case excluded). \\ \\ Another way of seeing that the argument is not correct is to observe that, if the $B_k$ operator performs the cloning of a maximally entangled state using a maximally entangled blank, as it is claimed, then $B_k$ must be proportional to a unitary operator, see Theorem~1(iii) of \cite{PhysRevA.76.032310} and Sec. 3.1 of \cite{NewJPhys.6.164}. It then follows that the closure condition for the Kraus operators is not satisfied, with $A_k$ as defined in Eq. (3) of \cite{PhysRevA.73.012343}.},
 and therefore finding necessary conditions when the states are partially entangled remains an open problem.

In this chapter, we consider a set $\SC=\{\ket{\psi_i}^{AB}\}$ of full Schmidt rank qudit (of arbitrary dimension) partially entangled states. Actually, we will begin by considering sets $\SC$ in which only one state is required to be full Schmidt rank, and then we will see that in fact, all states in $\SC$ must be full rank. Previous work assumed the blank state $\ket{\phi}$ to be maximally entangled, but in the present chapter we do not impose any \emph{a priori} assumptions on $\ket{\phi}$ and find that its Schmidt rank must be at least that of the states in $\SC$. Furthermore, we do not restrict to LOCC cloning, but allow for the more general class of separable operations --- all the necessary conditions we find for separable operations will also be necessary for LOCC since the latter is a (proper) subset of the former \cite{PhysRevA.59.1070}.

The remainder of the chapter is organized as follows. In the next section we give a preliminary discussion and define some terms that will be used. Then, in Sec.~\ref{lcl_sct3}, we turn to the characterization of clonable sets of states, where we show that $\ket{\phi}$ and all states in $\SC$ must be full Schmidt rank, provide additional necessary conditions on $\SC$, and then prove the group structure of these sets. From this group structure, it is then shown that the number of states in $\SC$ must divide $D$ exactly, and this is followed by a proof of a necessary (``group-shifted") condition on the local cloning of a set of $D\times D$ maximally entangled states. Then, in Sec.~\ref{lcl_sct5}, we further consider group-shifted sets, now allowed to be not maximally entangled, showing that a maximally entangled blank state is sufficient by giving an LOCC protocol that clones these states. This demonstrates that the necessary condition found in the previous section for cloning maximally entangled states is also sufficient for LOCC cloning. In Sec.~\ref{lcl_sct6}, we provide necessary conditions on the minimum entanglement in the blank. In addition, we obtain necessary and sufficient conditions for local cloning of any set when $D=2$ (entangled qubits), and for any group-shifted set for $D=3$ (entangled qutrits); in both these cases we find that the blank state must be maximally entangled, even when the states to be cloned are not. For higher dimensions with these group-shifted sets, we also show that the blank must have strictly more entanglement than the states to be cloned. Finally, Sec.~\ref{lcl_sct7} provides concluding remarks as well as some open questions. Longer proofs are presented in the Appendices.

\section{Preliminary remarks and definitions\label{lcl_sct2}}
A separable operation $\Lambda$ on a bipartite quantum system $\HC_A\otimes\HC_B$ is a transformation that can be written as
\begin{equation}
\label{lcl_eqn4}
\rho'=\Lambda(\rho)=\sum_{m=0}^{M-1}(A_m\otimes B_m)\rho(A_m\otimes B_m)^\dagger
\end{equation}
where $\rho$ is an initial density operator on the Hilbert space $\HC_A\otimes\HC_B$. The Kraus operators are arbitrary product operators satisfying the closure condition
\begin{equation}
\label{lcl_eqn5}
\sum_{m=0}^{M-1}A_m^\dagger A_m\otimes B_m^\dagger B_m=I_A\otimes I_B,
\end{equation}
with $I_A$ and $I_B$ the identity operators. The extension to multipartite systems is obvious, but here we will only consider the bipartite case. To avoid technical issues the sums in \eqref{lcl_eqn4} and \eqref{lcl_eqn5}, as well as the dimensions of $\HC_A$ and $\HC_B$, are assumed to be finite.

The local cloning protocol is described as follows.
Suppose Alice and Bob are two spatially separated parties, each holding a pair of quantum systems of dimension $D$, with Alice's systems described by a Hilbert space $\HC_A\otimes\HC_a$ and Bob's by $\HC_B\otimes\HC_b$. Let $\SC=\{\ket{\psi_i}^{AB}\}_{i=0}^{N-1}$ be a set of orthogonal bipartite entangled states on $\HC_A\otimes\HC_B$. Let $\ket{\phi}^{ab}\in\HC_a\otimes\HC_b$ be another bipartite entangled state that plays the role of a resource,  which we call the blank state, and is  shared in advance between Alice and Bob. Their goal is to implement deterministically (i.e. with probability one) the transformation
\begin{equation}
\label{lcl_eqn6}
\ket{\psi_i}^{AB}\otimes\ket{\phi}^{ab}\longrightarrow\ket{\psi_i}^{AB}\otimes\ket{\psi_i}^{ab}, \forall i=0\ldots N-1
\end{equation}
by a bipartite separable operation. Alice and Bob know exactly the states that belong to the set $\SC$ and also know the blank state $\ket{\phi}^{ab}$, but they do not know which state will be fed to the local cloning machine described by \eqref{lcl_eqn6} --- the machine has to work equally well for all states in $\SC$! Note that local cloning is defined up to local unitaries, i.e., a set $\SC=\{\ket{\psi_i}^{AB}\}_{i=0}^{N-1}$ can be locally cloned if and only if the set $\SC'=\{U^A\otimes V^B\ket{\psi_i}^{AB}\}_{i=0}^{N-1}$ can be locally cloned, where $U^A$ and $V^B$ are local unitaries. This is true because local unitaries can always be implemented deterministically at the beginning or at the end of the cloning operation.

The Schmidt coefficients of $\ket{\psi_i}^{AB}$ are labelled by $\lambda^{(i)}_r$ and by convention are sorted in decreasing order, with $\lambda^{(i)}_0 \geqslant \lambda^{(i)}_1\geqslant\cdots\geqslant\lambda^{(i)}_{D-1}$ and $\sum_{r=0}^{D-1}\lambda^{(i)}_r=1$, for all $i=0\ldots N-1$, and the Schmidt coefficients of $\ket{\phi}^{ab}$ are labelled by $\gamma_r$, with $\gamma_0\geqslant\gamma_1\cdots\geqslant\gamma_{D-1}$ and $\sum_{r=0}^{D-1}\gamma_r=1$. To remind the reader that the components of a vector $\vec{\lambda}$ are ordered in decreasing order we use the notation $\vec{\lambda}^{\downarrow}$.

The Schmidt rank of a bipartite state is the number of its non-zero Schmidt coefficients.
We say that a state of a $D\times D$ dimensional system has full Schmidt rank if its Schmidt rank is equal to $D$.

We use the concept of majorization, which is a partial ordering on $D$-dimensional real vectors. More precisely, if $\vec{x}=(x_0,\ldots,x_{D-1})$ and $\vec{y}=(y_0,\ldots,y_{D-1})$ are two real $D$-dimensional vectors, we say that $\vec{x}$ is majorized by $\vec{y}$ and write $\vec{x}\prec\vec{y}$ if and only if $\sum_{j=0}^{k}x_j^{\downarrow}\leqslant\sum_{j=0}^{k}y_j^{\downarrow}$ holds for all $k=0,\ldots, D-1$, with equality when $k=D-1$. 

For two $D\times D$ bipartite pure states ${|\chi\rangle}$ and ${|\eta\rangle}$, we use the shorthand notation ${|\chi\rangle\prec|\eta\rangle}$ to denote the fact that the vector of Schmidt coefficients of ${|\chi\rangle}$ is majorized by the vector of Schmidt coefficients of ${|\eta\rangle}$.
See \cite{PhysRevLett.83.436} or Chap. 12.5 of \cite{NielsenChuang:QuantumComputation} for more details about majorization.

The entanglement of a $D\times D$ bipartite pure state $\ket{\chi}$ can be quantified by various entanglement measures 
\footnote{Often called entanglement monotones, i.e., non-increasing under local operations and classical communication (LOCC).}, 
the ones used extensively in this chapter  being the \emph{entropy of entanglement}\index{entropy of entanglement}
\begin{equation}\label{lcl_eqn7}
E(\ket{\chi})=-\sum_{r=0}^{D-1}\lambda_r\log_D{\lambda_r}
\end{equation}
and the \emph{G-concurrence}\index{G-concurrence} \cite{PhysRevA.71.012318} 
\begin{equation}\label{lcl_eqn8}
C_G(\ket{\chi})=D\left(\prod_{r=0}^{D-1}{\lambda_r}\right)^{1/D},
\end{equation}
where $\lambda_r$ denotes the $r$-th Schmidt coefficient of $\ket{\chi}$. The base $D$ in the logarithm in \eqref{lcl_eqn7} as well as the prefactor $D$ in \eqref{lcl_eqn8} appear for normalization purposes, so that the entropy of entanglement as well as the $G$-concurrence of a maximally entangled state are both 1, regardless of the dimension.

\section{Characterizing sets of clonable states\label{lcl_sct3}}
\subsection{Preliminary analysis}
Mathematically, the local cloning problem can be formulated in terms of a separable transformation on a set of pure input states $\SC=\{\ket{\psi_i}^{AB}\}_{i=0}^{N-1}$, using a blank state $\ket{\phi}^{ab}$. 

If a set of states $\SC$ can be locally cloned using the blank state $\ket{\phi}^{ab}$, then there must exist a bipartite separable operation $\Lambda$ for which 
\begin{equation}
\label{lcl_eqn9}
\Lambda(\dya{\psi_i}^{AB}\otimes\dya{\phi}^{ab})=\dya{\psi_i}^{AB}\otimes\dya{\psi_i}^{AB},\text{ }\forall i=0\ldots N-1
\end{equation}
(note here that an overall phase factor in the definition of the individual states is of no significance). Since $\Lambda$ is separable, it can be represented by a set of product Kraus operators,
\begin{align}
\label{lcl_eqn10}
\sum_{m=0}^{M-1}(A_m\otimes B_m)(\dya{\psi_i}^{AB}\otimes\dya{\phi}^{ab})(A_m\otimes B_m)^{\dagger}\notag\\
=\dya{\psi_i}^{AB}\otimes\dya{\psi_i}^{AB},\text{ }\forall i=0\ldots N-1,
\end{align}
where operators $A_m$ act on $\HC_A\otimes\HC_a$, and $B_m$ on $\HC_B\otimes\HC_b$. The above equation is equivalent to
\begin{align}
\label{lcl_eqn11}
A_m&\otimes B_m(\ket{\psi_i}^{AB}\otimes\ket{\phi}^{ab})=\sqrt{p_{mi}}\mathrm{e}^{\mathrm{i}\varphi_{mi}}(\ket{\psi_i}^{AB}\otimes\ket{\psi_i}^{ab}),\notag\\
&\forall i=0\ldots N-1, \text{ }\forall m=0\ldots M-1
\end{align}
where $\mathrm{e}^{\mathrm{i}\varphi_{mi}}$ is a complex phase that may depend on $m$ and $i$, and $p_{mi}$ are probabilities for which
\begin{equation}
\label{lcl_eqn12}
\sum_{m=0}^{M-1} p_{mi}=1,\text{ }\forall i=0\ldots N-1.
\end{equation}

By map-state duality in the computational basis
\footnote{As an example of map-state duality, a bipartite state ${\ket{\chi}^{AB}\in\HC_A\otimes\HC_B}$, ${\ket{\chi}^{AB}=\sum c_{ij}\ket{i}^{A}\ket{j}^{B}}$, is transformed into a map ${\chi:\HC_B\longrightarrow\HC_A}$, ${\chi=\sum c_{ij}\ket{i}^{A}\bra{j}^{B}}$. Note that the rank of the operator ${\chi}$ is the Schmidt rank of ${\ket{\chi}^{AB}}$, and the squares of the singular values of ${\chi}$ (or, equivalently, the eigenvalues of ${\chi\chi^\dagger}$) are the Schmidt coefficients of ${\ket{\chi}^{AB}}$. For more details about map-state duality see Sec. II of \cite{PhysRevA.76.032310}.}
\cite{OSID.11.3,PhysRevA.73.052309,PhysRevA.76.032310,PhysRevA.78.020304} one can rewrite \eqref{lcl_eqn11} as
\begin{equation}
\label{lcl_eqn13}
A_m(\psi_i\otimes\phi)B^T_m=\sqrt{p_{mi}}\mathrm{e}^{\mathrm{i}\varphi_{mi}}\psi_i\otimes\psi_i, \text{ }\forall i,m,
\end{equation}
where $\psi_i$ and $\phi$ are now operators obtained from the corresponding kets by turning a ket into a bra, and $B^T_m$ is the transpose of $B_m$.

The superscripts in \eqref{lcl_eqn13} that label the Hilbert spaces have been dropped for clarity, since now one can regard everything as abstract linear operators, or matrices in the computational basis. 
Although map-state duality is basis-dependent, our results will not depend on the choice of a specific basis. 

We now state our first result characterizing sets of states $\SC$ that can be locally cloned.
\begin{theorem}[Rank of states in $\SC$]\label{lcl_thm1}
Let $\SC=\{\ket{\psi_i}^{AB}\}_{i=0}^{N-1}$ be a set of bipartite orthogonal states on $\HC_A\otimes\HC_B$ with one state, say $\ket{\psi_0}$, having full Schmidt rank. If the local cloning of $\SC$ is possible by a separable operation using a blank state $\ket{\phi}$, then $\ket{\phi}$ and all states in $\SC$ must be full rank.
\end{theorem}
\begin{proof}
This result follows directly from \eqref{lcl_eqn13}. If $\ket{\psi_0}$ has full Schmidt rank, then $\psi_0$ is a full rank operator. Then, since the rank of a tensor product is the product of ranks, it must be that $\phi$ is also full rank, as are $A_m$ and $B_m$ for each $m$ (a product of operators cannot have rank exceeding that of any of the individual operators in the product). Since the rank of a product of two operators is equal to that of the first whenever the second is full rank, \eqref{lcl_eqn13} with $i\ne0$ directly implies that $\psi_i$ has full rank for every $i$, and we are done.
\end{proof}
In this chapter, we are considering sets $\SC$ in which at least one state is full rank. Therefore by this theorem, we may instead restrict to sets in which every state is full rank, and we will do so throughout the remainder of the chapter.

As just argued in the proof of the previous theorem, all operators in \eqref{lcl_eqn13} are full rank, hence invertible. Now take the inverse of \eqref{lcl_eqn13}, replace $i$ by $j$, and right multiply \eqref{lcl_eqn13} by it to obtain
\begin{equation}
\label{lcl_eqn14}
A_m(\psi_i\psi_j^{-1}\otimes I)A_m^{-1}=\sqrt{\frac{p_{mi}}{p_{mj}}}\mathrm{e}^{\mathrm{i}(\varphi_{mi}-\varphi_{mj})}(\psi_i\psi_j^{-1}\otimes \psi_i\psi_j^{-1}).
\end{equation}
Define 
\begin{equation}\label{lcl_eqn15}
T_{ij}^{(m)}=\sqrt{\frac{p_{mi}}{p_{mj}}}\mathrm{e}^{\mathrm{i}(\varphi_{mi}-\varphi_{mj})}\psi_i\psi_j^{-1}
\end{equation}
so that \eqref{lcl_eqn14} can be written more compactly as
\begin{equation}
\label{lcl_eqn16}
A_m(T_{ij}^{(m)}\otimes I)A_m^{-1}=T_{ij}^{(m)}\otimes T_{ij}^{(m)}.
\end{equation}
Since for every $i$, $\psi_i$ is full rank, we see that $\det(\psi_i)\ne0$, so $\det(T_{ij}^{(m)})$ is also non-vanishing. Thus, taking the determinant on both sides of \eqref{lcl_eqn16} yields
\begin{equation}
\label{lcl_eqn17}
\det(T_{ij}^{(m)})^D=1,
\end{equation}
where we have used the fact that $\det(A\otimes B)=\det(A)^M\det(B)^N$, for $A$ and $B$ being $N\times N$ and $M\times M$ matrices, respectively.
Recalling the definition of $T_{ij}^{(m)}$ in \eqref{lcl_eqn15}, this condition becomes
\begin{equation}
\label{lcl_eqn18}
1=\left(\frac{p_{mi}}{p_{mj}}\right)^{D/2} \left|\frac{\det(\psi_i)}{\det(\psi_j)}\right|,
\end{equation}
or
\begin{equation}
\label{lcl_eqn19}
p_{mj}=p_{mi} \left|\frac{\det(\psi_i)}{\det(\psi_j)}\right|^{2/D}.
\end{equation}
Summing \eqref{lcl_eqn19} over $m$ yields
\begin{equation}
\label{lcl_eqn20}
\left|\det(\psi_i)\right|=\left|\det(\psi_j)\right|,
\end{equation}
implying
\begin{equation}
\label{lcl_eqn21}
p_{mi}=p_{mj},
\end{equation}
hence these determinants and probabilities are independent of the input state. As a consequence, we may write $T_{ij}^{(m)}$ in the simpler form,
\begin{equation}
\label{lcl_eqn22}
T_{ij}^{(m)} = \mathrm{e}^{\mathrm{i}(\varphi_{mi}-\varphi_{mj})}\psi_i\psi_j^{-1}.
\end{equation}

\textbf{Observation:}
The fact that $p_{mi}=p_{m}$, independent of $i$, implies that the cloning apparatus provides no information whatsoever about which state was input to that apparatus, nor can any such information ``leak" to an external environment that might be used to implement the local cloning separable operation. This is not without interest, since it rules out the possibility of local cloning by locally distinguishing while preserving entanglement \cite{PhysRevA.75.052313}. This result turns out to be valid in the much more general setting of one-to-one transformation of full Schmidt rank pure state ensembles by separable operations, but a discussion of these broader implications  will be presented in a future publication.

We can now provide additional conditions that must hold in order for $\SC$ to be a set of states that can be locally cloned by separable operations. These are stated in the following theorem, which holds under completely general conditions, applicable for any $N$ and $D$.
\begin{theorem}[Necessary conditions]\label{lcl_thm2}
Let $\SC=\{\ket{\psi_i}^{AB}\}$ be a set of full Schmidt rank bipartite orthogonal entangled states on $\HC_A\otimes\HC_B$. If the local cloning of $\SC$ using a blank state $\ket{\phi}^{ab}\in\HC_a\otimes\HC_b$ is possible by a separable operation, then the following must hold:
\begin{itemize}
\item[i)] 
All states in $\SC$ are equally entangled with respect to the $G$-concurrence measure,
\begin{equation}
\label{lcl_eqn23}
C_G(\ket{\psi_i}^{AB})=C_G(\ket{\psi_j}^{AB}),\text{ }\forall i,j.
\end{equation}

\item[ii)] Any two states in $\SC$ must either share the same set of Schmidt coefficients or  be incomparable under majorization.

\item[iii)] 
\begin{equation}\label{lcl_eqn24}
\mathrm{Spec}(T_{ij}^{(m)}\otimes I)=\mathrm{Spec}(T_{ij}^{(m)}\otimes T_{ij}^{(m)}),\text{ }\forall i,j,
\end{equation}
where Spec$(\cdot)$ denotes the spectrum of its argument and $T_{ij}^{(m)}$ is defined as in \eqref{lcl_eqn22}.
\end{itemize}
\end{theorem}

\begin{proof}
Proof of i) This follows at once from \eqref{lcl_eqn20}, the definition \eqref{lcl_eqn8} of $G$-concurrence, and the fact that for any state $\ket{\chi}$ the product of its Schmidt coefficients is equal to $|\det(\chi)|^2$.

Proof of ii) The proof follows from Theorem 1 (ii,iii) of \cite{PhysRevA.76.032310} which states that any two bipartite states $\ket{\chi}$ and $\ket{\eta}$ that are comparable under majorization (i.e. $\ket{\chi}\prec\ket{\eta}$ or $\ket{\eta}\prec\ket{\chi}$) and have equal $G$-concurrence must share the same set of Schmidt coefficients.

Proof of iii) The proof follows at once from \eqref{lcl_eqn16}.
\end{proof}

\subsection{Characterization of clonable sets in terms of finite groups}
\label{lcl_sct1a}
We next show that to any set $\SC$ of states that can all be cloned by the same apparatus, there can be associated a finite group, and the set is essentially generated by this group.
\begin{theorem}[Group structure of $\SC$]\label{lcl_thm3}
Let $\SC=\{\ket{\psi_i}^{AB}\}$ be a set of full Schmidt rank bipartite orthogonal entangled states on $\HC_A\otimes\HC_B$. If the local cloning of $\SC$ is possible by a separable operation, then the set $\SC$ can be extended to a larger set such that $\{T_{ij}^{(m)}\}$ of  \eqref{lcl_eqn22} for fixed $j,m$ constitutes an ordinary representation of a finite group, $G$. Since the states in $\SC$ are related as $\mathrm{e}^{\mathrm{i}\varphi_{mi}}|\psi_i\rangle=\mathrm{e}^{\mathrm{i}\varphi_{mj}}(T_{ij}^{(m)}\otimes I_B)|\psi_j\rangle$, then the larger set, with $N=|G|$ members, is generated by the action of the group $G$ on any individual state in the set.
\end{theorem}
\begin{proof}
The starting point of the proof is to multiply \eqref{lcl_eqn16} on the left of \eqref{lcl_eqn13} (with index $k$) to obtain
\begin{equation}
\label{lcl_eqn25}
A_m(T_{ij}^{(m)}\psi_k\otimes\phi)B^T_m=\sqrt{p_{m}}\mathrm{e}^{\mathrm{i}\varphi_{mk}}T_{ij}^{(m)}\psi_k\otimes T_{ij}^{(m)}\psi_k. 
\end{equation}
Using \eqref{lcl_eqn22} this becomes
\begin{equation}
\label{lcl_eqn26}
A_m(\psi_i\psi_j^{-1}\psi_k\otimes\phi)B^T_m
=\sqrt{p_{m}}\mathrm{e}^{\mathrm{i}(\varphi_{mi}-\varphi_{mj}+\varphi_{mk})}\psi_i\psi_j^{-1}\psi_k\otimes \psi_i\psi_j^{-1}\psi_k, 
\end{equation}
which by map-state duality implies that the state $\ket{\psi_i\psi_j^{-1}\psi_k}$ is cloned by the same apparatus as all the states in the original set $\SC$. Therefore $\ket{\psi_i\psi_j^{-1}\psi_k}$ --- which, by considering the version of \eqref{lcl_eqn26} that corresponds to states (as in \eqref{lcl_eqn11}), taking the squared norm of both sides and summing over $m$, is seen to be normalized --- must either (i) be orthogonal to the entire set $\SC$, or (ii) it is equal to one of those original states up to an overall phase factor. If this state is orthogonal to $\SC$, then $\SC$ can be extended by including this state as one of its members. So assume $\SC$ has been extended to its maximal size (since we are working in finite dimensions, this size will be finite), and then we can conclude that for every $i,j,k$,
\begin{equation}
\label{lcl_eqn27}
\psi_i\psi_j^{-1}\psi_k=\mathrm{e}^{\mathrm{i}(\varphi_{ml}-\varphi_{mi}+\varphi_{mj}-\varphi_{mk})}\psi_l,
\end{equation}
for some $l$, where the phase in the above expression has been determined by comparing \eqref{lcl_eqn26} to \eqref{lcl_eqn13}. Next multiply this latter expression on the right by $\mathrm{e}^{-\mathrm{i}\varphi_{mn}}\psi_n^{-1}$ to obtain
\begin{equation}
\label{lcl_eqn28}
T_{ij}^{(m)}T_{kn}^{(m)}=T_{ln}^{(m)}.
\end{equation}
Hence the collection of $T_{ij}^{(m)}$ is closed under matrix multiplication, which is associative. In addition, $T_{ii}^{(m)}=I$ for every $i$ and $T_{ij}^{(m)}T_{ji}^{(m)}=I$ for every $i,j$, so we see that the identity element and inverses are present, which concludes the proof that the set $\{T_{ij}^{(m)}\}$ with fixed $m$ form a representation of a group, $G$. Now, the number of index pairs $(i,j)$ is $N^2$, where $N$ is the number of states in $\SC$. However, we will now show that in fact the order $|G|$ of this group is equal to $N$ and not $N^2$.

Setting $n=j$ in \eqref{lcl_eqn28}, we have
\begin{equation}
\label{lcl_eqn29}
T^{(m)}_{ij}T^{(m)}_{kj}=T^{(m)}_{lj},
\end{equation}
so the product is closed even when the second index is constrained to be the same. If we set $l=j$, we see that with $T^{(m)}_{jj}=I$, then for each $i$ there exists $k$ such that $T^{(m)}_{kj}=(T^{(m)}_{ij})^{-1}$. Hence, for every fixed $j$ the set $\TC_j=\{T^{(m)}_{ij}\}$ also is a representation of $G$. Similarly, one can show the same holds if instead it is the first index that is held fixed.
Note now that by multiplying \eqref{lcl_eqn28} on the right by $(T_{kn}^{(m)})^{-1}$,
and given that \eqref{lcl_eqn28} holds for any $i,j,k,n$, we see that for every $i,j$, $T_{ij}$ is a member of the group formed by the $T_{kn}$ for fixed $n$. That is, the group of the $T_{kn}$ for fixed $n$ contains all elements $T_{ij}$.

Could two or more of the $T^{(m)}_{ij}$ be equal, for fixed $j$? We will now show this is not the case by demonstrating the linearly independence of the set $\TC_j$. Indeed,
\begin{align}
\label{lcl_eqn30}
0&=\sum_{k=0}^{N-1}c_kT^{(m)}_{kj}=\sum_{k=0}^{N-1}c_k\mathrm{e}^{\mathrm{i}(\varphi_{mk}-\varphi_{mj})}\psi_k\psi_j^{-1}\notag\\
&\Longleftrightarrow0=\sum_{k=0}^{N-1}c_k\mathrm{e}^{\mathrm{i}\varphi_{mk}}\psi_k.
\end{align}
However, the $\psi_k$ are mutually orthogonal, $\Tr(\psi_k^\dagger\psi_j)=\delta_{jk}$, so this can only be satisfied if all the $c_k$ vanish, implying that $\TC_j$ is linearly independent, and hence, that $|G|=N$: the (maximal) number of states in $\SC$ is equal to the order of $G$.
\end{proof}
For the remainder of the chapter, we will use labels $f,g,h$ instead of $i,j,k$, where the former represent elements of the group $G$; the group multiplication is denoted as $fg$, with $e$ the identity element. For example, instead of $\psi_0$ we will now write $\psi_e$, and in place of $T^{(m)}_{j0}$ we will simply write $T^{(m)}_f$.

We may now utilize the powerful tools of group theory to study sets $\SC$ of clonable states, obtaining a very strong constraint on how many states any given apparatus can possibly clone. Any group $G$ is characterized by its irreducible representations, which we denote as $\Gamma^{(\alpha)}(f),~f\in G$, and any representation of $G$ may be decomposed into a direct sum of irreducible representations with a given irreducible representation $\Gamma^{(\alpha)}(f)$ appearing some number $n_\alpha$ times in that sum. In general, a given representation may have $n_\alpha=0$ for some $\alpha$, but since here our representation is linearly independent, we know that every irreducible representation must appear at least once \cite{PhysRevA.81.062315}.

We can use character theory \cite{Hamermesh:GroupTheory} to calculate $n_\alpha$. Defining characters as $\chi(T^{(m)}_{f})=\Tr(T^{(m)}_{f})$ and $\chi^{(\alpha)}(f)=\Tr(\Gamma^{(\alpha)}(f))$, we have that
\begin{equation}
\label{lcl_eqn31}
n_\alpha=\frac{1}{|G|}\sum_{f\in G}\chi^{(\alpha)}(f)^\ast\chi(T^{(m)}_{f}).
\end{equation}
However, by taking the trace of \eqref{lcl_eqn16} and recalling that the trace of a tensor product is equal to the product of the traces, we see that $\chi(T^{(m)}_{f})$ is equal to either $0$ or $D$. Since the identity element of the group, $e$, is always in a conjugacy class by itself ($e=gfg^{-1}$ if and only if $f=e$), then we may conclude that $\chi(T^{(m)}_{f})$ vanishes except when $f=e$, in which case $\chi(T^{(m)}_{e})=D$. Thus, \eqref{lcl_eqn31} reduces to
\begin{equation}
\label{lcl_eqn32}
n_\alpha=\frac{Dd_\alpha}{|G|},
\end{equation}
where $d_\alpha=\chi^{(\alpha)}(e)$ is the dimension of the $\alpha^{\textrm{th}}$ irreducible representation. Since in every finite group there is always the trivial irreducible representation of all ones, $\Gamma^{(t)}(f)=1~\forall {f\in G}$, where this irreducible representation has dimension $d_t=1$, we have immediately that $n_t=D/|G|$ is an integer, implying that $N=|G|$ divides $D$. Thus,
\begin{theorem}[Number of clonable states]\label{lcl_thm4}
If an apparatus can locally clone more than one state on a $D\times D$ system, where at least one (and therefore all, see Theorem~\ref{lcl_thm1}) of these states has full Schmidt rank, then that apparatus can in fact clone a number of states that divides $D$ exactly. In particular if $D$ is prime, then any such apparatus can clone exactly $D$ states, no more and no less.
\end{theorem}

Now we see from \eqref{lcl_eqn32} that $n_\alpha$ is an integer multiple of $d_\alpha$. If $|G|=D$ so that $n_\alpha=d_\alpha$, we have what is known as the regular representation of $G$. Otherwise, our representation is a direct sum of an integer number $n_t=D/|G|$ of copies of the regular representation. As is well known, there is always a choice of basis in which the matrices in a \textit{unitary} regular representation appear as permutation matrices $L(f)$, with each row (column) having only a single non-zero entry equal to one. In this basis, denoted as $\{|g\rangle\}_{g\in G}$, we have that $L(f)|g\rangle=|fg\rangle$. The representation $L(f)$ is called the \emph{left regular representation}\index{regular representation, left}. One can as well use the \emph{right regular representation}\index{regular representation, right} $R(f)$ with $R(f)|g\rangle=|gf^{-1}\rangle$, but without loss of generality in the rest of the chapter we restrict only to $L(f)$, since for finite groups the right and left regular representations are equivalent \cite{Hazewinkel:Encyclopaedia}.

In our case the representation will generally not be unitary, so when $|G|=D$ we will have that
\begin{equation}
\label{lcl_eqn33}
T^{(m)}_{f}=SL(f)S^{-1},
\end{equation}
for some invertible matrix $S$. 

In the remainder of the chapter we restrict consideration to $|G|=D$ (or, equivalently, to $n_t=1$), and note that all results obtained in the remainder of the chapter are valid (with small modifications) also when $|G|<D$. However, the notation becomes a bit cumbersome, so we defer detailed discussion about the $|G|<D$ case to Appendix~\ref{lcl_apdx5}.

\subsection{Form of the clonable states when all are maximally entangled}
It was shown in \cite{NewJPhys.6.164} that when at least one of the states in $\SC$ is maximally entangled, then all states in $\SC$ must also be maximally entangled. In this section, we consider such sets, in which case the $T^{(m)}_f$ must all be unitary. This follows directly from the fact that when $\psi_e$ is proportional to the identity then $\psi_f$ is proportional to $T^{(m)}_f$, and also that $\ket{\psi_f}$ is maximally entangled if and only if $\psi_f$ is proportional to a unitary.

We have seen that when $N=D$, then $T^{(m)}_f=SL(f)S^{-1}$ for some invertible $S$, and $L(f)$ is the permutation form of the regular representation of group $G$. However, we have
\begin{lemma}[Unitary equivalence]\label{lcl_thm5}
For any two unitary representations $T_f$ and $L(f)$ of a finite group $G$, which are equivalent in the sense that $T_f=SL(f)S^{-1}$ for some invertible matrix $S$, then these two representations are also equivalent by a unitary similarity transformation, $T_f=WL(f)W^\dagger$, with $W$ unitary.
\end{lemma}
A proof of this lemma is given in Chap. 3.3 of \cite{Ma:GroupThPhys}, and we provide an alternative proof in Appendix~\ref{lcl_apdx1}.

What this lemma tells us is that $\psi_f$ is proportional to $WL(f)\psi_eW^\dagger$ (since by local unitaries, $\psi_e$ can be made proportional to the identity, we will assume here that this is the case, and then $\psi_e$ commutes with $W^\dagger$), or 
\begin{align}\label{lcl_eqn34}
\ket{\psi_f}&=c_f(W L(f)\otimes W^\ast)\sum_{g\in G}\ket{g}^A\ket{g}^B\notag\\
&=\frac{1}{\sqrt{D}}(W\otimes W^\ast)\sum_{g\in G}\ket{fg}^A\ket{g}^B,
\end{align}
where $W^\ast$ is the complex conjugate of $W$, the states $\{\ket{g}\}_{g\in G}$ are some orthonormal basis, $\ip{g}{h}=\delta_{g,h}$, and we have omitted an unimportant overall phase (from $c_f$, of magnitude $D^{-1/2}$) in the last line. 
Note that up to unimportant local unitaries and relabeling of group elements, the set of states \eqref{lcl_eqn34} can be written either as
\begin{equation}\label{lcl_eqn35}
\ket{\psi_f}=\frac{1}{\sqrt{D}}\sum_{g\in G}\ket{fg}^A\ket{g}^B
\end{equation}
or
\begin{equation}\label{lcl_eqn36}
\ket{\psi_f}=\frac{1}{\sqrt{D}}\sum_{g\in G}\ket{g}^A\ket{fg}^B.
\end{equation}
The states above are of a form that we will refer to as ``group-shifted".

In Section~\ref{lcl_sct5}, we provide an explicit LOCC protocol that accomplishes cloning of such shifted sets of states. Thus, we have
\begin{theorem}[Maximally entangled states]\label{lcl_thm6}
A set of maximally entangled states on a $D\times D$ system can be cloned by LOCC if and only if there exists a choice of Schmidt bases shared by those states such that they have a group-shifted form, as in \eqref{lcl_eqn35} or \eqref{lcl_eqn36}.
\end{theorem}
\noindent This extends the result of \cite{PhysRevA.74.032108}, which applied only for prime $D$.

Additionally, we remark that in our protocol presented in Sec.~\ref{lcl_sct5}, there is no need for classical communication (the measurement $M_r$ and the additional corrections $Q_r$ appearing in that protocol can be omitted when the states to be cloned are maximally entangled). This result was first proven in \cite{NewJPhys.6.164}, where it was shown that the Kraus operators implementing the cloning of maximally entangled states have to be proportional to unitary operators. A completely different proof of this fact was later provided in \cite{PhysRevA.76.032310}, in which it was shown that a separable operation that maps a pure state to another pure state, both sharing the same set of Schmidt coefficients, must have its Kraus operators proportional to unitaries; in our case $\ket{\psi_f}\otimes\ket{\phi}$ and $\ket{\psi_f}\otimes\ket{\psi_f}$ do share the same set of Schmidt coefficients, since they are maximally entangled. We here have another simple proof of this result, since we have proved in Theorem \ref{lcl_thm6} that a set of maximally entangled states must be group-shifted in order that they can be cloned, and since our protocol in Sec.~\ref{lcl_sct5} clones any set that is group-shifted without using communication.

\subsection{Form of the clonable states when $D=2$ (qubits)}
Here, we restrict our attention to local cloning of qubit entangled states, $D=2$. As $D$ is prime, we know from Theorem~\ref{lcl_thm3} that exactly two states can be cloned, $\SC=\{ \ket{\psi_e}^{AB}, \ket{\psi_g}^{AB}\}$. Both are assumed to be entangled (non-product), but not maximally entangled.

Since there is only one independent Schmidt coefficient for a two-qubit state, any two such states are comparable under majorization, and then from part ii) of Theorem~\ref{lcl_thm2} it follows at once that these states have to share the same set of Schmidt coefficients. This is already a surprising result, implicitly assumed (but not proved) in recent work on local cloning of qubit states \cite{PhysRevA.76.052305}. We can actually prove a stronger condition: not only do the states have to share the same set of Schmidt coefficients, but  they must also share the same Schmidt basis and be of a shifted form, as summarized by the following theorem.

\begin{theorem}[Entangled qubits]\label{lcl_thm7}
Let $\SC=\{\ket{\psi_e}^{AB}, \ket{\psi_g}^{AB}\}$ be a set of 2 orthogonal two-qubit entangled states and let $\lambda$ be the largest Schmidt coefficient of $\ket{\psi_e}^{AB}$, assumed to satisfy $1/2<\lambda<1$. If the local cloning of $\SC$ using a two-qubit entangled blank state $\ket{\phi}^{ab}$ is possible by a separable operation, then, up to local unitaries (that is, the same local unitaries acting on both states), the states must either  be of the form
\begin{align}\label{lcl_eqn37}
\ket{\psi_e}^{AB}&=\sqrt{\lambda}\ket{0}^A\ket{0}^{B}+\sqrt{1-\lambda}\ket{1}^{A}\ket{1}^{B}\notag\\
\ket{\psi_g}^{AB}&=\sqrt{\lambda}\ket{0}^{A}\ket{1}^{B}+\sqrt{1-\lambda}\ket{1}^{A}\ket{0}^{B}
\end{align}
or
\begin{align}\label{lcl_eqn38}
\ket{\psi_e}^{AB}&=\sqrt{\lambda}\ket{0}^{A}\ket{0}^{B}+\sqrt{1-\lambda}\ket{1}^{A}\ket{1}^{B}\notag\\
\ket{\psi_g}^{AB}&=\sqrt{\lambda}\ket{1}^{A}\ket{0}^{B}+\sqrt{1-\lambda}\ket{0}^{A}\ket{1}^{B}.
\end{align}
\end{theorem}
\noindent Note that a relative phase $\mathrm{e}^{\mathrm{i}\vartheta}$ may be introduced into $\ket{\psi_g}$, without altering $\ket{\psi_e}$, by Alice and Bob doing local unitaries on systems $A$ and $B$, $U^{\!A,B}=\dya{0}+\mathrm{e}^{\pm\mathrm{i}\vartheta/2}\dya{1}$ (one of them chooses the upper sign, the other does the lower, which accomplishes the task up to an unimportant overall phase). Therefore, the theorem allows cloning of states with these phases.

\begin{proof}
First note that without loss of generality one can always assume that the first state $\ket{\psi_e}^{AB}$ is already in Schmidt form, 
\begin{equation}\label{lcl_eqn39}
\ket{\psi_e}^{AB}=\sqrt{\lambda}\ket{0}^{A}\ket{0}^{B}+\sqrt{1-\lambda}\ket{1}^{A}\ket{1}^{B},
\end{equation}
since this can be done by a local unitary map $U^{Aa}\otimes V^{Bb}$. Therefore, the operators $\psi_e$ and $\psi_g$ obtained by map-state duality can be assumed to have the form
\begin{align}
\label{lcl_eqn40}
\psi_e&=\left(
\begin{array}{cc}
\sqrt{\lambda} & 0\\
0 & \sqrt{1-\lambda}
\end{array}
\right)
,\\
\label{lcl_eqn41}
\psi_g&=\left(
\begin{array}{cc}
a_{00} & a_{01}\\
a_{10} & a_{11}
\end{array}
\right),
\end{align}
where $\lambda$ is the largest Schmidt coefficient of $\ket{\psi_e}^{AB}$ and $a_{ij}$ are complex numbers with $\sum |a_{ij}|^2=1$, which is equivalent to the requirement that $\ket{\psi_g}$ be normalized. 

Orthogonality between these two states implies that
\begin{equation}\label{lcl_eqn42}
0=\sqrt{\lambda}a_{00}+\sqrt{1-\lambda}a_{11}.
\end{equation}
Since the only group of order $2$ is cyclic with elements $e,g$ and $g^2=e$, we have from Theorem~\ref{lcl_thm3} that $(T^{(m)}_{g})^2=SL(g)^2S^{-1}=I$. Thus, we require
\begin{align}
\label{lcl_eqn43}
(\psi_g\psi_e^{-1})^2&=
\left(
\begin{array}{cc}
\mathrm{e}^{\mathrm{i}\vartheta} & 0\\
0 & \mathrm{e}^{\mathrm{i}\vartheta}
\end{array}
\right),
\end{align}
where the factor of $\mathrm{e}^{\mathrm{i}\vartheta}$ arises from the phases that appear in the definition of $T^{(m)}_g$, see \eqref{lcl_eqn22}. Thus, \eqref{lcl_eqn43} implies
\begin{equation}\label{lcl_eqn44}
\frac{a_{00}^2}{\lambda}=\frac{a_{11}^2}{1-\lambda}=\mathrm{e}^{\mathrm{i}\vartheta}-\frac{a_{01}a_{10}}{\sqrt{\lambda(1-\lambda)}},
\end{equation}
and either (i) $a_{00}\sqrt{1-\lambda}=-a_{11}\sqrt{\lambda}$; or (ii) $a_{01}=0=a_{10}$. The condition that $\psi_g$ be normalized in the latter case (ii), along with \eqref{lcl_eqn42} and \eqref{lcl_eqn44}, can only be satisfied if $\lambda=1/2$, a case we are not considering here. The former case (i) along with \eqref{lcl_eqn42} implies that $a_{00}=0=a_{11}$ (again, assuming $\lambda\ne 1/2$). This concludes the proof, since it implies that $\ket{\psi_g}^{AB}$ has to have either the form \eqref{lcl_eqn37} or the form \eqref{lcl_eqn38}, up to an unimportant global phase.
\end{proof}

Now one can immediately see that one of the families of states considered in \cite{PhysRevA.76.052305}, of the form 
$\ket{\psi_e}=\sqrt{\lambda}\ket{0}^A\ket{0}^B+\sqrt{1-\lambda}\ket{1}^A\ket{1}^B$ and $\ket{\psi_g}=\sqrt{1-\lambda}\ket{0}^A\ket{0}^B-\sqrt{\lambda}\ket{1}^A\ket{1}^B$ cannot be locally cloned with a blank state of Schmidt rank 2, unless they are maximally entangled, case already studied in \cite{NewJPhys.6.164}.

\section{Local cloning of group-shifted states: explicit protocol using a maximally entangled blank state}\label{lcl_sct5}

Consider now a set of group-shifted partially entangled states $\SC=\{\ket{\psi_f}^{AB}\}_{f\in G}$ on $\HC_A\otimes\HC_B$, where the dimension of both Hilbert spaces $\HC_A$ and $\HC_B$ is equal to $D$,
\begin{equation}
\label{lcl_eqn45}
\ket{\psi_f}^{AB}=\sum_{g\in G}\sqrt{\lambda_g}\ket{g}^{A}\ket{fg}^{B},
\end{equation}
and we remind the reader that throughout this section we restrict to the $|G|=D$ case (see Appendix~\ref{lcl_apdx5} for the $|G|<D$ case).

In the following we present a protocol that locally clones $\SC$ using a maximally entangled blank state of Schmidt rank $D$. Our protocol, which works for any group $G$, is a direct generalization of the one presented for the special case of a cyclic group in \cite{PhysRevA.73.012343}.
\begin{theorem}[Group shifted states]\label{lcl_thm8}
Let $\SC=\{\ket{\psi_f}^{AB}\}_{f\in G}$ be a set of group-shifted full Schmidt rank bipartite orthogonal entangled states on $\HC_A\otimes\HC_B$ as defined by \eqref{lcl_eqn45}. The local cloning of $\SC$ is always possible using a maximally entangled blank state $\ket{\phi}^{ab}$ of Schmidt rank $D$.
\end{theorem}
\begin{proof}
Without loss of generality the maximally entangled blank state can be written as
\begin{equation}\label{lcl_eqn46}
\ket{\phi}^{ab}=\frac{1}{\sqrt{D}}\sum_{h\in G}\ket{h}^a\ket{h}^b.
\end{equation}
The local cloning protocol is summarized below and the quantum circuit is displayed in Fig.~\ref{lcl_fgr1}. 
\begin{figure}
\begin{center}
\includegraphics{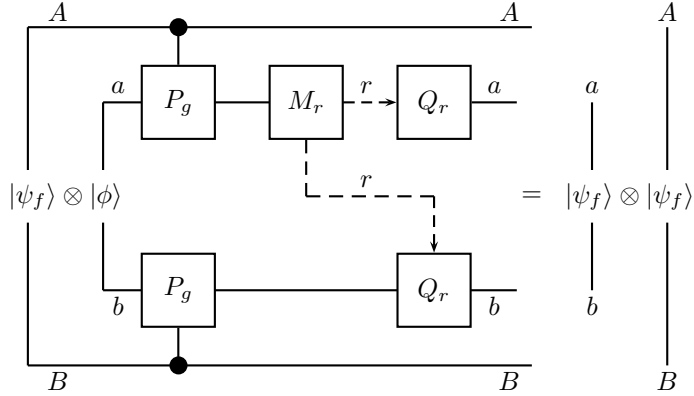}
\caption{Circuit diagram for the local cloning of group-shifted states with a maximally entangled blank state. There is no need to perform the measurement $M_r$ and the corrections $Q_r$ whenever the states to be cloned are maximally entangled.}
\label{lcl_fgr1}
\end{center}
\end{figure}
\begin{enumerate}
	\item Starting with $\ket{\psi_f}^{AB}\otimes\ket{\phi}^{ab}$, both Alice and Bob apply the ``controlled-group" unitary 
	\begin{equation}\label{lcl_eqn47}
		\sum_{g\in G} \dyad{g}{g} \otimes P_g,\quad\text{with } P_g=\sum_{h\in G}\dyad{gh}{h},
	\end{equation}
	where the permutation $P_g$ acts on system $a$ ($b$) and is controlled by system $A$ ($B$), to obtain
	\begin{align}\label{lcl_eqn48}
		&\sum_{g\in G}\sqrt{\lambda_g}\ket{g}^A\ket{fg}^B
		\frac{1}{\sqrt{D}}\sum_{h\in G}\ket{gh}^a\ket{fgh}^b\notag\\
		&=\sum_{g\in G}\sqrt{\lambda_g}\ket{g}^A\ket{fg}^B
		\frac{1}{\sqrt{D}}\sum_{h\in G}\ket{h}^a\ket{fh}^b.
	\end{align}
	\item Next Alice performs a generalized measurement on system $a$ with Kraus operators
	\begin{equation}\label{lcl_eqn49}
			M_{r}=\sum_{h\in G}\sqrt{\lambda_{hr}}\dyad{h}{h},\quad \sum_{r\in G}{M_{r}}^{\dagger}M_{r}=I,
	\end{equation}
	and communicates the result $r$ to Bob. Conditioned on the result $r$, the output state is
	\begin{equation}\label{lcl_eqn50}
			\sum_{g\in G}\sqrt{\lambda_g}\ket{g}^A\ket{fg}^B
		\sum_{h\in G}\sqrt{\lambda_{hr}}\ket{h}^a\ket{fh}^b.
	\end{equation}
	\item Both Alice and Bob apply the unitary correction 
	\begin{equation}\label{lcl_eqn51}
		Q_r=\sum_{h\in G}\dyad{hr}{h}
	\end{equation}
	on systems $a$ and $b$, respectively,
	to obtain
	\begin{align}\label{lcl_eqn52}
		&\sum_{g\in G}\sqrt{\lambda_g}\ket{g}^A\ket{fg}^B
		\sum_{h\in G}\sqrt{\lambda_{hr}}\ket{hr}^a\ket{fhr}^b\notag\\
		&=\sum_{g\in G}\sqrt{\lambda_g}\ket{g}^A\ket{fg}^B
		\sum_{h\in G}\sqrt{\lambda_{h}}\ket{h}^a\ket{fh}^b\notag\\
		&=\ket{\psi_f}^{AB}\otimes\ket{\psi_f}^{ab},
	\end{align}
	which is the desired output.
\end{enumerate}
\end{proof}

Note that from symmetry considerations states of the form $\sum_{g\in G}\sqrt{\lambda_g}\ket{fg}^{A}\ket{g}^{B}$ (with the term $fg$ appearing now on Alice's side instead of Bob's side) can also be locally-cloned, by interchanging the roles of Alice and Bob in the protocol, e.g. performing the measurement $M_r$ on system $b$  instead of $a$, then sending the result back to $a$. Therefore in the following, when discussing group-shifted states, we will restrict to the states of the form \eqref{lcl_eqn45}.

\section{Local cloning of group-shifted states: minimum entanglement of the blank\label{lcl_sct6}}
Here again, we restrict for simplicity to the $|G|=D$ case, and discuss the extension of the results for $|G|<D$ in Appendix~\ref{lcl_apdx5}.
\subsection{Necessary conditions for arbitrary $D$}
We now turn our attention to the task of characterizing the blank state, which essentially amounts to determining the amount of entanglement it must have in order for the local cloning to be possible. We first give a very general lower bound as,
\begin{theorem}[Minimum entanglement of the blank]\label{lcl_thm9}
Let $\SC=\{\ket{\psi_f}^{AB}\}_{f\in G}$ be a set of full Schmidt rank bipartite orthogonal entangled states on $\HC_A\otimes\HC_B$. If the local cloning of $\SC$ using a blank state $\ket{\phi}^{ab}\in\HC_a\otimes\HC_b$ is possible by a separable operation, then it must be that
\begin{equation}\label{lcl_eqn53}
Ent(\ket{\phi}^{ab})\geqslant\max_{f\in G} Ent(\ket{\psi_f}^{AB}),
\end{equation}
where $Ent(\cdot)$ denotes any pure-state entanglement measure.
\end{theorem}

\begin{proof} We recently proved in \cite{PhysRevA.78.020304} that any pure state entanglement monotone is non-increasing on average under the general class of separable operations. The theorem follows directly, since otherwise the local cloning machine increases entanglement across the $Aa/Bb$ cut.
\end{proof}

Providing a more detailed lower bound appears to be difficult in general, but turns out to be possible in the special case of group-shifted states.

Consider again the set of $D$ group-shifted entangled states \eqref{lcl_eqn45}, and allow for arbitrary phases, $\vartheta_{\!f\!,g}$,
\begin{equation}
\label{lcl_eqn54}
\ket{\psi_f}^{AB}=\sum_{g\in G}\sqrt{\lambda_g}\mathrm{e}^{\mathrm{i}\vartheta_{\!f\!,g}}\ket{g}^{A}\ket{fg}^{B}.
\end{equation}
Without loss of generality, the blank state $\ket{\phi}^{ab}$ can be written as
\begin{equation}\label{lcl_eqn55}
\ket{\phi}^{ab}=\sum_{h\in G}\sqrt{\gamma_h}\ket{h}^a\ket{h}^b,
\end{equation}
where $\gamma_h$ are its Schmidt coefficients, $\sum_{h\in G}{\gamma_h}=1$.

All states in $\SC$ have the same Schmidt coefficients, and hence the same entanglement. As shown above, the local cloning of the above set of states is possible using a maximally entangled blank state when all phases $\mathrm{e}^{\mathrm{i}\vartheta_{\!f\!,g}}$ are chosen to be $1$, but it is not yet known if one can accomplish this task using less entanglement. One might hope that the local cloning of $\SC$ is possible using a blank state having the same entanglement as each of the states in $\SC$, which could be regarded as an ``optimal" local cloning. However we prove below that such an optimal local cloning is impossible with these states. Indeed we find a sizeable gap between the entanglement needed in the blank state and the entanglement of the states of $\SC$. For $D=2$ and $D=3$, we prove that a maximally entangled blank state is \emph{always} necessary. 

In the rest of this section we will use the \emph{rearrangement inequality}\index{rearrangement inequality} (see Chap. X of \cite{Hardy:Inequalities}), which states that 
\begin{equation}\label{lcl_eqn56}
x_ny_1+\cdots +x_1y_n\leqslant x_{\sigma(1)}y_1+\cdots+x_{\sigma(n)}y_n\leqslant x_1y_1+\cdots +x_ny_n
\end{equation}
for every choice of real numbers $x_1\leqslant\cdots\leqslant x_n$ and $y_1\leqslant\cdots\leqslant y_n$ and every permutation $x_{\sigma(1)},\ldots,x_{\sigma(n)}$ of $x_1,\ldots,x_n$.

The following Lemma is the most important technical result of this section (note that in the statement of this result, we will use $\overline g$ for inverses $g^{-1}$ of elements in the group $G$, which will make the notation somewhat more readable). 

\begin{lemma}[Majorization conditions]\label{lcl_thm10}
Let $\SC=\{\ket{\psi_f}^{AB}\}_{f\in G}$ be a set of $D$ group-shifted full Schmidt rank bipartite orthogonal entangled states on $\HC_A\otimes\HC_B$ as defined by \eqref{lcl_eqn54} and considered to be not maximally entangled. 
If the local cloning of $\SC$ using a blank state $\ket{\phi}^{ab}$ is possible by a separable operation, then
\begin{itemize}
\item[i)] The majorization condition,
\begin{equation}\label{lcl_eqn57}
\vec{\alpha}\prec\vec{\beta},
\end{equation}
must hold. Here, $\vec{\alpha}$ and $\vec{\beta}$ are  vectors with $D^2$ components indexed by elements $g,h\in G$,
\begin{align}\label{lcl_eqn58}
\alpha_{g,h}=\gamma_h\sum_{f\in G}\mu_{\overline f}\lambda_{f g},\quad
\beta_{g,h}=\sum_{f\in G}\mu_{\overline f}\lambda_{fg}\lambda_{fh},
\end{align}
and $\{\mu_f\}_{f\in G}$ is an arbitrary set of non-negative real coefficients that satisfy $\sum_{f}\mu_f=1$.

\item[ii)]
The smallest Schmidt coefficient $\gamma_{\min}$ of the blank state has to satisfy
\begin{equation}\label{lcl_eqn59}
\gamma_{\min}\geqslant
\max_{\{\mu_f\}}\frac{\min_{g,h\in G}\sum_{f\in G}\mu_{\overline f}\lambda_{fg}\lambda_{fh}}
{\min_{g\in G}\sum_{f\in G}\mu_{\overline f}\lambda_{fg}}.
\end{equation}

\item[iii)] In particular, a good choice of $\{\mu_f\}$  is given by
\begin{equation}\label{lcl_eqn60}
\mu_f=\frac{\eta}{\lambda_{\overline f}},\quad\text{with }\eta^{-1}=\sum_{g\in G}1/\lambda_g,
\end{equation}
for which \eqref{lcl_eqn59} becomes
\begin{equation}\label{lcl_eqn61}
\gamma_{\min}\geqslant\frac{1}{D}\min_{g,h\in G}\sum_{f\in G}\frac{1}{\lambda_{f}}\lambda_{fg}\lambda_{fh}.
\end{equation}
\end{itemize}
\end{lemma}

The majorization relation \eqref{lcl_eqn57} restricts the possible allowed Schmidt coefficients for the blank state and can easily be checked numerically, but an analytic expression is difficult to find, since there is no simple way of ordering \eqref{lcl_eqn58}. That is why parts ii) and iii) of the Lemma have their importance, since they focus only on the smallest Schmidt coefficient of the blank state. In particular, the bound iii) is crucial in deriving the necessity of a maximally entangled blank state for the local cloning of qubit and group-shifted qutrit states.

The proof of the Lemma is rather technical and is presented in Appendix~\ref{lcl_apdx2}. However, the main idea of the proof consists of adding an ancillary system $\HC_E$ of dimension $D$ on Alice's side and then considering a superposition $\sum_{f\in G}\sqrt{\mu_f}
\ket{\psi_f}^{AB}\otimes\ket{\phi}^{ab}\otimes\ket{f}^E$, that will be mapped by the deterministic separable operation to an ensemble 
$\{p_m,\ket{\Psi_{m,\textrm{out}}}^{AaBbE}\}$, with $\ket{\Psi_{m,\textrm{out}}}^{AaBbE}=\sum_{f\in G}\mathrm{e}^{\mathrm{i}\varphi_{mf}}\sqrt{\mu_f}
\ket{\psi_f}^{AB}\otimes\ket{\psi_f}^{ab}\otimes\ket{f}^E,
$
and we have used the fact discovered above that $p_{mf}=p_m$, independent of $f$. 
The average Schmidt vector of the output ensemble over the $AaE/Bb$ cut has to majorize the input Schmidt vector, see \cite{PhysRevA.78.020304}, and this yields i). Parts ii) and iii) are direct implications of i).

\subsection{Qubits and Qutrits}

When $D=2$ or $D=3$, one can easily show that the minimum in \eqref{lcl_eqn61} is exactly one, and therefore
\begin{theorem}[Necessity of maximally entangled blank]\label{lcl_thm11}
The following must hold.

\begin{itemize}
\item[i)] A maximally entangled state of Schmidt rank 2 is the minimum required resource for the local cloning of 2 entangled qubit states.  
\item[ii)] A maximally entangled state of Schmidt rank 3 is the minimum required resource for the local cloning of 3 group-shifted entangled qutrit states.
\end{itemize}
\end{theorem}

The proof of both i) and ii) follows easily from Lemma~\ref{lcl_thm10}, iii), by applying the rearrangement inequality to \eqref{lcl_eqn61}, and is presented in Appendix~\ref{lcl_apdx3}. 

When $D=2$, or when $D=3$ and all phases $\expo{\ii\vartheta_{\!f\!,g}}=1$, an explicit protocol for cloning these states exists \cite{PhysRevA.73.012343} (alternatively, see the proof of our Theorem~\ref{lcl_thm8}), and therefore Theorem~\ref{lcl_thm11} becomes a necessary and sufficient condition for the local cloning of such states. In particular, together with Theorem \ref{lcl_thm7}, it provides a complete solution to the problem of local cloning when $D=2$.

\subsection{$D>3$, finite gap in the necessary entanglement}
For $D>3$, preliminary numerical studies indicate that the minimum \eqref{lcl_eqn61} in Lemma~\ref{lcl_thm10}, iii) is often equal to one, with few exceptions. It might be the case that a better  choice of $\{\mu_f\}$ in \eqref{lcl_eqn59} of Lemma~\ref{lcl_thm10}, ii) may provide the $1/D$ lower bound, but we were unable to prove this.

However, for any set of group-shifted states, we can prove that there is a rather sizeable gap between the entanglement needed in the blank state and the entanglement of the states of $\SC$, as stated by the following theorem.
\begin{theorem}[Finite gap]\label{lcl_thm12}
Let $\SC=\{\ket{\psi_f}^{AB}\}_{f\in G}$ be a set of $D$ group-shifted full Schmidt rank bipartite orthogonal entangled states on $\HC_A\otimes\HC_B$ as defined by \eqref{lcl_eqn54} and considered to be not maximally entangled. If the local cloning of $\SC$ using a blank state $\ket{\phi}^{ab}$ is possible by a separable operation, then
the entanglement of the blank state has to be strictly greater that the entanglement of the states in $\SC$, often by a wide margin. Specifically,
\begin{equation}\label{lcl_eqn62}
E(\ket{\phi}^{ab})\geqslant H(\{q_r\})>E(\ket{\psi_f}^{AB}), \forall f\in G,
\end{equation}
where $E(\cdot)$ denotes the entropy of entanglement and $H(\{q_r\})$ is the Shannon entropy of the probability distribution $\{q_r\}$, $q_r:=\sum_{f\in G} \lambda_f\lambda_{fr}$, $\sum_{r\in G}{q_r}=1$.
\end{theorem}

The proof follows by setting $\mu_f=1/D$ in  Lemma~\ref{lcl_thm10}, i), but is rather long and is presented in Appendix~\ref{lcl_apdx4}.

\section{Conclusion and open questions\label{lcl_sct7}}
We have investigated the problem of local cloning of a set $\SC$ of bipartite $D\times D$  entangled states by separable operations, at least one of which is full Schmidt rank. We proved that all states in $\SC$ must be full rank and that the maximal set of clonable states must be generated by a finite group $G$ of order $N$, the number of states in this maximal set, and then we showed that $N$ has to divide $D$ exactly. We further proved that all states in $\SC$ must be equally entangled with respect to the $G$-concurrence measure, and this implied that any two states in $\SC$ must either share the same set of Schmidt coefficients or otherwise be incomparable under majorization.

We have completely solved two important problems in local cloning. For $D=2$ (entangled qubits), we proved that no more than two states can be locally cloned, and that these states must be locally-shifted. We showed that a two-qubit maximally entangled state is a necessary and sufficient resource for such a cloning. In addition, we provided necessary and sufficient conditions when the states are maximally entangled, valid for any dimension $D$, showing that the states must be group-shifted, and then we also provided an LOCC protocol that clones such a set of states. 

We have studied in detail the local cloning of partially entangled group-shifted states and provided an explicit protocol for local cloning of such states with a maximally entangled resource. For $D=3$ (entangled qutrits) we showed that a maximally entangled blank state is also necessary and sufficient, whereas for $D>3$ we proved that the blank state has to be strictly more entangled than any state in $\SC$, often by a sizeable amount.

The necessary form of the clonable states for $D>2$ remains an open problem. One might guess that the states have to be of a group-shifted form, but a proof of such a claim is not presently available. Although we proved the necessity of a maximally entangled resource for the $D=2$ case and for group-shifted states in the $D=3$ case, in higher dimensions it is still not clear if a maximally entangled state of Schmidt rank $D$ is always necessary. Finally it would be of interest to investigate the local cloning of less than full Schmidt rank states, a problem that is likely to bring in additional complications, such as the possibility of first distinguishing amongst the states in $\SC$ while preserving the states intact \cite{PhysRevA.75.052313}, and then once the state is known, the cloning becomes straightforward with a blank state having Schmidt coefficients that are majorized by those of each of the states in $\SC$ \cite{PhysRevLett.83.436,PhysRevA.78.020304}.

\begin{subappendices}
\section{Mathematical proofs}

\subsection{Proof of Lemma~\ref{lcl_thm5}}\label{lcl_apdx1}
Consider the singular value decomposition of $S$, $S=V\DC U$ with $\DC$ diagonal and positive definite, and $V$ and $U$ unitary operators. Using this expression for $S$ in $T_f=SL(f)S^{-1}$ shows that
\begin{align}\label{lcl_eqn63}
V^\dagger T_fV=\DC (UL(f)U^\dagger)\DC^{-1},
\end{align}
or with $\tilde T_f=V^\dagger T_fV$ and $\tilde L(f)=UL(f)U^\dagger$,
\begin{align}\label{lcl_eqn64}
\tilde T_f\DC=\DC\tilde L(f).
\end{align}
Since $\tilde T_f$ and $\tilde L(f)$ are both unitary, it is not difficult to see from this that each commutes with $\DC^\dagger\DC=\DC^2$. That is,
\begin{align}\label{lcl_eqn65}
\DC_i^2[\tilde T_f]_{ij}=[\tilde T_f]_{ij}\DC_j^2\notag\\
\DC_i^2[\tilde L(f)]_{ij}=[\tilde L(f)]_{ij}\DC_j^2,
\end{align}
from which we conclude that when $\DC_i\ne \DC_j$, $[\tilde T_f]_{ij}=0=[\tilde L(f)]_{ij}$. By a judicious choice of $U$ and $V$, we may arrange for $\DC$ to be a direct sum of scalar matrices (some may be one-dimensional). That is, $\DC=\oplus_\nu\alpha_\nu I_\nu$, and then we see that $T_f$ and $L(f)$ share the same block-diagonal structure, with blocks corresponding to this direct sum decomposition of $\DC$.

We also have directly from \eqref{lcl_eqn64} that
\begin{equation}\label{lcl_eqn66}
[\tilde T_f]_{ij}\DC_j=\DC_i[\tilde L(f)]_{ij}.
\end{equation}
Therefore, when $\DC_j=\DC_i$, $[\tilde T_f]_{ij}=[\tilde L(f)]_{ij}$, and we see that the blocks of $\tilde T_f$ are identical to those of $\tilde L(f)$. In other words, we have shown that $\tilde T_f=\tilde L(f)$ or equivalently, $T_f=WL(f)W^\dagger$ with $W=VU$, completing the proof.

\subsection{Proof of Lemma~\ref{lcl_thm10}\label{lcl_apdx2}}

\textbf{Proof of i)}
Let us introduce an ancillary system $\HC_E$ of dimension $D$ on Alice's side and 
construct the superposition
\begin{equation}\label{lcl_eqn67}
\ket{\Psi_\textrm{in}}^{ABabE}:=\sum_{f\in G}\sqrt{\mu_f}
\ket{\psi_f}^{AB}\otimes\ket{\phi}^{ab}\otimes\ket{f}^E,
\end{equation}
with $\{\mu_f\}_{f\in G}$ an arbitrary set of non-negative real coefficients that satisfy $\sum_{f}\mu_f=1$. The proof is based on the fact that if $\ket{\psi_f}^{AB}\otimes\ket{\phi}^{ab}$ is deterministically mapped to $\mathrm{e}^{\mathrm{i}\varphi_{m\!f}}\ket{\psi_f}^{AB}\otimes\ket{\psi_f}^{ab}$ (see \eqref{lcl_eqn11}), then
$\ket{\Psi_\mathrm{in}}^{ABabE}$ will be deterministically mapped to an ensemble
$\{p_m,\ket{\Psi_{m,\textrm{out}}}^{AaBbE}\}$, where
\begin{equation}\label{lcl_eqn68}
\ket{\Psi_{m,\textrm{out}}}^{AaBbE}=\sum_{f\in G}\mathrm{e}^{\mathrm{i}\varphi_{m\!f}}\sqrt{\mu_f}
\ket{\psi_f}^{AB}\otimes\ket{\psi_f}^{ab}\otimes\ket{f}^E.
\end{equation}
Note that this conclusion rests crucially on the fact, discovered in the main text, that $p_{mf}=p_m$, independent of $f$. 

Let us now write $\ket{\Psi_\mathrm{in}}^{ABabE}$ in Schmidt form over the $AaE/Bb$ cut. One has (again we use $\overline f=f^{-1}$)
\begin{align}\label{lcl_eqn69}
\ket{\Psi_\mathrm{in}}^{ABabE}&=\sum_{f\in G}\sqrt{\mu_f}
\left(
\sum_{g,h\in G}\expo{\ii\vartheta_{\!f\!,g}}\sqrt{\lambda_{g}\gamma_{h}}\,\ket{g}^A\ket{fg}^B
\ket{h}^a\ket{h}^b
\right)
\ket{f}^E\notag\\
&=\sum_{f,g,h\in G}\expo{\ii\vartheta_{\!f\!,g}}\sqrt{\mu_f\lambda_g\gamma_{h}}\,\ket{g}^A\ket{h}^a\ket{f}^E\otimes
\ket{fg}^B\ket{h}^b\notag\\
&=\sum_{g,h\in G}\left(
\sum_{f\in G}
\expo{\ii\vartheta_{f,\overline f g}}\sqrt{\mu_{f}\lambda_{\overline fg}\gamma_h}\ket{\overline f g}^{A}\ket{f}^{E}
\right)\ket{h}^a
\otimes
\ket{g}^B 
\ket{h}^b\notag\\
&=\sum_{g,h\in G}\left(
\sum_{f\in G}
\expo{\ii\vartheta_{\overline f,f g}}\sqrt{\mu_{\overline f}\lambda_{fg}\gamma_h}
\ket{f g}^{A}\ket{\overline f}^{E}
\right)\ket{h}^a
\otimes
\ket{g}^B 
\ket{h}^b,
\end{align}
where  we used the group property of $G$ and replaced $g$ by $\overline f g$ and summation over $f$ by summation over $\overline f$ where necessary.
The states on the $AaE$ system are orthogonal for different pairs of $g,h$, and therefore \eqref{lcl_eqn69} represents a Schmidt decomposition, with Schmidt coefficients $\alpha_{g,h}$ given by the squared norm of the states on the $AaE$ system,
\begin{equation}\label{lcl_eqn70}
\alpha_{g,h}=\gamma_h\sum_{f\in G}\mu_{\overline f}\lambda_{fg}.
\end{equation}
A similar calculation yields for the Schmidt coefficients $\beta_{g,h}$ of $\ket{\Psi_{m,\mathrm{out}}}^{ABabE}$ the expression
\begin{equation}\label{lcl_eqn71}
\beta_{g,h}=\sum_{f\in G}\mu_{\overline f}\lambda_{fg}\lambda_{fh},
\end{equation}
independent of $m$, which means that the average Schmidt vector of the output ensemble under the $Aa/BbE$ cut is the same as the Schmidt vector of an individual state $\ket{\Psi_{m,\mathrm{out}}}^{ABabE}$.

We have proven in \cite{PhysRevA.78.020304} that the average Schmidt vector of the output ensemble produced by a separable operation acting on a pure state has to majorize the input Schmidt vector, and this concludes i).

\textbf{Proof of ii)}
The proof follows as a direct consequence of i). A particular majorization inequality imposed by Lemma~\ref{lcl_thm10}~i) requires that the smallest Schmidt coefficients $\alpha_{\min}$ and $\beta_{\min}$  have to satisfy
\begin{equation}\label{lcl_eqn72}
\alpha_{\min}\geqslant\beta_{\min},
\end{equation}
where $\alpha$ and $\beta$ were defined in \eqref{lcl_eqn70} and \eqref{lcl_eqn71}, respectively. 
This is equivalent to
\begin{equation}\label{lcl_eqn73}
\gamma_{\min}\geqslant
\frac{\min_{g,h\in G}\sum_{f\in G}\mu_{\overline f}\lambda_{fg}\lambda_{fh}}
{\min_{g\in G}\sum_{f\in G}\mu_{\overline f}\lambda_{fg}}.
\end{equation}
The above equation must hold regardless of which set of $\{\mu_f\}$ was chosen, hence taking the maximum over all possible sets $\{\mu_f\}$ concludes the proof of ii).

\textbf{Proof of iii)}
Inserting the expression \eqref{lcl_eqn60} for $\{\mu_f\}$ in \eqref{lcl_eqn73} yields
\begin{align}\label{lcl_eqn74}
\gamma_{\min}&\geqslant
\frac{\min_{g,h\in G}\sum_{f\in G}\frac{1}{\lambda_{f}}\lambda_{fg}\lambda_{fh} }
{\min_{g\in G}\sum_{f\in G}\frac{1}{\lambda_f}\lambda_{fg}}\\
\label{lcl_eqn75}
&=\frac{1}{D}\min_{g,h\in G}\sum_{f\in G}\frac{1}{\lambda_{f}}\lambda_{fg}\lambda_{fh},
\end{align}
where \eqref{lcl_eqn75} follows from applying the rearrangement inequality to the denominator in \eqref{lcl_eqn74}, which in this case reads as
\begin{equation}\label{lcl_eqn76}
\min_{g\in G}\sum_{f\in G}\frac{1}{\lambda_f}\lambda_{fg}=\sum_{f\in G}\frac{1}{\lambda_f}\lambda_{f}=D.
\end{equation}

\subsection{Proof of Theorem~\ref{lcl_thm11}\label{lcl_apdx3}}

\textbf{Proof of i)}
In this case the group $G$ is the cyclic group of order 2, and we identify its group elements by $\{0,1\}$. We proved in Theorem~\ref{lcl_thm7} that the qubit states have to be locally shifted. 
The minimum in \eqref{lcl_eqn61} of Lemma~\ref{lcl_thm10}, iii) becomes explicitly a minimum over 4 quantities that correspond to all possible pairings of $g,h$; a straightforward  calculation shows that 3 out of these 4 quantities are equal to 1, except for $g=h=1$, in which case the sum in \eqref{lcl_eqn75} equals 
${\lambda_1^2}/{\lambda_0}+{\lambda_0^2}/{\lambda_1}$. 
Order the $\lambda$'s such that $\lambda_0\geqslant \lambda_1$ and note that
\begin{align}\label{lcl_eqn77}
&\frac{1}{\lambda_0}\leqslant\frac{1}{\lambda_1}\text{ and }\\
\label{lcl_eqn78}
&\lambda_1^2\leqslant \lambda_0^2.
\end{align}
From the rearrangement inequality applied to \eqref{lcl_eqn77} and \eqref{lcl_eqn78} it follows that
\begin{equation}\label{lcl_eqn79}
\frac{\lambda_1^2}{\lambda_0}+\frac{\lambda_0^2}{\lambda_1}\geqslant\frac{\lambda_0^2}{\lambda_0}+\frac{\lambda_1^2}{\lambda_1}=1,
\end{equation}
and hence the minimum in case i) equals 1.

\textbf{Proof of ii)}
Now the group $G$ is isomorphic to the cyclic group of order 3 and again we identify its elements by $\{0,1,2\}$. We order the $\lambda$'s such that $\lambda_0\geqslant\lambda_1\geqslant\lambda_2$.  The minimum in \eqref{lcl_eqn75} is now taken over $9$ possible pairs $g,h$. Again straightforward algebra shows that most expressions sum up to $1$, except for the following three cases for which we show that the sum exceeds $1$.
\begin{enumerate}
\item $g=h=1$, for which the sum in \eqref{lcl_eqn75} equals $\lambda_1^2/\lambda_0+\lambda_2^2/\lambda_1+\lambda_0^2/\lambda_2$;
\item $g=h=2$, for which the sum in \eqref{lcl_eqn75} equals $\lambda_2^2/\lambda_0+\lambda_0^2/\lambda_1+\lambda_1^2/\lambda_2$;
\item $g=1,h=2$ or $g=2,h=1$, for which the sum in \eqref{lcl_eqn75} equals $\lambda_1\lambda_2/\lambda_0+\lambda_2\lambda_0/\lambda_1+\lambda_0\lambda_1/\lambda_2$.
\end{enumerate}
Note first that 
\begin{align}\label{lcl_eqn80}
&\frac{1}{\lambda_0}\leqslant \frac{1}{\lambda_1}\leqslant \frac{1}{\lambda_2}\\
\label{lcl_eqn81}
&{\lambda_2}^2\leqslant {\lambda_1}^2\leqslant {\lambda_0}^2\text{ and }\\
\label{lcl_eqn82}
&\lambda_1\lambda_2\leqslant \lambda_2\lambda_0 \leqslant \lambda_0\lambda_1.
\end{align}

From the rearrangement inequality applied to \eqref{lcl_eqn80} and \eqref{lcl_eqn81} it follows  that
\begin{align}\label{lcl_eqn83}
&\frac{1}{\lambda_0}\lambda_1^2+\frac{1}{\lambda_1}\lambda_2^2+\frac{1}{\lambda_2}\lambda_0^2\geqslant\notag\\
&\geqslant \frac{1}{\lambda_0}\lambda_0^2+\frac{1}{\lambda_1}\lambda_1^2+\frac{1}{\lambda_2}\lambda_2^2=1,
\end{align}
which proves case 1,
and
\begin{align}\label{lcl_eqn84}
&\frac{1}{\lambda_0}\lambda_2^2+\frac{1}{\lambda_1}\lambda_0^2+\frac{1}{\lambda_2}\lambda_1^2\geqslant\notag\\
&\geqslant \frac{1}{\lambda_0}\lambda_0^2+\frac{1}{\lambda_1}\lambda_1^2+\frac{1}{\lambda_2}\lambda_2^2=1,
\end{align}
which proves case 2.

Next apply the rearrangement inequality to \eqref{lcl_eqn80} and \eqref{lcl_eqn82} to get
\begin{align}\label{lcl_eqn85}
&\frac{1}{\lambda_0}(\lambda_1\lambda_2)+\frac{1}{\lambda_1}(\lambda_2\lambda_0)+\frac{1}{\lambda_2}(\lambda_0\lambda_1) \notag\\
&\geqslant
\frac{1}{\lambda_0}\lambda_0\lambda_1+\frac{1}{\lambda_1}\lambda_1\lambda_2+\frac{1}{\lambda_2}\lambda_0\lambda_2=1
\end{align}
and this proves case 3.

\subsection{Proof of Theorem~\ref{lcl_thm12}\label{lcl_apdx4}}
By setting $\mu_f=1/D$ in Lemma~\ref{lcl_thm10}, i), for all $f\in G$, the majorization relation \eqref{lcl_eqn57} reads as
\begin{equation}\label{lcl_eqn86}
\frac{1}{D}\vec{\gamma}\times\vec{1}\prec \vec{\beta},
\end{equation}
where $(1/D)\vec{\gamma}\times\vec{1}$ represents a $D^2$ component vector with components 
$\gamma_h/D$, each component repeated $D$ times; here $\vec{\gamma}$ is the Schmidt vector of the blank state $\ket{\phi}^{ab}$.
The $D^2$ components $\beta_{g,h}$ of $\vec{\beta}$ are given by
\begin{equation}\label{lcl_eqn87}
\beta_{g,h}=\frac{1}{D}\sum_{f\in G}\lambda_{fg}\lambda_{fh}=\frac{1}{D}\sum_{f\in G}\lambda_{f}\lambda_{f\overline g h}.
\end{equation}
Note that it is also the case that $\beta$ has $D$ components each repeated $D$ times, so the majorization relation \eqref{lcl_eqn86} implies a majorization relation between 2 $D$-component vectors
\begin{equation}\label{lcl_eqn88}
\vec{\gamma}\prec\vec{q},
\end{equation} 
where the $r$-th component of $\vec{q}$ is given by
\begin{equation}\label{lcl_eqn89}
q_r:=D\cdot\beta_{g,h}|_{\overline gh=r}=\sum_{f\in G}\lambda_{f}\lambda_{fr}.
\end{equation}
Note that both $\vec{\gamma}$ and $\vec{q}$ are normalized probability vectors.
Since the Shannon entropy is a Schur-concave function, \eqref{lcl_eqn88}  implies at once that 
\begin{equation}\label{lcl_eqn90}
E(\ket{\phi}^{ab})\geqslant H(\{q_r\}).
\end{equation}

We now show that the second inequality in \eqref{lcl_eqn62} is strict. First we will prove that the ordered vector of probabilities $\vec{q}^{\downarrow}$ with components defined in \eqref{lcl_eqn89} and decreasing magnitudes of entries down its column, is majorized 
by  $\vec{\lambda}^{\downarrow}$, the ordered vector of the $\lambda_f$,
\begin{equation}\label{lcl_eqn91}
\vec{q}^\downarrow\prec\vec{\lambda}^\downarrow.
\end{equation}
Since the Shannon entropy is not just Schur-concave, but strictly Schur-concave, this will imply at once that 
\begin{equation}\label{lcl_eqn92}
H(\{q_r\})\geqslant H(\{\lambda_f\})=E(\ket{\psi_f}^{AB}),~\forall f\in G,
\end{equation}
with equality if and only if $\vec{q}^\downarrow$
equals $\vec{\lambda}^\downarrow$ (or, equivalently, if and only if the unordered vector $\vec{q}$ is the same as $\vec{\lambda}$ up to a permutation). One can see that $\vec{q}$ is not a permutation of $\vec{\lambda}$ unless all $\lambda$'s are equal, case that we exclude. Hence, once we show the majorization condition \eqref{lcl_eqn91} holds, the proof will be complete.

We will actually show that $\vec\lambda^{\downarrow}$ majorizes every vector $\vec{q}$ of the $q_r$'s no matter how $\vec q$ is ordered. Denote by $S_n$, with $|S_n|=n$ and $n=1,\cdots,D-1$, the subset consisting of those elements $f\in G$ such that $\lambda_f$ is one of the largest $n$ of the $\lambda$'s. Then, we need to show that for each $n$,
\begin{eqnarray}
\label{lcl_eqn93}
\sum_{g\in S_n}\lambda_g\geqslant\sum_{g\in S_n}q_{\sigma(g)}=\sum_{g\in S_n}\sum_{f\in G}\lambda_f\lambda_{f\sigma(g)},
\end{eqnarray}
where $\sigma$ is an arbitrary permutation of the group elements. Since $\sum_f\lambda_f=1$, this is equivalent to
\begin{eqnarray}
\label{lcl_eqn94}
\sum_{f\in G}\lambda_f\left[\sum_{g\in S_n}\lambda_g-\sum_{g\in S_n}\lambda_{f\sigma(g)}\right]\geqslant0.
\end{eqnarray}
However, given the way we have defined $S_n$, it is always true that the quantity in square brackets is non-negative. The reason is that the first term in this quantity is the sum of the $n$ largest of the $\lambda$'s. Therefore the second term, which is also a sum of $n$ of the $\lambda$'s, cannot possibly be greater than the first. In fact, it is clear that for general sets of Schmidt coefficients $\{\lambda_f\}$, the quantity in square brackets will not be particularly small, implying that the gap between the required entanglement of the blank state and the entanglement of the states in $\SC$ will be sizable. This ends the proof.

\section{$|G|<D$ case\label{lcl_apdx5}}
In the main body of the current chapter, we restricted our consideration to the $|G|=D$ case. All of our results remain valid also when $|G|<D$, with minor modifications. Briefly, when $|G|<D$, $T_f^{(m)}$ is a direct sum of $n_t=D/|G|$ copies of $L(f)$, and the following Theorems/Lemmas have to be modified accordingly.

\textbf{Theorem~\ref{lcl_thm6}.}

Since Lemma~\ref{lcl_thm5} holds for any two unitary representations, it will hold when the regular representation $L(f)$ is replaced by a direct sum of a number of copies of $L(f)$. In this case, the maximally entangled group-shifted states \eqref{lcl_eqn35} and \eqref{lcl_eqn36} of Theorem~\ref{lcl_thm6} have the form
\begin{equation}
\label{lcl_eqn95}
\ket{\psi_f}^{AB} = \frac{1}{\sqrt{D}}\sum_{n=1}^{n_t}\sum_{g\in G}\ket{fg,n}^{A}\ket{g,n}^{B},
\end{equation}
or
\begin{equation}
\label{lcl_eqn96}
\ket{\psi_f}^{AB} = \frac{1}{\sqrt{D}}\sum_{n=1}^{n_t}\sum_{g\in G}\ket{g,n}^{A}\ket{fg,n}^{B},
\end{equation}
respectively. Here the states $\{\ket{g,n}\}_{g\in G,n=1,\ldots,n_t}$ are an orthonormal basis, $\ip{g,n}{h,m}=\delta_{g,h}\delta_{n,m}$. The symbols $f,g\in G$ label the group elements and $m,n=1,\ldots,n_t$ label the copies of the regular representation.

\textbf{Theorem~\ref{lcl_thm8}.}

When the family of partially entangled group-shifted states \eqref{lcl_eqn45} is replaced by
\begin{equation}\label{lcl_eqn97}
\ket{\psi_f}^{AB} = \sum_{n=1}^{n_t}\sum_{g\in G}\sqrt{\lambda_{g,n}}\ket{g,n}^{A}\ket{fg,n}^{B}
\end{equation}
and the maximally entangled blank state \eqref{lcl_eqn46} is modified to
\begin{equation}\label{lcl_eqn98}
\ket{\phi}^{ab}=\frac{1}{\sqrt{D}}
\sum_{m=1}^{n_t}\sum_{h\in G}\ket{h,m}^{a}\ket{h,m}^{b},
\end{equation}
the local cloning protocol of Theorem~\ref{lcl_thm8} continues to work, provided that
\begin{enumerate}
\item The controlled-group unitary \eqref{lcl_eqn47} is replaced by
\begin{align}\label{lcl_eqn99}
&\sum_{n=1}^{n_t}\sum_{g\in G} \dyad{g,n}{g,n} \otimes P_{g},\text{with}\notag\\
&P_{g}=\sum_{m=1}^{n_t}\sum_{h\in G}\dyad{gh,m}{h,m}.
\end{align}
\item The measurement \eqref{lcl_eqn49} Alice performs is changed to
\begin{align}\label{lcl_eqn100}
			M_{r}=\sum_{m=1}^{n_t}\frac{1}{(\sum_{k\in G}\lambda_{k,m})^{1/2}}\sum_{h\in G}\sqrt{\lambda_{hr,m}}\dyad{h,m}{h,m}.
\end{align}
where the factor involving the sum over $k$ is needed to insure that this set of measurement operators corresponds to a complete measurement.
\item Finally the unitary correction \eqref{lcl_eqn51} Alice and Bob perform is modified to
\begin{equation}\label{lcl_eqn101}
Q_r=\sum_{m=1}^{n_t}\sum_{h\in G}\dyad{hr,m}{h,m}.
\end{equation}
\end{enumerate}

\textbf{Lemma~\ref{lcl_thm10}.}

First the blank state has to be modified to
\begin{equation}\label{lcl_eqn102}
\ket{\phi}^{ab}=\frac{1}{\sqrt{D}}
\sum_{m=1}^{n_t}\sum_{h\in G}\sqrt{\gamma_{h,m}}\ket{h,m}^{a}\ket{h,m}^{b}.
\end{equation}
Next we follow the line of thought in Appendix~\ref{lcl_apdx2}. Even though there are only $|G|<D$ states in the clonable set $\SC$, we still use a $D$ dimensional ancillary  system $\HC_E$ on Alice's side, with a basis now given by $\{\ket{f,n}^E\}_{f\in G,n=1,\ldots,n_t}$. Restricting to an ancillary system of dimension $|G|$ leads to unnecessary complications, since the rearrangement inequality can no longer be applied in part ii) to obtain iii). 

We consider again an input superposition 
\begin{equation}\label{lcl_eqn103}
	\sum_{n=1}^{n_t}\sum_{f\in G}\sqrt{\mu_{f,n}}\ket{\psi_f}^{AB}\otimes\ket{\phi}^{ab}\otimes\ket{f,n}^E
\end{equation}
and look at the Schmidt vector of the output ensemble produced by the separable operation acting on \eqref{lcl_eqn103}, where $\{\mu_{f,n}\}$ is an arbitrary set of coefficients satisfying $\sum_{n=1}^{n_t}\sum_{f\in G}\mu_{f,s}=1$. 
We then have:
\begin{itemize}
	\item[i)] The majorization condition $\vec{\alpha}\prec\vec{\beta}$ corresponding to \eqref{lcl_eqn57} holds, provided the vectors $\vec{\alpha}$ and $\vec{\beta}$ in \eqref{lcl_eqn58} are redefined as
	\begin{align}\label{lcl_eqn104}
		\alpha_{g,h}^{n,m}&=\gamma_{h,m}\sum_{s=1}^{n_t}\sum_{f\in G}\mu_{\overline f,s}\lambda_{fg,n},\notag\\
		\beta_{g,h}^{n,m}&=\sum_{s=1}^{n_t}\sum_{f\in G}\mu_{\overline{f},s}\lambda_{fg,n}\lambda_{fh,m}.
	\end{align}
	
\item[ii)] The smallest Schmidt coefficient $\gamma_{\min}$ of the blank has to satisfy
\begin{equation}\label{lcl_eqn105}
\gamma_{\min}\geqslant
\max_{\{\mu_{f,s}\}}\frac{\min_{m,n}\min_{g,h\in G}\sum_{s=1}^{n_t}\sum_{f\in G}\mu_{\overline f,s}\lambda_{fg,n}\lambda_{fh,m}}
{\min_{n}\min_{g\in G}\sum_{s=1}^{n_t}\sum_{f\in G}\mu_{\overline f, s}\lambda_{fg,n}}.
\end{equation}

\item[iii)] A good choice of $\{\mu_{f,s}\}$ is given by $\mu_{f,s}=1/\lambda_{\overline f,s}$ (ignore the normalization, since $\mu_{f,s}$ appears both on the numerator and denominator of \eqref{lcl_eqn105}). Then \eqref{lcl_eqn105} becomes
\begin{equation}\label{lcl_eqn106}
\gamma_{\min}\geqslant\frac{1}{|D|}\min_{m,n}\min_{g,h\in G}\sum_{s=1}^{n_t}\sum_{f\in G}\frac{1}{\lambda_{\overline f,s}}\lambda_{fg,n}\lambda_{fh,m}.
\end{equation}
\end{itemize}

\textbf{Theorem~\ref{lcl_thm12}.}

Theorem~\ref{lcl_thm12} still provides a finite gap between the entanglement needed in the blank state and the entanglement of group shifted states \eqref{lcl_eqn97}. The proof follows the same ideas as before, by setting $\mu_{f,s}=1/D$, for all $f\in G$ and $s=1,\ldots,n_t$ in the majorization relation of the ``modified" Lemma~\ref{lcl_thm10},i) above.
\end{subappendices}

\chapter{Quantum error correcting codes using qudit graph states\label{chp5}}

\section{Introduction\olabel{Introduction}}
\label{quditgrcodes_sct1}

Quantum error correction is an important part of various schemes for quantum
computation and quantum communication, and hence quantum error correcting
codes, first introduced about a decade ago \cite{PhysRevA.52.R2493,PhysRevA.55.900,PhysRevLett.77.793} have
received a great deal of attention.  For a detailed discussion see Ch.~10 of
\cite{NielsenChuang:QuantumComputation}. Most of the early work dealt with codes for qubits, with a
Hilbert space of dimension $D=2$, but qudit codes with $D>2$ have also been
studied \cite{IEEE.45.1827, IEEE.47.3065, PhysRevA.65.012308, quantph.0111080, quantph.0202007, IJQI.2.55,quantph.0210097}. They are of intrinsic interest and could turn out to be of some
practical value.

Cluster or graph states, which were initially introduced in connection with
measurement based or one-way quantum computing \cite{PhysRevLett.86.5188}, are also
quite useful for constructing quantum codes, as shown in
\cite{PhysRevA.65.012308,quantph.0111080,quantph.0202007} in a context in which both the encoding
operation and the resulting encoded information are represented in terms of
graph states.  In the present chapter we follow \cite{quantph.0202007} in focusing on
qudits with general $D$, thought of as elements of the additive group
$\mathbb{Z}_D$ of integers mod $D$.  However, our strategy is somewhat
different, in that we use graph states and an associated basis (graph basis)
of the $n$-qudit Hilbert space in order to construct the coding subspace,
while \emph{not} concerning ourselves with the encoding process.  This leads
to a considerable simplification of the problem along with the possibility of
treating nonadditive graph codes on exactly the same basis as additive or
stabilizer codes.  It also clarifies the relationship (within the context of
graph codes as we define them) of degenerate and nondegenerate codes, though
in this chapter we focus mainly on the latter.  The approach used here was
developed independently in \cite{IEEE.55.433} and \cite{quantph.0709.1780} for $D=2$, and in
\cite{PhysRevA.78.012306} for $D>2$; thus several of our results are similar to those
reported in these references.

Following an introduction in Sec.~\ref{quditgrcodes_sct2} to Pauli operators, graph states,
and the graph basis, as used in this chapter, the construction of graph codes is
the topic of Sec.~\ref{quditgrcodes_sct5}.  In Sec.~\ref{quditgrcodes_sct6} we review the conditions
for an $(\!(n,K,\dist)\!)_D$ code, where $n$ is the number of carriers, $K$
the number of codewords or dimension of the coding space, $\dist$ the distance
of the code, and $D$ the dimension of the Hilbert space of one qudit.  We also
consider the distinction between degenerate and nondegenerate codes.  Our
definition of graph codes follows in Sec.~\ref{quditgrcodes_sct7}, and the techniques we
use to find nondegenerate codes, which are the main focus of this chapter, are
indicated in Sec.~\ref{quditgrcodes_sct8}, while various results in terms of specific
codes are the subject of Sec.~\ref{quditgrcodes_sct9}.

In Sec.~\ref{quditgrcodes_sct11} we show how to construct graph codes with $\dist=2$ that
saturate the quantum Singleton (QS) bound for arbitrarily large $n$ and $D$,
except when $n$ is odd and $D$ is even, and we derive a simple sufficient
condition for graphs to yield such codes.  For $n$ odd and $D=2$ we have an
alternative and somewhat simpler method of producing nonadditive codes of the
same size found in \cite{PhysRevLett.99.130505}.  For both $D=2$ and $D=3$ we have studied
nondegenerate codes on sequences of cycle and wheel graphs, in
Secs.~\ref{quditgrcodes_sct12} and \ref{quditgrcodes_sct13}. These include a number of cases which
saturate the QS bound for $\dist=2$ and 3, and others with $\dist=3$ and 4
which are the largest possible additive codes for the given $n$, $D$, and
$\dist$.  Section~\ref{quditgrcodes_sct13} contains results for a series of hypercube graphs
with $n=4$, 8, and 16, and in particular a $(\!(16,128,4)\!)_2$ additive code.

In Sec.~\ref{quditgrcodes_sct15} we show that what we call G-additive codes are stabilizer
codes (hence ``additive'' in the sense usually employed in the literature),
using a suitable generalization of the stabilizer formalism to general $D$.
In this perspective the stabilizer is a dual representation of a code which
is equally well represented by its codewords.
The final Sec.~\ref{quditgrcodes_sct16} has a summary of our results and indicates
directions in which they might be extended.

\section{Pauli operators and graph states\olabel{Graph states}}
\label{quditgrcodes_sct2}

\subsection{Pauli operators}
\label{quditgrcodes_sct3}

Let $\{\ket{j}\}$, $j=0,1,\ldots D-1$ be an orthonormal basis for the
$D$-dimensional Hilbert space of a qudit, and define the unitary operators
\footnote{
See \cite{PhysRevA.71.042315} for a list of references to work that employs operators of this type.
}
\begin{equation}
\label{quditgrcodes_eqn1}
  Z = \sum_{j=0}^{D-1} \omega^j \ket{j} \! \bra{j},\quad
  X = \sum_{j=0}^{D-1} \ket{j} \! \bra{j \oplus 1},
\end{equation}
with $\oplus$ denoting addition mod $D$. They satisfy
\begin{equation}
\olabel{XZcommute}
\label{quditgrcodes_eqn2}
 Z^D=I=X^D, \quad XZ=\omega ZX,\quad
  \omega := \mathrm{e}^{2 \pi \mathrm{i} /D}.
\end{equation}
We shall refer to the collection of $D^2$ operators $\{X^\mu
Z^\nu\}$, $\mu,\nu=0,1,\ldots, D-1$, as (generalized) \emph{Pauli operators}\index{Pauli operators},
as they generalize the well known $I,X,Z,XZ\,(=-iY)$ for a qubit.  Together
they form the \emph{Pauli basis}\index{Pauli basis} of the space of operators on a qudit.

For a collection of $n$ qudits with a Hilbert space
$\HC=\HC_1\otimes\HC_2\otimes\cdots\HC_n$ we use subscripts to identify the
corresponding Pauli operators: thus $Z_l$ and $X_l$ operate on the space
$\HC_l$ of qudit $l$.  An operator of the form
\begin{equation}
\olabel{eqv1}
\label{quditgrcodes_eqn3}
 P = \omega^\lambda X^{\mu_1}_1 Z^{\nu_1}_1 X^{\mu_2}_2 Z^{\nu_2}_2
 \cdots X^{\mu_n}_n Z^{\nu_n}_n,
\end{equation}
where $\lambda$, and $\mu_l$ and $\nu_l$ for $1\leq l\leq n$, are integers in
the range $0$ to $D-1$, will be referred to as a \emph{Pauli product}\index{Pauli product}.  If
$\mu_l$ and $\nu_l$ are both 0, the operator on qudit $l$ is the identity, and
can safely be omitted from the right side of \eqref{quditgrcodes_eqn3}.  The collection
$\QC$ of all operators $P$ of the form \eqref{quditgrcodes_eqn3} with $\lambda=0$, i.e., a
prefactor of 1, forms an orthonormal basis of the space of operators on $\HC$
with inner product $\langle A,\,B\rangle = D^{-n}\Tr(A^\dagger B)$; we call it
the (generalized) \emph{Pauli basis} $\QC$.

If $P$ and $Q$ are Pauli products, so is $PQ$, and hence the collection $\PC$
of all operators of the form \eqref{quditgrcodes_eqn3} for $n$ fixed form a multiplicative
group, the \emph{Pauli group}\index{Pauli group}.  While $\PC$ is not Abelian, it has the
property that
\begin{equation}
\label{quditgrcodes_eqn4}
 PQ=\omega^\mu QP,
\end{equation}
where $\mu$ is an integer that depends on $P$ and $Q$.  (When $D=2$ and
$\omega=-1$ it is customary to also include in the Pauli group operators of
the form \eqref{quditgrcodes_eqn3} multiplied by $i$.  For our purposes this makes no
difference.)

The \emph{base}\index{base, of an operator} of an operator $P$ of the form \eqref{quditgrcodes_eqn3} is the collection
of qudits, i.e., the subset of $\{1,2,\ldots n\}$, on which the operator acts
in a nontrivial manner, so it is not just the identity, which is to say those
$j$ for which either $\mu_j$ or $\nu_j$ or both are greater than 0.  A general
operator $R$ can be expanded in the Pauli basis $\QC$, and its base is the
union of the bases of the operators which are present (with nonzero
coefficients) in the expansion.  The \emph{size}\index{size, of an operator} of an operator $R$ is defined
as the number of qudits in its base, i.e., the number on which it acts in a
nontrivial fashion.  For example, the base of $P=\omega^2 X^2_1 X^{}_4Z^{}_4$
(assuming $D\geq 3$) is $\{1,4\}$ and its size is 2; whereas the size of
$R=X^{}_1 + 0.5X^{}_2Z^2_2 Z^{}_3 +i X^{}_4$ is 4.

For two distinct qudits $l$ and $m$ the
\emph{controlled-phase}\index{Controlled-phase gate} operation $\CP_{lm}$ on $\HC_l\otimes\HC_m$,
generalizing the usual controlled-phase for qubits, is
defined by
\begin{equation}
\olabel{CPdef}
\label{quditgrcodes_eqn5}
  \CP_{lm} = \sum_{j=0}^{D-1} \sum_{k=0}^{D-1}  \omega^{jk}
\ket{j} \! \bra{j} \otimes \ket{k} \! \bra{k}
= \sum_{j=0}^{D-1} \ket{j} \! \bra{j} \otimes Z^j_m.
\end{equation}
Of course, $C_{lm}=C_{ml}$, and it is easily checked that $(C_{lm})^D=I$.
It follows from its definition that $C_{lm}$ commutes with $Z_l$ and $Z_m$, and
thus with $Z_p$ for any qudit $p$.

\subsection{Graph states}
\label{quditgrcodes_sct4}

Let $G=(V,E)$ be a graph with $n$ vertices $V$, each corresponding to a qudit,
and a collection $E$ of undirected edges connecting pairs of distinct vertices
(no self loops).  Multiple edges are allowed, as in Fig.\ref{quditgrcodes_fgr1} for
the case of $D=4$, as long as the multiplicity (weight) does not exceed $D-1$,
thus at most a single edge in the case of qubits.  The $lm$ element
$\Gamma_{lm}=\Gamma_{ml}$ of the \emph{adjacency matrix}\index{adjacency matrix} $\Gamma$ is the
number of edges connecting vertex $l$ with vertex $m$.  The graph state
\begin{equation}
\label{quditgrcodes_eqn6}
 \ket{G} = \mathcal{U}\ket{G^0}= \mathcal{U}\left(\ket{+}^{\otimes n}\right),
\end{equation}
is obtained by applying the unitary operator
\begin{equation}
\olabel{bigU}
\label{quditgrcodes_eqn7}
 \mathcal{U} =
 \prod_{\{l,m\} \in E} \left(C_{lm}\right)^{\Gamma_{lm}}.
\end{equation}
to the product state
\begin{equation}
\label{quditgrcodes_eqn8}
 \ket{G^0} := \ket{+}\otimes\ket{+}\otimes\cdots\ket{+},
\end{equation}
where
\begin{equation}
  \ket{+} := D^{-1/2} \sum_{j=0}^{D-1} \ket{j}
\label{quditgrcodes_eqn9}
\end{equation}
is a normalized eigenstate of $X$, with eigenvalue 1.  In \eqref{quditgrcodes_eqn7} the
product is over all distinct pairs of qudits, with $(\CP_{lm})^0=I$ when $l$
and $m$ are not joined by an edge.  Since the $C_{lm}$ for different $l$ and
$m$ commute with each other, and also with $Z_p$ for any $p$, the order of the
operators on the right side of \eqref{quditgrcodes_eqn7} is unimportant.

\begin{figure}
\begin{center}
\includegraphics{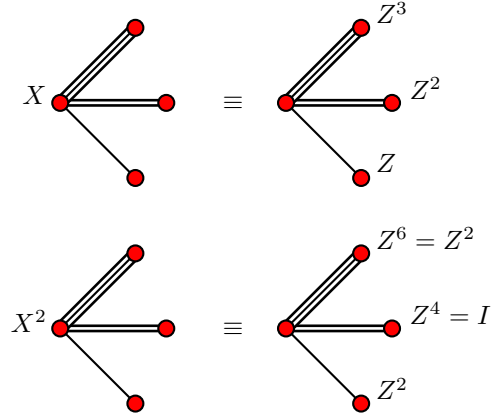}
\caption{Action of $X$ and $X^2$ on graph state ($D=4$).}
\olabel{fig:xonqudit}
\label{quditgrcodes_fgr1}
\end{center}
\end{figure}

Given the graph $G$ we define the \emph{graph basis}\index{graph basis} to be the set of
$D^n$ states
\begin{align}
\label{quditgrcodes_eqn10}
\ket{\vect{a}} &:= \ket{  a_1, a_2, \ldots ,a_n} = Z^{\vect{a}} \ket{G}
\notag\\
&= Z_1^{a_1}  Z_2^{a_2}  \cdots   Z_n^{a_n} \ket{G}
\end{align}
where $\vect{a}=(a_1,\ldots a_n)$ is an $n$-tuple of integers, each taking a
value between $0$ and $D-1$. The original graph state $\ket{G}$ is
$\ket{0,0,\ldots,0}$ in this notation.  That this collection forms an
orthonormal basis follows from the fact that the $Z_p$ operators commute with
the $\CP_{lm}$ operators, so can be moved through the unitary $\mathcal{U}$ on
the right side of \eqref{quditgrcodes_eqn6}.  As the states $Z^\nu\ket{+}$, $0\leq \nu \leq
D-1$, are an orthonormal basis for a single qudit, their products form an
orthonormal basis for $n$ qudits. Applying the unitary $\mathcal{U}$ to this
basis yields the orthonormal graph basis. The $n$-tuple representation in
\eqref{quditgrcodes_eqn10} is convenient in that one can define
\begin{align}
\label{quditgrcodes_eqn11}
  \ket{ \vect{a} \oplus \vect{b} } &:= \ket{  a_1 \oplus b_1, a_2 \oplus b_2,
  \ldots ,a_n \oplus b_n},
\notag\\
 \ket{ j  \vect{a}} &:= \ket{ j a_1, j a_2, \ldots ,j a_n  },
\end{align}
where $j$ is an integer between $0$ and $D-1$, and arithmetic operations are
mod $D$.

One advantage of using the graph basis is that its elements are mapped to each
other by a Pauli product (up to powers of $\omega$), as can be seen by
considering the action of $Z_l$ or $X_l$ on a single qudit.  The result for
$Z_l$ follows at once from \eqref{quditgrcodes_eqn10}.  And as shown in App.~\ref{quditgrcodes_sct17} and
illustrated in Fig.~\ref{quditgrcodes_fgr1}, the effect of applying $X_l$ to $\ket{G}$ is
the same as applying $(Z_m)^{\Gamma_{lm}}$ to each of the qudits corresponding
to neighbors of $l$ in the graph.  Applying these two rules and keeping track
of powers of $\omega$ resulting from interchanging $X_l$ and $Z_l$,
see \eqref{quditgrcodes_eqn2}, allows one to easily evaluate the action of any Pauli
product on any $\ket{\vect{a}}$ in the graph basis.

\section{Code construction\olabel{Code}}
\label{quditgrcodes_sct5}

\subsection{Preliminaries}
\label{quditgrcodes_sct6}

Consider a quantum code corresponding to a $K$-dimensional subspace,  with
orthonormal basis $\{ \ket{\vect{c}_q}\}$, of the Hilbert
space $\HC$ of $n$ qudits. When the Knill-Laflamme \cite{PhysRevA.55.900} condition
\begin{equation}
\olabel{eq:orthoOriginal}
\label{quditgrcodes_eqn12}
\bra{\vect{c}_q} Q \ket{\vect{c}_r} = f(Q) \delta_{qr}
\end{equation}
is satisfied for all $q$ and $r$ between $0$ and $K-1$, and every operator $Q$
on $\HC$ such that $1 \leq \mbox{size}(Q) < \dist$, but fails for some
operators of size $\dist$, the code is said to have \emph{distance}\index{distance} $\dist$,
and is an $(\!(n, K, \dist)\!)_D$ code; the subscript is often omitted when
$D=2$.  (See the definition of size in Sec.~\ref{quditgrcodes_sct3}. The only operator of
size 0 is a multiple of the identity, so \eqref{quditgrcodes_eqn12} is trivially
satisfied.)  A code of distance $\dist$ allows the correction of any error
involving at most $\lfloor (\dist-1)/2\rfloor$ qudits, or an error on
$\dist-1$ (or fewer) qudits if the location of the corrupted qudits is already
known (e.g., they have been stolen).

It is helpful to regard \eqref{quditgrcodes_eqn12} as embodying two conditions: the obvious
off-diagonal condition saying that the matrix elements of $Q$ must vanish when
$r\neq q$; and the diagonal condition which, since $f(Q)$ is an arbitrary
complex-valued function of the operator $Q$, is nothing but the requirement
that all diagonal elements of $Q$ (inside the coding space) be identical.  The
off-diagonal condition has a clear analog in classical codes, whereas the
diagonal one does not. Both must hold for all operators of size up to and
including $\delta-1$, but need not be satisfied for larger operators.

In the coding literature it is customary to distinguish \emph{nondegenerate}\index{nondegenerate codes}
codes for which $f(Q)=0$ for all operators of size between 1 and $\delta-1$,
i.e., for \emph{all} $q$ and $r$
\begin{equation}
\olabel{eq:ortho}
\label{quditgrcodes_eqn13}
\bra{\vect{c}_q} Q \ket{\vect{c}_r} = 0 \;\; \text{for} \;\;
 1\leq\text{size}(Q)<\dist,
\end{equation}
and \emph{degenerate}\index{degenerate codes} codes for which $f(Q)\neq 0$ for at least one $Q$ in the
same range of sizes.  See p.~444 of \cite{NielsenChuang:QuantumComputation} for the motivation behind
this somewhat peculiar terminology when $\delta$ is odd. In this chapter our
focus is on nondegenerate codes. For the most part they seem to perform as well
as degenerate codes, though there are examples of degenerate codes that provide
a larger $K$ for given values of $n$, $\dist$, and $D$ than all known
nondegenerate codes. Examples are the $(\!(6, 2, 3)\!)_2$ 
\footnote{
While there seems to be no proof that the $(\!(6, 2, 3)\!)_2$ degenerate code has a larger $K$ than any nondegenerate code with $n=6$ and $\dist=3$, some support comes from the fact that we performed an exhaustive search of all graphs with 6 vertices and did not find a nondegenerate graph code with $\dist=3$ and $K>1$. But the notion that this degenerate code is superior to nondegenerate codes is undercut by the observation that the well known nondegenerate $(\!(5, 2, 3)\!)_2$ code uses only 5 instead of 6 qubits to achieve equal values of $K$ and $\dist$.
} 
and
$(\!(25, 2, 9)\!)_2$ codes mentioned in \cite{IEEE.44.1369}.

\subsection{Graph codes}
\label{quditgrcodes_sct7}

When each basis vector $\ket{\vect{c}_q}$ is a member of the graph basis, of
the form \eqref{quditgrcodes_eqn10} for some graph $G$, we shall say that the corresponding
code is a \emph{graph codes}\index{graph code} associated with this graph.  As noted in
Sec.~\ref{quditgrcodes_sct1}, this differs from the definition employed in
\cite{PhysRevA.65.012308,quantph.0111080,quantph.0202007}, but agrees with that in more recent $D=2$
studies \cite{quantph.0709.1780, IEEE.55.433}, because we do not concern ourselves with the
processes of encoding and decoding.  In what follows we shall always assume
$\dist\geq 2$, since $\delta=1$ is trivial.  As the left side of \eqref{quditgrcodes_eqn12}
is linear in $Q$, it suffices to check it for appropriate operators drawn from
the Pauli basis $\QC$ as defined in Sec.~\ref{quditgrcodes_sct3}.  It is helpful to note
that for any $Q\in\QC$, any pair $\ket{\vect{c}_q}$ and $\ket{\vect{c}_r}$ of
graph basis states and any $n$-tuple $\vect{a}$,
\begin{align}
\label{quditgrcodes_eqn14}
\bra{\vect{c}_q \oplus \vect{a} } Q \ket{\vect{c}_r \oplus \vect{a} } &=
\bra{\vect{c}_q} Z^{-\vect{a}} Q Z^{\vect{a}} \ket{\vect{c}_r}
\notag\\
 &=\omega^\mu\bra{\vect{c}_q} Q \ket{\vect{c}_r}
\end{align}
for some integer $\mu$ depending on $Q$ and $\vect{a}$; see \eqref{quditgrcodes_eqn10},
\eqref{quditgrcodes_eqn11} and \eqref{quditgrcodes_eqn4}.  Therefore, if \eqref{quditgrcodes_eqn12} is satisfied for
some $Q$ and a collection $\{\ket{\vect{c}_q}\}$ of codewords, the same will be
true for the same $Q$ and the collection $\{ \ket{\vect{c}_q \oplus \vect{a} }
\}$ (with an appropriate change in $f(Q)$).  Thus we can, and hereafter always
will, choose the first codeword to be
\begin{equation}
\label{quditgrcodes_eqn15}
 \ket{\vect{c}_0}=\ket{0,0,\ldots ,0}=\ket{G}.
\end{equation}

Analogous to Hamming distance in classical information theory we define
the \emph{Pauli distance}\index{Pauli distance} $\Delta$ between two graph basis states as
\begin{equation}
\label{quditgrcodes_eqn16}
\Delta(\vc_q,\vc_r) =\Delta \left( \ket{\vect{c}_q} ,\ket{\vect{c}_r} \right)
 :=\\ \min \{ \mbox{size($Q$)} : \bra{\vect{c}_q} Q \ket{\vect{c}_r} \neq 0 \},
\end{equation}
where it suffices to take the minimum for $Q\in\QC$, the Pauli basis.
(Ket symbols can be omitted from the arguments of $\Delta$ when the
meaning is clear.)
Also note the identities
\begin{align}
\olabel{eq:distIdentity}
\label{quditgrcodes_eqn17}
\Delta(\vc_q,\vc_r) &= \Delta(\vc_r,\vc_q) 
= \Delta(\vc_q \oplus \vect{a} ,\vc_r \oplus \vect{a}) \notag \\
&= \Delta(\vc_0,\vc_r\ominus\vc_q),
\end{align}
where $\vect{a}$ is any $n$-tuple, and $\ominus$ means difference mod $D$, see
\eqref{quditgrcodes_eqn11}.  The second equality is a consequence of \eqref{quditgrcodes_eqn14}.  Note
that if in \eqref{quditgrcodes_eqn16} we minimize only over $Q$ operations which are tensor
products of $Z$'s (no $X$'s), $\Delta$ is exactly the Hamming distance between
the $n$-tuples $\vect{c}_q$ and $\vect{c}_r$, see \eqref{quditgrcodes_eqn10}.

For the case $q=r$, where \eqref{quditgrcodes_eqn16} gives 0 (for $Q=I$), we introduce a
special \emph{diagonal distance}\index{diagonal distance} $\Delta'$ which is the minimum size of the
right side of \eqref{quditgrcodes_eqn16} when one restricts $Q$ to be an element of $\QC$
of size 1 or more.  The diagonal distance does not depend on the particular
value of $q=r$, but is determined solely by the graph state $\ket{G}$---see
\eqref{quditgrcodes_eqn14} with $r=q$---and thus by the graph $G$.  This has the important
consequence that if we consider a particular $G$ and want to find the optimum
codes for a given $\delta$ that is no larger than $\Delta'$, the collection of
operators $Q\in\QC$ for which \eqref{quditgrcodes_eqn12} needs to be checked will all have
zero diagonal elements, $f(Q)=0$, and we can use \eqref{quditgrcodes_eqn13} instead of
\eqref{quditgrcodes_eqn12}.  In other words, for the graph in question and for
$\delta\leq\Delta'$, all graph codes are nondegenerate, and in looking for an
optimal code one need not consider the degenerate case. Our computer results
in Sec.~\ref{quditgrcodes_sct9} are all limited to the range $\delta\leq\Delta'$ where no
degenerate codes exist for the graph in question. Any code with $\dist >
\Delta'$ will necessarily be degenerate, since there is at least
one nontrivial $Q$ for which \eqref{quditgrcodes_eqn12} must be checked for the diagonal
elements.

A code is \emph{G-additive (graph-additive)}\index{G-additive (graph-additive) code} if given any two codewords
$\ket{\vect{c}_q}$ and $\ket{\vect{c}_r}$ belonging to the code,
$\ket{\vect{c}_q \oplus \vect{c}_r}$ is also a codeword. As shown in
Sec.~\ref{quditgrcodes_sct15}, this notion of additivity implies the code is additive in the
sense of being a stabilizer code. For this reason, we shall omit the G in
G-additive except in cases where it is essential to make the distinction.
Codes that do not satisfy the additivity condition are called nonadditive.
The additive property allows one to express all codewords as ``linear
combinations'' of $k$ suitably chosen codeword generators. This implies an
additive code must have $K=D^r$, $r$ an integer, whenever $D$ is prime. We
will see an example of this in Sec.~\ref{quditgrcodes_sct9} for $D=2$.

The \emph{quantum Singleton}\index{quantum Singleton bound} (QS) bound \cite{PhysRevA.55.900}
\begin{equation}
\olabel{eq:bound1}
\label{quditgrcodes_eqn18}
n \geq \log_D K + 2(\dist - 1) \;\;\; \text{or}
 \;\;\; K \leq D^{n - 2(\dist - 1)}
\end{equation}
is a simple but useful inequality.  We shall refer to codes which saturate
this bound (the inequality is an equality) as \emph{quantum Singleton} (QS)
codes. Some authors prefer the term MDS, but as it is not clear to us how
the concept of ``maximum distance separable,'' as explained in
\cite{MacWilliamsSloane:TheoryECC}, carries over to quantum codes, we prefer to use QS.

\subsection{Method}
\label{quditgrcodes_sct8}

We are interested in finding ``good'' graph codes in the sense of a large $K$
for a given $n$, $\dist$, and $D$. The first task is to choose a graph $G$ on
$n$ vertices, not a trivial matter since the number of possibilities increases
rapidly with $n$.  We know of no general principles for making this choice,
though it is helpful to note, see App.~\ref{quditgrcodes_sct17}, that the diagonal
distance $\Delta'$ cannot exceed 1 plus the minimum over all vertices of the
number of neighbors of a vertex.  Graphs with a high degree of symmetry are,
for obvious reasons, more amenable to analytic studies and computer searches
than those with lower symmetry.

Given a graph $G$ and a distance $\dist$, one can in principle search for the
best nondegenerate code by setting $\ket{\vc_0}=\ket{G}$, finding a
$\ket{\vc_1}$ with $\Delta(\vc_0,\vc_1)\geq \dist$, after that $\ket{\vc_2}$
with both $\Delta(\vc_0,\vc_2)\geq \dist$ and $\Delta(\vc_1,\vc_2)\geq \dist$,
and so forth, until the process stops.  However, this may happen before one
finds the largest $K$, because a better choice could have been made for
$\ket{\vc_q}$ at some point in the process. Exhaustively checking all
possibilities is rather time consuming, somewhat like solving an optimal
packing problem.

In practice what we do is to first construct a lookup table containing the
$D^n-1$ Pauli distances from $\ket{G}$ to all of the other graph basis states,
using an iterative process starting with all $Q\in\QC$ of size 1, then of size
2, etc. This process also yields the diagonal distance $\Delta'$.  As we are
only considering nondegenerate codes, we choose some $\dist\leq \Delta'$, so
that \eqref{quditgrcodes_eqn13} can be used in place of \eqref{quditgrcodes_eqn12}, and use the table to
identify the collection $S$ of all graph basis states with a distance greater
than or equal to $\dist$ from $\ket{\vc_0}=\ket{G}$.  If $S$ is empty there
are no other codewords, so $K=1$.  However, if $S$ is not empty then $K$ is at
least 2, and a search for the optimum code (largest $K$) is carried out as
follows.

We produce a graph $\SC$ (not to be confused with $G$) in which the nodes are
the elements of $S$, and an edge connects two nodes if the Pauli distance
separating them---easily computed from the lookup table with the help of
\eqref{quditgrcodes_eqn17}---is \emph{greater than or equal to} $\dist$. An edge in this
graph signifies that the nodes it joins are sufficiently (Pauli) separated to
be candidates for the code, and an optimal code corresponds to a largest
complete subgraph or \emph{maximum clique}\index{maximum clique} of $\SC$.  Once a maximum clique has
been found, the corresponding graph basis states, including $\ket{\vc_0}$,
satisfy \eqref{quditgrcodes_eqn13} and span a coding space with the largest possible $K$ for
this graph $G$ and this $\dist$.

The maximum clique problem on a general graph is known to be NP-complete
\cite{GareyJohnson:CompIntract} and hence computationally difficult, and we do not know if $\SC$
has special properties which can be exploited to speed things up.  We used the
relatively simple algorithm described in \cite{OpRsrchLett.9.375} for finding a
maximum clique, and this is the most time-consuming part of the search
procedure. 

The method just described finds additive as well as nonadditive codes. In fact
one does not know beforehand whether the resultant code will be additive or
not. If one is only interested in additive codes, certain steps can be modified
to produce a substantial increase in speed as one only has to find a set of
generators for the code.


\section{Results\olabel{Results}}
\label{quditgrcodes_sct9}

\subsection{Introduction}
\label{quditgrcodes_sct10}

Results obtained using methods described above are reported here for various
sequences of graphs, each sequence containing graphs of increasing $n$ while
preserving certain basic properties.  We used a computer search to find the
maximum number $K$ of codewords for each graph in the sequence, for distances
$\dist\leq \Delta'$ and for $D=2$ or 3, qubits and qutrits, up to the largest
number $n$ of qudits allowed by our resources (running time). Sometimes this
revealed a pattern which could be further analyzed using analytic arguments or
known bounds on the number of codewords.

In the case of distance $\dist=2$ we can demonstrate the existence of QS codes
for arbitrarily large values of $n$ and $D$, except when $n$ is odd and $D$ is
even, see Part~A.  In the later subsections we report a significant collection
of $D=2$ and 3 codes for $\dist=2$, 3, and 4, including QS codes; codes which
are the largest possible additive codes for that set of $n$, $D$ and $\dist$;
and a new $(\!(16,128,4)\!)_2$ additive code.

Tables show the $K$ found as a function of other parameters. The meaning of
superscripts used in the tables is given below.
\begin{itemize}
\item $a$ -- Indicates the maximum clique search was terminated
before completion. This means the code we found might not be optimal, i.e.
there might be another code with larger $K$ for this graph. We can only say 
the code is \emph{maximal}\index{maximal code} in the sense that no codeword can be added without
violating \eqref{quditgrcodes_eqn13}. Absence of this superscript implies no code with a
larger $K$ exists for this $\dist$ and this graph, either because the program
did an exhaustive search, or because $K$ saturates a rigorous bound.
\item $b$ -- Indicates a nonadditive code. Codes without this superscript are
additive.
\item $c$ -- Indicates a QS code, one where $K$ saturates the Singleton bound
\eqref{quditgrcodes_eqn18}.
\item $d$ -- Indicates this is not a QS code, but the largest possible
  \emph{additive} (graph or other) code for the given $n$, $\dist$ and $D$,
  This follows from linear programming bounds in
  \cite{Grassl:codetables} for $D=2$ and \cite{IEEE.51.4892} for $D=3$, along with the fact,
  Sec.~\ref{quditgrcodes_sct7}, that for an additive code, $K$ must be an integer power
  of $D$ when $D$ is prime. A larger \emph{nonadditive} code for this graph
  might still be possible in cases flagged with $a$ as well as $d$. 

\end{itemize}

\subsection{Distance $\dist=2$; bar and star graphs}
\label{quditgrcodes_sct11}

It was shown in \cite{IEEE.44.1369} that for $D=2$ one can construct $\dist=2$
QS codes for any even $n$, and similar codes for larger $D$ are mentioned,
without giving details, in \cite{IEEE.45.1827}.  One way to construct graph codes
with $\dist=2$ is to use the method indicated in the proof, App.~\ref{quditgrcodes_sct18},
of the following result.

\noindent
\textbf{Partition theorem}. \emph{Suppose that for a given $D$ the vertices of
  a graph $G$ on $n$ qudits can be partitioned into two nonempty sets $V_1$
  and $V_2$ with the property that for each vertex in $V_1$ the sum of the
  number of edges (the sum of the multiplicities if multiple edges are
  present) joining it to vertices in $V_2$ is nonzero and coprime to $D$, and
  the same for the number of edges joining a vertex in $V_2$ to vertices in
  $V_1$.  Then there is an additive QS code on $G$ with distance $\dist=2$.}

A \emph{bar}\index{bar graph} graph is constructed by taking $n$ vertices and dividing them
into two collections $V_1$ and $V_2$, of equal size when $n$ is even, and one
more vertex in $V_2$ when $n$ is odd, as in Fig.~\ref{quditgrcodes_fgr2}(a). Next pair the
vertices by connecting each vertex in $V_1$ by a single edge to a vertex in
$V_2$, with one additional edge when $n$ is odd, as shown in the figure.
(Multiple edges are possible for $D>2$, but provide no advantage in
constructing codes.)
When $n$ is even the conditions of the partition theorem are satisfied:
1 is always coprime to $D$.  For odd $n$, the last vertex in $V_1$ has
2 edges joining it to $V_2$, which is coprime to $D$ when $D$ is odd.  Hence
bar graphs yield $\dist=2$ QS codes for all $n$ when $D$ is odd, and for even
$n$ when $D$ is even.

\begin{figure}
\begin{center}
\includegraphics{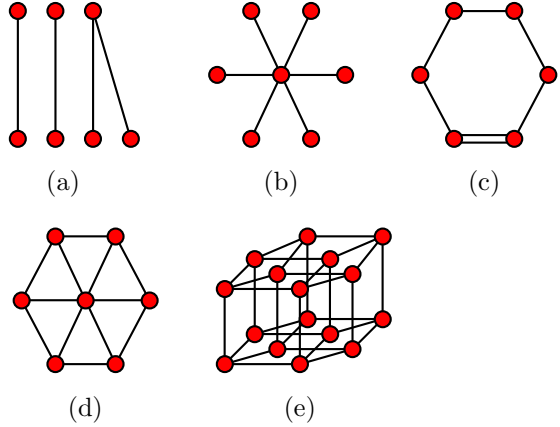}
\caption{Examples from different graph sequences: (a) bar (odd $n$), (b) star,
(c) cycle, (d) wheel, (e) $n=16$ hypercube.}
\label{quditgrcodes_fgr2}
\end{center}
\end{figure}

A \emph{star}\index{star graph} graph, Fig.~\ref{quditgrcodes_fgr2}(b), has a \emph{central} vertex joined by
single edges to every \emph{peripheral} vertex, and no edges connecting pairs
of peripheral vertices. Since the diagonal distance $\Delta'$ is 2,
nondegenerate star codes cannot have $\dist$ larger than 2.  As in the case of
bar codes, one can construct additive QS codes for any $n$ when $D$ is odd,
and for even $n$ when $D$ is even 
\footnote{
We omit the details. In some but not all cases one can use the Partition theorem with $V_1$ and $V_2$ the center and the peripheral vertices. Allowing some double edges when $D>2$ extends the range of $n$ values where the Partition theorem can be employed.
}
. For odd $n$ and $D=2$ there
are nonadditive codes with
\begin{equation}
\label{quditgrcodes_eqn19}
K(n)=2^{n-2} - \frac{1}{2} \binom{n-1}{(n-1)/2};
\end{equation}
see App.~\ref{quditgrcodes_sct19} for details.  Codes with these parameters were discovered
earlier by Smolin et al.\ \cite{PhysRevLett.99.130505} using a different approach. Computer
searches show that for all odd $n\leq 7$ star graphs cannot yield a $K$ larger
than \eqref{quditgrcodes_eqn19}.

\subsection{Cycle graphs}
\label{quditgrcodes_sct12}

We used computer searches to look for graph codes based on cycle (loop)
graphs, Fig.~\ref{quditgrcodes_fgr2}(c). Table~\ref{quditgrcodes_tbl1} shows the maximum number $K$ of
codewords for codes of distance $\dist=2$ and $\dist=3$ for both $D=2$ qubits
and $D=3$ qutrits. In the qutrit case the best codes were obtained by including
one double edge (weight 2), as in Fig.~\ref{quditgrcodes_fgr2}(c), though when $n$ is odd
equally good codes emerge with only single edges. In the qubit case all edges
have weight 1.

\renewcommand{\thefootnote}{\alph{footnote}}

\begin{table}
\begin{center}
\caption{Maximum $K$ for qubit and qutrit cycle graphs. See Sec.~\ref{quditgrcodes_sct10}
for detailed meaning of superscripts.}
\olabel{tab:qubitCycle}
\label{quditgrcodes_tbl1}
\begin{tabular}{ c | c  c | c c }
\multicolumn{1}{c |}{ } &
\multicolumn{2}{c |}{$D=2$} &
\multicolumn{2}{c }{$D=3$}\\
$n$ &
\multicolumn{1}{c}{$\dist=2$} &
\multicolumn{1}{c|}{\hsp$\dist=3$\hsp} &
\multicolumn{1}{c}{$\dist=2$} &
\multicolumn{1}{c}{$\dist=3$} \\
\hline
4     & 4\ftnc         & 0             & 9\ftnc     & 1\ftnc   \\
5     & 6\ftnb         & 2\ftnc        & 27\ftnc    & 3\ftnc   \\
6     & 16\ftnc        & 1             & 81\ftnc    & 9\ftnc   \\
7     & 22\ftnb        & 2\ftnd        & 243\ftnc   & 27\ftnc  \\
8     & 64\ftnc        & 8\ftnd        & 729\ftnc   & 81\ftnc  \\
9     & 96\ftna \ftnb   & 12\ftnb       & 2187\ftnc  & 243\ftnc \\
10    & 256\ftnc       & 18\ftnb       & 6561\ftnc  & 729\ftnc \\
11    & 272\ftna \ftnb  & 32\ftna \ftnd  & 19683\ftnc & 729\ftna \ftnd \\
\hsp12\hsp    &\hsp1024\ftnc\hsp  &  \hsp64\ftna \ftnd\hsp
              &\hsp59049\ftnc\hsp &  \hsp2187\ftna \ftnd \hsp
\end{tabular}
\begin{itemize}
\item[a]{Non-exhaustive search}
\item[b]{Nonadditive code}
\item[c]{Code saturating Singleton bound \eqref{quditgrcodes_eqn18}}
\item[d]{Largest possible additive code}
\end{itemize}
\end{center}
\end{table}
\renewcommand{\thefootnote}{\arabic{footnote}}

The $D=2$ entries in Table~\ref{quditgrcodes_tbl1} include for $n=5$ the well known $(\!(5,
2, 3)\!)_2$, the nonadditive $(\!(5, 6, 2)\!)_2$ presented in \cite{PhysRevLett.79.953},
and, for larger $n$, a $(\!(9, 12, 3)\!)_2$ code similar to that in \cite{PhysRevLett.101.090501}
and the $(\!(10, 18, 3)\!)_2$ of \cite{IEEE.55.433} based upon the same graph.

The $D=3$, $\dist=3$ entries are interesting because the QS bound is saturated
for $4 \leq n \leq 10$ but \emph{not} for $n = 11$. The $(\!(11, 3^6=729,
3)\!)_3$ code we found, the best possible \emph{additive} code according to
the linear programming bound in \cite{IEEE.51.4892}, falls short by a factor of 3 of
saturating the $K=3^7=2187$ QS bound, and even a nonadditive code based on
this graph must have $K\leq 1990$ 
\footnote{
Since the distance $\dist=3$ does not exceed the diagonal distance $\Delta'=3$ for this graph, a graph code is necessarily nondegenerate, see Sec.~\ref{quditgrcodes_sct7}, and hence the quantum Hamming bound---see p.~444 of \cite{NielsenChuang:QuantumComputation}---extended to $D=3$ applies, and this yields an upper bound of $K\leq 1990$.
}
.

One can ask to what extent the results for $\dist=2$ in Table~\ref{quditgrcodes_tbl1} could
have been obtained, or might be extended to larger $n$, by applying the
Partition theorem of Part A to a suitable partition of the cycle graph.  It
turns out---we omit the details---that when $D$ is odd one can use the
Partition theorem to produce codes that saturate the QS bound for any $n$, but
when $D$ is even the same approach only works when $n$ is a multiple of 4. In
particular, the $(\!(6, 16, 2)\!)_2$ additive QS code in Table~\ref{quditgrcodes_tbl1}
cannot be obtained in this fashion since the cycle graph cannot be partitioned
in the required way.

\subsection{Wheel graphs}
\label{quditgrcodes_sct13}

If additional edges are added to a star graph so as to connect the peripheral
vertices in a cycle, as in Fig.~\ref{quditgrcodes_fgr2}(d), the result is what we call a
\emph{wheel}\index{wheel graph} graph. Because each vertex has at least three neighbors, our
search procedure, limited to $\dist\leq \Delta'$, can yield $\dist=4$ codes on
wheel graphs, unlike cycle or star graphs.  The construction of $\dist=2$
codes for any $D$ is exactly the same as for star graphs, so in
Table~\ref{quditgrcodes_tbl2} we only show results for $\dist=3$ and 4, for both $D=2$ and
3. The $(\!(16, 128, 4)\!)_2$ additive code  appears to be new, and
its counterpart in the hypercube sequence is discussed below.

\renewcommand{\thefootnote}{\alph{footnote}}
\begin{table}[ht]
\begin{center}
\caption{Maximum $K$ for qubit and qutrit wheel graphs. See Sec.~\ref{quditgrcodes_sct10}
for detailed meaning of superscripts.}
\olabel{tab:qubitCycle}
\label{quditgrcodes_tbl2}
\begin{tabular}{ c | c  c | c c }
\multicolumn{1}{c |}{ } &
\multicolumn{2}{c |}{$D=2$} &
\multicolumn{2}{c }{$D=3$}\\
$n$ &
\multicolumn{1}{c}{$\dist=3$} &
\multicolumn{1}{c|}{\hsp$\dist=4$\hsp} &
\multicolumn{1}{c}{$\dist=3$} &
\multicolumn{1}{c}{$\dist=4$} \\ \hline
6    & 1               & 1\ftnc        & 1                & 1\ftnc    \\
7    & 2\ftnd          & 0             & 27\ftnc          & 1         \\
8    & 8\ftnd          & 1\ftnd        & 27               & 9\ftnc    \\
9    & 8\ftnd          & 1\ftnd        & 243\ftnc         & 9         \\
10   & 20\ftnc         & 4\ftnd        & 243\ftna         & 27        \\
11   & 32\ftna\ftnd    & 4\ftnd        & 729\ftna\ftnd    & 81        \\
12   & 64\ftna\ftnd    & 8             & 2187\ftna\ftnd   & 81\ftna   \\
13   & 128\ftna\ftnd   & 16            & 6561\ftna\ftnd   & 243\ftna  \\
14   & 256\ftna\ftnd   & 32\ftna       & 19683\ftna\ftnd  & 729\ftna  \\
15   & 512\ftna\ftnd   & 64\ftna\ftnd  & \hsp59049\ftna\ftnd\hsp &
\hsp 2187\ftna\hsp \\
\hsp16\hsp  &\hsp1024\ftna\ftnd\hsp  & \hsp128\ftna\ftnd\hsp & &
\end{tabular}
\begin{itemize}
\item[a]{Non-exhaustive search}
\item[b]{Nonadditive code}
\item[c]{Code saturating Singleton bound \eqref{quditgrcodes_eqn18}}
\item[d]{Largest possible additive code}
\end{itemize}
\end{center}
\end{table}
\renewcommand{\thefootnote}{\arabic{footnote}}

\subsection{Hypercube graphs}
\label{quditgrcodes_sct14}

Hypercube graphs, Fig.~\ref{quditgrcodes_fgr2}(e), have a high symmetry, and as $n$
increases the coordination bound, App.~\ref{quditgrcodes_sct17}, allows $\Delta'$ to
increase with $n$, unlike the other sequences of graphs discussed above.  We
have only studied the $D=2$ case, with the results shown in Table~\ref{quditgrcodes_tbl3}.
Those for $\dist=2$ are an immediate consequence of the Partition theorem:
each hypercube is obtained by adding edges between two hypercubes of the next
lower dimension, and these are the $V_1$ and $V_2$ of the theorem.  The
generators for the $(\!(16, 128, 4)\!)_2$ additive code are given in
Table~\ref{quditgrcodes_tbl4}. The $2^7=128$ codewords are of the form, see \eqref{quditgrcodes_eqn11},
$\ket{\alpha_1 \vect{g}_1\oplus \alpha_2 \vect{g}_2 \oplus \cdots \alpha_7
\vect{g}_7}$, where each $\alpha_j$ can be either 0 or 1.

\renewcommand{\thefootnote}{\alph{footnote}}
\begin{table}
\begin{center}
\caption{Maximum $K$ for qubit hypercube graphs. See Sec.~\ref{quditgrcodes_sct10}
for detailed meaning of superscripts.}
\olabel{tab:qubitHypercube}
\label{quditgrcodes_tbl3}

\begin{tabular}{c|ccc}
\multicolumn{1}{c|}{} &
\multicolumn{3}{c}{$D=2$} \\
$n$ &
$\dist=2$ &  $\dist=3$ & \hsp$\dist=4$\hsp\\ \hline
4    & 4\ftnc     &  0        & 0 \\
8    & 64\ftnc    &  8\ftnd   & 1\ftnd \\
\hsp 16\hsp\hsp   & \hsp16384\ftnc\hsp &  512\ftna & 128\ftna\ftnd
\end{tabular}
\begin{itemize}
\item[a]{Non-exhaustive search}
\item[c]{Code saturating Singleton bound \eqref{quditgrcodes_eqn18}}
\item[d]{Largest possible additive code}
\end{itemize}
\end{center}
\end{table}

\begin{table}
\begin{center}
  \caption{Generators of  $(\!(16, 128, 4)\!)_2$ additive code for hypercube
graph}
\olabel{tab:hypercubeGen}
\label{quditgrcodes_tbl4}
\begin{tabular}{cc}
Generator & Bit notation \\
\hline
$\ket{\,\vect{g}_1}$ & \;\; $\ket{0,0,0,0, 0,0,0,0, 0,0,0,0, 1,1,1,1}$ \;\; \\
$\ket{\,\vect{g}_2}$ & $\ket{0,0,0,0, 0,0,0,0, 0,0,1,1, 0,0,1,1}$ \\
$\ket{\,\vect{g}_3}$ & $\ket{0,0,0,0, 0,0,0,0, 1,1,0,0, 0,0,1,1}$ \\
$\ket{\,\vect{g}_4}$ & $\ket{0,0,0,0, 0,0,1,1, 0,1,0,0, 0,1,0,0}$ \\
$\ket{\,\vect{g}_5}$ & $\ket{0,0,0,0, 1,1,0,0, 0,0,0,1, 0,0,0,1}$ \\
$\ket{\,\vect{g}_6}$ & $\ket{0,0,1,1, 0,0,0,0, 0,1,0,0, 0,1,0,0}$ \\
$\ket{\,\vect{g}_7}$ & $\ket{1,1,0,0, 0,0,0,0, 0,0,0,1, 0,0,0,1}$
\end{tabular}
\end{center}
\end{table}
\renewcommand{\thefootnote}{\arabic{footnote}}

\section{G-Additive codes as stabilizer codes}
\olabel{stabilizer}
\label{quditgrcodes_sct15}

The stabilizer formalism introduced by Gottesman in \cite{quantph.9705052} for $D=2$
(qubits) provides a compact and powerful way of generating quantum error
correcting codes. It has been extended to cases where $D$ is prime or a prime
power in \cite{IEEE.47.3065, quantph.0610267, IEEE.51.4892}.  In \cite{quantph.0111080} stabilizer
codes were extended in a very general fashion to arbitrary $D$ from a point of
view that includes encoding.  However, our approach to graph codes is somewhat
different, see Sec.~\ref{quditgrcodes_sct1}, and hence its connection with stabilizers
deserves a separate discussion. We will show that for any $D\geq 2$ a
G-additive (as defined near the end of Sec.~\ref{quditgrcodes_sct7}) code is a stabilizer
code, and the stabilizer is effectively a dual representation of the code.

The Pauli group $\PC$ for general $n$ and $D$ was defined in Sec.~\ref{quditgrcodes_sct3}.
Relative to this group we define a \emph{stabilizer}\index{stabilizer code} code (not necessarily a
graph code) $\CC$ to be a $K\geq 1$-dimensional subspace of the Hilbert space
satisfying three conditions:

\begin{description}

\item[C1.] There is a subgroup $\SC$ of $\PC$ such that
for \emph{every} $T$ in $\SC$ and \emph{every} $\ket{\psi}$ in $\CC$
\begin{equation}
\olabel{stabCondition1}
\label{quditgrcodes_eqn20}
 T \ket{\psi} = \ket{\psi}
\end{equation}

\item[C2.] The subgroup $\SC$ is maximal in the sense that every $T$ in $\PC$
  for which \eqref{quditgrcodes_eqn20} is satisfied for all $\ket{\psi}\in\CC$ belongs to
  $\SC$.

\item[C3.]  The coding space $\CC$ is maximal in the sense that any ket
  $\ket{\psi}$ that satisfies \eqref{quditgrcodes_eqn20} for every $T\in\SC$ lies in
  $\CC$.
\end{description}

If these conditions are fulfilled we call $\SC$ the \emph{stabilizer}\index{stabilizer, of the code} of the
code $\CC$. That it is Abelian follows from \eqref{quditgrcodes_eqn4}, since for $K>0$
there is some nonzero $\ket{\psi}$ satisfying \eqref{quditgrcodes_eqn20}.  One can also
replace \eqref{quditgrcodes_eqn20} with
\begin{equation}
\label{quditgrcodes_eqn21}
 T \ket{c_q} = \ket{c_q}
\end{equation}
where the $\{\ket{c_q}\}$ form an orthonormal basis of $\CC$. Note that one
can always find a subgroup $\SC$ of $\PC$ satisfying C1 and C2 for any
subspace $\CC$ of the Hilbert space, but it might consist of nothing but the
identity. Thus it is condition C3 that distinguishes stabilizer codes from
nonadditive codes.  A stabilizer code is uniquely determined by $\SC$ as
well as by $\CC$, since $\SC$ determines $\CC$ through C3.

As we shall see, the stabilizers of G-additive graph codes can be described in
a fairly simple way.  Let us begin with one qudit, $n=1$, where the trivial
graph $G$ has no edges, and the graph basis states are of the form
$\{Z^c\ket{+}\}$ for $c$ in some collection $C$ of integers in the range
$0\leq c\leq D-1$. The subgroup $\SC$ of $\PC$ satisfying C1 and C2 must be of
the form $\{X^s\}$ for certain values of $s$, $0\leq s\leq D-1$, belonging to
a collection $S$.  This is because $Z$ and its powers map any state
$Z^c\ket{+}$ to an orthogonal state, and hence $T$ in \eqref{quditgrcodes_eqn21} cannot
possibly contain a (nontrivial) power of $Z$.  Furthermore, since
\begin{equation}
\label{quditgrcodes_eqn22}
  X^s Z^c \ket{+} = \omega^{cs}Z^c\ket{+},
\end{equation}
see \eqref{quditgrcodes_eqn2}, $X^s$ will leave $\{Z^c\ket{+}\}$ unchanged only if
$\omega^{cs}=1$, or
\begin{equation}
\label{quditgrcodes_eqn23}
cs \equiv 0 \pmod{D}.
\end{equation}
Thus for $\SC$ to satisfy C1, it is necessary and sufficient that
\eqref{quditgrcodes_eqn23} hold for every $c\in C$, as well as every $s\in S$. Further,
$\SC=\{X^s\}$ is maximal in the sense of C2 only if $S$ contains every $s$
satisfying \eqref{quditgrcodes_eqn23} for each $c\in C$. As shown in App.~\ref{quditgrcodes_sct20}, such
a collection $S$ must either (depending on $C$) consist of $s=0$ alone, or
consist of the integer multiples $\nu s_1$, with $\nu=0,1,\ldots (D/s_1-1)$, of
some $s_1>0$ that divides $D$.  In either case, $S$ is a subgroup of the group
$\mathbb{Z}_D$ of integers under addition mod $D$, and indeed any such subgroup
must have the form just described.

We now take up C3. Given the maximal collection $S$ of solutions to
\eqref{quditgrcodes_eqn23}, we can in turn ask for the collection of $C'$ of integers $c$
in the range $0$ to $D-1$ that satisfy \eqref{quditgrcodes_eqn23} for every $s$ in $S$.
Obviously, $C'$ contains $C$, but as shown in App.~\ref{quditgrcodes_sct20}, $C'=C$ if and
only if $C$ is a subgroup of $\mathbb{Z}_D$, i.e., $\CC$ is G-additive. Next
note that every $T$ in $\SC$, as it is a power of $X$ and because of
\eqref{quditgrcodes_eqn22}, maps every graph basis state to itself, up to a phase. Thus when
(and only when) $\CC$ is G-additive, the codewords are just those graph basis
states for which this phase is 1 for every $T\in\SC$.  To check C3, expand an
arbitrary $\ket{\psi}$ in the graph basis. Then $T\ket{\psi}=\ket{\psi}$ for
all $T\in\SC$ means that all coefficients must vanish for graph basis states
that do not belong to $\CC$. Hence C3 is satisfied if and only if $\CC$ is
G-additive.

The preceding analysis generalizes immediately to $n>1$ in the case of the
trivial graph $G^0$ with no edges.  A graph code $\CC$ has a basis of the form
$\{Z^{\vect{c}}\ket{G^0}\}$ for a collection $C$ of integer $n$-tuples
$\vect{c} \in \mathbb{Z}_D^n$, and is G-additive when the collection
$C=\{\vect{c}\}$ is closed under component-wise addition mod $D$, i.e., is a
subgroup of $\mathbb{Z}_D^n$. Whether or not $\CC$ is G-additive, the subgroup
$\SC$ of $\PC$ satisfying C1 and C2 consists of all operators of the form
$X^{\vect{s}}=X_1^{s_1}X_2^{s_2}\cdots$ with the $n$-tuple $\vect{s}$
satisfying
\begin{equation}
\label{quditgrcodes_eqn24}
\vect{c}\cdot\vect{s} := \sum_{l=1}^n c_l s_l \equiv \vect{0} \pmod{D}
\end{equation}
for every $\vect{c}\in C$.  Just as for $n=1$, $\SC$ cannot contain Pauli
products with (nontrivial) powers of $Z$ operators. Let $S$ denote the
collection of all such $\vect{s}$. The linearity of \eqref{quditgrcodes_eqn24} means $S$ is
an additive subgroup of $\mathbb{Z}_D^n$.

One can also regard \eqref{quditgrcodes_eqn24} as a set of conditions, one for every
$\vect{s}\in S$, that are satisfied by certain $\vect{c}\in\mathbb{Z}_D^n$.
The set $C'$ of all these solutions is itself an additive subgroup of
$\mathbb{Z}_D^n$, and contains $C$.  In App.~\ref{quditgrcodes_sct20} we show that $C'=C$
if and only if $C$ (the collection we began with) is an additive subgroup
of $\mathbb{Z}_D^n$, and when this is the case the sizes of $C$ and $S$ are
related by
\begin{equation}
\label{quditgrcodes_eqn25}
 |C|\cdot|S| = D^n.
\end{equation}
Just as for $n=1$, any $X^{\vect{s}}$ maps a graph basis state for the trivial
graph $G^0$---they are all product states---onto itself up to a multiplicative
phase, and the same argument used above for $n=1$ shows that C3 is satisfied
for all $T\in\SC$ if and only if $\CC$ is G-additive.

To apply these results to a general graph $G$ on $n$ qubits, note that the
unitary $\UC$ defined in \eqref{quditgrcodes_eqn7} provides, through \eqref{quditgrcodes_eqn6} and
\eqref{quditgrcodes_eqn10}, a one-to-one map of the graph basis states of the trivial $G^0$
onto the graph basis states of $G$. At the same time the one-to-one map $\UC P
\UC^\dagger$ carries the $\SC$ satisfying C1 and C2 (and possibly C3) for the
$G^0$ code to the corresponding $\SC$, satisfying the same conditions for the
$G$ code. (The reverse maps are obtained by interchanging $\UC^\dagger$ and
$\UC$.)  Consequently, the results obtained for $G^0$ apply at once to $G$,
and the transformation allows the elements of the stabilizer for the $G$ graph
code to be characterized by integer $n$-tuples $\vect{s}$ satisfying
\eqref{quditgrcodes_eqn24}. Thus we have shown that G-additive codes are stabilizer codes,
and for these the coding space and stabilizer group descriptions are dual,
related by \eqref{quditgrcodes_eqn24}: each can be derived from the other.

\section{Conclusion and discussion}
\olabel{conclusion}
\label{quditgrcodes_sct16}

In this chapter we have developed an approach to graph codes which works for
qudits with general dimension $D$, and employs graphical methods to search for
specific examples of such codes.  It is similar to the approaches developed
independently in \cite{IEEE.55.433,quantph.0709.1780,PhysRevA.78.012306}.  We have used it for computer searches
on graphs with a relatively small number $n$ of qudits, and also to construct
certain families of graphs yielding optimum distance $\dist=2$ codes for
various values of $D$ and $n$ which can be arbitrarily large.  It remains a
challenging problem to do the same for codes with distance $\dist>2$.

In a number of cases we have been able to construct what we call quantum
Singleton (QS) codes that saturate the quantum Singleton bound \cite{PhysRevA.55.900}:
these include the $\dist=2$ codes for arbitrarily large $n$ and $D$ mentioned
above, and also a number of $\dist=3$ codes in the case of $D=3$ (qutrits),
see Tables~\ref{quditgrcodes_tbl1} and \ref{quditgrcodes_tbl2}.  The results for cycle graphs for $D=3$
and $\dist=3$ in Table~\ref{quditgrcodes_tbl1} are interesting in that the QS bound is
saturated for $n\leq 10$, but fails for $n=11$, as it must for nondegenerate
codes; see the discussion in Sec.~\ref{quditgrcodes_sct12}.  Our results are consistent
with the difficulty of finding QS codes for larger $\dist$ \cite{Grassl:codetables}, but
suggest that increasing $D$ may help, as observed in \cite{PhysRevA.65.012308}.  It
is worth noting that we have managed to construct many of the previously known
nonadditive codes, or at least codes with the same $(\!(n,K,\dist)\!)_D$,
using simple graphs. Some other nonadditive codes not discussed here, such as
the $(\!(10, 24, 3)\!)_2$ code in \cite{quantph.0709.1780}, can also be obtained from
suitably chosen graphs.  While all these results are encouraging, they
represent only a beginning in terms of understanding what properties of graphs
lead to good graph codes, and how one might efficiently construct such codes
with arbitrarily large $n$ and $\dist$, for various $D$.

As noted in Sec.~\ref{quditgrcodes_sct7}, all graph codes with distance $\dist\leq
\Delta'$, where $\Delta'$ is the diagonal distance of the graph, are
necessarily nondegenerate, and our methods developed for such codes will (in
principle) find them all.  All codes with $\dist> \Delta'$ are necessarily
degenerate codes, and their systematic study awaits further work.
It should be noted that our extension of graph codes to $D>2$ is based on
extending Pauli operators in the manner indicated in \cite{PhysRevA.71.042315}.  Though
the extension seems fairly natural, and it is hard to think of alternatives
when $D$ is prime, there are other ways to approach the matter when $D$ is
composite (including prime powers), which could yield larger or at least
different codes, so this is a matter worth exploring.

The relationship between stabilizer (or additive) codes and G-additive
(as defined in Sec.~\ref{quditgrcodes_sct7}) graph codes has been clarified by showing
that they are dual representations, connected through a simple equation,
\eqref{quditgrcodes_eqn24}, of the same thing.  One might suspect that such duality extends
to nongraphical stabilizer codes, but we have not studied the problem outside
the context of graph codes. Nonadditive codes, which---if one uses our
definition, Sec.~\ref{quditgrcodes_sct15}---do not have stabilizers, are sometimes of larger
size than additive codes, so they certainly need to be taken into account in
the search for optimal codes. The graph formalism employed here works in
either case, but computer searches are much faster for additive codes.

\begin{subappendices}
\section{The X-Z rule and related}
\olabel{app:xzRule}
\label{quditgrcodes_sct17}

\noindent
\textbf{X-Z Rule}\index{X-Z rule}.  \emph{Acting with an $X$ operator on the $i'th$ qudit of a
  graph state $\ket{G}$ produces the same graph basis state as the action of
  $Z$ operators on the neighbors of qudit $i$, raised to the power given by
  the edge multiplicities $\Gamma_{im}$.}

The operator $X_i$ commutes with $\CP_{lm}$ when $i \neq l$ and $i \neq m$, but
if $i=l$ (or similarly $i=m$) one can show using \eqref{quditgrcodes_eqn5} and \eqref{quditgrcodes_eqn1}
that
\begin{equation}
\olabel{apxCeqn2}
\label{quditgrcodes_eqn26}
X_{l} \CP_{lm}  =  \CP_{lm} Z_m X_l = Z_m \CP_{lm} X_{l}.
\end{equation}
That is, an $X_i$ operator can be pushed from left to right through a
$\CP_{lm}$ with at most the cost of producing a $Z$ operator associated with
the \emph{other} qudit: if $i=l$ one gets $Z_m$, if $i=m$ one gets $Z_l$.
Since all $Z$ commute with all $\CP$, one can place the resulting $Z_m$ either
to the left or to the right of $\CP_{lm}$.

Now consider pushing $X_i$ from the left to the right through $\mathcal{U}$,
the product of $\CP_{lm}$ operators defined in \eqref{quditgrcodes_eqn7}. Using
\eqref{quditgrcodes_eqn26} successively for those $\CP_{lm}$ that do not commute with
$X_i$, one sees that this can be done at the cost of generating a $Z_m$ for
every edge of the graph connecting $i$ to another vertex $m$. Let the product
of these be denoted as $\hat Z:= \prod_{(l=i,m) \in E}Z_{m}^{\Gamma_{lm}}$.
Then, with definition \eqref{quditgrcodes_eqn6}, we can show
\begin{align}
\label{quditgrcodes_eqn27}
X_i \ket{G} &= X_i \mathcal{U} \ket{G^0} = \hat Z \mathcal{U} X_i \ket{G^0}
\notag\\
 &= \hat Z\mathcal{U}\ket{G^0} =\hat Z \ket{G},
\end{align}
which completes the proof of the X-Z Rule.

For graph codes satisfying \eqref{quditgrcodes_eqn13}, the X-Z Rule leads to the:

\noindent
\textbf{Coordination bound}. \emph{The diagonal distance $\Delta'$ for a graph
  $G$ cannot exceed $\nu+1$, where $\nu$ is the minimum over all vertices of
  the number of neighbors of a vertex, this being the number of vertices
  joined to the one in question by edges, possibly of multiplicity greater
  than 1.}

To make the counting absolutely clear consider Fig.~\ref{quditgrcodes_fgr1}, where the
vertex on the left has 3 neighbors, and each of the others has 1 neighbor, so
that in this case $\nu=1$.  To derive the bound, apply $X$ to a vertex which
has $\nu$ neighbors.  By the X-Z rule the result is the same as applying
appropriate powers of $Z$ to each neighbor.  Let $P$ be this $X$ tensored with
appropriate compensating powers of $Z$ at the neighboring vertices in such a
way that $P\ket{G}=\ket{G}$. The size of $P$ is $\nu+1$, and $\Delta'$ can be
no larger.  \smallskip

Another useful result follows from the method of proof of the X-Z Rule:

\noindent
\textbf{Paulis to Paulis}. \emph{Let $P$ be a Pauli product \eqref{quditgrcodes_eqn3},
  and for $\UC$ defined in \eqref{quditgrcodes_eqn7} let
\begin{equation}
\label{quditgrcodes_eqn28}
P' = \UC^\dagger P \UC,\quad P'' = \UC P \UC^\dagger.
\end{equation}
Then both $P'$ and $P''$ are Pauli products.}

To see why this works, rewrite the first equality as $\UC P' = P \UC$, and
imagine pushing each of the single qudit operators, of the form
$X_j^{\mu_j}Z_j^{\nu_j}$, making up the product $P$ through $\UC$ from left to
right. This can always be done, see the discussion following \eqref{quditgrcodes_eqn26}, at
the cost of producing some additional $Z$ operators, which can be placed on
the right side of $\UC$, to make a contribution to $P'$.  At the end of the
pushing the final result can be rearranged in the order specified in
\eqref{quditgrcodes_eqn3} at the cost of some powers of $\omega$, see \eqref{quditgrcodes_eqn2}.  The
argument for $P''$ uses pushing in the opposite direction.

\section{Partition theorem proof}
\olabel{app:bipartite}
\label{quditgrcodes_sct18}

Given the partition of the $n$ qudits into sets $V_1$ and $V_2$ containing
$n_1$ and $n_2$ elements, the code of interest consists of the graph basis
states $\ket{ \vect{c}} = \ket{c_1, c_2, \ldots,c_n}$ satisfying the two
conditions
\begin{eqnarray}
\sum_{i \in V_1} c_i & \equiv & 0 \pmod{D}
\olabel{eq:parity1}
\label{quditgrcodes_eqn29}
\\
\sum_{j \in V_2} c_j & \equiv & 0 \pmod{D}
\olabel{eq:parity2}
\label{quditgrcodes_eqn30}
\end{eqnarray}
This code is additive and contains $K=D^{n_1-1} \times D^{n_2-1} = D^{n-2}$
codewords. (The counting can be done by noting that \eqref{quditgrcodes_eqn29} defines a
subgroup of the additive group $\mathbb{Z}_D^{n_1}$, and its cosets are
obtained by replacing 0 with some other integer on the right side of
\eqref{quditgrcodes_eqn29}.)

We first demonstrate that this code has $\dist\geq 2$ by showing that any
Pauli operator, except the identity, applied to a single qudit maps a codeword
into a graph basis state not in the code. If $Z^\nu$ for $0 < \nu < D$ is
applied to a qudit in $V_1$, the effect will be to replace 0 on the right side
of \eqref{quditgrcodes_eqn29} with $\nu$, so this graph state is not in the code. If
$X^\mu$, $0 < \mu < D$ is applied to a qudit in $V_1$ the result according to
the X-Z Rule, App.~\ref{quditgrcodes_sct17}, will be the same as placing $Z$ operators on
neighboring qudits in $V_2$ (as well as $V_1$) in such a way that 0 on the
right side of \eqref{quditgrcodes_eqn30} is replaced by $g\mu$, where $g$ is the total
number of edges (including multiplicities) joining the $V_1$ qudit with qudits
in $V_2$.  But as long as $g$ is coprime to $D$, as specified in the condition
for the theorem, $g\mu$ cannot be a multiple of $D$, and \eqref{quditgrcodes_eqn30} will no
longer be satisfied.  The same is true if $Z^\nu X^\mu$ is a applied to a
qudit in $V_1$.  Obviously the same arguments work for Pauli operators applied
to a single qudit in $V_2$. Thus we have shown that $\dist\geq 2$.

But $\dist>2$ is excluded by the QS bound, so we conclude that we have an
additive code of $K=D^{n-2}$ elements and distance $\dist=2$ that saturates
the QS bound.

\section{Construction of qubit star graph codes\olabel{app:qubitStarGraph}}
\label{quditgrcodes_sct19}

As noted in Sec.~\ref{quditgrcodes_sct11} a star graph for $n$-qubits consists of a central
vertex joined by edges to $n-1$ peripheral vertices.  Let $V_1$ be the
central vertex and $V_2$ the set of peripheral vertices.  When $n$ is even and
$D=2$ the conditions of the Partition theorem, Sec.~\ref{quditgrcodes_sct11}, are
satisfied, and the $\dist=2$ code constructed in App.~\ref{quditgrcodes_sct18} consists of
the $2^{n-2}$ graph basis states with no $Z$ on the central qubit and an even
number $r$ of $Z$'s on the peripheral qubits, thus satisfying \eqref{quditgrcodes_eqn29}
and \eqref{quditgrcodes_eqn30}, and yielding an additive QS code.

When $n$ is odd the central vertex is connected to an even number $n-1$ of
vertices in $V_2$, so the conditions of the Partition theorem no longer hold.
A reasonably large $\dist=2$ nonadditive code can, however, be constructed by
again assuming no codeword has $Z$ on the central qubit, and that the code
contains all graph basis states with $r$ $Z$'s on the peripheral qubits
\emph{for a certain selected set $R$ of $r$ values}.

The set $R$ must satisfy two conditions. First, it cannot contain both $r$ and
$r+1$, because applying an additional $Z$ to a codeword with $r$ $Z$'s yields
one with $r+1$, and one cannot have both of them in a code of distance
$\dist=2$.  Second, applying $X$ to the central vertex and using the
X-Z rule, App.~\ref{quditgrcodes_sct17}, maps a codeword with $r$ $Z$'s to one with
$r'=n-1-r$; hence $R$ cannot contain both $r$ and $n-1-r$.  For example, when
$n=7$ ($n-1=6$ peripheral qubits) the set $R=\{0,2,5\}$ satisfies both
conditions, as does $R=\{1,4,6\}$, whereas $R=\{1,2,6\}$ violates the first
condition and $R=\{1,3,5\}$ the second.

By considering examples of this sort, and noting that the number of such graph
basis states with $r$ $Z$'s is $\binom{n-1}{r}$ which is equal to
$\binom{n-1}{n-1-r}$, one sees that for $n$ odd one can construct in this way a
nonadditive code with
\begin{equation}
\label{quditgrcodes_eqn31}
\sum_{i=0}^{(n-3)/2} \binom{n-1}{i} =
 2^{n-2}  - \frac{1}{2} \binom{n-1}{(n-1)/2}
\end{equation}
codewords.

\section{Solutions to $\vect{c}\cdot\vect{s} \equiv 0 \pmod{D}$
\olabel{app:Gadditive}}
\label{quditgrcodes_sct20}

Let $\AC$ be the collection of all $n$-component integer vectors (i.e.,
$n$-tuples) of the form $\vect{a} = (a_1,a_2,\ldots a_n)$,
$0\leq a_j \leq D-1$, with component-wise sums and scalar multiplication
defined using arithmetic operations mod $D$. In particular, $\AC$ is a group
of order $D^n$ under component-wise addition mod $D$.  We shall be
interested in subsets $C$ and $\SC$ of $\AC$ that satisfy
\begin{equation}
  \vect{c}\cdot\vect{s} := \sum_{l=1}^n c_l s_l \equiv 0 \pmod{D}
\label{quditgrcodes_eqn32}
\end{equation}
for all $\vect{c}\in C$ and $\vect{s}\in S$.  Given some collection $C$,
we shall say that $S$ is \emph{maximal}\index{maximal, stabilizer} relative to $C$ if it includes
\emph{all} solutions $\vect{s}$ that satisfy \eqref{quditgrcodes_eqn32} for every
$\vect{c}\in C$. It is easily checked that a maximal $S$ is an additive
subgroup of $\AC$: it includes the zero vector and $-\vect{s}$ mod $D$ whenever
$\vect{s}\in S$.  A similar definition holds for $C$ being maximal relative to
a given $S$. We use $|C|$ to denote the number of elements in a set or
collection $C$.

\textbf{Theorem}.  \emph{Let $C$ be an additive subgroup of $\AC$, and let
  $S$ be maximal relative to $C$, i.e., the set of all $\vect{s}$ that
  satisfy \eqref{quditgrcodes_eqn32} for every $\vect{c}\in C$.  Then $C$ is also
  maximal relative to $S$, and
\begin{equation}
  |C|\cdot |S| = D^n.
\label{quditgrcodes_eqn33}
\end{equation}}

The proof is straightforward when $D$ is a prime, since $\mathbb{Z}_D$ is a
field, and one has the usual rules for a linear space. The composite case is
more difficult, and it is useful to start with $n=1$:

\textbf{Lemma}.
  Let $C$ be a subgroup under addition mod $D$ of the integers lying between
  $0$ and $D-1$, and $S$ all integers in the same range satisfying
\begin{equation}
  cs \equiv 0 \pmod{D}
\label{quditgrcodes_eqn34}
\end{equation}
for every $c\in C$. Then $C$ consists of \emph{all} integers $c$ in the
range of interest which satisfy \eqref{quditgrcodes_eqn34}, and $|C|\cdot |S|= D$.

When $C=\{0\}$ the proof is obvious, since $|C|=1$ and $|S|=D$.  Otherwise,
because it is an additive subgroup of $\mathbb{Z}_D$, $C$ consists of the
multiples $\{\mu c_1\}$ of the smallest positive integer $c_1$ in $C$,
necessarily a divisor of $D$, when $\mu$ takes the values $0, 1, \ldots
s_1-1$, where $s_1=D/c_1$.  One quickly checks that all integer multiples
$s=\nu s_1$ of this $s_1$ satisfy \eqref{quditgrcodes_eqn34} and are thus contained in $S$.
But $S$ is also an additive subgroup, and $s_1$ is its minimal positive
element (except in the trivial case $c_1=1$), for were there some smaller
positive integer $s'$ in $S$ we would have $0<c_1s' <D$, contradicting
\eqref{quditgrcodes_eqn34}.  Similarly there is no way to add any additional integers to
$C$ while preserving the subgroup structure under addition mod $D$ without
including a positive $c$ less than $c_1$, which will not satisfy \eqref{quditgrcodes_eqn34}
for $s=s_1$.

For $n>1$ it is helpful to use a \emph{generator matrix}\index{generator matrix} $F$, with components
$F_{rl}$, each between $0$ and $D-1$, with the property that $\vect{c}\in C$
if and only if it can be expressed as linear combinations of rows of $F$,
i.e.,
\begin{equation}
  c_l \equiv \sum_r b_r F_{rl} \pmod{D}
\label{quditgrcodes_eqn35}
\end{equation}
for a suitable collection of integers $\{b_r\}$. This collection will of course
depend on the $\vect{c}$ in question, and for a given $\vect{c}$ need not be
unique, even assuming (as we shall) that $0\leq b_r\leq D-1$.  In particular
the matrix $F$ for which each row is a distinct $\vect{c}$ in $C$, with $r$
running from $1$ to $|C|$, is a generator matrix.  It is straightforward to
show that if $F$ is any generator matrix for $C$, $S$ consists of all solutions
$\vect{s}$ to the equations
\begin{equation}
  \sum_{l=1}^n F_{rl} s_l \equiv 0 \pmod{d} \text{ for } r=1, 2, \ldots.
\label{quditgrcodes_eqn36}
\end{equation}

The collections $C$ and $S$, vectors of the form \eqref{quditgrcodes_eqn35} and those
satisfying \eqref{quditgrcodes_eqn36}, remain the same if $F$ is replaced by another
generator matrix $F'$ obtained by one of the following \emph{row operations}\index{row operations}:
(i) permuting two rows; (ii) multiplying (mod $D$) any row by an
\emph{invertible}\index{invertible integer} integer, i.e., an integer which has a multiplicative inverse
mod $D$; (iii) adding (mod $D$) to one row an \emph{arbitrary} multiple (mod
$D$) of a different row; (iv) discarding (or adding) any row that is all
zeros, to get a matrix of a different size. Of these, (i) and (iv) are
obvious, and (ii) is straightforward. For (iii), consider what happens if the
second row of $F$ is added to the first, so that $F'_{rl}=F^{}_{rl}$ except
for
\begin{equation}
  F'_{1l} \equiv F^{}_{1l}+ F^{}_{2l} \pmod{D}.
\label{quditgrcodes_eqn37}
\end{equation}
Then setting
\begin{equation}
  b'_1=b_1,\; b'_2\equiv b_2-b_1\pmod{d},\; b'_l=b_l \text{ for } l\geq 3
\label{quditgrcodes_eqn38}
\end{equation}
leads to the same $\vect{c}$ in \eqref{quditgrcodes_eqn35} if $b$ and $F$ are replaced by
$b'$ and $F'$ on the right side. Likewise, any $\vect{c}$ that can be written
as a linear combination of $F'$ rows can be written as a combination of those
of $F$, so the two matrices generate the same collection $C$, and hence have
the same solution set $S$ to \eqref{quditgrcodes_eqn36}.  Since adding to one row a
different row can be repeated an arbitrary number of times, (iii) holds for
an arbitrary (not simply an invertible) multiple of a row.

The corresponding column operations on a generator matrix are (i) permuting
two columns; (ii) multiplying a column by an invertible integer; (iii) adding
(mod $D$) to one column an arbitrary multiple (mod $D$) of a different column.
Throwing away (or adding) columns of zeros is \emph{not} an allowed operation.
When column operations are carried out to produce a new $F'$ from $F$, the new
collections $C'$ and $S'$ obtained using \eqref{quditgrcodes_eqn35} and \eqref{quditgrcodes_eqn36}
will in general be different, but $C'$ is an additive subgroup of the same
size (order), $|C'|=|C|$, and likewise $|S'|=|S|$. The argument is
straightforward for (i) and (ii), and for (iii) it is an easy exercise to show
that if the second column of $F$ is added to the first to produce $F'$, the
collection $C$ is mapped into $C'$ by the map
\begin{equation}
 c_1' \equiv c_1+c_2 \pmod{D}\: ;\quad c'_l = c_l \text{ for } l\geq 2
\label{quditgrcodes_eqn39}
\end{equation}
whose inverse will map $C'$ into $C$ when one generates $F$ from $F'$ by
subtracting the second column from the first. Thus $|C|= |C'|$.
The same strategy shows that $|S'|=|S|$; instead of \eqref{quditgrcodes_eqn39}
use $s'_2\equiv s_2-s_1\pmod{D}$, and $s'_l=s_l$ for $l\neq 2$.

The row and column operations can be used to transform the generator matrix to
a (non unique) diagonal form, in the following fashion.  If each $F_{rl}$ is
zero the problem is trivial.  Otherwise use row and column permutations so
that the smallest positive integer $f$ in the matrix is in the upper left
corner $r=1=l$. Suppose $f$ does not divide some element, say $F_{13}$, in the
first row.  Then by subtracting a suitable multiple of the first column from
the third column we obtain a new generator $F'$ with $0<F'_{13}<f$, and
interchanging the first and third columns we have a generator with a smaller,
but still positive, element in the upper left corner.  Continue in this
fashion, considering both the first row and the first column, until the upper
left element of the transformed generator divides \emph{every} element in
both. When this is the case, subtracting multiples of the first column from
the other columns, and multiples of the first row from the other rows, will
yield a matrix with all zeros in the first row and first column, apart from
the nonzero upper left element at $r=1=l$, completing the first step of
diagonalization.

Next apply the same overall strategy to the sub matrix obtained by ignoring the
first row and column.  Continuing the process of diagonalization and
discarding rows that are all zero (or perhaps adding them back in again), one
arrives at a diagonal $n\times n$ generator matrix
\begin{equation}
  \hat F_{rl} = f_l\delta_{rl},
\label{quditgrcodes_eqn40}
\end{equation}
where some of the $f_l$ may be zero.  The counting problem is now much
simplified, because for each $l$ $c_l$ can be any multiple mod $D$ of $f_l$,
and $s_l$ any solution to $f_l s_l\equiv 0 \pmod{D}$, independent of what
happens for a different $l$.  Denoting these two collections by $C_l$ and
$S_l$, the lemma implies that $|C_l|\cdot|S_l|=D$ for every $l$, and taking
the product over $l$ from $1$ to $n$ yields \eqref{quditgrcodes_eqn33}.  This in turn
implies that $C$ consists of \emph{all possible} $\vect{c}$ that satisfy
\eqref{quditgrcodes_eqn32} for all the $\vect{s}\in S$. To see this, note that the size
$|C|$ of $C$ is $D^n/|S|$.  If we interchange the roles of $C$ and $S$ in the
above argument (using a generator matrix for $S$, etc.), we again come to the
result \eqref{quditgrcodes_eqn33}, this time interpreting $|C|$ as the number of solutions
to \eqref{quditgrcodes_eqn32} with $S$ given.  Thus since it cannot be made any larger, the
original additive subgroup $C$ we started with is maximal relative to $S$.
This completes the proof.
\end{subappendices}

\chapter{Location of quantum information in additive graph codes\label{chp6}}

\section{Introduction\label{locinf_sct1}}


Quantum codes in which quantum information is redundantly encoded in a
collection of code carriers play an important role in quantum information, in
particular in systems for error correction and in schemes for quantum
communication \cite{PhysRevA.51.2738, PhysRevA.52.R2493, PhysRevLett.76.722, PhysRevA.54.1098}. They are a generalization of the classical codes well known and widely used in everyday communication systems \cite{MacWilliamsSloane:TheoryECC}. While for the latter it is fairly obvious where the information is located,
the quantum case is more complicated for two reasons.  First, a quantum
Hilbert space with its non-commuting operators is a more complex mathematical
structure than the strings of bits or other integers used in classical codes.
Second, the very concept of ``information'' is not easy to define in the
quantum case.  However, in certain cases one is able to make quite precise
statements.  Thus in the five qubit code \cite{PhysRevLett.77.198} that
encodes one qubit of information, none of the encoded information is present
in any two qubits taken by themselves, whereas all the information can be
recovered from any set of three qubits \cite{PhysRevA.71.042337}.


Similar precise statements can be made, as we shall see, in the case of an
\emph{additive graph code}\index{additive graph code} on a collection of $n$ qudits which constitute the
\emph{carriers}\index{carrier qudits} of the code, provided each qudit has the same dimension $D$,
with $D$ some integer greater than one (not necessarily prime). It was shown
in \cite{PhysRevA.78.042303} that all additive graph codes are stabilizer
codes, and in \cite{quantph.0111080,quantph.0703112} that all stabilizer codes
are equivalent to graph codes for prime $D$. A detailed discussion of
non-binary quantum error correcting codes can be found in
\cite{quantph.9802007,IEEE.45.1827,PhysRevA.65.012308,
  PhysRevA.78.012306, PhysRevA.78.042303}.  The five qubit code just mentioned
is an example of a quantum code that is locally equivalent to an additive
graph code \cite{PhysRevA.65.012308}, and the information location has an
``all or nothing'' character.  In general the situation is more
interesting in that some subset of carriers may contain some but not all of
the encoded information, and what is present can be either ``classical'' or
``quantum,'' or a mixture of the two. Since many of the best codes currently known are additive graph
codes, identifying the location of information could prove useful when
utilizing codes for error correction, or designing new or better codes, or
codes that correct some types of errors more efficiently than others
\cite{PhysRevA.75.032345}. Our formalism can also be applied to study quantum secret sharing schemes employing graph states and can even handle a more general setting where there might be subsets that contain partial information and hence are neither authorized (contain the whole quantum secret) nor unauthorized (contain no information whatsoever about the secret).


Our approach to the problem of information location is algebraic, based upon
the fact that generalized Pauli operators on the Hilbert space of the carriers
form a group.  Subgroups of this group can be associated with different types
of information, and the information available in some subset of the carriers
can also be identified with, or is isomorphic to, an appropriate subgroup, as
indicated in the isomorphism theorem of Sec.~\ref{locinf_sct10}.  In the process of
deriving this theorem we go through a series of steps which amount to an
\emph{encoding procedure}\index{encoding procedure} that takes the initial quantum information and
places it in the coding subspace of the carrier Hilbert space.  These steps
can in turn be transformed into a set of quantum gates to produce an explicit
circuit that carries out the encoding.  This result, although somewhat
subsidiary to our main aims, is itself not without interest, and is an alternative to a previous scheme \cite{PhysRevA.65.012308} limited to prime $D$.


There have been some previous studies of quantum channels using an algebraic
approach similar to that employed here.  Those most closely related to our
work are by B\'{e}ny et al.\ \cite{PhysRevLett.98.100502,PhysRevA.76.042303}
(and see B\'{e}ny \cite{quantph.0907.4207}) and Blume-Kohout et al.\
\cite{PhysRevLett.100.030501}.  These authors have provided a set of very
general conditions under which an algebraic structure is preserved by a
channel.  In App.~\ref{locinf_apdx4} we show that our results fit within the
framework of a ``correctable algebra'' as defined in
\cite{PhysRevLett.98.100502,PhysRevA.76.042303,quantph.0907.4207}. See also the
remarks in Sec.~\ref{locinf_sct18}.


The remainder of this chapter is organized as follows.  Some general comments
about types of quantum information and their connection with certain ideal
quantum channels are found in Sec.~\ref{locinf_sct2}. Section \ref{locinf_sct3} contains
definitions of the Pauli group and of some quantum gates used later in the
chapter.  The formalism associated with additive graph codes as well as our
encoding strategy is in Sec.~\ref{locinf_sct6}; this along with some results on
partial traces leads to the fundamental isomorphism result in Sec.~\ref{locinf_sct10},
which also indicates some of its consequences for the types of information
discussed in Sec.~\ref{locinf_sct2}.  Section~\ref{locinf_sct14} contains various
applications to specific codes, for both qubit and qudit carriers.  Finally,
Sec.~\ref{locinf_sct18} contains a summary, conclusions, and some open questions.
Appendices~\ref{locinf_apdx1} and \ref{locinf_apdx2} contain longer proofs of theorems,
App.~\ref{locinf_apdx3} presents an efficient linear algebra based algorithm for
working out the results for any additive graph code, and App.~\ref{locinf_apdx4}
illustrates the connection with related work in \cite{PhysRevLett.98.100502}
and \cite{PhysRevA.76.042303}.

\section{Types of information}
\label{locinf_sct2}


Both classical and quantum information theory have to do with statistical
correlations between properties of two or more systems, or properties of a
single system at two or more times.  In the classical case information is
always related to a possible set of physical properties that are distinct and
mutually exclusive---e.g., the voltage has one of a certain number of
values---with one and only one of these properties realized in a particular
system at a particular time.  For quantum systems it is useful to distinguish
different \emph{types}\index{types (species) of information} or \emph{species} of information
\cite{PhysRevA.76.062320}, each corresponding to a collection of mutually
distinct properties represented by a (projective) decomposition $\JC=\{J_j\}$
of the identity $I$ on the relevant Hilbert space $\HC$:
\begin{equation}
\label{locinf_eqn1}
  I = \sum_j J_j,\quad J_j = J_j^\dagger = J_j^2,
\quad J_j J_k = \delta_{jk} J_j.
\end{equation}
Any normal operator $M$ has a spectral representation of the form
\begin{equation}
\label{locinf_eqn2}
M = \sum_j \mu_j J_j,
\end{equation}
where the $\mu_j$ are its eigenvalues, and the decomposition $\{J_j\}$ is
uniquely specified by requiring $\mu_j\neq \mu_k$ when $j\neq k$. This means
one can sensibly speak about the type of information $\JC(M)$ \emph{associated
  with} a normal operator $M$.  When $M$ is Hermitian this is the kind of
information obtained by measuring $M$.


This terminology allows one to discuss the transmission of information through
a quantum channel in the following way.  Let $\EC$ be the
completely positive, trace preserving superoperator that maps the space of
operators $\LC(\HC)$ of the channel input onto the corresponding operator
space $\LC(\HC')$ of the channel output $\HC'$ (which may have a different
dimension from $\HC$).  Provided
\begin{equation}
  \EC(J_j) \EC(J_k) = 0 \text{ for } j\neq k,
\label{locinf_eqn3}
\end{equation}
for all the operators $\{J_j\}$ associated with a decomposition $\JC$ of the
$\HC$ identity, we shall say the channel is \emph{ideal}\index{ideal (noiseless) channel} or \emph{noiseless}
for the $\JC$ species of information, or, equivalently, the $\JC$ type of
information is \emph{perfectly present}\index{perfectly present} in the channel output $\HC'$.
Formally, each physical property $J_j$ at the input corresponds in a
one-to-one fashion to a unique property, the support of $\EC(J_j)$ (or the
corresponding projector) at the output.  Thus we have a quantum version of a
noiseless classical channel, a device for transmitting symbols, in this case
the label $j$ on $J_j$, from the input to the output by associating distinct
symbols with distinct physical properties---possibly a different collection of
properties at the output than at the input.


The opposite extreme from a noiseless channel is one in which $\EC(J_j)$ is
\emph{independent of $j$} up to a multiplicative constant.  In this case no
information of type $\JC$ is available at the channel output: the channel is
\emph{blocked}\index{completely noisy (blocked) channel}, or completely noisy; equivalently, the $\JC$ species of
information is \emph{absent}\index{absent, information} from the channel output.  Hereafter we shall
always use ``absent'' in the strong sense of ``completely absent'', and the
term \emph{present}\index{present (partially present), information}, or \emph{partially present} for situations in which some
type of information is not (completely) absent but is also not perfectly
present: i.e., the channel is noisy but not completely blocked for this type
of information. 


In some cases all the projectors in $\{J_j\}$ will be of rank 1, onto pure
states, but in other cases some or all of them may be of higher rank, in which
case one may have a \emph{refinement}\index{refinement} $\LC = \{L_l\}$ of $\{J_j\}$ such that
each projector $J_j$ is a sum of one or more projectors from the $\LC$
decomposition.  It is then clear that if the $\LC$ information is absent/perfectly present from/in the channel output the same is true of the $\JC$
information, but the converse need not hold.  Thus it may be that the coarse grained $\JC$
information is perfectly present, but no additional information is available
about the refinement.  A particularly simple situation, which we will
encounter later, is one in which the output $\HC'$ is itself a tensor product,
say $\HC'_1\ot\HC'_2$, $\JC$ a decomposition of $\HC'_1$, $\JC=\{J_j\otimes I\}$ and $\KC$ a decomposition of $\HC'_2$, $\KC=\{I\otimes K_k\}$. It can then be the case that the information associated with
the $\JC$ decomposition is perfectly present and that associated with the 
$\KC$ decomposition is (perfectly) absent from the channel output. 


Suppose $\JC = \{J_j\}$ and $\KC = \{K_k\}$ are two types of quantum
information defined on the same Hilbert space. The species $\JC$ and $\KC$ are
\emph{compatible}\index{compatible, types of information} if all the projectors in $\JC$ commute with all the
projectors in $\KC$, in which case the distinct nonzero projectors in the
collection $\{J_j K_k\}$ provide a common refinement of the type discussed
above. Otherwise, if some projectors in one collection do not commute with
certain projectors in the other, $\JC$ and $\KC$ are \emph{incompatible}\index{incompatible, types of information} and
cannot be combined with each other. This is an example of the single framework
rule of consistent quantum reasoning, \cite{PhysRevA.54.2759} or Ch.~16 of
\cite{RBGriffiths:ConsistentQuantumTheory}.  The same channel may be ideal for
some $\JC$ and blocked for some $\KC$, or noisy for both but with different
amounts of noise.  From a quantum perspective, classical information theory is
only concerned with a single type of (quantum) information, or several
compatible types which possess a common refinement, whereas the task of
quantum, in contrast to classical, information theory is to analyze situations
where multiple incompatible types need to be considered.  


The term ``classical information'' when used in a quantum context can be
ambiguous or misleading.  Generally it is used when only a single type of
information, corresponding to a single decomposition of the identity, suffices
to describe what emerges from a channel, and other incompatible types can
therefore be ignored.  Even in such cases it is helpful to indicate explicitly
which decomposition of the identity is involved if that is not obvious from
the context. The contrasting term ``quantum information'' can then refer to
situations where two or more types of information corresponding to
incompatible decompositions are involved, and again it is helpful to be
explicit about what one has in mind if there is any danger of ambiguity.


An \emph{ideal quantum channel}\index{ideal quantum channel} is one in which there is an isometry $V$ from
$\HC$ to $\HC'$ such that 
\begin{equation}
  \EC(A) = VAV^\dagger
\label{locinf_eqn4}
\end{equation}
for every operator $A$ on $\HC$. In this case the superoperator $\EC$
preserves not only sums but also operator products:
\begin{equation}
  \EC(AB) = \EC(A) \EC(B).
\label{locinf_eqn5}
\end{equation}
Conversely, if \eqref{locinf_eqn5} holds for any pair of operators, one can show that
the quantum channel is ideal \cite{PhysRevLett.98.100502,PhysRevA.76.042303},
i.e. $\EC$ has the form \eqref{locinf_eqn4}.  As the isometry maps orthogonal
projectors to orthogonal projectors, \eqref{locinf_eqn3} will be satisfied for every
species of information, and we shall say that \emph{all} information is
perfectly present at the channel output.  The converse, that a channel which
is ideal for all species, or even for an appropriately chosen pair of
incompatible species is an ideal quantum channel, is also correct; see
\cite{PhysRevA.71.042337, PhysRevA.76.062320}.


The preservation of operator products, \eqref{locinf_eqn5}, can be a very useful tool
in checking for the presence or absence of various types of information in the
channel output, as we shall see in Sec.~\ref{locinf_sct10}.  When \eqref{locinf_eqn5} holds
for arbitrary $A$ and $B$ belonging to a particular decomposition of the
identity, this suffices to show that the channel is ideal for this
species. However, note that this sufficient condition is not necessary, since
\eqref{locinf_eqn3} could hold without the $\EC(A_j)$ being projectors, in which case
$\EC(A_j^2)$ is not mapped to $\EC(A_j)^2$.


We use the term \emph{ideal classical channel}\index{ideal classical channel} for a type of information 
$\JC = \{J_j\}$ to refer to a situation where \eqref{locinf_eqn3} is satisfied and,
in addition, 
\begin{equation}
\label{locinf_eqn6} 
\EC(J_j A J_k) = 0 \text{ for } j\neq k,
\end{equation}
where $A$ is any operator on the input Hilbert space $\HC$. 
That is, not only is type $\JC$ perfectly transmitted, but all other types
are ``truncated'' relative to this type, in the notation of
\cite{PhysRevA.54.2759}.  

\section{Preliminary remarks and definitions}
\label{locinf_sct3}

\subsection{Generalized Pauli operators on $n$ qudits}
\label{locinf_sct4}


We generalize Pauli operators to higher dimensional systems of arbitrary
dimension $D$ in the following way. The $X$ and $Z$ operators acting on a
single qudit are defined as
\begin{equation}
\label{locinf_eqn7}
Z=\sum_{j=0}^{D-1}\omega^j\dyad{j}{j},\quad X=\sum_{j=0}^{D-1}\dyad{j}{j+1},
\end{equation}
and satisfy
\begin{equation}
\label{locinf_eqn8}
X^D=Z^D=I,\quad XZ=\omega ZX,\quad \omega = \mathrm{e}^{2 \pi \ii /D},
\end{equation}
where \emph{the addition of integers is modulo $D$}, as will be 
assumed from now on. For a collection of $n$ qudits we use subscripts to
identify the corresponding Pauli operators: thus $Z_i$ and $X_i$ operate on
the space of qudit $i$. The Hilbert space of a single qudit is denoted by
$\HC$, and the Hilbert space of $n$ qudits by $\HC_n$,
respectively. Operators of the form
\begin{equation}
\label{locinf_eqn9}
\omega^{\lambda}X^{\vect{x}}Z^{\vect{z}} :=
\omega^{\lambda}X_1^{x_1}Z_1^{z_1}\otimes X_2^{x_2}Z_2^{z_2}\otimes\cdots
\otimes X_n^{x_n}Z_n^{z_n}
\end{equation} 
will be referred to as \emph{Pauli products}\index{Pauli product}, where $\lambda$ is an integer
in $\ZZ_D$ and $\vect{x}$ and $\vect{z}$ are $n$-tuples in $\ZZ_D^n$, the
additive group of $n$-tuple integers mod $D$.  For a fixed $n$ the collection
of all possible Pauli products \eqref{locinf_eqn9} form a group under operator
multiplication, the \emph{Pauli group}\index{Pauli group} $\PC_n$. If $p$ is a Pauli product,
then $p^D=I$ is the identity operator on $\HC_n$, and hence the order of any
element of $\PC_n$ is either $D$ or else an integer that divides $D$. While
$\PC_n$ is not abelian, it has the property that two elements \emph{commute up
  to a phase}: $p_1p_2 = \omega^{\lambda_{12}} p_2p_1$, with $\lambda_{12}$ an
integer in $\ZZ_D$ that depends on $p_1$ and $p_2$.


The collection of Pauli products with $\lambda=0$, i.e. a pre-factor of $1$, is
denoted by $\QC_n$ and forms an orthonormal basis of
$\LC(\HC_n)$, the Hilbert space of linear operators on
$\HC_n$, with respect to the inner product
\begin{equation}
\label{locinf_eqn10}
\frac{1}{D^n}\Tr[q_1^\dagger q_2]=\delta_{q_1,q_2}, \quad \forall q_1,q_2\in \QC_n.
\end{equation}
Note that $\QC_n$ is a \emph{projective group}\index{projective group} or group up
to phases. There is a bijective map between $\QC_n$ and the quotient group
$\PC_n /\{\omega^{\lambda}{I}\}$ for $\lambda\in\ZZ_D$ where
$\{\omega^{\lambda}{I}\}$, the center of $\PC_n$, consists of phases
multiplying the identity operator on $n$ qudits.

\subsection{Generalization of qubit quantum gates to higher dimensions}
\label{locinf_sct5}


In this subsection we define some one and two qudit gates generalizing
various qubit gates. The qudit generalization of the Hadamard gate is the
\emph{Fourier gate}\index{Fourier gate}
\begin{equation}\label{locinf_eqn11}
 \mathrm{F}:=\frac{1}{\sqrt{D}}\sum_{j=0}^{D-1}\omega^{jk}\dyad{j}{k}.
\end{equation}
For an invertible integer $q\in\ZZ_D$ (i.e. integer for which there exists $\bar q\in\ZZ_D$ such that $q \bar q \equiv 1 \bmod D$), we define a
\emph{multiplicative gate}\index{multiplicative gate}
\begin{equation}\label{locinf_eqn12}
 \mathrm{S}_q:=\sum_{j=0}^{D-1}\dyad{j}{jq},
\end{equation}
where $qj$ means multiplication mod $D$. The requirement that $q$ be
invertible ensures that $\mathrm{S}_q$ is unitary; for a qubit
$\mathrm{S}_q$ is just the identity.


For two distinct qudits $a$ and $b$ we define the CNOT\index{Controlled-NOT gate} gate as
\begin{equation}
\label{locinf_eqn13}
 \mathrm{CNOT}_{ab}:=\sum_{j=0}^{D-1}\dyad{j}{j}_a\otimes X_b^j=\sum_{j,k=0}^{D-1}\dyad{j}{j}_a\otimes \dyad{k}{k+j}_b,
\end{equation}
the obvious generalization of the qubit Controlled-NOT, where $a$ labels the control qudit and $b$ labels the target qudit. Next the SWAP\index{SWAP gate} gate is defined as
\begin{equation}
\label{locinf_eqn14}
 \mathrm{SWAP}_{ab}:=\sum_{j,k=0}^{D-1}\dyad{k}{j}_a\otimes \dyad{j}{k}_b.
\end{equation}
It is easy to check that SWAP gate is hermitian and does indeed swap
qudits $a$ and $b$. Unlike the qubit case, the qudit SWAP gate is not a
product of three CNOT gates, but can be expressed in terms of CNOT gates and
Fourier gates as
\begin{equation}\label{locinf_eqn15}
 \mathrm{SWAP}_{ab}=\mathrm{CNOT}_{ab}(\mathrm{CNOT}_{ba})^{\dagger}\mathrm{CNOT}_{ab}(\mathrm{F}_a^2\otimes I_b),
\end{equation}
with 
\begin{equation}\label{locinf_eqn16}
(\mathrm{CNOT}_{ba})^{\dagger}=(\mathrm{CNOT}_{ba})^{D-1}=(I_a\otimes \mathrm{F}_b^2)\mathrm{CNOT}_{ba} (I_a\otimes \mathrm{F}_b^2). 
\end{equation}
Finally we define the generalized Controlled-phase\index{Controlled-phase gate} or CP gate as
\begin{equation}
\label{locinf_eqn17}
\mathrm{CP}_{ab}=\sum_{j=0}^{D-1}\dyad{j}{j}_a\otimes Z^j_b=
\sum_{j,k=0}^{D-1}\omega^{jk}\dyad{j}{j}_a\otimes\dyad{k}{k}_b.
\end{equation}
The CP and CNOT gates are related by a local Fourier gate, similar to the qubit case
\begin{equation}\label{locinf_eqn18}
\mathrm{CNOT}_{ab}=(I_a\otimes \mathrm{F}_b) \mathrm{CP}_{ab} (I_a\otimes \mathrm{F}_b)^\dagger,
\end{equation}
since $\mathrm{F}$ maps $Z$ into $X$ under conjugation (see Table \ref{locinf_tbl1}).


The gates $\mathrm{F}$, $\mathrm{S}_q$, SWAP, CNOT and CP are unitary
operators that map Pauli operators to Pauli operators under conjugation, as
can be seen from Tables ~\ref{locinf_tbl1} and~\ref{locinf_tbl2}. They are elements of the
so called \emph{Clifford group}\index{Clifford group} on $n$ qudits \cite{quantph.9802007,PhysRevA.71.042315}, the group of $n$-qudit unitary
operators that leaves $\PC_n$ invariant under conjugation, i.e. if $O$ is a
Clifford operator, then $\forall p\in\PC_n$, $OpO^\dagger\in\PC_n$. From
Tables~\ref{locinf_tbl1} and~\ref{locinf_tbl2} one can easily deduce the result of
conjugation by $\mathrm{F}$, $\mathrm{S}_q$, SWAP, CNOT and CP on \emph{any}
Pauli product. 

\begin{table}
\begin{center}
\begin{tabular}{|l|l|l|}
\hline
Pauli operator    & $\mathrm{S}_q$ & $\mathrm{F}$  \\
\hline
\hline
$Z$ & $Z^{q}$ & $X$\\
\hline
$X$ & $X^{\bar q}$ & $Z^{D-1}$\\
\hline
\end{tabular}
\caption{The conjugation of Pauli operators by one-qudit gates $\mathrm{F}$ and $\mathrm{S}_q$ ($\bar q$ is the multiplicative inverse of $q$ mod $D$).}
\label{locinf_tbl1}
\end{center}
\end{table}

\begin{table}
\begin{center}
\begin{tabular}{|l|l|l|l|}
\hline
Pauli product & $\mathrm{CNOT}_{ab}$ & $\mathrm{SWAP}_{ab}$ & $\mathrm{CP}_{ab}$\\
\hline
\hline
$I_a\otimes Z_b$ & $Z_a\otimes Z_b$ & $Z_a\otimes I_b$ & $I_a\otimes Z_b$\\ 
\hline
$Z_a\otimes I_b$ & $Z_a\otimes I_b$ & $I_a\otimes Z_b$ & $Z_a\otimes I_b$\\
\hline
$I_a\otimes X_b$ & $I_a\otimes X_b$ & $X_a\otimes I_b$ & $Z_a^{D-1}\otimes X_b$\\
\hline
$X_a\otimes I_b$ & $X_a\otimes X_b^{D-1}$ & $I_a\otimes{X}_b$ & $X_a\otimes Z_b^{D-1}$\\
\hline
\end{tabular}
\caption{
  The conjugation of Pauli products on qudits $a$ and $b$ by two-qudit 
  gates CNOT, SWAP and CP. For the CNOT gate, the first qudit $a$ is the
  control and the second qudit $b$ the target.
}
\label{locinf_tbl2}
\end{center}
\end{table}

\section{Graph states, graph codes and related operator groups \label{locinf_sct6}}

\subsection{Graph states and graph codes\label{locinf_sct7}} 


Let $G=(V,E)$ be a graph with $n$ vertices $V$, each corresponding to a qudit, and a collection $E$ of undirected edges connecting pairs of distinct vertices (no self loops). Two qudits can be joined by multiple edges, as long as the multiplicity does not exceed $D-1$. The graph $G$ is completely specified by the \emph{adjacency matrix}\index{adjacency matrix} $\Gamma$, where the matrix element $\Gamma_{ab}$ represents the number of edges that connect vertex $a$ with vertex $b$. The \emph{graph state}\index{graph state}
\begin{equation}\label{locinf_eqn19}
\ket{G}=U\ket{G_0}=U\left(\ket{+}^{\otimes n}\right)
\end{equation}
is obtained by applying the unitary (Clifford) operator 
\begin{equation}
\label{locinf_eqn20}
U=\prod_{(a,b) \in E}\left(\mathrm{CP}_{ab}\right)^{\Gamma_{ab}},
\end{equation}
where each pair $(a,b)$ of vertices occurs only once in the product,
to the \emph{trivial graph state}\index{trivial graph state}
\begin{equation}\label{locinf_eqn21}
\ket{G_0}:=\ket{+}^{\otimes n},
\end{equation}
with
\begin{equation}\label{locinf_eqn22}
\ket{+}:=\frac{1}{\sqrt{D}}\sum_{j=0}^{D-1}\ket{j}.
\end{equation}


Define $\SC^G$ to be the stabilizer of $\ket{G}$, the subgroup of operators from $\PC_n$ that leave $\ket{G}$ unchanged. The stabilizer $\SC^G_0$ of the trivial graph state $\ket{G_0}$ is simply the set of all $X$-type Pauli products with no additional phases,
\begin{equation}
\label{locinf_eqn23}
\SC^G_0=\left\{ X^{\vect{x}}: \vect{x}=(x_1,x_2,\ldots,x_n)\right\},
\end{equation}
where $x_j$ are arbitrary integers between $0$ and $D-1$.
Since $\ket{G}$ is related to $\ket{G_0}$ through a Clifford operator (see \eqref{locinf_eqn19} and \eqref{locinf_eqn20}), it follows at once that the stabilizer $\SC^G$ of $\ket{G}$ is related to the stabilizer $\SC_0^G$ of the trivial graph through the Clifford conjugation
\begin{equation}\label{locinf_eqn24}
\SC^G=U\SC^G_0 U^\dagger,
\end{equation}
with $U$ defined in \eqref{locinf_eqn20}.


A \emph{graph code}\index{graph code} $C$ can be defined as the $K$-dimensional subspace $\HC_C$
of $\HC_n$ spanned by a collection of $K$ mutually orthogonal
codewords
\begin{equation}
\label{locinf_eqn25}
 \ket{\vect{c}_j} = 
Z^{\vect{c}_j}\ket{G},\quad j = 1,2,\ldots, K
\end{equation}
where 
\begin{equation}
\label{locinf_eqn26}
\vect{c}_j = (c_{j1}, c_{j2},\ldots, c_{jn})
\end{equation}
is for each $j$ an $n$-tuple in $\ZZ_D^n$.  The $c_{jk}$ notation suggests a
matrix $\vect{c}$ with $K$ rows and $n$ columns, of integers between $0$ and
$D-1$, and this is a very helpful perspective.  In this chapter we are concerned
with \emph{additive}\index{additive graph code} graph codes, meaning that the rows of this matrix form a
group under component-wise addition mod $D$, isomorphic to the abelian
\emph{coding group}\index{coding group} $\CC$, of order $|\CC|=K$, of the operators
$Z^{\vect{c}_j}$ under multiplication.  We use $(\CC,\ket{G})$ to denote the
corresponding graph code.  For more details about graph states and graph codes
for arbitrary $D$, see \cite{PhysRevA.78.042303}.


Note that the codeword $(0,0,\ldots,0)$ is just the graph state $\ket{G}$, and
in the case of the trivial graph $\ket{G_0}$ this is the tensor product of
$\ket{+}$ states, \eqref{locinf_eqn21}, not the tensor product of $\ket{0}$ states
which the $n$-tuple notation $(0,0,\ldots,0)$ might suggest. Overlooking this
difference can lead to confusion through interchanging the role of $X$ and $Z$
operators, which is the reason for pointing it out here.

\subsection{The encoding problem\label{locinf_sct8}}

A coding group $\CC$ can be used to create an additive code starting with any $n$ qudit
graph state, including the trivial graph $\ket{G_0}$, because the entangling
unitary $U$ commutes with $Z^{\vect{z}}$ for any $\vect{z}$; thus
\begin{equation}
\label{locinf_eqn27}
 \ket{\vect{c}_j} = Z^{\vect{c}_j}U\ket{G_0} = 
UZ^{\vect{c}_j}\ket{G_0} =U\ket{\vect{c}_j^0}
\end{equation}
where the $\ket{\vect{c}_j^0}$ span the code $(\CC,\ket{G_0})$. But in
addition the coding group $\CC$ is isomorphic, as explained below to a
\emph{trivial}\index{trivial code} code $\CC_0$,
\begin{equation}
\label{locinf_eqn28}
 \CC_0=\langle Z_1^{m_1}, Z_2^{m_2},\ldots, Z_k^{m_k}\rangle
\end{equation}
which is \emph{generated by}\index{generated by}, i.e., includes all products of, the operators
inside the angular brackets $\langle$ $\rangle$.  Here $k$ is an integer less
than or equal to $n$, and each $m_j$ is $1$ or a larger integer that divides
$D$.
The simplest situation is the one in which each of the $m_j$ is equal to 1, in
which case $\CC_0$ is nothing but the group, of order $D^k$, of products of
$Z$ operators to any power less than $D$ on the first $k$ qudits.  One can
think of these qudits as comprising the input system through which information
enters the code, while the remaining $n-k$ qudits, each initially in a
$\ket{+}$ state, form the ancillary system for the encoding operation.  


If, however, one of the $m_j$ is greater than 1, the corresponding generator
$Z_j^{m_j}$ is of order
\begin{equation}
\label{locinf_eqn29}  
d_j = D/m_j,
\end{equation}
and represents a qudit of dimensionality $d_j$ rather than $D$.  Thus for
example, if $D=6$ and $m_1=2$, applying $Z_1^2$ and its powers to $\ket{+}$
will produce three orthogonal states corresponding to a qutrit, $d_1=3$.
(Identifying operators $Z$ and $X$ on these three states which satisfy
\eqref{locinf_eqn8} with $D=3$ is not altogether trivial, and is worked out in
Sec.~\ref{locinf_sct9} below.)
In general one can think of the group $\CC_0$ in \eqref{locinf_eqn28} as
associated with a collection of $k$ qudits, the $j$'th qudit having dimension
$d_j$, and therefore the collection as a whole a dimension of $K=d_1d_2\cdots
d_k$, equal to that of the graph code.  If one thinks of the information to be
encoded as initially present in these $k$ qudits, the encoding problem is how
to map them in an appropriate way into the coding subspace $\HC$ of the $n$
($D$-dimensional) carriers.  


We address this by first considering
the connection between $\CC$ and $\CC_0$ in a simple 
example with $n=3$, $D=6$, and 
\begin{equation}
\label{locinf_eqn30}
\CC=\left\langle  Z_1^4Z_2^3Z_3^3, Z_2^3Z_3^3 \right\rangle,
\end{equation}
a coding group of order 6. The two generators in \eqref{locinf_eqn30} correspond, in
the notation introduced in \eqref{locinf_eqn26}, to the rows of the $2\times 3$
matrix
\begin{equation}
\label{locinf_eqn31}
\vect{f} =\mat{ 
 4 & 3 & 3\\
 0 & 3 & 3}.
\end{equation}
By adding rows or multiplying them by constants mod $D$ one can create 4
additional rows which together with those in \eqref{locinf_eqn31} constitute the
$6\times 3$ $\vect{c}$  matrix.


Through a sequence of elementary operations mod $D$\index{elementary operations mod $D$}---a) interchanging of
rows/columns, b) multiplication of a row/column by an \emph{invertible} integer, c)
addition of any multiple of a row/column to a \emph{different} row/column---a
matrix such as $\vect{f}$ can be converted to the Smith normal form 
\cite{Newman:IntegralMatrices,Storjohann96nearoptimal}
\begin{equation}
  \vect{s} = \vect{v}\cdot\vect{f}\cdot\vect{w},
\label{locinf_eqn32}
\end{equation}
where $\vect{v}$ and $\vect{w}$ are invertible (in the mod $D$ sense) square
matrices, and $\vect{s}$ is a diagonal rectangular matrix, as in
\eqref{locinf_eqn33}. It is proved in \cite{Storjohann96nearoptimal} that a $K\times n$ matrix can be reduced to the Smith form in only $\OC(K^{\theta-1}n)$ operations from $\ZZ_D$, where $\theta$ is the exponent for matrix multiplication over the ring $\ZZ_D$, i.e. two $m\times m$ matrices over $\ZZ_D$ can be multiplied in $\OC(m^{\theta})$ operations from $\ZZ_D$. Using standard matrix multiplication $\theta=3$, but better algorithms \cite{CoppersmithWinograd} allow for $\theta=2.38$. 


For the example above, the sequence
\begin{equation}
\label{locinf_eqn33}
  \mat{ 4 & 3 & 3\\ 0 & 3 & 3} \rightarrow 
  \mat{ 4 & 0 & 0\\ 0 & 3 & 3} \rightarrow 
  \mat{ 4 & 0 & 0\\ 0 & 3 & 0} \rightarrow 
  \mat{ 2 & 0 & 0\\ 0 & 3 & 0} =\vect{s}
\end{equation} 
proceeds by adding the second row of $\vect{f}$ to the first (mod 6), then the
second column to the third column, and finally multiplying the first row by 5
(which is invertible mod 6).  The final step is needed so that the diagonal
elements divide $D$: $m_1=2$, $m_2=3$, so that $d_1=3$ and $d_2=2$. Thus we
arrive at the trivial coding group
\begin{equation}
\label{locinf_eqn34}
 \CC_0=\left\langle Z_1^2, Z_2^3 \right\rangle,
\end{equation}
isomorphic to $\CC$ in \eqref{locinf_eqn30}.


Since the procedure for reducing a matrix to Smith normal form is quite
general, the procedure illustrated in this example can be applied to
any coding group $\CC$, as defined following \eqref{locinf_eqn26}, to find a
corresponding trivial coding group $\CC_0$. 
The row operations change the collection of generators but not the coding
group that they generate; i.e., the final collection of $K$ rows is the same.
The column operations, on the other hand, produce a different, but isomorphic,
coding group, and one can think of these as realized by a unitary operator $W$
which is a product of various SWAP, CNOT and $\mathrm{S}_q$ gates, so that 
\begin{equation}
\label{locinf_eqn35}
\CC=W\CC_0 W^\dagger,
\end{equation}
that is, conjugation by $W$ maps each operator in $\CC_0$ to its counterpart
in $\CC$.  In our example, $W=\mathrm{CNOT}_{32}$ is the only column
operation, the second arrow in \eqref{locinf_eqn33}, and represents the first step
in the encoding circuit for this example, Fig.~\ref{encoding}(b). It is left
as an exercise to check that this relates the generators in \eqref{locinf_eqn30} and
\eqref{locinf_eqn34} through \eqref{locinf_eqn35}.  Table~\ref{locinf_tbl3} indicates how
different matrix column operations are related to the corresponding gates in
the encoding circuit. 

\begin{figure}
\begin{center}
\includegraphics{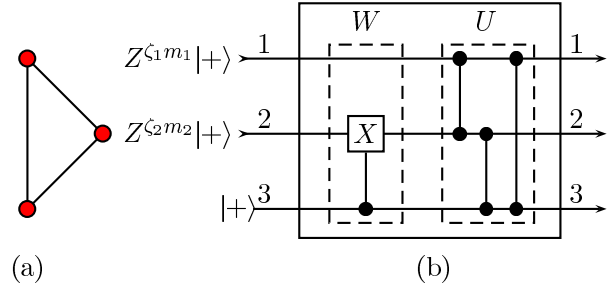}
\caption{(a) The graph state used in the example; (b) The encoding circuit: the input states ${Z_1^{\zeta_1 m_1}Z_2^{\zeta_2 m_2}\ket{++}}$ that correspond to the trivial code $\CC_0$ are mapped by
$W$ to $\CC$, then $U$ entangles the qudits. Here $m_1=2$, $m_2=3$ and $\zeta_j$ are integers such that $0\leqslant \zeta_j\leqslant d_j-1$, with $d_1=3$, $d_2=2$.}
\label{encoding}
\end{center}
\end{figure}

\begin{table}
\begin{center}
\begin{tabular}{|l|c|}
\hline
Matrix operation in $\ZZ_D$ & Clifford conjugation\\
\hline
\hline
Interchange of columns $a$ and $b$    & $\mathrm{SWAP}_{ab}$\\
\hline
Multiplication of column $a$          & $\mathrm{S}_q$ on qudit $a$\\
by invertible integer $q$             & \\
\hline
Addition of m times column $b$ to             & $(\mathrm{CNOT}_{ab})^m$ \\
column $a$             &  \\
\hline
\end{tabular}
\caption{The correspondence between matrix column operations in $\ZZ_D$ and 
conjugation by Clifford gates. For the CNOT gate, the first qudit $a$ is the control
and the second qudit $b$ the target.}
\label{locinf_tbl3}
\end{center}
\end{table}


The overall encoding operation 
\begin{equation}
\label{locinf_eqn36}
\ket{\vect{c}_j}=UW\ket{\vect{c}_j^{0}}
\end{equation}
starting with the trivial code on the trivial graph $(\CC_0,\ket{G_0})$  and
ending with the desired code $(\CC,\ket{G})$ is shown for our example in
Fig.~\ref{encoding}(b) for the case of a graph indicated in (a) in this
figure.  It is important to notice that both $W$ and $U$, and therefore their
product, are Clifford operators, unitaries that under conjugacy map Pauli products to Pauli products.  This follows from the fact
that the gates in Table~\ref{locinf_tbl3} are Clifford gates, and will allow us in
what follows to extend arguments that are relatively straightforward for
trivial codes on trivial graphs to more general additive graph codes.  

\subsection{The information group}
\label{locinf_sct9}


In this section we define the \emph{information group}\index{information group} that plays a central
role in the isomorphism theorem in Sec.~\ref{locinf_sct10} below.  The basic strategy
is most easily understood in terms of $C_0 = (\CC_0,\ket{G_0})$, the trivial
code on the trivial graph.  However, because the overall encoding map $UW$ in
\eqref{locinf_eqn36} is a Clifford operation mapping Pauli products to Pauli
products, various results that apply to $C_0$ can be immediately translated
to the general graph code $C = (\CC,\ket{G})$ we are interested in, and for
this reason most of the formulas valid for both are written in the form
valid for $C$ even if the derivations are based on $C_0$.

The pointwise stabilizer%
\footnote{Also called the ``fixer'' or ``fixator''.  It is important to
  distinguish this subgroup from the group theoretical notion of the
  stabilizer of the coding space in the sense of the
  subgroup of $\PC_n$ that maps the coding space onto itself without necessarily
  leaving the individual vectors fixed.  As we shall not employ the latter, it
  should cause no confusion if we hereafter follow the usual convention in
  quantum codes and omit ``pointwise,'' even though retaining it would add
  some precision.}%
\ of $C_0$, the subgroup of operators from $\PC_n$ that leave every codeword
$\ket{\vect{c}^0_j}$ unchanged, is given by
\begin{equation}
\label{locinf_eqn37}
\SC_0=\left\{X^{\vect{x}}:\vect{x}
 =(\eta_1 d_1,\eta_2 d_2,\ldots,\eta_k d_k,x_{k+1},\ldots,x_n)
\right\},
\end{equation}
where the $d_j$ are defined in \eqref{locinf_eqn29}, $\eta_j$ is any integer between
0 and $m_j-1$, and the $x_j$ for $j>k$ are arbitrary integers between 0 and
$D-1$.  That this is correct can be seen as follows.  First, Pauli products
belonging to $\SC_0$ cannot contain $Z_j$ operators, for such operators map
each codeword onto an orthogonal state. On the other
hand, every $X_j^{x_j}$ leaves $\ket{G_0}$, \eqref{locinf_eqn21}, unchanged, so it
belongs to $\SC_0$ if and only if it commutes with $Z_j^{m_j}$, which means
$x_jm_j$ must be a multiple of $D$, or $x_j$ a multiple of $d_j$, see
\eqref{locinf_eqn29}.  Thus elements of $\SC_0$ commute with elements of $\CC_0$,
\eqref{locinf_eqn28}.  Since its operators cannot alter the phases of the codewords,
no additional factors of $\om^\lambda$ are allowed, and thus $\SC_0$ is given
by \eqref{locinf_eqn37}.  The stabilizer of the (nontrivial) code $C$ is then the
isomorphic group $\SC$ obtained using the unitary $UW$ of \eqref{locinf_eqn36}:
\begin{equation}
\label{locinf_eqn38}
\SC = (UW)\SC_0(UW)^\dagger \equiv \{ (UW) s (UW)^\dag \ : s\in \SC_0 \}, 
\end{equation}
a collection of Pauli products because the unitary $UW$, as remarked earlier,
is a Clifford unitary.  The order of $\SC_0$, and thus of $\SC$, is given by
\begin{equation}
\label{locinf_eqn39}
|\SC|=D^{n-k}\prod_{j=1}^k{m_j}=
\frac{D^n}{\prod_{j=1}^kd_j}=\frac{D^n}{|\CC|}=\frac{D^n}{K}.
\end{equation}

Next define the subgroup $\WC$ of $\PC_n$
\begin{equation}
\label{locinf_eqn40}
\WC=\langle \SC^G, \CC\rangle
\end{equation}
generated by operators belonging to the stabilizer $\SC^G$ of the graph state
or to the coding group $\CC$, and denote it by $\WC_0=\langle \SC^G_0,
\CC_0^{}\rangle$ in the case of the trivial code. The elements of $\SC_0$
commute with those of $\SC_0^G$ (both are abelian and the former is a subgroup
of the latter), and also with those of $\CC_0$, as noted above.  As group
properties are preserved under the $UW$ map, as in \eqref{locinf_eqn38}, we conclude
that all elements in $\SC$ commute with those in $\WC$, even though $\WC$ is
not (in general) abelian, and hence $\SC$ is a normal subgroup of $\WC$. Now
define the \emph{abstract information group}\index{information group, abstract} as the quotient group
\begin{equation}
\label{locinf_eqn41}
\GCbar=\WC/\SC=\langle \SC^G, \CC\rangle/\SC
\end{equation}
consisting of cosets of $\SC$, written as $g\SC$ or $\SC g$ for $g$ in $\WC$.
Note that because any element $g$ of $\WC$ is a Pauli product, $g^D=I$ is the
identity, and the order of $g$ is either $D$ or an integer that divides $D$. 
Consequently the order of any element of $\GCbar$ is also $D$ or an integer
that divides $D$. 

To understand the significance of $\GCbar$ consider a
trivial code on a single qudit, with 
\begin{equation}
\label{locinf_eqn42} 
\CC_0 = \avg{Z_1^{m_1}},\quad \SC_0^G = \avg{X_1},\quad 
\SC_0 = \avg{X_1^{d_1}}
\end{equation}
The elements of $\GCbar_0$ can be worked out using its identity $\bar I$ and
the generators $\bar X$ and $\bar Z$:
\begin{align}
\label{locinf_eqn43} 
 \bar I &= \SC_0 = \{I_1, X_1^{d_1}, X_1^{2d_1},\ldots \}
\notag\\
 \bar X &= X_1\SC_0 = \{X_1, X_1^{d_1+1}, X_1^{2d_1+1},\ldots \}
\notag\\
 \bar Z &= Z_1^{m_1}\SC_0 = \{Z_1^{m_1}, Z_1^{m_1} X_1^{d_1},\ldots \}.
\end{align}
It is evident that the cosets $\bar X$, $\bar X^2 =
X_1^2\SC_0$ and so forth up to $\bar X^{d_1-1}$ are distinct, whereas
$\bar X^{d_1}=\bar I = \SC_0$.  The same is true for powers of $\bar Z$.
Furthermore,
\begin{equation}
\label{locinf_eqn44}
\bar X\bar Z = X_1 Z_1^{m_1}\SC_0 = \omega^{m_1} Z_1^{m_1}  X_1\SC_0 = 
 \bar\omega \bar Z\bar X,
\end{equation}
with $\bar\omega = \omega^{m_1} = \mathrm{e}^{2\pi \mathrm{i}/d_1}$.  Thus $\GCbar_0$ is 
generated by operators $\bar X$ and $\bar Z$ that satisfy \eqref{locinf_eqn8}
with $D$ replaced by $d_1$, which is to say the corresponding group is what
one would expect for a qudit of dimension $d_1$.  The same  argument
extends easily to the trivial code on $k$ carriers produced by $\CC_0$, see
\eqref{locinf_eqn28}: $\GCbar_0$ is isomorphic to the group of Pauli products on a
set of qudits of dimension $d_1, d_2,\ldots, d_k$.  The same
structure is inherited by the abstract information group $\GCbar$ for 
the code $C=(\CC,\ket{G})$ obtained by applying the $UW$ map as
in \eqref{locinf_eqn38}.

The abstract information group $\GCbar$ is isomorphic to the \emph{information
  group}\index{information group} $\GC$ of \emph{information operators}\index{information operator} acting on the coding space $\HC_C$ and defined in
the following way.  Its identity is the operator
\begin{equation}
\label{locinf_eqn45}
 P = |\SC|^{-1}\Sigma(\SC) =  |\SC|^{-1}\sum_{s\in\SC} s,
\end{equation}
where $\Sigma(\AC)$ denotes the sum of the operators that make up a collection
$\AC$.  In fact, $P$ is just the projector onto $\HC_C$, as
can be seen as follows. Since $\SC$ is a group, $P^2=P$; and since a group
contains the inverse of every element, and $s\in \SC$ is unitary (a Pauli
product),  $P^\dagger = P$.  These two conditions mean that $P$ is a projector
onto some subspace of $\HC_n$.  Since $\SC$ is the (pointwise) stabilizer of
the coding space each $s$ in $\SC$ maps a codeword onto itself, and
thus $P$ maps each codeword to itself. Consequently, all the codewords lie in
the space onto which $P$ projects.  Finally, the rank of $P$ is
\begin{equation}
\label{locinf_eqn46}
\Tr[P] = D^n/|\SC| = |\CC| = K
\end{equation}
(see \eqref{locinf_eqn39}), since the trace of every $s$ in $\SC$ is zero except for
the identity with trace $D^n$. (Note that while $\PC_n$ contains the identity
multiplied by various phases, only the identity operator occurs in $\SC$.)
Therefore $P$ projects onto $\HC_C$, and is given by the formula
\begin{equation}
\label{locinf_eqn47}
P= \sum_{j=1}^{K} \dya{\vect{c}_j}. 
\end{equation}

The other information operators making up the information group
$\GC=\{\hat g\}$ are formed in a similar way from the different cosets
making up $\WC/\SC$:
\begin{equation}
\label{locinf_eqn48} 
 \hat g = |\SC|^{-1}\Sigma(g\SC) =  g P = P g P = P\hat g P.
\end{equation}
That is, for each coset form the corresponding sum of operators and divide by
the order of the stabilizer $\SC$.  The second and third equalities in
\eqref{locinf_eqn48} reflect the fact that the product of the cosets $\SC$ and $g\SC$
in either order is $g\SC$, which is to say $P$ forms the group identity of
$\GC$. They also tell us that the operators that make up $\GC$ act only on the
coding space, mapping $\HC_C$ onto itself, and give zero when applied to any
element of $\HC_n$ in the orthogonal complement of $\HC_C$.  Because $\SC$ is
a normal subgroup of $\WC$, products of operators of the form \eqref{locinf_eqn48}
mirror the products of the corresponding cosets, so the map from the abstract
$\GCbar$ to the group $\GC$ is a homomorphism.  That it is actually an
isomorphism is a consequence of the following, proved in App.~\ref{locinf_apdx1}:

\begin{lemma}\label{locinf_thm1}

Let $\RC$ be a linearly independent collection of Pauli product operators that
form a subgroup of $\PC_n$, and for a Pauli product $p$ let $p\RC=\{pr: r\in
\RC\}$. Then

i) The operators in $p\RC$ are linearly independent.

ii) If $p$ and $q$ are two Pauli products, one or the other of the
following two mutually exclusive possibilities obtains:

$\alpha$) 
\begin{equation}
\label{locinf_eqn49} 
p\RC = \mathrm{e}^{\ii\phi} q\RC
\end{equation}
in the sense that each operator in $p\RC$ is equal to $\mathrm{e}^{\ii\phi}$ times an
operator in $q\RC$

$\beta$)
The union $p\RC\cup q\RC$ is a collection of $2|\RC|$ linearly independent
operators. 

\end{lemma}

Since the collection of Pauli products $\QC_n$ with fixed phase forms a basis
of $\LC(\HC_n)$, a collection of Pauli products can be linearly
\emph{dependent} if and only if it contains both an operator and that operator
multiplied by some phase.  As the (pointwise) stabilizer $\SC$ leaves each
codeword unchanged, the corresponding operators are linearly independent, and
the lemma tells us that distinct cosets $g\SC\neq h\SC$ give rise to distinct
operators $\hat g\neq\hat h$.  Either $g\SC = \mathrm{e}^{\ii\phi} h\SC$, in which
case $\hat g = \mathrm{e}^{\ii\phi} \hat h \neq \hat h$ (since if $\mathrm{e}^{\ii\phi}=1$ the cosets are identical). Or else the $g\SC$ operators are linearly independent
of the $h\SC$ operators, and therefore $\hat g$ and $\hat h$ are linearly
independent. Consequently the homomorphic map from $\GCbar$ to $\GC$ is a
bijection, and the two groups are isomorphic.

The single qudit example considered in \eqref{locinf_eqn42} provides an example of how
$\GCbar$ and $\GC$ are related.  In this case the projector
\begin{equation}
\label{locinf_eqn50} 
P_0 = (1/m_1) (I_1 + X_1^{d_1} + \cdots)
\end{equation}
projects onto the subspace spanned by $\ket{+},Z_1^{m_1}\ket{+}, Z_1^{2
  m_1}\ket{+},\ldots$. While each of the operators that make up a coset such as
$\bar X$ in \eqref{locinf_eqn43} is unitary, their sum, an operator times $P_0$, is
no longer unitary, though when properly normalized acts as a unitary on the
subspace onto which $P_0$ projects. That the different sums of operators
making up the different cosets are distinct is in this case evident from
inspection without the need to invoke Lemma~\ref{locinf_thm1}.

Let us summarize the main results of this subsection.  For an additive graph
code $C$ we have defined the information group $\GC$ of operators acting on
the coding subspace $\HC_C$, whose group identity is the projector $P$ onto $\HC_C$.  
It is isomorphic to the group of Pauli products acting on a tensor product of
qudits of dimensions $d_1$, $d_2$, \dots, $d_k$, which can be thought of as the
input to the code, see Sec.~\ref{locinf_sct8}.  Each element $\hat g$ of $\GC$ is
of the form $P\hat g P$, so as an operator on $\HC_n$ it commutes with $P$ and
yields zero when applied to any vector in the orthogonal complement of
$\HC_C$. The dimension of $\HC_C$ is $K=d_1 d_2\cdots d_k$, the size of the
code, and hence the elements of $\GC$ span the space of linear operators
$\LC(\HC_C)$ on $\HC_C$.

\section{Subsets of carriers and the isomorphism theorem}
\label{locinf_sct10}

\subsection{Subsets of carriers}
\label{locinf_sct11}

Before stating the isomorphism theorem, which is the principal technical
result of this chapter, let us review some facts established in Sec.~\ref{locinf_sct6}.
The additive graph code $(\CC,\ket{G})$ we are interested in can be thought of
as arising from an encoding isometry that carries the channel input onto a
subspace $\HC_C$ of the $n$-qudit carrier space $\HC_n$, as in Fig~\ref{encoding}.  
This isometry, as explained in Sec.~\ref{locinf_sct2} in connection with \eqref{locinf_eqn4}, constitutes a perfect quantum channel, and thus all the information of interest can be said
to be located in the $\HC_C$ subspace, where it is represented by the
information group $\GC$, a multiplicative group of operators for which the
projector $P$ on $\HC_C$ is the group identity, and which as a group is
isomorphic to the abstract information group $\GCbar$ defined in
\eqref{locinf_eqn41}.

We are interested in what kinds of information are available in some subset
$\pt$ of the carriers, where $\ptc$ denotes the complementary set.  For this
purpose it is natural to consider the partial traces over $\ptc$, i.e., the
traces down to the Hilbert space $\HC_\pt$, of the form
\begin{equation}
\label{locinf_eqn51}
 g_\pt =  N^{-1} \Tr_{\ptc}[\hat g],
\end{equation}
where $\hat g$ is an element of the information group $\GC$, and the positive
constant $N$ is defined in \eqref{locinf_eqn58} below.  In those cases in which
$g_{\pt}=0$ the $\JC(\hat g)$ information has disappeared and is not available in the subset $B$, so we shall be interested in those $\hat g$ for which the partial trace does not vanish, that is to say in the elements of the \emph{subset information group}\index{information group, subset}
\begin{equation}
\label{locinf_eqn52}
\GC^{\pt}=\left\{ \hat g \in \GC: \Tr_{\ptc}[\hat g]\neq 0\right\}.
\end{equation}
We show below that $\GC^{\pt}$ is a subgroup of $\GC$, thus justifying its name,
and that it is isomorphic to the group $\GC_{\pt}$ of nonzero operators of the form
$g_{\pt}$ defined in \eqref{locinf_eqn51}.  To actually determine which elements of $\GC$
belong to $\GC^\pt$ one needs to take partial traces of the $\hat g \in \GC$
to see which of them do not trace down to
zero. In App.~\ref{locinf_apdx3} we present an efficient linear algebra algorithm
based on solving systems of linear equations mod $D$ that can find $\GC^{\pt}$
in $\mathcal{O}(K^2n^\theta)$ operations from $\ZZ_D$ where $\theta$ is defined
in Sec. \ref{locinf_sct8}.

If an operator $A$ on the full Hilbert space $\HC_n$ of the $n$ carriers
can be written as a tensor product of an operator on $\HC_\pt$ times the
identity operator $I_{\ptc}$ on $\HC_{\ptc}$ we shall say that $A$ is
\emph{based in $\pt$}\index{operator, based in}. Let $\BC$ be the collection of all operators on $\HC_n$
that are based in $\pt$.  Obviously, $\BC$ is closed under sums, products,
and scalar multiplication. In addition the partial trace $\Tr_{\ptc}[A]$ of an
operator $A$ in $\BC$ is ``essentially the same'' operator, apart from
normalization in the sense that
\begin{equation}
\label{locinf_eqn53} 
A = D^{-|\ptc|}\cdot \Tr_{\ptc}[A] \otimes I_{\ptc}.
\end{equation}
If $A\notin\BC$ is a Pauli product, then its partial trace over $\ptc$
vanishes, since $\Tr[X]$ and $\Tr[Z]$ and their powers (when not equal to
$I$) are zero.  Consequently the partial trace over $\ptc$ of $\Sigma(g\SC)$ in
\eqref{locinf_eqn48} is the same as the partial trace of $\Sigma[(g\SC)\cap\BC]$,
which suggests that it is useful to consider the properties of collections of
Pauli operators of the form $(g\SC)\cap\BC$ with $g$ an element of $\WC$.
The following result, proved in App.~\ref{locinf_apdx1}, turns out to be useful.

\begin{lemma}\label{locinf_thm2}
Let $g,h$ be two arbitrary elements of $\WC$, and $\BC$ the collection of operators with
base in $\pt$.  

i) The set $(g\SC)\cap\BC$ is empty if and only if $(g^{-1}\SC)\cap\BC$ is
empty. 

ii) Every nonempty set of the form $(g\SC)\cap\BC$ contains precisely 
\begin{equation}
\label{locinf_eqn54} M = |\SC\cap \BC| \geq 1
\end{equation}
elements.

iii) Two nonempty sets $(g\SC)\cap\BC$ and $(h\SC)\cap\BC$ are either
identical, which means $g\SC = h\SC$ and $\Sigma[(g\SC)\cap\BC] =
\Sigma[(h\SC)\cap\BC]$, or  else they have no elements in common
and the operators $\Sigma[(g\SC)\cap\BC]$ and $\Sigma[(h\SC)\cap\BC]$
are distinct.

iv) If both $(g\SC)\cap\BC$ and $(h\SC)\cap\BC$ are nonempty, their product
as sets, including multiplicity, is given by
\begin{equation}
\label{locinf_eqn55}
[ (g\SC)\cap\BC]\cdot [(h\SC)\cap\BC] = M[ (gh\SC)\cap\BC].
\end{equation}

\end{lemma}

By \eqref{locinf_eqn55} we mean the following.  The product (on the left) of any
operator from the collection $(g\SC)\cap\BC$ with another
operator from the collection $(h\SC)\cap\BC$ belongs to the collection
$(gh\SC)\cap\BC$ (on the right), and every operator in $(gh\SC)\cap\BC$ can be written 
as such a product in precisely $M$ different ways.

We are now in a position to state and prove our central result:

\subsection{Isomorphism theorem}
\label{locinf_sct12}

\begin{theorem}[Isomorphism]
\label{locinf_thm3}
  Let $C$ be an additive graph code with information group $\GC$, $P$ the projector onto the coding space $\HC_C$ and $\pt$ be some subset of the carrier qudits. Then the collection $\GC^\pt$ of members of $\GC$ with nonzero partial trace down to $\pt$, \eqref{locinf_eqn52}, is a subgroup of the information group $\GC$, and the mapping $\hat g\rightarrow g_{\pt}$ in \eqref{locinf_eqn51} carries $\GC^{\pt}$ to an
isomorphic group $\GC_{\pt}$ of nonzero operators on $\HC_\pt$.  Furthermore,

i) If $\hat g$ and $\hat h$ are any two elements of $\GC^\pt$, then
\begin{equation}
\label{locinf_eqn56} 
  \Tr_{\ptc}[\hat g \hat h] = \Tr_{\ptc}[\hat g]\;\Tr_{\ptc}[\hat h]/N \quad \text{or}\quad (gh)_\pt = g_\pt h_\pt 
\end{equation}

ii) If $\hat g\neq\hat h$ are distinct elements of $\GC^\pt$, 
$g_{\pt}\neq h_{\pt}$ are distinct elements of $\GC_\pt$. 

iii) The identity element
\begin{equation}
\label{locinf_eqn57}
P_{\pt}:=\Tr_{\ptc}[P]/N,
\end{equation}
of $\GC_{\pt}$ is a projector onto a subspace of $\HC_\pt$
(possibly the whole space) with rank equal to $\Tr[P]/N = K/N$. 

The normalization constant $N$ is given as
\begin{equation}
\label{locinf_eqn58}
 N := |\SC\cap\BC| \cdot D^{|\ptc|}/|\SC|
\end{equation}
where $\BC$ are the operators based in $\pt$.
\end{theorem}

\begin{proof}
The proof is a consequence of Lemma~\ref{locinf_thm2} and the following observations.
The trace $\Tr_{\ptc}[\hat g]$ in \eqref{locinf_eqn51} is, apart from a constant,
the trace of $\Sigma[(g\SC)\cap \BC]$, and is zero if $(g\SC)\cap \BC$ is
empty.  If the collection $(g\SC)\cap \BC$ is not empty, then by
Lemma~\ref{locinf_thm1} it consists of a collection of linearly independent
operators, and the trace of its sum cannot vanish. Thus there is a one-to-one,
see part (iii) of Lemma~\ref{locinf_thm2}, correspondence between nonempty sets of
the form $(g\SC)\cap\BC$ and the elements $\hat g$ in $\GC^\pt$. Then (i) and
(iv) of Lemma~\ref{locinf_thm2} imply both that $\GC^\pt$ is a group, and also that
the map from $\GC^\pt$ to $\GC_\pt$ is a homomorphism, whereas (ii) shows
that this is actually an isomorphism: $g_\pt=h_\pt$ is only possible when
$g\SC=h\SC$. That $N$ in \eqref{locinf_eqn58} is the correct normalization follows
from \eqref{locinf_eqn54}, \eqref{locinf_eqn55}, and \eqref{locinf_eqn48}.
\end{proof}

A significant consequence of Theorem~\ref{locinf_thm3} is the following result on the
presence and absence of information in the subset $B$, using
the terminology of Sec.~\ref{locinf_sct2}:

\begin{theorem} 
\label{locinf_thm4} 
Let $C$ be an additive graph code on $n$ carrier qudits, with information
group $\GC$. Let $\pt$ be a subset of the carrier qudits, $\GC^{\pt}$ the
corresponding subset information group, and $\JC(\hat g)$ the type of
information corresponding to $\hat g$ (as defined in Sec.~\ref{locinf_sct2}).
 Then
\begin{itemize}
 \item[i)] The $\JC(\hat g)$ type of information is perfectly
   present in $\pt$ if and only if $\hat g \in \GC^{\pt}$.
\item[ii)] The $\JC(\hat g)$ type of information is
absent from $\pt$ if and only if $\hat g^k \notin \GC^{\pt}$ for all integers
$k$ between $1$ and $D-1$.
 \item[iii)] All information is perfectly present in $\pt$ if and only if
   $\GC^{\pt}=\GC$.
 \item[iv)] All information is absent from $\pt$ if and only if $\GC^{\pt}$
   consists entirely of scalar multiples of the identity element $P$ of $\GC$.

\end{itemize}
\end{theorem}

The proof of the theorem can be found in App.~\ref{locinf_apdx2}. Statement (iii) is useful because
the check of whether there is a perfect quantum channel from the input
to $\pt$ involves a finite group $\GC$; one does not have to consider 
all normal operators of the form \eqref{locinf_eqn2}. 
Statement (ii) deserves further comment.  If $D$ is prime then the order of
any element of the Pauli group (apart from the identity) is $D$, see the
remark following \eqref{locinf_eqn9}. The same is true of elements of the quotient
group $\GCbar$, \eqref{locinf_eqn41}, and thus of members $\hat g$ of the isomorphic
group $\GC$.  Consequently, for any $k$ in the interval $1 < k < D$, there is
some $m$ such that $1=km \mod D$, which means $\hat g = (\hat g^k)^m$.  And
since $\GC^\pt$ is a group, $\hat g^k \in\GC^\pt$ implies $\hat
g\in\GC^\pt$. Thus when $D$ is prime, $\hat g\notin \GC^\pt$ is equivalent to
$\hat g^k \notin \GC^{\pt}$ for all integers $k$ between $1$ and $D-1$, and
the latter can be replaced by the former in statement (ii).
However, when $D$ is composite it is quite possible
to have $\Tr_{\ptc}[\hat g] = 0$ but $\Tr_{\ptc}[\hat g^{k'}] \neq 0$ for some
$k'$ larger than 1 and less than $D$; see the example below. In this situation
we can still say that $\JC(\hat g^{k'})$ is perfectly present, but it is not
true that $\JC(\hat g)$ is absent. One can regard the type $\JC(\hat g)$ as a
\emph{refinement}\index{refinement} of $\JC(\hat g^{k'})$, and as explained in Sec.~\ref{locinf_sct2},
although the coarse-grained $\JC(\hat g^{k'})$ information is perfectly
present in $\pt$, the additional information associated with the refinement is
not.

As an example, suppose $\hat g$ has a spectral decomposition
\begin{equation}
\label{locinf_eqn59}
\hat g=J_0+\ii J_1-J_2-\ii J_3,
\end{equation}
with the $J_j$ orthogonal projectors such that
\begin{align}
\label{locinf_eqn60}
\Tr_{\ptc}[J_0]=\Tr_{\ptc}[J_2]\neq\Tr_{\ptc}[J_1]=\Tr_{\ptc}[J_3].
\end{align}
Then $\Tr_{\ptc}[\hat g]=0$, whereas 
\begin{equation}
\label{locinf_eqn61}
\hat g^2=(J_0+J_2)-(J_1+J_3),
\end{equation}
and thus $\Tr_{\ptc}[\hat g^2]\neq 0$.  Thus $\hat g^2$ is an element of
$\GC^{\pt}$, whereas $\hat g$ is not, and so the coarse grained $\JC(\hat g^2)$
information corresponding to the decomposition on the right side of
\eqref{locinf_eqn61} is present in $\pt$, while the further refinement corresponding
to the right side of \eqref{locinf_eqn59} is not.  Precisely this structure is
produced by a graph code on two carriers of dimension $D=4$, with graph state
$\ket{G}=\ket{++}$, coding group $\CC=\langle Z_1Z_2\rangle$, information
group $\GC=\langle X_1P , Z_1Z_2P \rangle$, coding space projector
\begin{equation}
\label{locinf_eqn62}
 P = (I + X_1 X_2^{3} + X_1^2 X_2^2 +X_1^3 X_2)/4,
\end{equation}
and 
\begin{equation}
\label{locinf_eqn63}
\hat g=X_1P=\dyad{\bar 0\bar0}{\bar0\bar0}+
\ii\dyad{\bar 1\bar 2}{\bar 1 \bar 2}-\dyad{\bar 2\bar0}{\bar 2\bar0}
-\ii\dyad{\bar 3\bar 2}{\bar 3 \bar 2},
\end{equation}
where $\ket{\bar j}=Z^j\ket{+}$ are the eigenvectors of the $X$ operator.

\subsection{Information flow}
\label{locinf_sct13}

At this point let us summarize how we think about information ``flowing'' from
the input via the encoding operation into a subset $\pt$ of the code
carriers. At the input the information is represented by the quotient group
$\GCbar_0=\WC_0/\SC_0$, see \eqref{locinf_eqn41}, or more concretely by the
isomorphic group $\GC_0$ of operators generated by the cosets, as in
\eqref{locinf_eqn48}.  The encoding operation $UW$, see \eqref{locinf_eqn36} and
\eqref{locinf_eqn38}, maps $\GCbar_0$ to the analogous $\GCbar=\WC/\SC$ associated
with the code $C$, and likewise $\GC_0$ to the group of operators $\GC$ acting
on the coding space $\HC_C$.  Tracing away the complement $\ptc$ of $\pt$ maps
some of the $\hat g$ operators of $\GC$ to zero, and the remainder form the
subset information group $\GC^{\pt}$. Applying the inverse $UW$ map to
$\GC^{\pt}$ gives $\GC^{\pt}_0$, a subgroup of $\GC_0$ that tells us what
types of information at the input (i.e. before the encoding) are available in
the subset of carriers $\pt$.  This is illustrated by various examples in the
next section.

\section{Examples}
\label{locinf_sct14}

\subsection{General principles}
\label{locinf_sct15}

In this section we apply the principles developed earlier in the chapter to some
simple $[[n,k,\delta]]_D$ additive graph codes, where $n$ is the number of
qudit carriers, each of dimension $D$, the dimension of the coding space
$\HC_C$ is $K=D^k$, and $\delta$ is the distance of the code; see Chapter 10 of \cite{NielsenChuang:QuantumComputation} for a
definition of $\delta$.  We shall be interested in the subset information
group $\GC^\pt$, \eqref{locinf_eqn52}, that represents the information about the
input that is present in the subset $\pt$ of carriers.  Rather than discussing
$\GC^\pt$ or its traced down counterpart $\GC_\pt$, it will often be simpler
to use $\GC_0^\pt$, the subset information group referred back to the channel
input, see Sec.~\ref{locinf_sct13} above, and in this case we add an initial
subscript $0$ to operators: $X_{01}$ means the $X$ operator on the first qudit
of the input.  Since all three groups are isomorphic to
one another, the choice of which to use in any discussion is a matter of
convenience. (In the examples below for the sake of brevity we sometimes 
omit a term $\mathrm{e}^{\ii \phi}I$ from the list of generators of $\GC^{\pt}_0$.)

Before going further it is helpful to list some general principles of quantum
information that apply to all codes, and which can simplify the analysis of
particular examples, or give an intuitive explanation of why they work.  In
the following statements ``information'' always means information about the
input which has been encoded in the coding space through some isometry.

1. If all information is perfectly present in $\pt$,
then all information is absent from $\ptc$.

2. If all information is absent from $\ptc$ then all information is perfectly
present in $\pt$.

3. If the information about some orthonormal basis (i.e., the type
corresponding to this decomposition of the identity) is perfectly present in
$\pt$, then the information about a mutually-unbiased basis is absent from
$\ptc$.

4. If two types of information that are ``sufficiently incompatible'' are both
perfectly present in $\pt$, then all information is perfectly present in
$\pt$.  In particular this is so when the two types are associated with
mutually unbiased bases.

5. For a code of distance $\delta$ all information is absent from any $\pt$ if
$|\pt|<\delta$, and all information is perfectly present in $\pt$ if $|\pt| >
n-\delta$.

Items 1, 2, 3 and 4 correspond to the No Splitting, Somewhere, Exclusion
and Presence theorems of \cite{PhysRevA.76.062320}, which also gives weaker
conditions for ``sufficiently incompatible.''  The essential idea behind 5 is
found in Sec.~III~A of \cite{PhysRevA.56.33}
\footnote{It is shown in \cite{PhysRevA.56.33} that if noise only affects a certain 
subset $\ptc$ of the carriers with $|\ptc|<\delta$, then the errors can be corrected
using the complementary set $\pt$. In our notation this is equivalent to saying that all the information is in $\pt$.
}.

\subsection{One encoded qudit \label{locinf_sct16}}

\begin{figure}
\begin{center}
\includegraphics{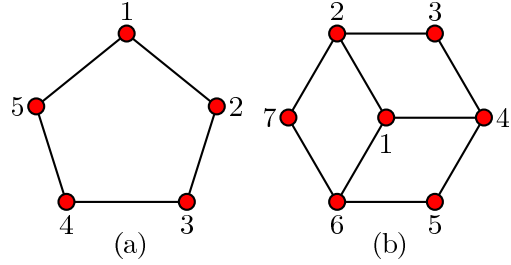}
\caption{(a) The graph state for the $[[5,1,3]]_D$ code; (b) The graph state
  for Steane $[[7,1,3]]_2$ code}
\label{five_steane}
\end{center}
\end{figure}

It was shown in \cite{IEEE.45.1827} that a $[[5,1,3]]_D$
code exists for all $D$. Here we consider the graph version
\cite{PhysRevA.65.012308} where the coding group is
\begin{equation}
\label{locinf_eqn64}
\CC=\langle Z_1 Z_2 Z_3 Z_4 Z_5\rangle
\end{equation}
and the graph state is shown in Fig.~\ref{five_steane}(a). Our formalism shows
that, whatever the value of $D$, there are only two possibilities.  When $|B|$
is 1 or 2 $\GC^{\pt}$ is the just the group identity, the projector $P$ on the
coding space, so all information is absent whereas if $|B|$ is 3, 4 or
(obviously) 5, $\GC^{\pt} = \GC$, so the subsystem $B$ is the output of a
perfect quantum channel.  To be sure, these results also follow from principle
5 in the above list, given that $\delta=3$ for this code.

The Steane $[[7,1,3]]_2$ code, a graphical version of which
\cite{quantph.0709.1780} has a coding group
\begin{equation}
\label{locinf_eqn65}
\CC=\langle Z_3 Z_5 Z_7\rangle
\end{equation}
for the graph state shown in Fig.~\ref{five_steane} (b), is more interesting
in that while principles 5 ensures that all $|\pt|\leq 2 = \delta -1$ subsets
of carriers contain zero information and all $|\pt|\geq 5 = n-\delta+1$
subsets contain all the information, one qubit, it leaves open the question of
what happens when $|\pt|=3$ or 4. We find that all information is perfectly
present when $\pt$ is $\{1,2,5\}$, $\{1,3,6\}$, $\{1,4,7\}$, $\{2,3,4\}$,
$\{2,6,7\}$, $\{4,5,6\}$, or $\{3,5,7\}$---representing three different
symmetries in terms of the graph in the figure---and absent for all other
cases of $|B|=3$.  Therefore all information is absent from the $|\pt|=4$ subsets which
are complements of the seven just listed, and perfectly present in all others
of size $|\pt|=4$.  So far as we know, generalizations of this code to $D>2$
have not been studied.

\begin{figure}
\begin{center}
\includegraphics{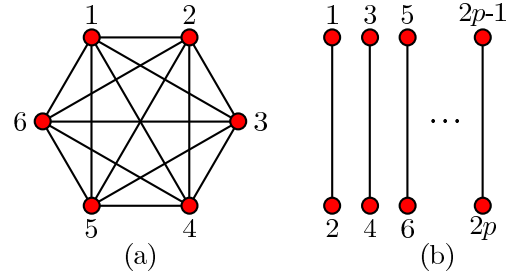}
\caption{(a) Complete graph (on $6$ qudits); (b) Bar graph with $n=2p$ carriers and $p$ bars}
\label{bar_complete}
\end{center}
\end{figure}

A simple code in which a specific type of information is singled out is
$[[n, 1, 1]]_D$ generated by
\begin{equation}
\label{locinf_eqn66}
\CC=\langle Z_1Z_2 \dotsm Z_n\rangle
\end{equation}
on the \emph{complete graph}\index{complete graph}, illustrated in Fig.~\ref{bar_complete}(a) for
$n=6$.  Whereas all information is (of course) present when $|\pt|=n$, it
turns out that for any subset $B$ with $1\leq |\pt| < n$ one has 
$\GC^{\pt}_{0}= \langle X_{01} Z_{01} \rangle$, i.e., the abelian group consisting
of all powers of the operator $X_1 Z_1$ on the input qudit. 
Thus the information is ``classical,'' corresponding to that
decomposition of the input identity that diagonalizes $X_1 Z_1$.
The intuitive explanation for this situation is that this $X_1 Z_1$ type
of information is separately copied as an ideal classical channel, see \eqref{locinf_eqn6},
 to each of the carrier qudits, and as a
consequence other mutually unbiased types of information are ruled out by
principle 3. This, of course, is typical of ``classical'' information, which
can always be copied. 

A more interesting example in which distinct types of information come into
play is the bar graph, Fig.~\ref{bar_complete} (b), in which $n$ qudits are
divided up into $p=n/2$ pairs or ``bars,'' and the code is generated by 
\begin{equation}
\label{locinf_eqn67}
\CC=\langle Z_1Z_2\dotsm Z_n\rangle.
\end{equation}
Let us say that a subset of carriers $\pt$ has property I if the corresponding
subgraph contains at least one of bars, and property II if it contains at
least one qudit from each of the bars.  Then:

(i) If $\pt$ has property I but not II, $\GC^{\pt}_{0} = \langle
X_{01} \rangle$, an abelian group.

(ii) If $\pt$ has property II but not I, $\GC^{\pt}_{0} = \langle
X_{01}^p Z_{01} \rangle$, another abelian group

(iii) If $\pt$ has both property I and property II, all information (1 qudit)
is perfectly present. 

(iv) When $\pt$ has neither property I nor II, all information is absent.

While both (i) and (ii) are ``classical'' in an appropriate sense and indeed represent an ideal classical channel, the two
abelian groups do not commute with each other, so the two types of information
are incompatible, and it is helpful to distinguish them. Case (iii)
illustrates principle 4, since $X_{01}$ and $X_{01}^p Z_{01}$ (whatever the value
of $p$) correspond to mutually unbiased bases. In case (iv) the complement
$\ptc$ of $\pt$ possesses both properties I and II, and therefore contains all
the information, so its absence from $\pt$ is an illustration of principle 1.

\subsection{Two encoded qudits\label{locinf_sct17}}

\begin{figure}
\begin{center}
\includegraphics{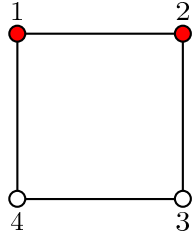}
\caption{The graph state  of the $[[4,2,2]]_D$ code}
\label{K4code}
\end{center}
\end{figure}

Consider a $[[4,2,2]]_D$ code based on the graph state shown in Fig.~\ref{K4code} whose
coding group 
\begin{equation}
\label{locinf_eqn68}
\CC=\langle Z_1 Z_2, Z_3 Z_4 \rangle,
\end{equation}
employs two generators of order $D$, and thus encodes two qudits.  Note that
while the graph state has the symmetry of a square the coding group has a
lower symmetry corresponding to the different types of nodes employed in the
figure.

Let us begin with the qubit case $D=2$.  Our analysis shows that
when $|\pt|=1$ all information is absent, and thus for $|\pt|\geq 3$ all
information is present, consistent with the fact that this code has
$\delta=2$ \cite{PhysRevA.78.042303},  see principle 5. Thus the
interesting cases are those in which $|\pt|=|\ptc|=2$, for which one finds: 
\begin{align}
\label{locinf_eqn69}
 \pt=\{1,3\}, \ptc=\{2,4\}:&\quad  
 \GC^{\pt}_{0}=\GC^{\ptc}_{0} = \langle X_{01} Z_{01} Z_{02}, X_{01} X_{02}  \rangle;\\ 
\label{locinf_eqn70}
\pt=\{1,4\}, \ptc=\{2,3\}:&\quad  
 \GC^{\pt}_{0}=\GC^{\ptc}_{0} = \langle X_{01} Z_{01}, X_{02} Z_{02} \rangle;\\ 
\label{locinf_eqn71}
 \pt=\{1,2\}, \ptc=\{3,4\}:&\quad  
 \GC^{\pt}_{0}=\GC^{\ptc}_{0} = \langle X_{01} Z_{01}, X_{02} Z_{02} \rangle.
\end{align}
In each case the generators commute and thus the subgroup $\GC^{\pt}_{0}$ is
abelian. Hence the information is ``classical'', and the same type is present
both in $\pt$ and $\ptc$, not unlike the situation for the complete graph
considered earlier. However, the three subgroups do not commute with each
other, so the corresponding types of information are mutually incompatible, a
situation similar to what we found for the bar graph.

For $D>2$ it is again the case that all information is absent when $|\pt|=1$
completely present for $|\pt|\geq 3$.  And \eqref{locinf_eqn69} and \eqref{locinf_eqn70}
remain correct (with each generator of order $D$), and these subgroups are
again abelian.  However, when $\pt=\{1,2\}$ and $\ptc=\{3,4\}$, \eqref{locinf_eqn71}
must be replaced with
\begin{equation}
\label{locinf_eqn72}
\GC^{\pt}_{0} = \langle Z_{01}^{} X_{02}^{2}, Z_{02}^{} \rangle,\quad 
\GC^{\ptc}_{0} = \langle Z_{01}^{}, X_{01}^2 Z_{02}^{} \rangle.
\end{equation}
In each case the two generators do not commute with each other, so neither
subgroup is abelian.  However, all elements of $\GC^{\pt}_{0}$ commute with
all elements of $\GC^{\ptc}_{0}$. Also, the two subgroups are isomorphic
(interchange subscripts 1 and 2).

For \emph{odd} $D \geq 3$ one can use for $\GC^{\pt}_{0}$ an alternative pair
of generators
\begin{equation}
\label{locinf_eqn73}
\GC^{\pt}_{0} = \langle Z_{01}^m X_{02}^{}, Z_{02}^{} \rangle,\quad 
m:=(D+1)/2,
\end{equation}
whose order is $D$ and whose commutator is
\begin{equation}
\label{locinf_eqn74}
\bigl( Z_{01}^m X_{02}^{}\bigr) Z_{02} = \omega Z_{02} \bigl( Z_{01}^m X_{02}^{}\bigr).
\end{equation}
This means---see \eqref{locinf_eqn8}---that $\GC^{\pt}_{0}$, and thus also the
(isomorphic) $\GC^{\ptc}_{0}$, is isomorphic to the
Pauli group of a single qudit. Since  $\GC^{\pt}_{0}$ and
$\GC^{\ptc}_{0}$ commute with each other, it is natural to think of the pair
as associated with the tensor product of two qudits with the same $D$.  That
this is correct can be confirmed by explicitly constructing
a ``pre-encoding'' circuit embodying the unitary
\begin{equation}
\label{locinf_eqn75}
(\mathrm{F}_1\otimes \mathrm{F}_2)^\dag \mathrm{CP}^{-m}_{12}
(\mathrm{F}_1\otimes \mathrm{F}_2),
\end{equation} 
expressed in terms of the Fourier and CP gates defined in Sec.~\ref{locinf_sct5},
that carries the Pauli groups on ``pre-input" qudits 1 and 2 onto $\GC^{\pt}_{0}$ and
$\GC^{\ptc}_{0}$, respectively. 

Things become more complicated for \emph{even} $D\geq 4$, where $\GC^{\pt}_{0}$ (and also $\GC^{\ptc}_{0}$) are no longer isomorphic to the Pauli group of a single qudit.

\section{Conclusion}
\label{locinf_sct18}

We have shown that for additive graph codes with a set of $n$ carrier qudits,
each of the same dimension $D$, where $D$ is any integer greater than 1, it is
possible to give a precise characterization of the information from the
coding space that is present in an arbitrary subset $\pt$ of the carriers.
This information corresponds to a subgroup $\GC^\pt$ of a group $\GC$, the
information group of operators on the coding space, that spans the coding
space and provides a useful representation of the information that it
contains.  We discuss how what we call a trivial code, essentially a tensor
product of qudits of (possibly) different dimensions, can be encoded into the
coding space in a manner which gives one a clear intuitive interpretation of
$\GC$.  The subgroup $\GC^\pt$ is then simply the subset of operators in $\GC$
whose trace down to $\pt$ is nonzero, and the traced-down operators when
suitably normalized form a group $\GC_\pt$ that is isomorphic to $\GC^\pt$.
The information present in those operators in $\GC$ that are not in $\GC^\pt$
disappears so far as the subsystem $\pt$ is concerned, as their partial traces
are zero.  This is the central result of our chapter and is illustrated by a
number of simple examples in Sec.~\ref{locinf_sct14}. We also provide in
App.~\ref{locinf_apdx3} a relatively simple algorithm for finding the elements of
$\GC^B$.

These results can be extended to arbitrary qudit stabilizer codes even if they
are not graph codes, by employing appropriate stabilizer and information
groups, as in Sec.~\ref{locinf_sct6}. Here, however, the concept of a trivial code,
and thus our perspective on the encoding step, may not apply.  The extension
of these ideas, assuming it is even possible, to more general codes, such as
nonadditive graph codes, remains an open question.

As shown in App.~\ref{locinf_apdx4} our formalism can be fitted within the general
framework of invariant algebras as discussed in \cite{PhysRevLett.98.100502,
  PhysRevA.76.042303,quantph.0907.4207,PhysRevLett.100.030501}.  The overall
conceptual framework we use is somewhat different from that found in these
references in that we directly address the question of what information is
present in the subsystem of interest, rather than asking whether there exists
some recovery operation (the $\RC$ in App.~\ref{locinf_apdx4}) that will map an
algebra of operators back onto its original space. Thus in our work the
operator groups $\GC^\pt$ on the coding space and $\GC_\pt$ on the subsystem
are isomorphic but not identical.  Hence, even though there is, obviously, a
close connection between our ``group approach'' and the ``algebraic
approach,'' the algebra of interest being generated from the group of
operators, further relationships remain to be explored.  The fact that the
arguments in App.~\ref{locinf_apdx4} are not altogether straightforward suggests that
the use of groups in cases where this is possible may provide a useful
supplement, both mathematically and intuitively, to other algebraic ideas.  In
particular the additional structure present in an additive graph code allows
one to determine $\GC^\pt$ in $\OC(n^{\theta}+K^2n^2)$, App.~\ref{locinf_apdx3}, as
against $\OC(K^6)$ for the algorithm presented in
\cite{PhysRevLett.100.030501} for a preserved matrix algebra, where $K$ is the
dimension of the input and output Hilbert space.

\begin{subappendices}
\section{Proof of Lemmas~\getrefnumber{locinf_thm1} and \getrefnumber{locinf_thm2} }
\label{locinf_apdx1}

Proof of Lemma~\ref{locinf_thm1}

The operators in $p\RC$ are linearly independent when those in $\RC$ are
linearly independent, since $p$ is unitary and thus invertible.  This
establishes (i).  For (ii), consider the case where $q$ is the identity $I$.
As the collection $\RC$ is linearly independent, there is at most one $r\in
\RC$ such that $pr$ is a multiple of the identity.  If such an $r$ exists, $p$
is of the form $\mathrm{e}^{\ii\phi} r^{-1}$, and since $\RC$ is a group,
$p\RC=\mathrm{e}^{\ii\phi}r^{-1}\RC=\mathrm{e}^{\ii\phi}\RC$, we have situation $(\alpha)$, with the collection
$p\RC\cup\RC$ linearly dependent.  Next assume the collection
$p\RC\cup\RC$ is linearly dependent, which means there are complex numbers
$\{a_r\}$ and $\{b_r\}$, not all zero, such that
\begin{equation}
\label{locinf_eqn76} 
\sum_{r\in\RC} \left[ a_r r + b_r pr\right] = 0. 
\end{equation}
This is not possible if all the $a_r$ are zero, since this would mean $p\sum_r
b_r r=0$, thus $\sum_r b_r r=0$ implying $b_r=0$ for every $r$, since the
$\RC$ collection is by assumption linearly independent.  Thus at least one
$a_r$, say $a_s$ is not zero. Multiply both sides of \eqref{locinf_eqn76} by $s^{-1}$
on the right and take the trace:
\begin{equation}
\label{locinf_eqn77}
a_s\Tr[I] + \sum_{r\in\RC} b_r\Tr[prs^{-1}] = 0, 
\end{equation}
implying there is at least one $r$ for which $\Tr[prs^{-1}]\neq 0$.  But then
$p$ is of the form $\mathrm{e}^{\ii\phi} sr^{-1}= \mathrm{e}^{\ii\phi}\bar r^{-1}$ for 
$\bar r = rs^{-1}\in\RC$, so we are back to situation $(\alpha)$. Hence the alternative
to $(\alpha)$ is $(\beta)$: the operators in $p\RC\cup\RC$ are linearly
independent.  Finally, if $q$ is not the identity $I$, simply apply the
preceding argument with $\bar p = q^{-1}p$ in place of $p$.
\medskip

Proof of Lemma~\ref{locinf_thm2}

Statement (i) is a consequence of the fact that if an invertible operator is in
$\BC$, so is its inverse, and since
$\SC$ is a group, $g\SC$ consists of the inverses of the elements in
$g^{-1}\SC$.

Statements (ii) and (iv) follow from a close examination of \eqref{locinf_eqn55}.
Assume both sets on the left side are nonempty. If $gs_1$ 
and  $hs_2$ are both in $\BC$, so is their product $gs_1hs_2= ghs_1s_2$, where
we use the fact that $g$ and $h$ commute with every element of $\SC$. 
If, on the other hand, $(gh\SC)\cap\BC$ and $(g\SC)\cap\BC$ are nonempty,
any element, say $ghs_1$, in the former can be written using a specific
element, say $g\bar s$, in the latter, as
\begin{equation}
\label{locinf_eqn78}
 ghs_1 = (g\bar s)(h s_2) 
\end{equation}
where $s_2 = s_1 \bar s^{-1}$ is uniquely determined by this equation, and the
fact that both $ghs_1$ and $g\bar s$ are (by assumption) in $\BC$ means the
same is true of $h s_2$.  Thus not only can every element of $(gh\SC)\cap\BC$
be written as a product of elements of $(g\SC)\cap\BC$, but there is a one-to-one
correspondence between $(gh\SC)\cap\BC$ and $(g\SC)\cap\BC$, which must
therefore be of equal size. A similar argument shows that $(gh\SC)\cap\BC$
and  $(h\SC)\cap\BC$ are of the same size. This establishes both (ii)
and (iv).

As for (iii), use the fact that the cosets $g\SC$ and $h\SC$ are either
identical or have no elements in common, so the same is true of their
intersections with $\BC$.  If $g\SC$ and $h\SC$ have no elements in common,
Lemma~\ref{locinf_thm1} with $\RC=\SC$ tells us that either $g\SC= \mathrm{e}^{\ii\phi} (h\SC)$
for some nonzero $\phi$, in which case 
$\Sigma[(g\SC)\cap\BC] = \mathrm{e}^{\ii\phi}\Sigma[(h\SC)\cap\BC]$
is distinct from $\Sigma[(h\SC)\cap\BC]$, or else the collection $(g\SC)\cup
(h\SC)$ is linearly independent, which means that its intersection with $\BC$
shares this property and the operators $\Sigma[(g\SC)\cap\BC]$ and
$\Sigma[(h\SC)\cap\BC]$ are linearly independent.

\section{Proof of  Theorem~\getrefnumber{locinf_thm4} \label{locinf_apdx2}}

The proof of Theorem~\ref{locinf_thm4} makes use of the following:
\begin{lemma}
\label{locinf_thm5}

Let $\hat g=P\hat gP$ be an information operator in $\GC$ with spectral decomposition
\begin{equation}
\label{locinf_eqn79} 
\hat g  = \sum_{j=0}^{m-1} \lambda_j J_j, 
\end{equation}
where the mutually orthogonal projectors $J_j$ sum to $P$.  Then each
projector $J_j$ can be written as a polynomial in $\hat g$ with $\hat g^0=P$:
\begin{equation}
\label{locinf_eqn80} 
J_j = \sum_{k=0}^{m-1} \alpha_{jk} \hat g^k.
\end{equation}
\end{lemma}

\begin{proof}
The proof consists in noting that
\begin{equation}
\label{locinf_eqn81} 
\hat g^k = \sum_{j=0}^{m-1} \lambda_j^k J_j=\sum_{j=0}^{m-1}\beta_{kj}J_j,
\end{equation}
is a linear equation in the $J_j$ with $\beta_{kj} = \lambda_j^k$ an $m\times
m$ Vandermonde matrix whose determinant is $\prod_{j > k}(\mu_j-\mu_k)$ (see
p.~29 of \cite{HornJohnson:MatrixAnalysis}).  As the $\mu_j$ are distinct the
matrix $\beta_{kj}$ has an inverse $\alpha_{jk}$.
\end{proof}

To prove (i) of Theorem~\ref{locinf_thm4}, first assume that $\hat g$ is in
$\GC^\pt$.  Since $\GC^{\pt}$ is a group with identity $P$, this means that all
powers of $\hat g$, including $\hat g^0=P$, are also in $\GC^{\pt}$.
Consequently, the projectors entering the spectral decomposition \eqref{locinf_eqn79}
of $\hat g$ satisfy
\begin{equation}
\label{locinf_eqn82}
N^{-1}\Tr_{\ptc}[J_j] \; \Tr_{\ptc}[J_k]
= \Tr_{\ptc}[J_jJ_k] 
= \delta_{jk}\Tr_{\ptc}[J_j],
\end{equation}
with the first equality obtained by expanding $J_j$ and $J_k$ in powers of
$\hat g$, \eqref{locinf_eqn80}, and using \eqref{locinf_eqn56} along with the linearity of
the  partial trace.  This orthogonality of the partial traces of different projectors,
see \eqref{locinf_eqn3}, implies that the $\JC(\hat g)$ type of information is
perfectly present in $\pt$. Conversely, if the $\JC(\hat g)$ type of
information is perfectly present in $\pt$ then the partial traces down to
$\pt$ of the different $J_j$, which cannot be zero, are mutually orthogonal
and thus linearly independent.  Therefore by \eqref{locinf_eqn79}, $\Tr_{\ptc}[\hat
g]$ cannot be zero, and $\hat g$ is in $\GC^\pt$.

The prove (ii) note that $\hat g^k$ absent from $\GC^{\pt}$ for $1\leq k < D$ 
means that 
$\Tr_{\ptc}[\hat g^k]=0$ for these values of $k$, and thus by taking the
partial trace of both sides of  \eqref{locinf_eqn80} and using \eqref{locinf_eqn57},  
\begin{equation}
\label{locinf_eqn83} 
\Tr_{\ptc}[J_j] = N\alpha_{j0} P_{\pt}.
\end{equation}
Since these partial traces are identical up to a multiplicative constant 
there is no information of the $\JC(\hat g)$ type in $\pt$.  For the converse,
if there is no $\JC(\hat g)$ information in $\pt$ then there is also no
$\JC(\hat g^2)$, $\JC(\hat g^3)$, etc. information in $\pt$, since the
projectors which arise in the spectral decomposition of $\hat g^k$ are already
in the spectral decomposition of $\hat g$, see \eqref{locinf_eqn81}.
Consequently, by (i), these $\hat g^k$ must be absent from $\GC^\pt$.  

To prove (iii), note that if all information is perfectly present in $\pt$
this means that for every $\hat g\in\GC$ the $\JC(\hat g)$ information is
present in $\pt$, and therefore, by (i), $\hat g\in\GC^\pt$, so $\GC =
\GC^\pt$. For the converse, let $Q_1$ and $Q_2$ be two orthogonal but
otherwise arbitrary projection operators on subspaces of the coding space
$\HC_C$.  Because the elements of the information group $\GC$ form a basis for
the set of linear operators on $\HC_C$, see comments at the end of
Sec.~\ref{locinf_sct9}, $Q_1$ and $Q_2$ can both be written as sums of elements
$\hat g$ in $\GC$, and the same argument that was employed in \eqref{locinf_eqn82}
shows that the orthogonality of $Q_1$ and $Q_2$ implies the orthogonality of
$\Tr_{\ptc}[Q_1]$ and $\Tr_{\ptc}[Q_2]$.

To prove (iv), note that if $\GC^{\pt}$ consists entirely of scalar
multiples of $P$, the partial trace down to
$\pt$ of any projector $Q$ on a subspace of $\HC_C$, since it can be written
as a linear combination of the partial traces of the $\hat g$ in $\GC$, most
of which vanish, will be some multiple of $P_{\pt}$, and thus all information is
absent from $\pt$.  Conversely, if $\GC^{\pt}$ contains a $\hat g$ which is
not proportional to $P$ the corresponding $\JC(\hat g)$ type of information
will be present in $\pt$ by (i), so it is not true that all information is
absent from $\pt$, a contradiction.

\section{Algorithm for finding $\GC^{\pt}$\label{locinf_apdx3}}

Here we present an algorithm for determining the subset information group
$\GC^{\pt}$ by finding the elements $\hat g$ of $\GC$ whose partial trace down
to $\pt$ is nonzero.  If two or more elements differ only by a phase it is
obviously only necessary to check one of them.  For what follows it is helpful
to adopt the abbreviation
\begin{equation}
\label{locinf_eqn84}
E^{(\vect{x}|\vect{z})} := X^{\vect{x}} Z^{\vect{z}}
\end{equation}
with $(\vect{x}|\vect{z})$ an $n$-tuple row
vector pair, and thus a $2n$-tuple of integers between $0$ and $D-1$. 
Arithmetic operations in the following analysis are assumed to be mod $D$.

First consider the trivial code on the trivial graph, Sec.~\ref{locinf_sct8}, with
information group $\GC_0^\pt$ consisting of elements of the form $\hat g_0=
g_0 P_0$, see \eqref{locinf_eqn48}, with $g_0 = E^{(\vect{x}_0|\vect{z}_0)}$ some 
element of $\WC_0=\langle \SC_0^G, \CC_0\rangle$, and
\begin{equation}
P_0 = |\SC_0|^{-1}\sum_{\vect{x}\in\XZ} X^{\vect{x}},
\label{locinf_eqn85} 
\end{equation}
where $\XC_0$ denotes the collection of $n$-tuples that enter the stabilizer
$\SC_0$, \eqref{locinf_eqn37}.
 By choosing $\vect{x}_0$ and $\vect{z}_0$ to be of the form
\begin{align}
\label{locinf_eqn86}
\vect{x}_0 &= (\xi_1,\xi_2,\ldots \xi_k,0,0,\ldots 0),\notag\\
\vect{z}_0 &= (\zeta_1 m_1,\zeta_2 m_2,\ldots \zeta_k m_k,0,0,\ldots 0),
\end{align}
using integers in the range
\begin{equation}
\label{locinf_eqn87}
 0\leq \xi_j \leq (d_j-1),\quad  0\leq \zeta_j \leq (d_j-1), 
\end{equation}
we obtain a single representative $g_0 =E^{(\vect{x}_0|\vect{z}_0)}$ for each
coset $g_0\SC_0$ in $\WC/\SC_0$.  The corresponding information operator,
which depends only on the coset, is 
\begin{equation}
\label{locinf_eqn88} 
\hat g_0 =  E^{(\vect{x}_0|\vect{z}_0)} P_0 = 
 |\SC_0|^{-1} \sum_{\vect{x}\in\XZ} \omega^{-\vect{z}_0\vect{x}} 
 E^{(\vect{x}+\vect{x}_0|\vect{z}_0)},
\end{equation}
where the addition of $\vect{x}$ and $\vect{x}_0$ is component-wise mod $D$,
and $\vect{z}_0\vect{x}$ denotes the scalar product of $\vect{z}_0$ and
$\vect{x}$ mod $D$ (multiply corresponding components and take the sum mod
$D$).

Elements of the information group $\GC^\pt$ of the nontrivial code of interest
to us are then of the form
\begin{align}
\label{locinf_eqn89}
\hat g &= (UW)\hat g_0(UW)^{\dagger}\notag\\
 &=  
|\SC_0|^{-1} \sum_{\vect{x}\in\XZ}
\omega^{\nu(\vect{x},\vect{x}_0,\vect{z})-\vect{z}_0\vect{x}}
E^{(\vect{x}+\vect{x}_0|\vect{z}_0)\QZ},
\end{align}
where we use the fact that because the conjugating operator $UW$,
\eqref{locinf_eqn36}, is a Clifford operator there is a $2n\times 2n$ matrix $Q$ over
$\ZZ_D^{2n}$, representing a symplectic automorphism
\cite{PhysRevA.71.042315}, such that
\begin{equation}
\label{locinf_eqn90} 
(UW) E^{(\vect{x}|\vect{z})}(UW)^{\dagger} = 
 \omega^{\nu(\vect{x},\vect{z})} E^{(\vect{x}|\vect{z})\QZ}.
\end{equation}
with $(\vect{x}|\vect{z})Q$ the $2n$-tuple, interpreted as an $n$-tuple pair,
obtained by multiplying $(\vect{x}|\vect{z})$ on the right by $Q$, and
$\nu(\vect{x},\vect{z})$ an integer whose value does not concern us. The
explicit form of $Q$ can be worked out by means of the encoding procedure
presented in Sec.~\ref{locinf_sct8}, using Tables~\ref{locinf_tbl1} and~\ref{locinf_tbl2}.

The operators appearing in the sum on the right side of \eqref{locinf_eqn89}
are linearly independent Pauli products, since $Q$ is nonsingular. The 
trace down to $\pt$ of such a product is nonzero if and only if its base is in 
$\pt$, and when  nonzero the result after the trace is essentially the same
operator: see \eqref{locinf_eqn53} and the associated discussion.  Consequently
$g_B=N^{-1} \Tr_{\ptc}[\hat g]$ is nonzero if and only if the trace down to
$\pt$ of at least one operator on the right side of \eqref{locinf_eqn89} is nonzero.
A useful test takes the form
\begin{equation}
\label{locinf_eqn91} 
\Tr_{\ptc}[ E^{(\vect{x}|\vect{z})} ]\neq 0
\Longleftrightarrow (\vect{x}|\vect{z})J  =\vect{0},
\end{equation}
where $\vect{0}$ is the zero row vector, and $J$ is a diagonal $2n\times 2n$
matrix with 1 at the diagonal positions $j$ and $2j$ whenever qudit $j$
belongs to $\ptc$, and 0 elsewhere.  Therefore the $\hat g$ associated with
$\vect{x}_0$ and $\vect{z}_0$ through \eqref{locinf_eqn88} and \eqref{locinf_eqn89} is a
member of $\GC^B$ if and only if there is at least one $\vect{x}\in\XC_0$
such that
\begin{equation}
\label{locinf_eqn92}
(\vect{x}+\vect{x}_0|\vect{z}_0)Q J 
= \vect{0} \text{\ \ or\ \ } (\vect{x}|\vect{0})Q J 
= -(\vect{x}_0|\vect{z}_0)Q J.
\end{equation}

The $\vect{x}$ that belong to $\XZ$ are characterized by the equation
\begin{equation}
\label{locinf_eqn93}
\vect{x} M = \vect{0},
\end{equation}
where $M$ is an $n\times k$ matrix that is everywhere 0 except for
$M_{jj}=m_j$ for $1\leq j\leq k$, using the $m_j$ that appear in
\eqref{locinf_eqn28}.  Consequently, instead of asking whether \eqref{locinf_eqn92} has a
solution $\vect{x}$ belonging to $\XZ$ one can just as well ask if there is
any solution to the pair \eqref{locinf_eqn92} and \eqref{locinf_eqn93}, or equivalently to
the equation
\begin{equation}
\label{locinf_eqn94}
\vect{x} T = \vect{u}_0 
\end{equation}
where $T$ is an $n\times (2n+k)$ matrix whose first $2n$ columns consist of
the top half of the matrix $QJ$, (upper $n$ elements of each column), and
whose last $k$ columns are the matrix $M$ in \eqref{locinf_eqn93}; while
$\vect{u}_0$ is a row vector whose first $2n$ elements are
$-(\vect{x}_0|\vect{z}_0)Q J$ and last $k$ elements are 0. Deciding if
\eqref{locinf_eqn94} has a solution $\vect{x}$ becomes straightforward once one has
transformed $T$ to Smith normal form, including determining the associated
invertible matrices, see \eqref{locinf_eqn32}.  As this needs to be done just once
for a given additive code and a given subset $\pt$, the complexity of the
algorithm for finding $\GC^\pt$ is $\OC(n^{\theta})$ for finding the Smith
form plus $\OC(n^2K^2)$ for testing the $K^2$ elements of $\GC$ once the Smith
form is available.  By using the group property of $\GC^{\pt}$ one can construct 
a faster algorithm, but that is beyond the scope of this Dissertation.

\section{Correctable $*$-algebra\label{locinf_apdx4}}

The counterpart in \cite{PhysRevA.76.042303} of our notion of information
perfectly present at the output of a quantum channel, see Sec.~\ref{locinf_sct2}, is
that of a \emph{correctable $*$-algebra}\index{correctable algebra} $\AC$ of operators acting on a
Hilbert space. The $*$ (sometimes denoted C$^*$) means that $\AC$, as well as
being an algebra of operators in the usual sense, contains $a\ad$ whenever
it contains $a$. Let the channel superoperator $\EC$ be represented by Kraus
operators,
\begin{equation}
\label{locinf_eqn95}
\EC(\rho) = \sum_j E_j\rho E_j\ad, 
\end{equation}
satisfying the usual closure condition $\sum_j E_j\ad E_j=I$, and let $P$ be a
projector onto some subspace $P\HC$ of the Hilbert space $\HC$.  Then a $*$-algebra $\AC$ is defined in \cite{PhysRevA.76.042303} to be \emph{correctable for $\EC$ on states in $P\HC$}\index{correctable algebra} provided $a=PaP$ for every $a$ in $\AC$, and there exists a superoperator $\RC$ (the recovery operation in an error correction scheme) whose domain is the range of $\EC$, whose range is $\LC(\HC)$, and such that
\begin{equation}
\label{locinf_eqn96}
P[(\RC \circ \EC)^\dagger (a)] P = a = P a P
\end{equation}
for all $a\in\AC$.  Here the dagger denotes the adjoint of the superoperator
in the sense that
\begin{equation}
\label{locinf_eqn97} 
\Tr\left[b \left((\RC\circ\EC)(c) \right)\right] = 
\Tr\left[ \left((\RC \circ \EC)^\dagger(b) \right) c\right]
\end{equation}
for any $b$ and $c$ in $\LC(\HC)$. 
In \cite{PhysRevA.76.042303}, see Theorem~9 and Corollary~10, it is shown that
any correctable algebra in this sense is a subalgebra of (what we call) a
\emph{maximal} correctable algebra\index{correctable algebra, maximal}
\begin{equation}
\label{locinf_eqn98} 
\AC_M = \left\{ a \in \LC({P\HC}) : 
[ a, P E^\dag_i E_j P]=0 \quad \forall \: i,j \right\}.
\end{equation}

We can apply this to our setting described in Secs.~\ref{locinf_sct6} and \ref{locinf_sct10} where $P$ is the projector on the coding space $\HC_C$ and $\EC_B$ is the superoperator for the partial trace down to the subset $\pt$ of carriers,
\begin{equation}
\label{locinf_eqn99} 
\EC_B(\rho) = \Tr_{\ptc}[\rho] = \sum_j E_j \rho E_j\ad \quad \text{ for }\rho \in \LC(\HC) 
\end{equation}
with Kraus operators
\begin{equation}
\label{locinf_eqn100}
E_j := I_B \otimes \bra{j}_{\ptc},
\end{equation}
where $\ket{j}_{\ptc}$ is any orthonormal basis of $\HC_{\ptc}$, so
\begin{equation}
\label{locinf_eqn101}
E_i\ad E_j = I_\pt \otimes \ket{i}\bra{j}_{\ptc}.
\end{equation}

We shall now show that collection of operators in $\GC^\pt$ (defined in Theorem~\ref{locinf_thm3}) spans a $*$-algebra which is correctable for $\EC_\pt$ on states in $P\HC = \HC_C$, and is the maximal algebra of this kind, i.e. $\text{span}(\GC^{\pt}) = \AC_M$.  First note that $\text{span}(\GC^{\pt})$ is indeed a $*$-algebra: every $\hat g \in \GC$ is a unitary operator and $\GC$ contains the adjoint of each of its elements; replacing $g$ with $g\ad$ in \eqref{locinf_eqn48} yields $\hat g\ad$. Of course $\Tr_{\ptc}[\hat g]=0$ if and only if $\Tr_{\ptc}[\hat g\ad]=0$ and in addition, $a = P a P$ for $a \in \text{span}(\GC^{\pt})$ because $\hat g = P\hat g P$, \eqref{locinf_eqn48}.

By definition $\Tr_{\ptc}[\hat g]\neq 0$ for $\hat g \in \GC^\pt$, and this means that the partial trace down to $\pt$ of at least one element in the corresponding coset $g\SC$, see \eqref{locinf_eqn48}, must be nonzero.  Let $h$ be such an element; since it is a Pauli product it must be of the form $h=h_\pt\otimes I_{\ptc}$.  As a consequence,
\begin{align}
\label{locinf_eqn102}
[\hat g,PE_i\ad E_j P] &=  [\hat h,PE_i\ad E_j P]= 
 P[h, E_i\ad E_j] P\notag\\
 &= P[\,h_\pt\otimes I_{\ptc}, I_\pt \otimes \ket{i}\bra{j}_{\ptc}\,] P = 0,
\end{align}
where the successive steps are justified as follows.  Since $\hat g$ depends
only on the coset $g\SC$ and $h$ belongs to this coset, $h\SC=g\SC$ and $\hat
h = Ph = hP =\hat g$.  This means we can move the projector $P$ outside the
commutator bracket, and once outside it is obvious that the latter vanishes
for every $i$ and $j$.   Thus any $\hat g$ in $\GC^\pt$ belongs
to the maximal $\AC_M$ defined in \eqref{locinf_eqn98}, as do all linear
combinations of the elements in $\GC^\pt$.

To show that $\AC_M$ is actually spanned by $\GC^\pt$ we note that any 
$a$ belonging to  $\AC_M$ can be written as  
\begin{equation}
\label{locinf_eqn103} 
a = b+c,
\end{equation}
where $b$ is a linear combination of elements of $\GC^\pt$ and $c$ of elements
of $\GC$ that do not belong to $\GC^\pt$, so
$\Tr_{\ptc}[c]=\Tr_{\ptc}[c\ad]=0$. Thus it is the case that
\begin{equation}
\label{locinf_eqn104} 
P(\RC\circ\EC_\pt)\ad (b)P = b,\quad 
P(\RC\circ\EC_\pt)\ad (c)P = c,
\end{equation}
where the first follows, see \eqref{locinf_eqn96}, from the previous argument showing
that the span of $\GC^\pt$ is a subalgebra of $\AC_M$, and the second from
linearity and the assumption that $a$ belongs to $\AC_M$.  Multiply the second
equation by $c\ad$ and take the trace:
\begin{align}
\label{locinf_eqn105}
 \Tr[c\ad c] &= \Tr \left[ c\ad P \left( (\RC\circ\EC_\pt)\ad (c) \right) P \right]\notag\\
  &= \Tr \left[ \left( \RC\circ\EC_\pt(c\ad) \right) c \right] = 0,
\end{align}
where we used the fact that $Pc\ad P=c\ad$, and $\EC_\pt(c\ad) = \Tr_{\ptc}[c\ad]=0$.  Thus $c=0$ and any element of $\AC_M$ is a linear combination of the operators in $\GC^\pt$.

In conclusion, we have shown for any additive graph code $C$ and any subset of
carrier qudits $\pt$, the $*$-algebra spanned by operators in $\GC^\pt$ is
exactly the maximal correctable algebra $\AC_M$ defined in
 \eqref{locinf_eqn98}. In App.~\ref{locinf_apdx3} we outline an algorithm that
enumerates the elements in $\GC^{\pt}$ for any $\HC_C$ and $\EC_\pt$, which in
light of the result above is an operator basis of $\AC_M$.
\end{subappendices}

\chapter{Bipartite equientagled bases\label{chp7}}


\section{Introduction\label{eqent_sct1}}
We present two different solutions to the problem posed by Karimipour and Memarzadeh in \cite{PhysRevA.73.012329} of constructing an orthonormal basis of two qudits with the following properties: (i) The basis continuously changes from a product basis (every basis state is a product state) to a maximally entangled basis (every basis state is maximally entangled), by varying some parameter $t$, and (ii) for a fixed $t$, all basis states are equally entangled. As mentioned in \cite{PhysRevA.73.012329}, such a family of bases may find applications in various quantum information protocols including quantum cryptography, optimal Bell tests, investigation of the enhancement of channel capacity due to entanglement and the study of multipartite entanglement. For a more detailed motivation the interested reader may consult \cite{PhysRevA.73.012329}.

The chapter is organized as follows : In Sec.~\ref{eqent_sct2} we summarize the main results of \cite{PhysRevA.73.012329} and then introduce the concept of Gauss sums and some useful related properties. Next we provide an explicit parameterization of a family of equientangled bases and we prove that it interpolates continuously between a product basis and a maximally entangled basis, for all dimensions. We illustrate the behaviour of our solution with explicit examples. In Sec.~\ref{eqent_sct7} we construct another such family using a completely different method based on graph states, describe a simple extension of it to multipartite systems, and then illustrate its behaviour with specific examples. Finally in Sec.~\ref{eqent_sct11} we compare the two solutions and make some concluding remarks.

\section{Construction based on Gauss sums\label{eqent_sct2}}
\subsection{Summary of previous work\label{eqent_sct3}}
Let us start by summarizing the main results of \cite{PhysRevA.73.012329}. Consider a bipartite Hilbert space $\HC\otimes\HC$, where both Hilbert spaces have the same dimension $D$. The authors first defined an arbitrary normalized bipartite state
\begin{equation}
\label{eqent_eqn1}
\ket{\psi_{0,0}}=\sum_{k=0}^{D-1}a_k\ket{k}\ket{k}.
\end{equation}
Next for $m,n=0,1,\ldots, D-1$, they considered the collection of $D^2$ ``shifted" states
\begin{align}
\ket{\psi_{m,n}} &=X^{m}\otimes X^{m+n}\ket{\psi_{0,0}} \notag\\
&=\sum_{k=0}^{D-1}a_k\ket{k\oplus m}\ket{k\oplus m\oplus n}, \label{eqent_eqn2}
\end{align}
where 
\begin{equation}\label{eqent_eqn3}
X:=\sum_{k=0}^{D-1}\dyad{k\oplus 1}{k}
\end{equation}
is the generalized Pauli (or shift) operator and $\oplus$ denotes addition modulo $D$. 
They noted that all states have the same value of entropy of entanglement \cite{NielsenChuang:QuantumComputation} given by the von-Neumann entropy
\begin{equation}\label{eqent_eqn4}
E(\ket{\psi_{m,n}})=E(\ket{\psi_{0,0}})=-\sum_{k=0}^{D-1}|a_k|^2\log_D|a_k|^2,
\end{equation}
where the logarithm is taken in base $D$ for normalization reasons so that all maximally entangled states have entanglement equal to one regardless of $D$.

Demanding the states in \eqref{eqent_eqn2} be orthonormal yields
\begin{equation}\label{eqent_eqn5}
\sum_{k=0}^{D-1} (a_k)^{*}a_{k\oplus m}=\delta_{m,0},\text{ }\forall m=0,\ldots,D-1,
\end{equation}
and the authors proved (see their Eqn. (36)) that \eqref{eqent_eqn5} is satisfied if and only if the coefficients $a_k$ have the form
\begin{equation}\label{eqent_eqn6}
a_k=\frac{1}{D}\sum_{j=0}^{D-1}\ee^{\ii \theta_j}\omega^{kj},
\end{equation}
where $\theta_j$ are arbitrary real parameters and $\omega=\expo{2\pi\ii/D}$ is
 the $D$-th root of unity. 
 
Therefore the authors found a family of $D^2$ orthonormal states, all having the same Schmidt coefficients and hence the same value of entanglement.  To ensure it interpolates from a product basis to a maximally entangled basis, it is sufficient to find a set of parameters $\{\theta^0_j\}_{j=0}^{D-1}$ for which the magnitude of $a_k$ is  $|a_k|=1/\sqrt{D}$ for all $k$. Then the problem is solved by defining
\begin{equation}\label{eqent_eqn7}
a_k(t):=\frac{1}{D}\sum_{j=0}^{D-1}\ee^{\ii t\theta^0_j}\omega^{kj},
\end{equation}
where $t\in[0,1]$ is a real parameter. When $t=0$ we have $a_k=\delta_{k,0}$ so the basis states are product states and when $t=1$, the basis is maximally entangled by assumption. We also observe there is a continuous variation in between these two extremes as a function of $t$. 

Karimipour and Memarzadeh considered the existence of such a set $\{\theta^0_j\}_{j=0}^{D-1}$ in arbitrary dimensions (see the last paragraph of Sec. V in \cite{PhysRevA.73.012329}). They found particular solutions for $D\leqslant 5$, but did not find a general solution for arbitrary $D$.

\subsection{Quadratic Gauss Sums\label{eqent_sct4}}
We now define the basic mathematical tools we will make use of in the rest of this section. The most important concept is that of a \emph{quadratic Gauss sum}\index{quadratic Gauss sum}, defined below.
\begin{QGS}
Let $p,m$ be positive integers. The quadratic Gauss sum is defined as
\begin{equation}\label{eqent_eqn8}
\sum_{j=0}^{p-1}\expo{2\pi\ii j^2 m/p}.
\end{equation}
\end{QGS}
The quadratic Gauss sums satisfy a reciprocity relation known as 
\begin{LS}
Let $p,m$ be positive integers. Then
\begin{equation}\label{eqent_eqn9}
\frac{1}{\sqrt{p}}\sum_{j=0}^{p-1}
\expo{2\pi\ii j^2m/p}=
\frac{\ee^{\pi\ii/4}}{\sqrt{2m}}\sum_{j=0}^{2m-1}\expo{-\pi\ii j^2 p/2m}
\end{equation}
\end{LS}
The quadratic Gauss sums can be generalized as follows.
\begin{GQGS}
Let $p,m,n$ be positive integers. The generalized quadratic Gauss sum is defined as
\begin{equation}\label{eqent_eqn10}
\sum_{j=0}^{p-1}\expo{2\pi\ii( j^2 m + j n )/p}.
\end{equation}
\end{GQGS}
Finally the following reciprocity formula for generalized Gauss sums holds.
\begin{RFGQGS}
Let $p,m,n$ be positive integers such that $mp\neq0$ and $mp+n$ is even. Then
\begin{align}\label{eqent_eqn11}
\frac{1}{\sqrt{p}}&\sum_{j=0}^{p-1}\expo{\pi\ii(j^2m+jn)/p}=\expo{\pi\ii(mp-n^2)/4mp}\frac{1}{\sqrt{m}}\sum_{j=0}^{m-1}\expo{-\pi\ii(j^2p+jn)/m}.
\end{align}
\end{RFGQGS}
The definitions of the Gauss sums \eqref{eqent_eqn8} and \eqref{eqent_eqn10} as well as the Landsberg-Schaar's identity \eqref{eqent_eqn9} can be found in standard number theory books \cite{HardyWright:IntroTheoryNumbers,EverestWard:IntroNumberTheory,Nathanson:ElementaryMethodsNumberTheory}. The reciprocity formula for the generalized quadratic Gauss sum is not as well-known, and can be found in \cite{BerndtEvans81}.

\subsection{Explicit Solution\label{eqent_sct5}}
We now show that a family of equientangled bases that interpolates continuously between the product basis and the maximally entangled basis exists for all dimensions $D$, as summarized by the following Theorem.
\begin{theorem}\label{eqent_thm1}
The collection of $D^2$ normalized states
\begin{align}\label{eqent_eqn12}
\ket{\psi_{m,n}(t)}&=\sum_{k=0}^{D-1}a_k(t)\ket{k\oplus m}{\ket{k\oplus m\oplus n}},\\
m,n &= 0,\ldots, D-1,\notag
\end{align}
 indexed by a real parameter 
\mbox{$t\in[0,1]$} with
\begin{equation}\label{eqent_eqn13}
a_k(t)=\frac{1}{D}\sum_{j=0}^{D-1}\ee^{\ii t\theta^0_j}\omega^{kj},
\qquad \omega=\ee^{2\pi\ii/D},
\end{equation}
with the particular choice of
\begin{equation}\label{eqent_eqn14}
\theta^0_j = \left\{ 
\begin{array}{l l}
  \pi  j^2/D & \quad \text{if $D$ is even}\\
  2\pi j^2/D & \quad \text{if $D$ is odd},
\end{array} \right.
\end{equation}
defines a family of equientangled bases that continuously interpolates between a product basis at $t=0$ and  a maximally entangled basis at $t=1$.
\end{theorem}
That \eqref{eqent_eqn12} defines a family of equientangled bases that consists of a product basis at $t=0$ follows directly from the remarks of Sec.~\ref{eqent_sct1}, $a_k(0)=\delta_{k,0}$. 
Next note that a continuous variation of $t$ in the interval $[0,1]$ corresponds to a continuous variation of the Schmidt coefficients of the states in the basis. The latter implies that no matter which measure one uses to quantify the entanglement, the measure will vary continuously with $t$, since any pure state entanglement measure depends only on the Schmidt coefficients of the state \cite{PhysRevLett.83.1046}.

The only thing left to show is that the basis states in Theorem~\ref{eqent_thm1} are maximally entangled when $t=1$, or, equivalently, that $|a_k(1)|=1/\sqrt{D}$ for all $k$. We prove this by explicitly evaluating the value of $a_k(1)$ in the following Lemma.
\begin{lemma}\label{eqent_lma1}
Let $a_k(t)$ and $\{\theta_j^0\}_{j=0}^{D-1}$ be as defined by Theorem~\ref{eqent_thm1}. Then for all $k$
\begin{equation}\label{eqent_eqn15}
a_k(1) = 
\frac{\expo{\pi\ii/4}}{\sqrt{D}}\times 
\left\{ 
\begin{array}{l l}
   \omega^{-k^2/2}, & \quad \text{if $D$ is even}\\
  \omega^{-k^2/4} \left(\frac{1-\ii^{2k+D}}{\sqrt{2}} \right), & \quad \text{if $D$ is odd}
\end{array} \right..
\end{equation}
\end{lemma}
Lemma~\ref{eqent_lma1} implies at once that $|a_k(1)|=1/\sqrt{D}$, and therefore proves Theorem~\ref{eqent_thm1}. 
\begin{proof}
(of Lemma~\ref{eqent_lma1}) Note first that the expression for $a_k(1)$ in \eqref{eqent_eqn13} with $\theta_j^0$ defined in \eqref{eqent_eqn14} resembles the generalized quadratic Gauss sum \eqref{eqent_eqn10}. We will use the  reciprocity formula \eqref{eqent_eqn11} to prove Lemma~\ref{eqent_lma1}. There are two cases to be considered: Even $D$ and odd $D$.

\textbf{Even $D$}. Note that one can rewrite $a_k(1)$ in \eqref{eqent_eqn13} with $\theta_j^0$ defined in \eqref{eqent_eqn14} as
\begin{equation}\label{eqent_eqn16}
a_k(1)=\frac{1}{D}\sum_{j=0}^{D-1}\expo{\pi\ii j^2/D} \; \expo{2\pi\ii j k/D}=\frac{1}{D}\sum_{j=0}^{D-1}\expo{\pi\ii(j^2+2 k j)/D}.
\end{equation}
Applying the reciprocity formula \eqref{eqent_eqn11} to last term in \eqref{eqent_eqn16} with $m=1,n=2k,p=D$ (noting that $mp+n=D+2k$ is even) yields 
\begin{align}\label{eqent_eqn17}
a_k(1)&=\frac{1}{\sqrt{D}}\ee^{\pi\ii (D-4k^2)/4D} = \frac{\ee^{\pi\ii/4}}{\sqrt{D}}(-1)^D \ee^{-\pi\ii k^2/D} \notag\\
&=\frac{\ee^{\pi\ii/4}}{\sqrt{D}}\omega^{-k^2/2},\text{ since $(-1)^D=1$ for even $D$}.
\end{align}

\textbf{Odd $D$}. The proof is essentially the same as in the even $D$ case, but we explicitly write it below for the sake of completeness. Using a similar argument we rewrite $a_k(1)$ in \eqref{eqent_eqn13} with $\theta_j^0$ defined in \eqref{eqent_eqn14} as
\begin{align}\label{eqent_eqn18}
a_k(1)=\frac{1}{D}\sum_{j=0}^{D-1}\expo{2\pi\ii j^2/D} \; \expo{2\pi\ii j k/D}=\frac{1}{D}\sum_{j=0}^{D-1}\expo{\pi\ii(2j^2+2 k j)/D}.
\end{align}
Applying again the reciprocity formula \eqref{eqent_eqn11} to last term in \eqref{eqent_eqn18} with $m=2,n=2k,p=D$ (noting that $mp+n=2D+2k$ is even) yields 
\begin{align}\label{eqent_eqn19}
a_k(1)&=\frac{1}{\sqrt{D}}\frac{\ee^{\pi\ii(2D-4k^2)/8D}}{\sqrt{2}}
(1+\ee^{-\pi\ii(D+2k)/2})\notag\\
&=\frac{\ee^{\pi\ii/4}}{\sqrt{D}}\omega^{-k^2/4} \left( \frac{1-\ii^{2k+D}}{\sqrt{2}} \right),
\end{align}
where we used $(-\ii)^{2k+D}=-(\ii^{2k+D})$ since $2k+D$ is odd.
This concludes the proof of Lemma~\ref{eqent_lma1} and implicitly of Theorem~\ref{eqent_thm1}.
\end{proof}

\subsection{Examples\label{eqent_sct6}}
In this section we present some examples that illustrate the behaviour of the solution we provided in Theorem~\ref{eqent_thm1}, for various dimensions. First we consider $D=5$ and we plot the absolute values of the $a_k(t)$ coefficients as a function of $t$ in Fig.~\ref{eqent_fgr1}. It is easy to see that indeed the basis interpolates between a product basis and a maximally entangled one in a continuous manner. We observe that all coefficients are non-zero for $t>0$ and we believe that this is  probably also the case for all odd $D$'s.
\begin{figure}
\begin{center}
\includegraphics[scale=0.45]{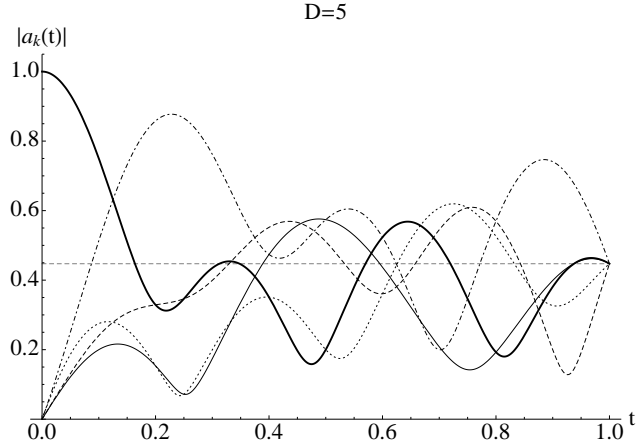}
\caption{The variation of $|a_k(t)|$ with $t$ for $D=5$. Note how at $t=0$ all coefficients but one are zero, and how at $t=1$ all coefficients are equal in magnitude to $1/\sqrt{5}$, with a continuous variation in between. The dashed line represents the $1/\sqrt{5}$ constant function. }
\label{eqent_fgr1}
\end{center}
\end{figure}

In Fig.~\ref{eqent_fgr2} we perform the same analysis as above, but now for $D=8$. We observed that some coefficients vanish for some values of $t$, which seems to be true in general for even $D$.
\begin{figure}
\begin{center}
\includegraphics[scale=0.45]{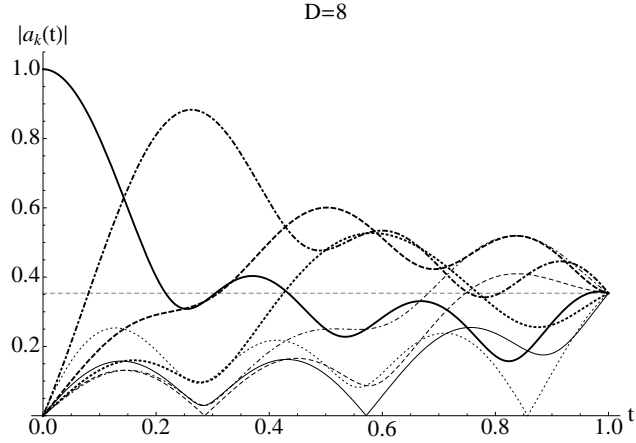}
\caption{The variation of $|a_k(t)|$ with $t$ for $D=8$. Again note how at $t=0$ all coefficients but one are zero, and how at $t=1$ all coefficients are  equal in magnitude to $1/\sqrt{8}$, with a continuous variation in between. The dashed line represents the $1/\sqrt{8}$ constant function. }
\label{eqent_fgr2}
\end{center}
\end{figure}

In Fig.~\ref{eqent_fgr3} we plot the entropy of entanglement of the states in the basis as a function of $t$ for dimensions $D=2,3,5,8$ and $100$.
We see how the entanglement varies continuously but not monotonically between 0 and 1.
\begin{figure}
\begin{center}
\includegraphics[scale=0.45]{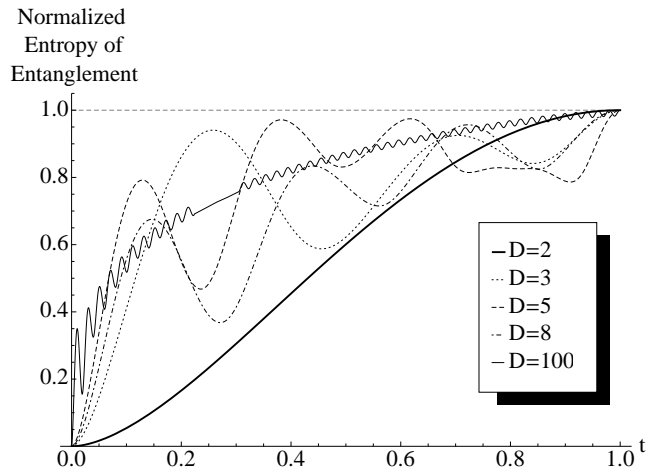}
\caption{The entropy of entanglement as a function of $t$ for various dimensions. Note that the variation is not monotonic (except for $D=2$), although for large $D$ the oscillations tend to be smoothed out.}
\label{eqent_fgr3}
\end{center}
\end{figure}

Finally in Fig.~\ref{eqent_fgr4} we display a parametric plot  of the variation of the second Schmidt coefficient $a_1(t)$ in the complex plane as $t$ is varied from $0$ to $1$ for $D=51$, so that the reader can get an idea of how the coefficients defined in \eqref{eqent_eqn13} look in general. The other coefficients $a_k$ look similar. 

\begin{figure}
\begin{center}
\includegraphics[scale=0.47]{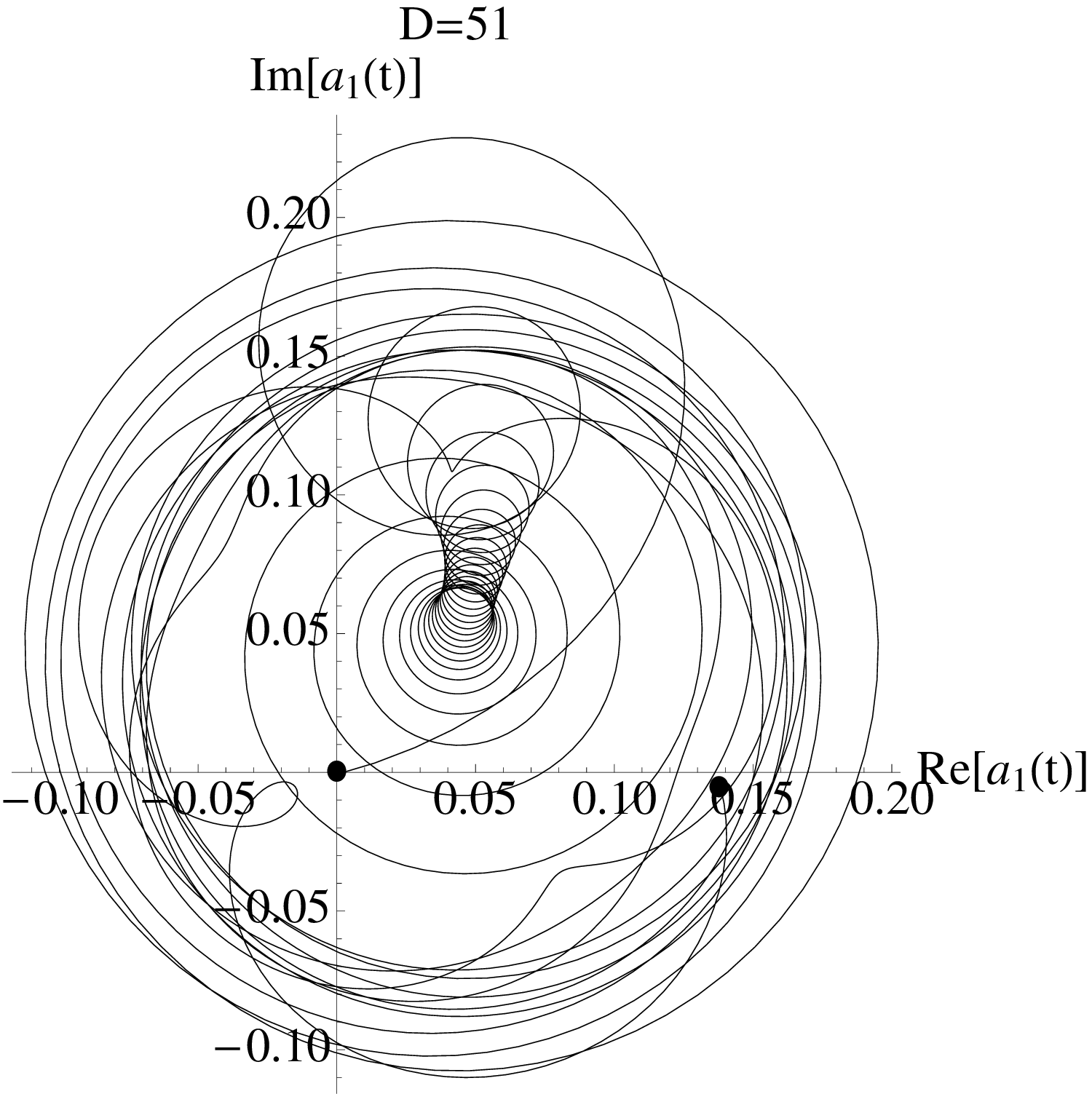}
\caption{Parametric plot of $a_1(t)$ in the complex plane as $t$ is varied from $0$ to $1$. Note that $a_1(0)=0$ and $a_1(1)=\frac{1-\ii}{\sqrt{2\cdot 51}}\ee^{\pi\ii(\frac{1}{4}-\frac{1}{2\cdot 51})}$, the value provided by Lemma~\ref{eqent_lma1}. The starting point $t=0$ and the ending point $t=1$ are marked by solid disks.}
\label{eqent_fgr4}
\end{center}
\end{figure}

\section{Construction based on Graph States\label{eqent_sct7}}
\subsection{Explicit solution\label{eqent_sct8}}
We provide below another solution to the problem that uses qudit graph states. 
Again having in mind a bipartite Hilbert space $\HC\otimes\HC$, both local spaces having dimension $D$, we define a one-qudit state
\begin{equation}\label{eqent_eqn20}
\ket{+}:=\frac{1}{\sqrt{D}}\sum_{k=0}^{D-1}\ket{k}.
\end{equation}
It is easy to see that the collection of $D$ states
\begin{equation}\label{eqent_eqn21}
\ket{\overline m}:= Z^m\ket{+},\quad {m=0,\ldots, D-1}
\end{equation}
defines an orthonormal basis of $\HC$ (also known as the Fourier basis), $\ip{\overline m}{\overline n} = \delta_{mn}$, where 
\begin{equation}\label{eqent_eqn22}
Z:=\sum_{k=0}^{D-1}\omega^k\dyad{k}{k},
\end{equation}
with  
$\omega=\expo{2\pi\ii/D}
$ being the $D$-th root of unity. It then follows at once that the collection of $D^2$ states
\begin{equation}\label{eqent_eqn23}
\ket{\overline m}\ket{\overline n}=(Z^m\otimes Z^n)\ket{+}\ket{+},
\quad{m,n=0,\ldots, D-1}
\end{equation}
defines an orthonormal product basis of the bipartite Hilbert space $\HC\otimes\HC$. 

Next we define the generalized controlled-Phase gate as
\begin{equation}\label{eqent_eqn24}
\CP:=\sum_{k=0}^{D-1} \dyad{k}{k} \otimes Z^k=\sum_{j,k=0}^{D-1}\omega^{jk}\dyad{j}{j}\otimes\dyad{k}{k}
\end{equation}
and note that $\CP$ is a unitary operator that commutes with $Z^m\otimes Z^n$, for all $m,n=0,\ldots,D-1$. The state
\begin{equation}\label{eqent_eqn25}
\ket{G}:=\CP\ket{+}\ket{+}=\frac{1}{D}\sum_{j,k=0}^{D-1}\omega^{jk}\ket{j}\ket{k}
\end{equation}
is an example of a two-qudit \emph{graph state}\index{graph state} and it is not hard to see that $\ket{G}$ is maximally entangled. Then the collection of $D^2$ states
\begin{align}\label{eqent_eqn26}
(&Z^m\otimes Z^n)\ket{G}=(Z^m\otimes Z^n) \CP\ket{+}\ket{+}\notag\\
&=\CP (Z^m\otimes Z^n)\ket{+}\ket{+}, \quad{m,n=0,\ldots, D-1}
\end{align}
defines an orthonormal basis of the bipartite Hilbert space $\HC\otimes\HC$, which we call a \emph{graph basis}\index{graph basis}. Since $Z^m\otimes Z^n$ are local unitaries and $\ket{G}$ is a maximally entangled state, then all the other graph basis states must also be maximally entangled. For more details about graph states of arbitrary dimension see \cite{quantph.0602096, PhysRevA.78.042303, PhysRevA.81.032326}.

We now have all the tools to construct a continuous interpolating family of equientangled bases, as summarized by the Theorem below.
\begin{theorem}\label{eqent_thm2}
The collection of $D^2$ normalized states
\begin{align}\label{eqent_eqn27}
\ket{G_{m,n}(t)}&=(Z^m\otimes Z^n) \CP(t)\ket{+}\ket{+},\\
m,n&=0,\ldots, D-1\notag,
\end{align}
indexed by a real parameter $t\in[0,1]$ where 
\begin{equation}\label{eqent_eqn28}
\CP(t)=\sum_{j,k}\omega^{jkt}\dyad{j}{j}\otimes\dyad{k}{k}
\end{equation}
defines a family of equientangled bases that continuously interpolates between a product basis at $t=0$ and a maximally entangled basis at $t=1$. 
\end{theorem}
\begin{proof}
We make the crucial observation that $\CP(t)$ commutes with $Z^m\otimes Z^n$ for all $m,n=0,\ldots, D-1$ and all $t \in [0,1]$ which implies that $\{\ket{G_{m,n}(t)}\}_{m,n=0}^{D-1}$ defines an orthonormal basis since it differs from the orthonormal basis in \eqref{eqent_eqn23} only by the unitary operator $\CP(t)$. All states in the basis are equally entangled, and moreover, share the same set of Schmidt coefficients since any two basis states are equivalent up to local unitaries of the form $Z^m\otimes Z^n$.

Finally note that \mbox{$\CP(t=0)=I\otimes I$} and \mbox{$\CP(t=1)=\CP$} (defined in \eqref{eqent_eqn24}), and therefore at $t=0$ the basis is product, see \eqref{eqent_eqn23}, and at $t=1$ the basis is maximally entangled, see \eqref{eqent_eqn26}. The operator $\CP(t)$ can be viewed as a controlled-Phase gate whose ``entangling strength" can be tuned continuously. The Schmidt coefficients of the states in the basis vary continuously with $t$ and hence the entanglement also varies continuously with $t$, regardless of which entanglement measure one uses (see the remarks following Theorem~\ref{eqent_thm1}).
\end{proof}
Our construction above can be expressed in the framework described in the last two paragraphs of Sec. II of \cite{PhysRevA.73.012329}, by setting $U_m=Z^m$ and $V_n=Z^n$.

Next we prove that the Schmidt coefficients of the basis states in Theorem \ref{eqent_thm2} are all non-zero for any $t>0$, so all bases consist of full Schmidt rank states whenever $t>0$.

\begin{lemma}\label{eqent_lma2}
The equientangled family of bases $\{\ket{G_{m,n}(t)}\}_{m,n=0}^{D-1}$ defined in Theorem~\ref{eqent_thm2} consists of full Schmidt rank states, for any $0<t\leqslant1$.
\end{lemma} 
\begin{proof}
We will show that the product of the Schmidt coefficients is always non-zero, which implies that no Schmidt coefficient can be zero, whenever $0<t\leqslant 1$.

Let 
\begin{equation}\label{eqent_eqn29}
\ket{\psi}=\sum_{j,k=0}^{D-1}\Omega_{jk}\ket{j}\ket{k}
\end{equation}
be an arbitrary normalized pure state in a bipartite Hilbert space $\HC\otimes\HC$ and
let $\{\lambda_k\}$ denote the set of Schmidt coefficients of $\ket{\psi}$ satisfying $\sum_k{\lambda_k}=1$; note that they are equal to the squares of the singular values of the coefficient matrix  $\Omega$ in \eqref{eqent_eqn29}. The product of the squares of the singular values is just the product of the eigenvalues of $\Omega\Omega^\dagger$, the latter product being equal to $\det(\Omega\Omega^\dagger)=|\det(\Omega)|^2$, so we conclude that
\begin{equation}\label{eqent_eqn30}
\prod_{k=0}^{D-1}\lambda_k=|\det(\Omega)|^2.
\end{equation}

The states $\ket{G_{m,n}(t)}$ of Theorem~\ref{eqent_thm2} share the same set of Schmidt coefficients (they are all related by local unitaries) so it suffices to show that the product of the Schmidt coefficients is non-zero only for the state $\ket{G_{0,0}(t)}$. Recall that 
\begin{equation}\label{eqent_eqn31}
\ket{G_{0,0}(t)} = \CP(t)\ket{+}\ket{+} = \sum_{j,k}\frac{\omega^{jkt}}{D}\ket{j}\ket{k}.
\end{equation}
Expressing the coefficients $\omega^{jkt}/D$ as a matrix $\Omega(t)$, one can easily see that $D\cdot\Omega(t)$ is a $D\times D$ Vandermonde matrix whose determinant is 
\begin{equation}\label{eqent_eqn32}
\det\left[D\cdot\Omega(t)\right] = \prod_{j > k}(\omega^{jt}-\omega^{kt}) = \prod_{j>k}\omega^{kt}\left[\omega^{(j-k)t}-1\right]
\end{equation}
(see p.~29 of \cite{HornJohnson:MatrixAnalysis} for more details on Vandermonde matrices).
For a given $t$, the product above is zero if and only if at least one term is zero, i.e. there must exist integers $j$, $k$, with $0\leqslant k<j\leqslant D-1$, such that
\begin{equation}\label{eqent_eqn33}
(j-k)t=n D \Longleftrightarrow t=n \frac{D}{j-k},
\end{equation}
for some positive integer $n\geqslant 0$.
Note that $0<j-k\leqslant D-1$, so $D/(j-k)>1$ and the above equation can never be satisfied for $0<t\leqslant 1$. We have therefore proved that $\det[\Omega(t)]\neq0$ for $0<t\leqslant 1$, which, in the light of \eqref{eqent_eqn30}, is equivalent to saying that the product of the Schmidt coefficients is non-zero for $0<t\leqslant 1$, and this concludes the proof of the Lemma.
\end{proof}

For this family of equientangled bases, we do not have an analytic expression for the Schmidt coefficients nor the entropy of entanglement for general $D$ though they can be easily found by numerically diagonalizing the coefficient matrix $\Omega(t)\Omega(t)^\dagger$. Having said that, we derived a simple analytic expression for the \emph{product} of all Schmidt coefficients, see \eqref{eqent_eqn30} and \eqref{eqent_eqn32}, which is simply related to an entanglement monotone called $G$-concurrence (first introduced in \cite{PhysRevA.71.012318}) which is defined for a pure bipartite state \eqref{eqent_eqn29} in terms of its Schmidt coefficients $\{\lambda_k\}$ as
\begin{equation}\label{eqent_eqn34}
C_G(\ket{\psi}) = D\left(\prod_{k=0}^{D-1}\lambda_k\right)^{1/D} = D|\det(\Omega)|^{2/D}
\end{equation}
where $\sum_k\lambda_k=1$. The $G$-concurrence is zero whenever at least one Schmidt coefficient is zero and is equal to one if and only if the state is maximally entangled. Unlike the entropy of entanglement, we are able to show two analytical results that are true for all $D$ expressed in Lemmas \ref{eqent_lma3} and \ref{eqent_lma4}.

\begin{lemma}\label{eqent_lma3}
The $G$-concurrence of the equientangled basis states, $\{\ket{G_{m,n}(t)}\}_{m,n=0}^{D-1}$ defined in Theorem~\ref{eqent_thm2} is
\begin{equation}\label{eqent_eqn35}
C_G(t) = \frac{2^{D-1}}{D} \prod_{r=1}^{D-1}
\left[\sin^2(\pi r t/D)\right]^{(D-r)/D} \text{ for all $m,n,D$.}
\end{equation}
\end{lemma}

\begin{proof}
Every basis state has the same $G$-concurrence since they all share the same set of Schmidt coefficients (recall that they differ only by local unitaries). Invoking the definition of $G$-concurrence, we have
\begin{align}
C&_G(t)=D|\det(\Omega(t))|^{2/D}
=D\left|\frac{1}{D^D}\det(D\cdot\Omega(t))\right|^{2/D}\label{eqent_eqn36}\\
&=\frac{1}{D}\prod_{j>k}\left|\omega^{(j-k)t}-1\right|^{2/D}=\frac{1}{D}\prod_{r=1}^{D-1}\left|\omega^{rt}-1\right|^{\frac{2(D-r)}{D}}\label{eqent_eqn37}\\
&=\frac{1}{D}\prod_{r=1}^{D-1}\left[2 - 2\frac{\omega^{rt}+\omega^{-rt}}{2} \right]^{\frac{D-r}{D}}\label{eqent_eqn38}\\
&=\frac{1}{D}\prod_{r=1}^{D-1}2^{(D-r)/D}\left[1-\cos(2\pi r t/D)\right]^{\frac{D-r}{D}}\label{eqent_eqn39}\\
&=\frac{1}{D}\prod_{r=1}^{D-1}2^{2(D-r)/D}\left[\sin^2(\pi r t/D)\right]^{\frac{D-r}{D}}\label{eqent_eqn40}\\
&=\frac{2^{D-1}}{D}\prod_{r=1}^{D-1}\left[\sin^2(\pi r t/D)\right]^{(D-r)/D} \label{eqent_eqn41},
\end{align}
where in \eqref{eqent_eqn36} we used the fact that $\det(cM)=c^D\det(M)$ for a $D\times D$ arbitrary matrix $M$ and an arbitrary constant $c$. The first equality in \eqref{eqent_eqn37} follows at once from \eqref{eqent_eqn32}, whereas the second equality in \eqref{eqent_eqn37} follows from a simple counting argument in which one replaces $j-k$ by $r$, making sure that the different pairs $(j,k)$, $j>k$, that give rise to the same 
$r$ are counted; for a given $r$ there are $D-r$ such pairs.
\end{proof}

It turns out that the $G$-concurrence has the following nice property:
\begin{lemma}\label{eqent_lma4}
The $G$-concurrence of the basis states $\{\ket{G_{m,n}(t)}\}_{m,n=0}^{D-1}$ in \eqref{eqent_eqn35} is strictly increasing in the open interval $t \in (0,1)$, for all dimensions $D$.
\end{lemma}
\begin{proof}
We prove this by showing that the first derivative of the $C_G(t)$ with respect to $t$ is strictly positive. First note that since $C_G(t)>0$ is a positive function in the interval $t \in (0,1)$. This means showing $\frac{\mathrm{d}}{\mathrm{d} t} C_G(t) > 0$ is equivalent to showing that $\frac{\mathrm{d}}{\mathrm{d} t} \log C_G(t) > 0$. The derivative of the logarithm of \eqref{eqent_eqn41} is
\begin{equation}\label{eqent_eqn42}
\frac{\mathrm{d}}{\mathrm{d} t} \log C_G(t) = \frac{2\pi}{D^2} \sum_{r=1}^{D-1}r(D-r)
\cot(\pi r t/D),
\end{equation}
where $\cot(\cdot)$ denotes the cotangent function. We differentiate again to get
\begin{equation}\label{eqent_eqn43}
\frac{\mathrm{d^2}}{\mathrm{d} t^2}\log C_G(t)=-\frac{2\pi^2}{D^3}\sum_{r=1}^{D-1}r^2(D-r)\frac{1}{\sin^2(\pi r t/D)}.
\end{equation}
Note that the right hand side of \eqref{eqent_eqn43} is strictly negative whenever $t>0$, which implies that the first derivative of the logarithm \eqref{eqent_eqn42} is a strictly decreasing function of $t$ and hence achieving its minimum value at $t=1$, which is given by
\begin{equation}\label{eqent_eqn44}
\frac{\mathrm{d}}{\mathrm{d} t} \log C_G(t)\mid_{t=1} = \frac{2\pi}{D^2}\sum_{r=1}^{D-1}r(D-r) \cot(\pi r/D)=0,
\end{equation}
where the last equality follows from symmetry considerations (terms cancel one by one). We have shown $\frac{\mathrm{d}}{\mathrm{d} t} \log C_G(t) > 0 \Leftrightarrow \frac{\mathrm{d}}{\mathrm{d} t} C_G(t) > 0$ and therefore we conclude that the $G$-concurrence is strictly increasing for $t \in (0,1)$.
\end{proof}

\begin{figure}[h!]
\begin{center}
\includegraphics[scale=0.56]{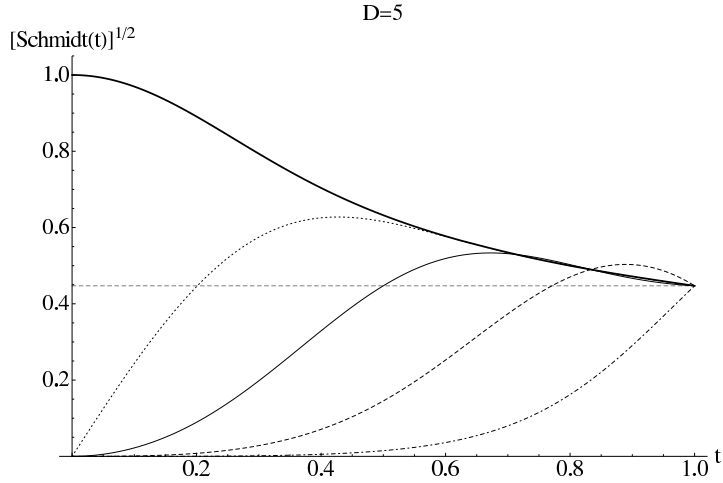}
\caption{The square roots of the Schmidt coefficients as functions of $t$ for $D=5$. Note how at $t=0$ all coefficients but one are zero, and how at $t=1$ all coefficients are equal in magnitude to $1/\sqrt{5}$, with a continuous variation in between. The dashed line represents the $1/\sqrt{5}$ constant function. }
\label{eqent_fgr5}
\end{center}
\end{figure}

\subsection{Extension to multipartite systems\label{eqent_sct9}}
The construction presented in Theorem~\ref{eqent_thm2} can be easily generalized to  multipartite systems of arbitrary dimension. The concept of maximally entangled states is not defined for three parties or more, so in this case, the family continuously interpolates between a product basis and a qudit graph basis. It is still true that for a fixed $t$, all basis states constructed this way have the same entanglement (as quantified by any entanglement measure) since they only differ by local unitaries. 

As a specific example, consider the tripartite GHZ state \mbox{$(\ket{000}+\ket{111})/\sqrt{2}$}. This state is a stabilizer state \cite{NielsenChuang:QuantumComputation} and therefore is local-unitary equivalent \cite{quantph.0111080}  to a graph state $\ket{G}=(\CP_{12} \CP_{23} \CP_{13}) \ket{+}_1 \ket{+}_2 \ket{+}_3$, where the subscripts on $\CP$ indicates which pair of qubits the $\CP$ gate is applied to. By varying the ``strength" of the controlled-Phase gate, one can now construct a family of equally entangled basis for the Hilbert space of 3 qubits that continuously interpolates between a product basis and the GHZ-like graph basis. This GHZ construction can be easily generalized to higher dimensions and also to $n$ parties by using the complete graph given by
\begin{equation}
\label{eqent_eqn45}
\left| G_{\mathrm{GHZ}}(t) \right \rangle := \prod_{i=1}^{n-1} \prod_{j>i}^n \CP_{ij}(t) \; \ket{+}^{\otimes n}
\end{equation}
where $\ket{+}$ is defined in \eqref{eqent_eqn20} and $\CP(t)$ is defined in \eqref{eqent_eqn28}. Finally note that this construction works for any graph state of any dimension, and not just for GHZ-like graph states. Such bases with tunable entanglement may be of use in the study of multipartite entanglement.

\subsection{Examples\label{eqent_sct10}}
In this subsection we perform a similar analysis as the one in Sec.~\ref{eqent_sct6}, so that one can easily compare the behaviour of both solutions.

We consider again a $D=5$ example, for which we plot in Fig.~\ref{eqent_fgr5} the square root of the Schmidt coefficients as functions of $t$. It is easy to see that indeed the basis interpolates between a product basis and a maximally entangled one in a continuous manner. As proven in Lemma~\ref{eqent_lma2}, all Schmidt coefficients are non-zero for $t>0$.

In Fig.~\ref{eqent_fgr6} we plot the same quantities for $D=8$. 
\begin{figure}[h!]
\begin{center}
\includegraphics[scale=0.56]{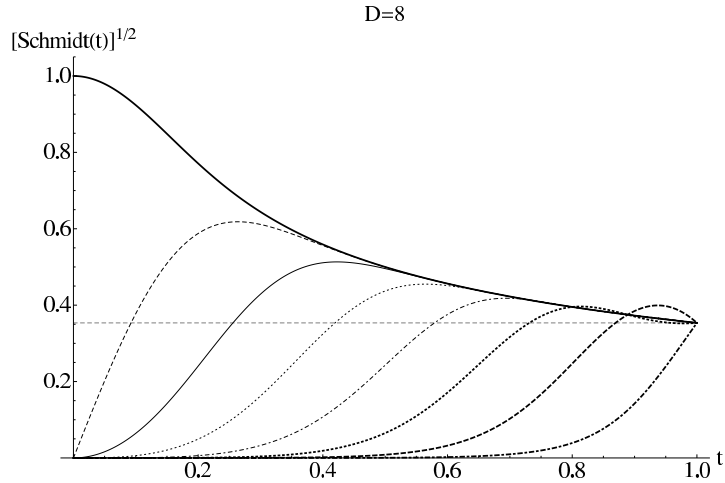}
\caption{The square roots of the Schmidt coefficients as functions of $t$ for $D=8$. Again note how at $t=0$ all coefficients but one are zero, and how at $t=1$ all coefficients are  equal in magnitude to $1/\sqrt{8}$, with a continuous variation in between. The dashed line represents the $1/\sqrt{8}$ constant function. }
\label{eqent_fgr6}
\end{center}
\end{figure}

\begin{figure}[h!]
\begin{center}
\includegraphics[scale=0.46]{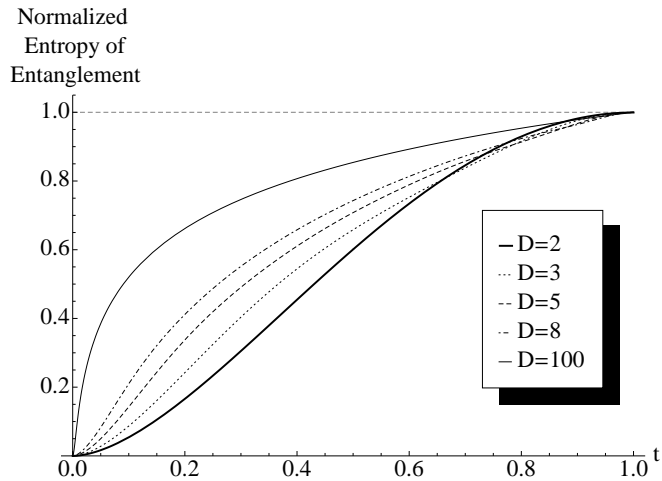}
\caption{The entropy of entanglement as function of $t$ for various dimensions. Note that the variation seems to be monotonically increasing for all $D$, a statement we did not prove.}
\label{eqent_fgr7}
\end{center}
\end{figure}

\begin{figure}[h!]
\begin{center}
\includegraphics[scale=0.44]{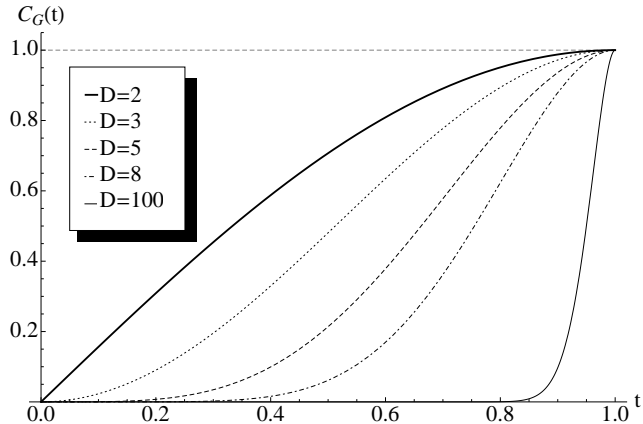}
\caption{The $G$-concurrence as function of $t$ for various dimensions. The variation is strictly increasing in $t$ for all $D$ as shown in Lemma \ref{eqent_lma4}.}
\label{eqent_fgr8}
\end{center}
\end{figure}

Observe how in both examples above the variation of the Schmidt coefficients is not oscillatory, as in the examples of Sec.~\ref{eqent_sct6}.

In Fig.~\ref{eqent_fgr7} we plot the entropy of entanglement of the basis states as a function of $t$ for dimensions $D=2,3,5,8$ and $100$. We observe that the entropy of entanglement varies continuously and monotonically between 0 and 1. It is not known if the entropy of entanglement is always strictly increasing for all $D$ although we verified
this by visual inspection for all $D\leqslant 10$. In Fig.~\ref{eqent_fgr8}, we plot the $G$-concurrence for the same dimensions. We see how the curves are strictly increasing and this is true for all $D$ as proven in Lemma~\ref{eqent_lma4}.

\section{Conclusion\label{eqent_sct11}}
We have solved the problem posed in \cite{PhysRevA.73.012329} by providing two families of equientangled bases for two identical qudits for arbitrary dimension $D$. The construction of the first solution is based on quadratic Gauss sums and follows along the lines of \cite{PhysRevA.73.012329}, whereas the second family is constructed using a different method based on qudit graph states.

The first solution based on quadratic Gauss sums has an explicit analytic expression for the Schmidt coefficients that is easy to evaluate since they are just sums with $D$ terms (see \eqref{eqent_eqn13} and \eqref{eqent_eqn14}). However some Schmidt coefficients can be zero and the entropy of entanglement of the states in the basis varies non-monotonically with $t$ for $D>2$.

The second solution based on graph states consists entirely of full Schmidt rank states for any $0<t\leqslant 1$ that seem to have an entropy of entanglement that is strictly increasing as $t$ increases. Unfortunately we did not find a simple analytic expression for the Schmidt coefficients, but they can be computed numerically without much difficulty. We found a simple analytic expression for another pure state entanglement measure, the $G$-concurrence, which we proved is strictly increasing as $t$ increases. Finally we remark that one can extend this construction to equally entangled bases of more than two parties that interpolate continuously between a product basis and a graph basis even if the concept of maximally entangled states is not defined for more than two parties. This construction may be of interest in studying multipartite entanglement.

\renewcommand{\baselinestretch}{1} 
\large
\normalsize

\bibliographystyle{alpha}

\newcommand{\etalchar}[1]{$^{#1}$}

\printindex

\end{document}